\documentclass[12pt]{article}
\usepackage[margin=.88in]{geometry}
\usepackage[utf8]{inputenc}
\usepackage{amsfonts}
\usepackage{amsgen,amsmath,amstext,amsbsy,amsopn,amssymb,subcaption}
\usepackage[dvips]{graphicx}
\usepackage{comment}

\usepackage[shortlabels]{enumitem}
\usepackage{natbib}
\usepackage{booktabs}
\usepackage{amsthm}
\usepackage{hyperref}
\usepackage{blkarray}
\usepackage{algorithm}
\usepackage{algorithmic}
\usepackage{url}
\usepackage{longtable}
\usepackage{stmaryrd}
\usepackage{mwe}
\usepackage{mathtools}
\DeclareMathAlphabet\mathbfcal{OMS}{cmsy}{b}{n}
\SetSymbolFont{stmry}{bold}{U}{stmry}{m}{n}

\makeatletter
\newcommand*{\rom}[1]{\expandafter\@slowromancap\romannumeral #1@}
\makeatother

\usepackage[dvipsnames]{xcolor}

\newcommand{\argmin}{\mathop{\rm arg\min}}

\newcommand{\Vect}{\operatorname{Vec} }
\newcommand{\tF}{{\rm F}}
\newcommand{\SVD}{{\rm SVD}}
\newcommand{\Mat}{{\rm Mat}}

\newcommand{\tr}{{\rm tr}}
\newcommand{\abs}[1]{\left\lvert#1\right\rvert}
\newcommand{\norm}[1]{\left\lVert#1\right\rVert}

\newcommand{\PP}{\mathbb{P}}
\newcommand{\R}{\mathbb{R}}
\newcommand{\E}{\mathbb{E}}
\newcommand{\cA}{\mathcal{A}}
\newcommand{\cE}{\mathcal{E}}
\newcommand{\cH}{\mathcal{H}}
\newcommand{\cP}{\mathcal{P}}
\newcommand{\cD}{\mathcal{D}}
\newcommand{\cM}{\mathcal{M}}
\newcommand{\cI}{\mathcal{I}}
\newcommand{\cN}{\mathcal{N}}
\newcommand{\cS}{\mathcal{S}}
\newcommand{\cX}{\mathcal{X}}
\newcommand{\bbI}{\mathbb{I}}

\def\wt{\widehat}
\def\sfw{\mathsf{w}}

\def\calA{{\mathcal A}}

\def\calD{{\mathcal D}}

\def\calH{{\mathcal H}}

\def\calM{{\mathcal M}}

\def\calP{{\mathcal P}}

\def\calW{{\mathcal W}}

\def\PP{{\mathbb P}}
\def\QQ{{\mathbb Q}}
\def\RR{{\mathbb R}}

\def\TT{{\mathbb T}}

\newcommand{\gap}{{\mathsf{gap} }}
\newcommand{\init}{{\mathsf{init}}}

\newcommand{\sfh}{{\mathsf{h}}}
\newcommand{\sfR}{{\mathsf{R}}}

\newcommand{\dmax}{\bar{d}}
\newcommand{\higher}{ {\mathsf{higher}} }

\newtheorem{Theorem}{Theorem}

\newtheorem{Lemma}{Lemma}

\theoremstyle{plain}

\newtheorem{Proposition}{Proposition}

\title{Multiple Testing of Linear Forms for Noisy Matrix Completion}

\author{Wanteng Ma$^1$, Lilun Du$^2$, Dong Xia$^3$ and Ming Yuan$^4$\\
~ \\
{\small $^{1}$ Department of Statistics and Data Science, University of Pennsylvania, USA}\\
{\small $^2$Department of Decision Analytics and Operations,  City University of Hong Kong, Hong Kong}\\
{\small $^{3}$Department of Mathematics, Hong Kong University of Science and Technology, Hong Kong}\\
{\small $^4$Department of Statistics,  Columbia University, USA}\\
}
\date{\today\\
} 

\begin{document}

\maketitle

\footnotetext[3]{Dong Xia’s research was partially supported by Hong Kong RGC Grant GRF 16301622 and GRF 16303224. Corresponding email: \href{madxia@ust.hk}{madxia@ust.hk}} 
\footnotetext[4]{Ming Yuan's research was supported in part by NSF Grants DMS-2015285 and DMS-2052955.}

\begin{abstract}
Many important tasks of large-scale recommender systems can be naturally cast as testing multiple linear forms in noisy matrix completion model. These problems, however, present unique challenges because of the subtle bias-and-variance tradeoff and the intricate dependence among the estimated entries induced by the low-rank structure. In this paper, we develop a general approach to overcome these difficulties by introducing new statistics for individual tests with sharp asymptotics both marginally and jointly, and utilizing them to control the false discovery rate (FDR) via a data splitting and symmetric aggregation scheme. We show that valid FDR control can be achieved asymptotically, with guaranteed power, under nearly optimal sample-size requirements for noisy matrix completion using the proposed methodology. Extensive numerical simulations and real data examples are also presented to further illustrate its practical merits.
\end{abstract}

\begin{sloppypar}

\section{Introduction}\label{sec:intro}
\subsection{Matrix completion and inference}
Popularized by the Netflix prize \citep{bennett2007netflix}, matrix completion techniques have emerged as an essential tool for large-scale collaborative-filtering-based recommender systems \citep{schafer2007collaborative}. Consider, more specifically, representing the ratings of $d_1$ users on $d_2$ products/items by a $d_1\times d_2$ matrix. For all practical purposes, both $d_1$ and $d_2$ can be very large yet only a rather small number of the entries can be observed. The idea is that if the interaction between users and products can be approximately captured by a handful of latent user-specific and product-specific characteristics, then it is possible to infer the whole user-item rating matrix from these sparsely observed entries, and hence recommend products to users who may be genuinely interested in them. 

Most existing works study recommender systems from an estimation perspective and investigate how well the user-item matrix can be estimated or reconstructed collectively. These are clearly relevant metrics for evaluating recommender systems. For example, the Netflix prize uses root mean squared error as the gold standard for the competition. Yet they do not account for the fact that only a subset of the products can be recommended to a user and as such estimation accuracy may not be directly translated into the quality of these recommendations. Instead, various classical notions for binary classification such as precision and recall are often adopted in practice to evaluate the quality of top recommendations. See, e.g., \cite{herlocker2004evaluating}. This subtlety has significant statistical implications. First of all, making quality recommendations requires a more careful uncertainty quantification. Consider recommending between a blockbuster movie and an independent film to a user. Even if both estimated ratings are similar and favorable, the uncertainty associated with the estimated rating for the former is likely to be much smaller as it has been viewed by a much greater number of people. It could therefore be more prudent to recommend it over the latter. On the other hand, as each recommendation incurs uncertainty, when making a list of recommendations, it is more helpful to assess their quality collectively rather than individually. For example, the percentage of relevant recommendations among all recommended products could be a more meaningful measure than the chance of a specific recommendation being relevant. Both aspects draw immediate comparison with multiple testing problems, for example, in high-throughput gene expression studies where, among thousands of genes, a small subset that is likely to behave differently between control and treatment groups is sought. See, e.g., \cite{storey2003statistical, efron2007correlation, efron2012large}. Our work is inspired by this analogy and examines the problem of item recommendations from a multiple testing perspective.

For the sake of generality, we shall adopt the framework of trace regression where each observation is a random pair $(X,Y)$ with $X\in\R^{d_1\times d_2}$ and $Y\in\R$. The random matrix $X$ is sampled uniformly from the orthonormal basis $\mathfrak{E}=\{e_{i}e_{j}^{\top}: 1\le i\le d_1, 1\le j\le d_2\}$ where $\{e_{i}\}$ are the canonical basis vectors of an Euclidean space of conformable dimensions. The response variable $Y$ is related to $X$ via
\begin{equation}\label{eq:NMC}
	Y=\langle  M, X\rangle +\xi
\end{equation}
where $\langle M, X\rangle=\tr(M^{\top}X)$, and the independent measurement error $\xi$ is assumed to be a zero-mean sub-Gaussian random variable. Our goal is to infer the true user-product preference matrix $M$ from i.i.d. copies of $(X,Y)$ when $M$ is of (approximately) low rank and the observations are incomplete. Specifically, the task of deciding if product $j$ should be recommended to user $i$ can be cast as testing the null hypothesis, denoted by $H_{0,ij}$, about the $(i,j)$ entry of the true user-product matrix $M$, e.g., product $j$ is irrelevant to user $i$, against the alternative, denoted by $H_{a,ij}$, that user $i$ is interested in product $j$. Likewise, item recommendations in general amount to testing collectively all null hypotheses $H_{0,ij}$, $1\le i\le d_1$ and $1\le j\le d_2$. More broadly, one may consider testing about multiple linear forms, $\langle M,T \rangle$ for a family of $T\in \cH\subset \R^{d_1\times d_2}$. For example, one may consider $T$ of the form $e_ie_{j_1}^\top-e_ie_{j_2}^\top$ to determine between two products ($j_1$ and $j_2$) which one to recommend to a user ($i$). This multiple testing framework allows us to address, among others, two most pertinent questions for recommender systems: which items should we recommend so that we can ensure a certain percentage of recommendations are relevant, or click-through rate; given a list of recommendations, what percentage of recommendations are relevant. Both questions can be naturally rephrased in terms of the false discovery rate (FDR), commonly used in the context of multiple testing.

\subsection{Overview of our approach}
Suppose that we obtain a small number of noisy observations from the entries of the matrix $M$, represented by $\calD=\{(X_i, Y_i)\}_{i=1}^{n}$ satisfying (\ref{eq:NMC}). Typically, the sample size $n$ is much smaller than the total number of entries $d_1d_2$. Let $\calH$ denote a family of linear forms corresponding to the unobserved entries of $M$ or their linear combinations of interest. We aim to test the collection of null hypotheses $H_{0T}:\langle M, T\rangle=\theta_T, \forall T\in\calH$ with the objective of controlling the false discovery proportion (FDP) at or below $\alpha$. Our approach employs sample splitting and symmetric data aggregation (SDA). Specifically, we split the data into two sub-samples, $\calD_1$ and $\calD_2$, each of equal sizes. For each subsample, we apply a debiasing procedure to an initial low-rank estimator to obtain two independent copies of nearly unbiased estimators, denoted by $\widehat M^{(1)}$ and $\widehat M^{(2)}$, respectively. For each $T\in\calH$, we construct the symmetrized statistic:
$$
W_T^{\rm rank}=W_T^{(1)} W_T^{(2)}\quad {\rm with}\quad W_T^{(1)}=\frac{\langle \widehat M^{(1)}, T\rangle-\theta_T}{\widehat \sigma_{\xi}^{(1)}\widehat s_T^{(1)}\cdot\sqrt{2d_1d_2/n}}\quad {\rm and}\quad W_T^{(2)}=\frac{\langle \widehat M^{(2)}, T\rangle-\theta_T}{\widehat \sigma_{\xi}^{(2)}\widehat s_T^{(2)}\cdot\sqrt{2d_1d_2/n}}.
$$
Here, $\widehat{\sigma}^{(i)}, \widehat s_T^{(i)}$ are estimates of the standard deviation, which will be formally defined later.
The individual statistics $W_T^{(1)}$ and $W_T^{(2)}$ are constructed using the sub-samples $\calD_1$ and $\calD_2$, respectively, which are mutually independent. Under nearly optimal sample size and signal-to-noise ratio conditions, both statistics are asymptotically normal. Moreover, these statistics feature a novel variance design that achieves faster convergence than existing methods \citep{chen2019inference, xia2021statistical}.

Under null hypothesis, the test statistical $W_T^{\rm rank}$ is approximately symmetric about zero. Therefore, we choose the smallest threshold $L>0$ such that $|\{T\in\calH: W_T^{\rm rank}< -L\}|\leq \alpha \cdot \big(|\{T\in\calH: W_T^{\rm rank} >L\}|\vee 1\big)$, and we reject those null hypothesis $H_{0T}$ for which $W_T^{\rm rank}>L$. The rationale is to estimate the number of falsely rejected null hypothesis, $|\{T\in\calH_0: W_T^{\rm rank}>L\}|$, by $|\{T\in\calH: W_T^{\rm rank}< -L\}|$. Here, $\calH_0\subset \calH$ denotes the set of linear forms for which the null hypotheses are true. We prove that this approach effectively controls the FDP asymptotically as long as the number of strongly correlated pairs of linear forms is small relative to the number of strong signals.

Furthermore, we propose a whitening method for the test statistics when many linear forms are strongly correlated, e.g., when testing an entire row, multiple linear forms on the same entries, or a complete small block. Specifically, our method decorrelates the test statistics by leveraging an explicit characterization of their dependencies. Consequently, this method achieves more effective FDR control in the presence of numerous strongly correlated pairs.

\subsection{Related work}

The pioneering works of \cite{candes2009power,candes2010matrix,candes2012exact} laid the theoretical foundations of matrix completion via convex optimization. Since then, substantial progress has been made toward improving the accuracy and scalability of these techniques and toward better understanding the statistical and computational underpinnings of the problem. See, e.g., \cite{cai2010singular,keshavan2010matrix,recht2010guaranteed,gross2011recovering,koltchinskii2011nuclear,liu2011universal,negahban2011estimation,rohde2011estimation,tsybakov2011nuclear,negahban2012restricted,sun2012calibrated,klopp2014noisy,cai2015rop,cai2016matrix,gao2016optimal}, among numerous others. Matrix completion techniques have been widely applied to collaborative-filtering-based recommender systems. See, e.g., \cite{koren2009matrix,davidson2010youtube,mcauley2013hidden,das2017survey} and, more recently, in \cite{chen2022review,han2024keyword}. Moreover, the observations are discrete/quantized in many applications, which has motivated generalized or 1-bit matrix completion formulations; see, e.g., \citep{davenport20141}.

Since its introduction in the seminal paper by \cite{benjamini1995controlling}, FDR has proven to be an extremely useful notion in a wide variety of areas including bioinformatics \citep{jung2005sample,roeder2009genome,brzyski2017controlling}, neuroimaging \citep{perone2004false,chumbley2010topological}, and finance \citep{barras2010false,bajgrowicz2012technical}, to name a few. Numerous methodologies have also been developed to control FDR in multiple testing. Notable examples includes \cite{benjamini2001control,sarkar2002some,wu2008false,clarke2009robustness,barber2015controlling,candes2018panning,barber2019knockoff}, among many others. There are, however, considerable new challenges when considering multiple testing in the context of item recommendations or matrix completion, both in defining test statistics for individual hypothesis and in how to utilize them effectively to improve the overall performance.

In most if not all of the existing literature of multiple testing, the individual test statistics are either given or naturally defined. For matrix completion, however, finding the right test statistics is arguably one of the most difficult steps for statistical inferences. Common estimators for entries of the underlying matrix do not admit an explicit expression, which creates technical obstacles to characterize their bias and variance. This challenge is already in full display when testing a single hypothesis which occurs, for example, when deciding on whether to recommend a specific product to a particular user. See, e.g., \cite{chen2019inference, xia2021statistical,farias2022uncertainty,chen2023statistical,gui2023conformalized, shao2023distribution}. The problem is exacerbated when dealing with multiple hypotheses where more refined bounds for the convergence of test statistics are needed both for controlling the FDR and to ensure power without unnecessary sample size and signal-to-noise ratio restriction. We shall introduce a new test statistic especially suitable for such purposes. It builds upon recent developments \citep[e.g.,][]{chen2019inference,xia2021statistical} for inferring a single entry and is based upon a more precise characterization of variance than earlier works. In particular, it can be shown that, with the improved variance estimate, the new statistic converges to normal distribution at a faster rate, both marginally and jointly, and is thus more suitable for use in multiple testing.

Most procedures for FDR control were developed, at least initially, assuming that the individual test statistics are independent of each other. How to handle complicated dependency structure, as is the case for matrix completion, remains a critical issue and an actively researched subject in multiple testing. See, e.g., \cite{efron2007correlation, leek2008general, fan2017estimation, li2017rate, du2021false, fithian2022conditional}. A common strategy to deal with dependence is data splitting. See, e.g., \cite{roeder2009genome, song2015split, barber2019knockoff, zou2020new, du2021false, dai2022false, dai2023scale}, for a number of recent examples and applications of data splitting schemes. In particular, \cite{du2021false} showed that the FDR can be properly controlled as long as the individual test statistics have nearly symmetric null distribution and the dependence among them is sufficiently weak. To make use of this insight, we derive the asymptotic correlation of our proposed individual test statistics. Interestingly, for many item recommendation tasks, these statistics are only weakly correlated and hence, the FDR can be controlled accordingly. In other settings where the test statistics can be strongly correlated, our explicit characterization of their dependence structure also suggests ways to ``whitening'' and ``screening'' so that FDR can still be controlled under minimal sample size requirement.

\subsection{Outline of the paper}
The rest of the paper is organized as follows. In the next section, we shall introduce our test statistics for a single linear form and study its asymptotic properties. Section \ref{sec:FDR-SDA} discusses how these individual test statistics can be aggregated to test multiple linear forms. Section \ref{sec:strong-corr} introduces a whitening and screening scheme to address situations where the test statistics could be strongly correlated. Numerical experiments, both simulated and real-world data examples, are presented in Section \ref{sec:experiments}. We conclude with a few remarks in Section \ref{sec:remark}. Due to space limitation, all proofs, as well as further examples and discussions, are relegated to the Supplement.

Throughout the paper, let $\|\cdot\|$ denote the spectral norm of a matrix and the $\ell_2$-norm of a vector, and  denote $\|M\|_{2,\max}:=\max_{i\in [d_1]}\|e_i^{\top}M\|$. Define $\|M\|_{\max}=\max_{i,j}|M_{ij}|$ and $\|M\|_{\infty}:=\max_{i\in[d_1]} \|e_i^{\top} M\|_{\ell_1}$ for a matrix $M$. Note that $\|\cdot\|_{\max}$ and $\|\cdot\|_{\infty}$ are equivalent for a vector. Also, we use $\|\cdot\|_{\rm F}$ to denote the Frobenius norm. Moreover, we use $\odot$ to denote matrix Hadamard (entry-wise) product. We also let $\widetilde{O}(\cdot)$ denote big-O notation when suppressing additional polylogarithmic factors.

\section{An Approach for Individual Test}\label{sec:clt}

We begin with testing a single hypothesis:
\begin{equation}\label{eq:general-rtt}
	H_{0T}: \langle M, T\rangle =\theta_{T}\qquad {\rm vs}\qquad  H_{aT}: \langle M, T\rangle \neq \theta_{T}
\end{equation}
for some fixed $T\in\R^{d_1\times d_2}$ and pre-specified value $\theta_T\in \R$, based on $n$ independent observations $\cD:=\{(X_i,Y_i)\}_{i=1}^n$ following the trace regression model \eqref{eq:NMC}. Recall that $\xi$ in \eqref{eq:NMC} is sub-Gaussian noise with mean $0$ and variance $\sigma_\xi$ such that $\E \exp{(\lambda \xi)}\le \exp{(c^2 \sigma_\xi^2\lambda^2 /2 )}$ for some constant $c>0$. Following the convention, we shall assume that the singular vectors of $M$ are incoherent:
\begin{equation}\label{eq:incoherence}
	\max \left\{\sqrt{\frac{d_1}{r}}\|U\|_{2, \max }, \sqrt{\frac{d_2}{r}}\|V\|_{2, \max }\right\} \leq \sqrt{\mu},
\end{equation}
where $r$ is the rank of $M$, and $M=U\Lambda V^\top$ its singular value decomposition. This ensures that the entries of $M$ are delocalized so that it can be recovered even if some entries are not observed. In what follows, we shall denote by $\lambda_{\max}$ and $\lambda_{\min}$ the largest and smallest nonzero singular values of $M$, and $\kappa_0$ the ratio $\kappa_0:=\lambda_{\max}/\lambda_{\min}$, i.e., its condition number.

{

For brevity, we focus on two-sided tests here, though our discussion extends straightforwardly to one-sided tests. Without loss of generality, we assume that $d_1 \ge d_2$ throughout. The goal of this section is to develop a test statistic for \eqref{eq:general-rtt} that can be effectively applied to testing a large number of hypotheses. The problem of testing a single linear form \eqref{eq:general-rtt} has been previously studied by \cite{xia2021statistical}; see also \cite{chen2019inference, farias2022uncertainty, chen2023statistical}, among others. However, the tests proposed in these works are neither sufficiently sharp nor directly applicable to multiple testing. For example, \citet{xia2021statistical}'s debiasing approach relies on independent initialization through data splitting, which may lead to potential power loss. Although \citet{chen2019inference,farias2022uncertainty} avoids data splitting by employing a leave-one-out analysis, their methods are restricted to the entrywise case where $T = e_i e_j^\top$, and they use Bernoulli sampling with a constrained sample size $n\le d_1 d_2$. Moreover, in both \citet{xia2021statistical}, and \citet{chen2019inference,farias2022uncertainty}, the rate of asymptotic normal convergence is no faster than $\sqrt{\log(d_1)/d_2}$. This convergence rate is too slow for our purposes, as it unnecessarily restricts the number of hypotheses that can be tested, regardless of how large the sample size $n$ is.

We now propose a novel approach to address these issues mentioned above. Our approach for estimating $\langle M, T\rangle$ consists of three steps: gradient-descent initialization, bias-correction, and an improved low-rank projection. For the initialization, we apply the gradient descent \citep{chen2019inference} to obtain an entry-wise consistent rank-$r$ estimator $\widehat{M}^{\mathsf{init}}$ such that for any $\tau\ge 1$,
\begin{equation}\label{eq:init-est}
	\norm{\widehat{M}^{\mathsf{init}} - M}_{\max} \le C_{\init}  \sigma_\xi \sqrt{\frac{\tau  d_1 \log d_1}{n }} ,
\end{equation}
with probability at least $1-d_1^{-\tau}$, for some parameter $C_{\init}>0$ that is only $\operatorname{Poly}(\tau,\kappa_0,\mu,r,\log d_1)$. Denote the  left and right $r$ singular vectors of $\widehat{M}^{\mathsf{init}} $ as $\widehat{U}^{\mathsf{init}}$, $\widehat{V}^{\mathsf{init}} $. A key observation is that $\widehat{U}^{\mathsf{init}}$, $\widehat{V}^{\mathsf{init}} $ are also incoherent.  For brevity, in what follows, we shall assume $\tau$ is a large enough constant to ensure that $n\le O(d_1^{2\tau})$. To correct the bias of initial estimates, we then define
\begin{equation}\label{eq:debias}
    \widehat{M}^{\mathsf {unbs }}=\widehat{M}^{\mathsf{init}}+\frac{d_1 d_2}{n} \sum_{i=1}^n\left(Y_i-\left\langle\widehat{M}^{\mathsf{init}}, X_i\right\rangle\right) X_i. 
\end{equation}
Unfortunately, as this debiasing may lead to a significant increase in variance, we shall trade off between bias and variance by low-rank projection. However, different from the popular simple SVD used in the literature \citep{chen2019inference,xia2021statistical}, we use an incoherence-assisted low-rank projection. More exactly, we compute

\begin{equation}\label{eq:left-SVD}
    \begin{aligned}
        \widehat{U} &= \SVD_{r}\big(\widehat{M}^{\mathsf {unbs }}\widehat{V}^{\mathsf{init}} \big)\quad {\rm and} \quad \widehat{V} &= \SVD_{r}\big((\widehat{M}^{\mathsf {unbs }})^\top \widehat{U}^{\mathsf{init}} \big),
    \end{aligned}
\end{equation}
where $ \SVD_{r}(\cdot)$ returns the top-$r$ left singular vectors. This yields a low-rank estimate:
\begin{equation}\label{eq:lr-retraction}
    \widehat{M}= \widehat{U}  \widehat{U} ^\top\widehat{M}^{\mathsf {unbs }} \widehat{V} \widehat{V}^\top.
\end{equation}

Finally we shall estimate $\langle M, T\rangle$ by $\langle \widehat{M}, T\rangle$. Under certain regularity conditions, we can show that
\begin{equation}\label{eq:clt-tangent}
\frac{\langle \widehat{M}, T\rangle-\langle M, T\rangle}{\sigma_\xi\norm{\cP_M(T)}_{\tF}\sqrt{d_1d_2/n}}\to_d N(0,1),
\end{equation}
where $\to_d$  represents convergence in distribution when $n\to \infty$, and
\begin{equation}\label{eq:tangent-proj}
    \cP_M(A)=UU^\top AVV^\top+UU^\top AV_\perp V_\perp^\top+U_\perp U_\perp^\top AVV^\top
\end{equation}
represents the projection onto the tangent space of low-rank matrix manifold $\calM_r$ at point $M$, and $U_\perp$ and $V_\perp$ are orthonormal matrices whose columns span the orthogonal complements of the left and right singular spaces of $M$, respectively.

We define the alignment ratio $\beta_T$ (which may goes to $0$ asymptotically) as
\begin{equation}\label{eq:alignment}
	\beta_T := \frac{\norm{\cP_M(T)}_\tF}{\|T\|_{\rm F}} \sqrt{\frac{d_2}{r} },
\end{equation} 
where $\norm{\cP_M(T)}_\tF$ is proportional to the (asymptotic) standard deviation of the test statistic with respect to a linear form $T$.  Here we assume $\norm{\cP_M(T)}_\tF>0$ throughout the discussion. The alignment ratio essentially quantifies how close $T$ is to the tangent space $\cP_M$ so that $\cP_{M}(T)$ accounts for a substantial, albeit small, fraction of $T$. When $\norm{\cP_M(T)}_\tF=0$, the linear form $\langle M, T\rangle=0$ and estimates with faster rate of convergence can be obtained (e.g., using a trivial estimate). This condition avoids such pathological situations. Similar assumptions are also made in earlier works. See, e.g., \cite{chen2019inference,xia2021statistical}.   

We can use this result to test \eqref{eq:general-rtt} for a fixed $T$. The theoretical guarantee of our approach is presented in Theorem \ref{thm:asymp-normal}:

\begin{Theorem}\label{thm:asymp-normal}
Given $\widehat{M}$ from \eqref{eq:lr-retraction} with dimension ratio $\alpha_d=d_1/d_2$, suppose that the sample size and SNR condition satisfies:
\begin{equation}\label{eq:condition-SNR-sample}
    n\ge  C_1 C_{\init}^2\kappa_0^2 \frac{\norm{T}_{\ell_1}^2}{\beta_T^2\norm{T}_{\tF}^2} \mu^{5}r^5 d_1\log^2 d_1, \ \frac{\lambda_{\min}}{\sigma_{\xi}} \ge  C_2 C_{\init}^2 \frac{\kappa_0\norm{T}_{\ell_1}}{\beta_T\norm{T}_\tF } \sqrt{\frac{ \alpha_{d} \mu^{6} r^5 d_1^2 d_2  \log^2 d_1 }{ n } }
\end{equation}
for some constants $C_1,C_2>0$. Then there exists a constant $C_3>0$ such that for any $T
\in \R^{d_1\times d_2}$ with alignment ratio $\beta_T>0$ defined in \eqref{eq:alignment} and
\begin{equation*}
    \begin{aligned}
        &\sup _{t \in \mathbb{R}}\left|\mathbb{P}\left(\frac{\langle \widehat{M}, T\rangle-\langle M, T\rangle}{\sigma_{\xi}\|\cP_M(T) \|_\tF \cdot \sqrt{d_1 d_2 / n}} \leq t\right)-\Phi(t)\right| \\
&\le C_3\left(  C_{\init}   \frac{\kappa_0\norm{T}_{\ell_1}}{\beta_T \norm{T}_{\tF}} \sqrt{\frac{ \mu^{5} r^5 d_1 \log^2 d_1  }{ n}} +  \frac{ C_{\init}^2\sigma_{\xi}  \kappa_0 \norm{T}_{\ell_1} }{\beta_T \lambda_{\min}  \norm{T}_\tF }\sqrt{\frac{\alpha_d \mu^{6} r^{5 } d_1^2 d_2  \log^2 d_1 }{ n } } \right). 
    \end{aligned}
\end{equation*}
\end{Theorem}

We highlight the advantage of our methods compared to existing literature. First, we apply the gradient descent directly to the observations for initialization, which avoids data splitting and potential power loss \citep{xia2021statistical}.  Second, we use a new low-rank projection method akin to the subspace iteration \citep{bathe1973solution,bathe2013subspace}, and tensor power iteration \citep{richard2014statistical}. The benefit of this new projection method can be explained from a theoretical point of view: recall that $ \widehat{M}^{\mathsf {unbs }}-M=\frac{d_1 d_2}{n} \sum_{i=1 }^n \xi_i X_i- \frac{d_1 d_2}{n} \sum_{i=1 }^n\left\langle\widehat{M}^{\mathsf{init}} - M, X_i\right\rangle X_i+\widehat{M}^{\mathsf{init}} -M$. When multiplied by  incoherent $\widehat{U}^{\mathsf{init}}$ or $\widehat{V}^{\mathsf{init}} $ on one side, the size of each summand (i.e., $\widehat{U}^{\mathsf{init}\top}X_i$, or $ X_i \widehat{V}^{\mathsf{init}}$)  can be effectively reduced, leading to sharper concentrations. The leave-one-out argument enables further bounding the negligible bias terms in $\langle \widehat{M}-M, T\rangle$, making our theory more general than those in \cite{chen2019inference}.


Moreover, the variance that we characterize is sharper than previous methods \citep{chen2019inference,xia2021statistical,farias2022uncertainty} in the sense that (i) as is shown in \cite{ma2024statistical}, it reaches the Cramér-Rao lower bound for the regression-based noisy matrix completion, indicating its efficiency; (ii) it returns a confidence interval with minimax optimal length, as given in our following Theorem \ref{thm:matrix-ci-minimax}; and (iii) it allows the convergence rate to be totally controlled by the sample size $n$ and SNR ($\lambda_{\min}/\sigma_{\xi}$), whenever $d_1$, $d_2$ is large or not. This convergence rate is especially helpful for multiple testing since the number of hypotheses we are interested may grow with respect to $d_1$, $d_2$. The proof of Theorem~\ref{thm:matrix-ci-minimax} can be found in Section~\ref{Minimax} of the Supplement. 

\begin{Theorem}[Minimax optimal length of confidence interval]\label{thm:matrix-ci-minimax} 
Fix low-rank $M$ and $T$ as discussed above.
Define the parameter space as 
$$
\begin{aligned}
    \boldsymbol{\Theta}=\Big\{ M'\in\R^{d_1\times d_2}: &\operatorname{rank}(M')\le r, (d_1\norm{U'}_{2,\max}^2)\vee  (d_2\norm{V'}_{2,\max}^2)\le {\mu r},
    \\
    &\lambda_{\min}(M')\ge \lambda_{\min}, \kappa(M')\le \kappa_0,\norm{\cP_{M' }(T)}_{\tF}\ge \norm{\cP_{M }(T)}_{\tF} \Big\}.
\end{aligned}
$$
Here $\cP_{M'}(\cdot)$ means the projection onto the tangent space at $M'$. Consider the set of any valid $1-\alpha$ confidence interval  with $\alpha<\frac{1}{4}$ as:
$$
\mathcal{I}_\alpha(\boldsymbol{\Theta}, T):=\left\{\operatorname{CI}_{T}^\alpha\left( \{(X_i,Y_i)\}_{i=1}^n \right)=[l, u]: \inf _{M' \in \boldsymbol{\Theta}} \mathbb{P}_{M',T}(l \leq\langle M', T\rangle \leq u ) \geq 1-\alpha\right\},
$$
where $l, u$ are any functions of observations $\{(X_i,Y_i)\}_{i=1}^n$, and $\mathbb{P}_{M',T}$  is taken with respect to the randomness of observations for  fixed underlying $M'\in\boldsymbol{\Theta}$ and $T$. Then, when the SNR satisfies
\begin{equation*}
         \frac{\lambda_{\min}}{\sigma_{\xi}} \ge C_{\gap} \kappa_0\frac{\norm{T}_{\ell_1}}{ \beta_T\norm{T }_{\tF} }\sqrt{\frac{ \mu^{6} r^{3 }d_1^2d_2 }{ n } }
\end{equation*}
for a numeric constant $C_{\gap}$, the length of the confidence interval has the minimax lower bound:
\begin{equation*}
    \inf_{\operatorname{CI}_{T}^\alpha\left(M', \{(X_i,Y_i)\}_{i=1}^n \right) \in \mathcal{I}_\alpha(\boldsymbol{\Theta}, T) } \sup _{M'\in \boldsymbol{\Theta} } \mathbb{E} L\left(\operatorname{CI}_{T}^\alpha\left(M', \{(X_i,Y_i)\}_{i=1}^n \right)\right) \ge c\sigma_{\xi}  \norm{\cP_{M }(T)}_{\tF}\sqrt{\frac{d_1 d_2}{n}}.
\end{equation*}
\end{Theorem}

To pursue a fully data-driven approach for hypothesis testing, we provide the following estimates for the variance term $\sigma_\xi^2\norm{\cP_M(T)}_{\tF}^2$:
\begin{equation*}
    \widehat{\sigma}^2_{\xi}=\frac{1}{n}\sum_{i=1}^n\left(Y_1-\left\langle \widehat{M}^{\mathsf{init}},X_i\right\rangle\right)^2, \widehat{s}_T^2= \norm{\cP_{\widehat{M}^{\mathsf{init}}} (T)}_{\tF}^2,
\end{equation*}
where $\cP_{\widehat{M}^{\mathsf{init}}}(\cdot)$ follows \eqref{eq:tangent-proj} by replacing $U$, $V$ with $\widehat{U}^{\mathsf{init}}$, $\widehat{V}^{\mathsf{init}}$.
Now, we define our test statistic formally as  
\begin{equation}\label{eq:WT-def}
    W_T\left(\{(X_i,Y_i)\}_{i=1}^n\right) = \frac{\langle \widehat{M}, T\rangle-\theta_T}{\widehat{\sigma}_{\xi} \widehat{s}_T \cdot \sqrt{d_1 d_2 / n}}. 
\end{equation}
The following result shows that the asymptotic normality continues to hold using these variance estimates.

\begin{Theorem}\label{thm:asymp-normal-varest} 
	Under the condition that \eqref{eq:condition-SNR-sample} holds, if $H_{0T}$ is true, then
$$
\begin{aligned}
    &\sup _{t\in \mathbb{R}}\left|\mathbb{P}\left(W_{T} \leq t\right)-\Phi\left(t\right)\right|
    \\
    &
\le C_3\left(  C_{\init}   \frac{\kappa_0\norm{T}_{\ell_1}}{\beta_T \norm{T}_{\tF}} \sqrt{\frac{ \mu^{5} r^5 d_1 \log^2 d_1  }{ n}} +  \frac{ C_{\init}^2\sigma_{\xi}  \kappa_0 \norm{T}_{\ell_1} }{\beta_T \lambda_{\min}  \norm{T}_\tF }\sqrt{\frac{\alpha_d \mu^{6} r^{5 } d_1^2 d_2  \log^2 d_1 }{ n } } \right).
\end{aligned} 
$$
\end{Theorem}
}


\section{Multiple Tests}\label{sec:FDR-SDA}

We now turn our attention to testing a family of hypotheses $\{H_{0T}: T\in \cH\}$ for a subset $\cH\subset \R^{d_1\times d_2}$. In particular, we can take $\cH=\{e_i e_j^\top: 1\le i\le d_1, 1\le j\le d_2\}$ for testing preferences of all user-item pairs. Denote the number of tests $\abs{\cH}=q$. Without loss of generality, assume that the linear forms are linearly independent so that the $q$ is no larger than $d_1 d_2$.  Denote the null set by $\cH_0$, i.e., $\cH_0=\left\{T\in \cH: \langle M,T\rangle=\theta_T \right\}$ and the non–null set $\cH_1=\cH\setminus \cH_0$, with cardinality $q_0$ and $q_1$ respectively. To establish the test statistics for all $\{H_{0T}: T\in \cH\}$, we denote the smallest alignment parameter as $\beta_0=\min_{T\in\cH_0}\beta_T$, and write
\begin{equation}\label{eq:fixed-asymp}
h_n:=  C_{\init} \max_{T\in\cH_0} \left\{ \frac{\kappa_0\norm{T}_{\ell_1}}{\beta_0 \norm{T}_{\tF}} \sqrt{\frac{ \mu^{5} r^5 d_1 \log^2 d_1  }{ n}} +  \frac{ C_{\init}^2\sigma_{\xi}  \kappa_0 \norm{T}_{\ell_1} }{\beta_0 \lambda_{\min}  \norm{T}_\tF }\sqrt{\frac{\alpha_d \mu^{6} r^{5 } d_1^2 d_2  \log^2 d_1 }{ n } } \right\},
\end{equation}
where, for brevity, we omit the dependence of $h_n$ on $d_1$. In light of Theorem \ref{thm:asymp-normal-varest}, with appropriate initial estimates, we have
$$
\left|\mathbb{P}\left(W_{T} \leq t\right)-\Phi\left(t\right)\right|\lesssim h_n
$$
for all $T\in \cH_0$. In what follows, we assume $h_n\to 0$ to ensure proper inference.

\subsection{Symmetric Data Aggregation}

With the asymptotic normality of $W_T$, it looks possible to directly apply \cite{benjamini1995controlling} style of methods to control the FDR in an asymptotic sense. However, doing so may put an unreasonable limit on the number ($q$) of tests under consideration. This is due to the fact that the test statistic $W_T$ has much heavier tail than that in classic multivariate normal mean problems. As a result, while $W_T$ converges to $\cN(0,1)$ in distribution for any linear form $T$ (as long as signal strength is large enough), it does not necessarily converge in fourth-order or higher-order moments. Indeed, it can be shown that the $2k$-th order moment ($k\ge 2$) of $W_T$ for a properly chosen linear form $T$ is lower bounded by
\begin{equation}\label{eq:W-kth-moment}
   \sqrt[2k]{ \E \abs{W_{T}}^{2k}} \gtrsim \left(\frac{d_1 d_2}{n}\right)^{1/4}.
\end{equation}
If $d_1\asymp d_2 \asymp d $, and $n \asymp d^{1+\epsilon }$ for some $\epsilon\in(0,1)$, then we have $\E \abs{W_{T}}^{2k} \gtrsim d^{(1-\epsilon)k/2}$. See Section  \ref{sec:proof-w-k-mom} in the supplement for proof of \eqref{eq:W-kth-moment}.

Thankfully, much more powerful approaches can be developed by exploiting other salient features of $W_T$ entailed by its asymptotic normality. In particular, we shall adopt a general strategy introduced by \cite{du2021false} by leveraging symmetricity and data aggregation. Assume that, without loss of generality, $n$ is even with $n=2n_0$. We split $\cD=\{(X_i,Y_i)\}_{i=1}^n$ into two sub-samples:
$$
\cD_1=\big\{(X_i, Y_i)\big\}_{i=1}^{n_0}\quad {\rm and}\quad \cD_2=\big\{(X_i, Y_i)\big\}_{i=n_0+1}^{n}.
$$
We describe our approach as following: we first construct two groups of independent asymptotically symmetric statistics $\{W_{T}^{(1)}: T\in\cH\}$ from $\cD_1$ and $\{W_{T}^{(2)}: T\in \cH\}$ from $\cD_2$ by data splitting.  After that, we aggregate them by multiplication: $W_T^{\mathsf{rank} }= W_T^{(1)} \cdot W_T^{(2)} $. Finally, we rank each $W_T^{\mathsf{rank} }$ and choose a data-driven threshold by taking advantage of symmetricity:
\begin{equation}\label{eq:dd-threshold}
    L:=\inf \left\{t>0: \frac{\#\left\{T: W_T^{\mathsf{rank} }<-t\right\}}{\#\left\{T: W_T^{\mathsf{rank} }>t\right\} \vee 1} \leq \alpha\right\}\cup\{+\infty\},
\end{equation}
given any FDR level $\alpha\in (0,1)$, and reject $H_{0T}$ if $W_T^{\mathsf{rank}}>L$. Details are given in Algorithm \ref{alg:matrix-fdr}. Hereafter, we denote $M_T:=\langle M, T\rangle$ for simplicity. 

\begin{algorithm}[ht!]
	\caption{Matrix FDP Control}
	\label{alg:matrix-fdr}
	\begin{algorithmic}[1]
		\REQUIRE Hypotheses $\left\{H_{0T}: M_T=\theta_T, T\in \cH\right\}$, data splits $\cD_1$, $\cD_2$, rank $r$, FDR level $\alpha$.
		\STATE{ Apply gradient descent \citep{chen2019inference} on $\cD_1$, $\cD_2$ to construct initial estimates $\widehat{M}_{\mathsf{init}}^{(1)}$, $\widehat{M}_{\mathsf{init}}^{(2)}$, respectively.}
		\STATE{ Apply \eqref{eq:lr-retraction}, \eqref{eq:WT-def} to construct two independent test statistics: $
    W_{T}^{(1)}:= W_T(\cD_1), \quad  W_{T}^{(2)}= W_T(\cD_2).
        $
}
        
\STATE Compute the final ranking statistics by $W_T^{\mathsf{rank}}=W_{T}^{(1)}W_{T}^{(2)}$, and then choose a data-driven threshold $L$ by \eqref{eq:dd-threshold}.
		\STATE{Reject $H_{0T}$ if $W_T^{\mathsf{rank}}>L$.}
	\end{algorithmic}
\end{algorithm}

 By the definition of $L$,  we have
$$
\text{FDP} = \frac{\sum_{T\in \cH  } \bbI(W_T^{\mathsf{rank}} <-L )}{\left( \sum_{T\in \cH  } \bbI(W_T^{\mathsf{rank}} >L ) \right) \vee 1 } \frac{\sum_{T\in \cH_0  } \bbI(W_T^{\mathsf{rank}} >L ) }{\sum_{T\in \cH  } \bbI(W_T^{\mathsf{rank}} <-L )  }  \le \alpha \frac{\sum_{T\in \cH_0  } \bbI(W_T^{\mathsf{rank}} >L ) }{ \sum_{T\in \cH_0  } \bbI(W_T^{\mathsf{rank}} <-L )}.
$$
The crux of our argument is that the ratio on the rightmost hand side is approximately $1$ by virtue of the symmetry of $W_T^{\mathsf{rank}}$ under $H_{0T}$. To do so, we need to investigate the dependence among multiple test statistics, which will be elaborated later. 

We remark that our data aggregation scheme resembles the ``mirror-statistic’’ approach in the literature \citep{xing2021controlling,dai2022false}. The sign of the newly assembled statistic is set as $\operatorname{sign}(W_T^{\mathsf{new}})=\operatorname{sign}(W_{T}^{(1)})\cdot \operatorname{sign}(W_{T}^{(2)})$, and the absolute values are then aggregated in suitable ways. In addition to the multiplicative aggregation $\abs{W_T^{\mathsf{new}}}=\big|W_{T}^{(1)}W_{T}^{(2)}\big|$, prior work \citep{xing2021controlling,dai2022false,dai2023scale} has also considered the minimum aggregation $\min\Big\{\big|W_{T}^{(1)}\big|,\big|W_{T}^{(2)}\big|\Big\}$ and the additive aggregation $\big|W_{T}^{(1)}\big|+\big|W_{T}^{(2)}\big|$. Moreover, \citep{dai2023scale,du2021false} suggests that, for testing multivariate normal means, multiplicative aggregation can yield higher power than the other two options. See Section~\ref{sec:compare} in the supplement for a more detailed discussion.

\subsection{Dependence among Test Statistics}\label{sec:dependence}
One of the main challenges for multiple testing is how to account for the dependence structure among test statistics. To this end, we shall first derive the asymptotic distribution for the joint distribution of two estimated linear forms. In particular, for two matrices $T_1, T_2\in \R^{d_1\times d_2}$, it can be shown that
\begin{equation}\label{eq:corr}
\texttt{corr}(\langle \widehat{M}, T_1\rangle, \langle \widehat{M}, T_2\rangle)\approx \frac{\left\langle \cP_M(T_1) ,\cP_M(T_2)  \right\rangle }{ \norm{\cP_M(T_1)}_{\tF}\norm{\cP_M(T_2)}_{\tF} }=:	\rho_{T_1,T_2},
\end{equation}
where $\cP_M(\cdot)$ is defined in \eqref{eq:tangent-proj}. More specifically, we have

\begin{Theorem}\label{thm:asymp-two-var} 
	Suppose \eqref{eq:condition-SNR-sample} holds for  $T_1,T_2\in \cH$, and $|\rho_{T_1,T_2}|<1$. Define  $\Phi_\rho(\cdot,\cdot)$ as the cumulative distribution function of bivariate normal distribution $N(0,((1, \rho)^\top, (\rho, 1)^\top))$.  If both $H_{0T_1}$ and $H_{0T_2}$ hold, then
	\begin{equation*}
		\begin{aligned}
			&\sup _{t_1,t_2 \in \mathbb{R}}\left|\mathbb{P}\left(W_{T_1} \leq t_1, W_{T_2} \leq t_2\right)-\Phi_{
				\rho_{T_1,T_2}}\left(t_1,t_2\right)\right| \le C_3(1-\rho_{T_1,T_2})^{-\frac{3}{2}} h_n.
		\end{aligned}
	\end{equation*}
\end{Theorem}

This result explicitly characterizes the dependence between two test statistics which is critical for the FDR control in multiple testing. In particular, we shall separate pairs of linear forms in the null hypotheses into strongly correlated:
\begin{equation}\label{eq:weak-corr-1}
	\begin{aligned}
		\cH_{0,\text{strong} }^2 := \left\{ (T_1,T_2)\in\cH_0\times \cH_0 : \rho_{T_1,T_2}\ge c q_0^{-\nu} \right\},
	\end{aligned}
\end{equation}
where $\nu>0$ can be any pre-specified non-vanishing number (e.g., $\nu\ge 0.1$) and $c>0$ is some universal constant, and weakly correlated $\cH_{0,\text{weak} }^2:=(\calH_0\times \calH_0)\setminus \cH_{0,\text{strong} }^2$. The proportion of  all linear form pairs that are strongly correlated is therefore
\begin{equation}\label{eq:beta-s}
    \beta_{\mathsf{s} } := \frac{ \abs{ \cH_{0,\text{strong} }^2 }   }{\abs{ \cH_{0}^2 }}.
\end{equation}
Here, $\beta_{\mathsf{s} }$ measures the degree of dependence of the test statistics. When $\beta_{\mathsf{s} }\to0$, most of the $\{W_T\}_{T\in\cH_0}$ are  weakly correlated, leading to weak dependency among all test statistics.

We demonstrate that the measure of dependency in \eqref{eq:beta-s} is particularly useful for matrix linear form inference as the incoherent structure naturally ensures small $\beta_{\mathsf{s} } $ for many practical instances. Under the incoherent assumption, we have
$$
\rho_{T_1,T_2}\le  \frac{ \mu^4 r \norm{T_1}_{\ell_1 }\norm{T_2}_{\ell_1 } }{\beta^2_0 \norm{T_1}_{\tF}\norm{T_2}_{\tF} }\frac{1}{d_2}+\frac{\abs{\left\langle T_1 T_2^\top , UU^\top \right\rangle} +\abs{\left\langle T_1^\top T_2 , VV^\top \right\rangle} }{\norm{\cP_M(T_1)}_\tF \norm{\cP_M(T_2)}_\tF}.
$$
Thus, two linear forms $(T_1, T_2)$ are weakly correlated if $T_1^\top T_2=\boldsymbol{0}$, $T_1 T_2^\top=\boldsymbol{0}$ and
\begin{equation}\label{eq:weak-corr-2}
		\frac{ \mu^2 \norm{T_1}_{\ell_1 }\norm{T_2}_{\ell_1 } }{\beta^2_0 \norm{T_1}_{\tF}\norm{T_2}_{\tF} } \le C.
\end{equation}
Condition \eqref{eq:weak-corr-2} holds when $T_1$, $T_2$ are sparse, i.e., the number, $s_0$, of nonzero entries in $T_1$ and $T_2$ is of the order $O(\beta_0^2)$. Note that these conditions concern the linear forms only and do not depend on $M$. We can use this to show that in the following practical examples related to item recommendations, the linear forms are weakly correlated (given $\beta_0\gtrsim 1$), regardless of the underlying matrix $M$:
\paragraph{Inference of a submatrix.} Consider the inference problem with indexing matrices $\calH=\{e_i e_j^\top : l_1\le i\le l_2, l_3\le j\le l_4  \}$, where $l_2-l_1 \asymp d_1 $, $l_4-l_3 \asymp d_2 $. This can represent recommendation tasks in problems including Netflix prize \citep{bennett2007netflix}, social network discovery \citep{pech2017link}, among others. Here we have the number of tests of order $O(d_1 d_2)$. Since $\|T\|_{\ell_1} /\|T\|_{\mathrm{F}}=1$ for any $T\in \cH$, condition \eqref{eq:weak-corr-2} is easily satisfied. Therefore, at most $O(d_1)$ pairs are strongly correlated (share the same row/column) for each linear form  so that $\beta_{\mathsf{s} }\lesssim 1/d_2$ for any $\nu<0.5$.  
	
\paragraph{Inference of entrywise comparisons.} We can also consider comparison between two entries $M_{i,j}$ and $M_{i+1,j}$: $\cH=\{e_{i} e_{j}^\top- e_{i+1} e_{j}^\top: l_1\le i\le l_2, l_3\le j\le l_4  \}$. If $l_2-l_1 \asymp d_1 $, $l_4-l_3 \asymp d_2 $, then the total number of tests is of the order $O(d_1 d_2)$. Similar to before, $\|T\|_{\ell_1} /\|T\|_{\mathrm{F}}=\sqrt{2}$ for any $T\in \cH$ so that there are at most $O(d_1^2 d_2)$ pairs that can be strongly correlated (share the same row/column). This again yields $\beta_{\mathsf{s} }\lesssim 1/{d_2}$ for any $\nu<0.5$.

\paragraph{Inference of several user/feature groups.} For many applications, groupwise recommendation \citep{bi2018multilayer} is of interest. This can be formulated as testing $H_{0T}: \sum_{i\in G_k} M_{ij}\le \theta_{kj}$ vs $H_{1T}: \sum_{i\in G_k} M_{ij}>\theta_{kj} $, where $(G_1,\ldots,G_K)$ is a partition of the $[d_1]$. In other words $\cH=\{\sum_{i\in G_k}e_ie_j^\top: 1\le k\le K, 1\le j\le d_2\}$.
Note that $\|T\|_{\ell_1} /\|T\|_{\mathrm{F}}=\sqrt{\abs{G_k} }$ for all $T\in \cH$. If $K=\Omega(d_2)$, then for any $\nu<0.5$,
$$\beta_{\mathsf{s} }\lesssim \frac{d_2K(K+d_2)}{d_2^2K^2}\lesssim\frac{1}{d_2}.$$

\subsection{Theoretical Guarantees}
A crucial aspect to understand the efficacy of a multiple testing procedure is the signal strength of the non-null set, i.e., $|\langle M, T\rangle-\theta_T|$ for $T\in\cH_1$. Recall that for any matrix completion estimator $\widetilde{M}$, the best entrywise error rate we can attain is $\|\widetilde{M}-M\|_{\max}\lesssim_p \sigma_{\xi}\sqrt{ d_1 \log (d_1)/n}$ \citep{koltchinskii2011nuclear}. In the case of gradient descent estimator $\widehat{M}^{\mathsf{init}}$, with high probability, one can expect
\begin{equation*}
\abs{\left\langle \widehat{M}^{\mathsf{init}}-M ,T  \right\rangle} \le \norm{\widehat{M}^{\mathsf{init}}-M}_{\max} \norm{T}_{\ell_1}\le  C_{\init}  \sigma_{\xi}\sqrt{\frac{ d_1 \log d_1}{n}} \norm{T}_{\ell_1}.
\end{equation*}
Thus, we say that a signal can be consistently identified if
\begin{equation}\label{eq:strong-T}
\frac{ \abs{\langle M, T\rangle-\theta_T } }{  \norm{T}_{\ell_1} \sqrt{ \log  d_1 } } \ge C_{\mathsf{gap}  } \cdot C_{\init} \sigma_{\xi}\sqrt{\frac{ d_1 \log d_1}{n}}
\end{equation}
for a sufficiently large constant $C_{\mathsf{gap}}>0$. Denote by $\cS$ the set of all linear forms $T\in \cH$ such that \eqref{eq:strong-T} holds, with its cardinality as $\eta_n:=\abs{\cS}$. Note that $\beta_s$ and $\eta_n=|\cS|$ are the most essential quantities in characterizing the effectiveness of FDR control and power guarantee for multiple testing. We are now in position to state our main result.

\begin{Theorem}\label{thm:weak-cor-fdr}
Suppose that $h_n\to 0$ and
$$
\left(\sqrt{\beta_{\mathsf{s}}} \vee  h_n \right)  \frac{q_0}{\eta_n}   \to 0. 
$$
Then, we have
$$
\mathrm{FDP}:=\frac{\sum_{T\in \cH_0  } \bbI(W_T^{\mathsf{rank}} >L ) }{\left( \sum_{T\in \cH  } \bbI(W_T^{\mathsf{rank}} >L ) \right) \vee 1 }\le \alpha(1+o_p(1))
$$
and
$$
\mathrm{POWER}:= \frac{\sum_{T\in \cH_1  } \bbI(W_T^{\mathsf{rank}} >L ) }{q_1} \ge \frac{\eta_n}{q_1}(1-o_p(1)).
$$
\end{Theorem}
Theorem \ref{thm:weak-cor-fdr} provides an asymptotic bound on the FDP, which is stronger than an FDR bound. Indeed, the first claim implies that
\begin{equation}\label{eq:fdr-exp}
\text{FDR} = \E(\text{FDP}) \le \alpha  (1 + o(1)),
\end{equation}
thereby establishing an asymptotic control of the FDR. On the other hand, if nearly all signals are strong in that $\eta_n/q_1\to 1$, then the second claim indicates that $\text{POWER}\to_p 1$. 

For clarity, we stated the asymptotic bounds for FDP and POWER in Theorem \ref{thm:weak-cor-fdr}. Our proof actually establishes stronger results in a nonasymptotic form. Theorem \ref{thm:weak-cor-fdr} is a direct consequence of these nonasymptotic results, which will be presented in the supplement (Section \ref{sec:add-res}). It is also worth noting that both the sample size and signal-to-noise ratio (implied by the condition on $h_n$) requirements of Theorem \ref{thm:weak-cor-fdr} are comparable to those for estimation \citep{keshavan2010matrix_b,ma2018implicit,xia2021statistical}. This immediately suggests that we can effectively control FDR under conditions of weak correlation, provided the underlying matrix can be consistently recovered.

We remark that (i) the condition $\left(\sqrt{\beta_{\mathsf{s}}} \vee h_n\right)\frac{q_0}{\eta_n} \to 0$ is easily satisfied in the three practical examples discussed above, provided that the proportion of identified signals is not too small (i.e., $\frac{\eta_n}{q_0} \gg \frac{1}{\sqrt{d_2}} \vee h_n$). For example, in submatrix inference for the Netflix Prize problem, the condition $\sqrt{\beta_{\mathsf{s}}} \vee h_n \to 0$ is rather mild due to the large scale of the dataset, and the condition $\frac{\eta_n}{q_0} \ge c$ easily holds for testing at a fixed movie rating threshold. Moreover, (ii) the threshold \eqref{eq:dd-threshold} ensures FDR control only when $W_T^{\mathsf{rank}}$ is highly likely to be positive for $T \in \cH_1$. Indeed, if $W_T^{\mathsf{rank}}$ are negative for many $T \in \cH_1$, \eqref{eq:dd-threshold} will yield an overly large $L$, which leads to poor tail symmetry of $\{W_T\}_{T\in\cH_0}$, leaving the FDR unbounded. For instance, simple asymptotic normal-based methods may yield negative test statistics for non-null hypotheses when $\langle {M}, T\rangle-\theta_T<0$, for $T\in\cH_1$. In contrast, we favor symmetric data aggregation, which tends to keep $W_T^{\mathsf{rank}}$ positive because $W_T^{(1)}$, $W_T^{(2)}$ usually share the same sign when $T \in \cH_1$. This preference for symmetric data aggregation, along with its power boost, is also highlighted and justified in \cite{dai2022false,du2021false}.

We now give a brief example illustrating the error rate described in (\ref{eq:fdr-exp}). Consider the submatrix-inference setting discussed in Section~\ref{sec:dependence}. Under a strong matrix-SNR condition $\lambda_{\min}/\sigma_{\xi}\gtrsim d_1$, and assuming a sufficient proportion of identified signals so that $\eta_n/q_0 \asymp 1$, we have $\text{FDR}\le \alpha + h_n + \widetilde{O}\big( (h_n q_0/\eta_n)^{1/6}  \big)\le \alpha + \widetilde{O} \big(  (d_1/n)^{1/12}  \big) $. A more detailed nonasymptotic bound can be found in Section~\ref{sec:non-saymp-results} of the supplement.


\subsection{Testing under Heterogeneity}\label{sec:hetero}
Although the asymptotic normality result in Section \ref{sec:clt}, together with our multiple testing procedure (Algorithm \ref{alg:matrix-fdr}), is developed under homogeneous noise and uniform sampling, we note that the framework extends naturally to settings with heterogeneous noise and, potentially, non-uniform missingness. In particular, when the noise variables $\xi_i$ have heterogeneous variances, the asymptotic normal approximation remains valid after a straightforward modification of the variance term.

\begin{Proposition}\label{prop:hetero} Assume $\{(X_i,\xi_i)\}_{i=1}^n$ are i.i.d. but $\xi_i$ depends on $X_i$ such that 
\begin{equation*}
    \E[\xi_i\mid X_i= x_0 ] = 0, \quad \E[\xi_i^2\mid X_i= x_0 ] = \left(\langle S, x_0 \rangle\right)^2,
\end{equation*}
for any entry element in  orthonormal basis: $x_0 \in \mathfrak{E}$. Here $S\in\R^{d_1\times d_2}$ represents the  heterogeneous standard deviation of $\xi_i$ with respect to the location where it is sampled. Suppose further that $\xi_i\mid X_i$ is still $\sigma_\xi$ sub-Gaussian, and $\underline{\sigma}\le S_{ij}\le \overline{\sigma}$, with $\overline{\sigma}/\underline{\sigma} \le \kappa_{\sigma}$ for some $\kappa_\sigma >0$. Then, we have  
\begin{equation}\label{eq:hetero-noise-clt}
    \frac{\langle \widehat{M}, T\rangle-\langle M, T\rangle}{\|\cP_M(T) \odot S \|_\tF \cdot \sqrt{d_1 d_2 / n}} \to_d N(0,1),
\end{equation}
provided that $n$ and $\frac{\lambda_{\min}}{\sigma_{\xi}} $ are large. Moreover, when 
    \begin{equation}\label{eq:condition-SNR-sample-hetero}
    n\ge  C_1 C_{\init}^2\kappa_0^2 \kappa_{\sigma}^6\frac{\norm{T}_{\ell_1}^2}{\beta_T^2\norm{T}_{\tF}^2} \mu^{5}r^5 d_1\log^2 d_1, \ \frac{\lambda_{\min}}{\sigma_{\xi}} \ge  C_2 C_{\init}^2 \frac{\kappa_0\kappa_{\sigma}^2\norm{T}_{\ell_1}}{\beta_T\norm{T}_\tF } \sqrt{\frac{ \alpha_{d} \mu^{6} r^5 d_1^2 d_2  \log^2 d_1 }{ n } }.
\end{equation}
we can take the test statistic under  heterogeneity as
\begin{equation}\label{eq:test-stat-hetero}
    W_{T}^{\mathsf{h}} =  \frac{\langle \widehat{M}, T\rangle-\theta_T}{ \widehat{s}_T^{ \mathsf{h} } \cdot \sqrt{d_1 d_2 / n}}, 
\end{equation} 
where
\begin{equation}
    \left(\widehat{s}_{T}^{\mathsf{h}}\right)^2 = \frac{d_1 d_2}{n}\sum_{i=1}^n\left[\left(Y_i-\left\langle \widehat M^{\init},X_i \right\rangle\right)\left\langle\cP_{\widehat M^{\init}}(T), X_i\right\rangle \right]^2.
\end{equation}
We then have the following result when $H_{0T}$ is true:
\begin{equation}\label{eq:hetero-CLT}
\begin{aligned}
    &\sup _{t\in \mathbb{R}}\left|\mathbb{P}\left(W_{T}^{\sfh} \leq t\right)-\Phi\left(t\right)\right|
    \\
    &
\le C_3\left(  C_{\init}   \frac{\kappa_0 \kappa_{\sigma}^3\norm{T}_{\ell_1}}{\beta_T \norm{T}_{\tF}} \sqrt{\frac{ \mu^{5} r^5 d_1 \log^2 d_1  }{ n}} +  \frac{ C_{\init}^2\sigma_{\xi}  \kappa_0 \kappa_{\sigma}^2 \norm{T}_{\ell_1} }{\beta_T \lambda_{\min}  \norm{T}_\tF }\sqrt{\frac{\alpha_d \mu^{6} r^{5 } d_1^2 d_2  \log^2 d_1 }{ n } } \right).
\end{aligned} 
\end{equation}

\end{Proposition}
Proposition \ref{prop:hetero} covers a broad range of heterogeneous settings, including cases where the observations follow Bernoulli noise models, such as $Y_i \sim \operatorname{Ber}(\langle M, X_i \rangle)$, as well as other discrete data scenarios, for example, the Netflix movie rating \citep{bennett2007netflix}. Moreover, since our our multiple testing procedure only relies on the asymptotic behavior of $W_T$, by replacing $W_T$ with $W_T^{\sfh}$ in \eqref{eq:test-stat-hetero} and $h_n$ with the rate in \eqref{eq:hetero-CLT}, our Algorithm \ref{alg:matrix-fdr} can still control FDP asymptotically.

The extension to non-uniform missingness is also possible. Our bias–variance decomposition shows that when the sampling distribution of $X_i$ is non-uniform, specifically, given any  $x_0 \in \mathfrak{E}$, $\PP(X_i=x_0)=\langle P,x_0 \rangle$ for some matrix $P\in \R_+^{d_1\times d_2}$ satisfying $\mathbf{1}^\top P \mathbf{1} = 1$, the leading variance component yields the following asymptotic normality result:
\begin{equation}\label{eq:non-uniform}
    \frac{\langle \widehat{M}, T\rangle-\langle M, T\rangle}{\|\cP_M(T) \odot S \odot  \sqrt{P} \|_\tF \cdot d_1 d_2 /\sqrt{n} } \to_d N(0,1), \text{  as }n, \lambda_{\min}/\sigma_{\xi} \to \infty. 
\end{equation}
Here, $\sqrt{P} $ denotes the entrywise square root of $P$. For uniform sampling case, we have $P=\frac{1}{d_1 d_2}\mathbf{1} \mathbf{1}^\top   $, which returns \eqref{eq:non-uniform} to \eqref{eq:hetero-noise-clt}. Establishing a rigorous proof of \eqref{eq:non-uniform}, however, is substantially more involved, as it requires revisiting matrix completion theory under non-uniform missingness. Moreover, consistently estimating the variance term in \eqref{eq:non-uniform} is challenging without additional structural assumptions. We leave these important directions to future work.

\section{Whitening and Screening}\label{sec:strong-corr}
Theorem \ref{thm:weak-cor-fdr} shows that the symmetric data aggregation method can control FDR effectively if the number of strongly correlated linear form pairs is sufficiently small relative to the number of strong signals, i.e., $\sqrt{\beta_{\mathsf{s}}}q_0/\eta_n\to 0$. While this is plausible in many applications, as we have argued, there are also situations in which this may not be the case. We now discuss how this condition can be further relaxed thanks to the explicit characterization of the correlation among test statistics. In particular, as advocated by \cite{du2021false}, we proceed to apply symmetric data aggregation after appropriate whitening  and screening. Interestingly, by exploiting the explicit characterization of the dependence among $W_T$s, we can develop a more general and intuitive theoretical framework to study the power and FDR control for matrix completion.

To illustrate the motivation for whitening, consider the following toy example of simultaneous mean testing. We observe $Y\sim N(\theta,\Sigma)$ in $\R^q$ and aim to test $H_{0,j}:\theta_j=0$ vs $H_{1,j}:\theta_j \neq 0$, $1\le j\le q$. If the covariance matrix $\Sigma$ is known, we can apply a ``whitening’’ transformation: $\Sigma^{-\frac{1}{2}}Y =\Sigma^{-\frac{1}{2}}\theta + \xi'$, where $\xi'\sim N(0,I_q)$ has independent components. Whitening thus transforms the problem to high-dimensional linear regression with independent noise. Based on this de-correlated model, it is natural to apply variable selection methods (e.g., Lasso) for preliminary signal screening, and subsequently use symmetric data aggregation to achieve more effective FDR control, as suggested in \cite{du2021false}.

More specifically, denote the collection of test statistics obtained from Algorithm \ref{alg:matrix-fdr} as $ Z^{(i)}=\left[W^{(i)}_{T_1},  W^{(i)}_{T_2}, \dots, W^{(i)}_{T_q} \right]^\top\in\RR^q$, for $i=1, 2$. By Theorem \ref{thm:asymp-two-var}, $Z^{(i)}{\approx}_d N(\sfw, R)$ where $\textsf{w}\in\RR^{q}$ with the $i$-th entry $\sfw_i=\big(\langle M, T_i\rangle-\theta_{T_i}\big)/\big(\sigma_{\xi}\|\calP_M(T_i)\|_{\rm F}\sqrt{d_1d_2/n}\big)$ and $R=(\rho_{T_j,T_k})_{1\le j,k\le q}$. If $R$ is known, then $R^{-1/2}Z^{(i)}\approx_d N(R^{-1/2} \sfw, I_q)$ has asymptotically independent coordinates and thus allows for better FDR control. However, such a whitening step can also mask the nonzero coordinates of $\sfw$, which, as suggested by \cite{du2021false}, can be estimated by Lasso. Of course, $\rho_{T_j,T_k}$ is unknown, but it can nonetheless be estimated by
$$
\widehat{\rho}_{T_j,T_k} = \frac{\left\langle \cP_{ \widehat{M}^{\mathsf{init}} } (T_j) ,\cP_{ \widehat{M}^{\mathsf{init}} }(T_k)  \right\rangle }{ \norm{\cP_{ \widehat{M}^{\mathsf{init}} }(T_j)}_{\tF}\norm{\cP_{ \widehat{M}^{\mathsf{init}} }(T_k)}_{\tF} },
$$
where $\cP_{ \widehat{M}^{\mathsf{init}} }(\cdot)$ is defined following \eqref{eq:tangent-proj}, and $\widehat{M}^{\mathsf{init}} $ is the initial GD estimate. In summary, we shall consider the following algorithm detailed in Algorithm~\ref{alg:matrix-sda}. Here, to ensure valid inversion of the population correlation matrix $R$, we confine $q$ as $q\le (d_1+d_2)r-r^2$, and assume $T\in\cH$ are linearly independent. We note that when $q$ is large, Algorithm \ref{alg:matrix-fdr} is preferable, since in many scenarios (e.g., those discussed in Section \ref{sec:dependence}), a large hypothesis set $\cH$ typically implies that a substantial proportion of the test statistics are only weakly correlated.

\begin{algorithm}[htbp]
\caption{Matrix FDP Control with Whitening and Screening}
\label{alg:matrix-sda}
\begin{algorithmic}[1]
\REQUIRE Hypotheses $\left\{H_{0T_i}: M_{T_i}=\theta_{T_i}, i\in [q]\right\}$, data splits  $\cD_1$, $\cD_2$, rank $r$, FDR level $\alpha$, regularization parameter $\lambda\ge 0$.
\STATE{ Apply Algorithm \ref{alg:matrix-fdr} to get  $Z^{(1)}\in\R^{q}$, $Z^{(2)}\in\R^{q}$ from $\{\mathcal{D}_1\}$ and $\{\mathcal{D}_2\}$ respectively}
\STATE {From $\mathcal{D}_1$, obtain a covariance estimate $\widehat{R}=(\wt\rho_{T_i, T_j})_{i,j=1}^q$ using $\widehat{M}_{\mathsf{init}}^{(1)}$ estimated from Algorithm \ref{alg:matrix-fdr}, that is
\begin{equation*}
    \widehat{\rho}_{T_i,T_j} = \frac{\left\langle \cP_{\widehat{M}_{\mathsf{init}}^{(1)}} (T_i) ,\cP_{\widehat{M}_{\mathsf{init}}^{(1)}}(T_j)  \right\rangle }{ \norm{\cP_{\widehat{M}_{\mathsf{init}}^{(1)}}(T_i)}_{\tF}\norm{\cP_{\widehat{M}_{\mathsf{init}}^{(1)}}(T_j)}_{\tF} } .
\end{equation*}
And solve Lasso estimator
\begin{equation*}
    \wt \sfw^{(1)}: = \argmin_{\sfw\in \R^q } \left\{ \frac{1}{2}\big\|\wt R^{-1/2}(Z^{(1)} - \sfw)\big\|^2 + \lambda\norm{\sfw}_{\ell_1} \right\}.
\end{equation*}

}
\STATE {Denote $\calA:={\rm supp}(\wt \sfw^{(1)})$ the support of $\wt \sfw^{(1)}$. Run linear regression on $\calA$ with new design matrix $\wt R^{-1/2}_{\calA}$ and response $\wt R^{-1/2}Z^{(2)}$ to get asymptotically symmetric statistics $\wt\sfw^{(2)}$, where
\begin{equation*}
\wt\sfw_{\calA}^{(2)}:=\big(\wt R_{\calA}^{-1/2\top} \wt R_{\calA}^{-1/2}\big)^{-1}\wt R_{\calA}^{-1/2\top}\wt R^{-1/2}Z^{(2)}\quad {\rm and}\quad \wt \sfw^{(2)}_{\calA^{\rm c}}=0
\end{equation*}
with variance estimate $\widehat{\sigma}_{\sfw i}^2:= e_i^{\top}\big(\wt R_{\calA}^{-1/2\top} \wt R_{\calA}^{-1/2}\big)^{-1} e_i $ for $i\in\calA$.
}
\STATE { Compute the final ranking statistics of each $T_i$ by $\sfw_{T_i}^{\mathsf{rank}}=\wt\sfw_{i}^{(1)}\wt \sfw_{i}^{(2)}/\widehat{\sigma}_{\sfw i}$, and then choose a data-driven threshold $L$ by 
\begin{equation*}
    L:=\inf \left\{t>0: \frac{\sum_{i=1}^q\bbI\left( \sfw_{T_i}^{\mathsf{rank}}<-t\right)}{\sum_{i=1}^q \bbI \left( \sfw_{T_i}^{\mathsf{rank}}>t\right) \vee 1} \leq \alpha\right\}.
\end{equation*}}
\STATE{Reject $H_{0T_i}$ if $\sfw_{T_i}^{\mathsf{rank}}>L$}
\end{algorithmic}
\end{algorithm}

Here $\wt R_{\calA}^{-1/2}$ is the submatrix of $\wt R^{-1/2}$ with only columns indexed by $\calA$. Similarly, $\wt\sfw_{\calA}$ is the subvector of $\wt\sfw$ with only coordinates indexed by $\calA$. Note that Algorithm \ref{alg:matrix-fdr} can be treated as a special case of Algorithm \ref{alg:matrix-sda} by choosing the regularization parameter $\lambda=0$. However, as we argue below, with an appropriate choice of $\lambda>0$, the whitening and screening may lead to a more effective multiple testing procedure. In addition, a more concrete example of testing entries of submatrix of $M$ is provided in the supplement (Section \ref{sec:add-res}) to demonstrate the impact of whitening and screening.

It is clear that the efficacy of Algorithm \ref{alg:matrix-sda} hinges upon the reduction of dependence among test statistics with Lasso screening. We can show that, under mild regularity conditions, the asymptotic covariance matrix of $\wt \sfw^{(2)}_\calA$ is given by
$$
Q^{\ast}:=\big(R_{\calA}^{-1/2\top} R_{\calA}^{-1/2}\big)^{-1}.
$$
Similar to before, write
$$
\cH_{0\mathcal{A} ,\text{strong} }^2 = \left\{ (T_i,T_j)\in\calA_0 \times \calA_0 :  \abs{ Q^{\ast}_{jk}}/\sqrt{Q^{\ast}_{kk}Q^{\ast}_{jj}} \ge c|\calA|^{-\nu}  \right\},
$$
where $\calA_0=\calA\cap \cH_0$. Denote by 
$$
\beta_{\mathsf{s} }': = \frac{ \abs{\cH_{0\mathcal{A},\text{strong} }^2}   }{\abs{\calA_0}^2 }.
$$
In other words, $\beta_{\mathsf{s} }'$ represents the proportion of strongly correlated pairs after whitening and screening. Likewise, we shall write $\eta_n'=\abs{\cS'}$ where $\cS'$ is the set of strong signals to be defined. To define strong signal, write
\begin{equation}\label{eq:def-TH}
    T_{\cH}= \left[ \begin{array}{c}
		\Vect(T_1)^\top  \\
		\Vect(T_2)^\top \\
		\vdots\\
		\Vect(T_q)^\top
	\end{array} \right]\in \R^{q\times d_1 d_2} 
\end{equation}
Then the limiting covariance matrix of $W_T$s is given by 
$$
\Sigma:= \big(\big<\calP_M(T_j),  \calP_M(T_k)\big>\big)_{1\leq j,k\leq q}=T_{\cH}(I_{d_1 d_2} - U_\perp U_\perp^\top \otimes V_\perp V_\perp^\top  ) T_{\cH}^\top.
$$
We assume that $\Sigma$ is invertible with $\lambda_{\min}(\Sigma)\ge c$ for some small constant $c$, and denote the condition number $\kappa_1=\lambda_{\max}(\Sigma)/\lambda_{\min}(\Sigma)$. We can then define new strong signals as:
\begin{equation}\label{eq:strong-text}
\cS' = \left\{ T\in \cH: \frac{ \abs{\langle M, T\rangle-\theta_T } }{   \norm{T}_{\ell_1} \sqrt{ q_1 \log d_1 } } \ge C_{\mathsf{gap} } \cdot C_{\init} \kappa_1^{3/2} \sqrt{\frac{d_1 \log d_1}{n}} \right\},
\end{equation}
with its cardinality as $\eta_n'=\abs{\cS'}$. We have the following theoretical guarantee for Algorithm \ref{alg:matrix-sda}.
\begin{Theorem}\label{thm:matrix-fdr-strong}
 Let $T_{\calH}$ be a $q\times d_1d_2$ matrix with $i$-th row being ${\rm vec}(T_i)$ and define ${\rm supp}(T_{\calH}):=\cup_{i=1}^q {\rm supp}(T_i)$. Suppose that $q_{0}'$ a uniform upper bound for $\abs{\calA_0}$ and  
$$
\left(\sqrt{\beta_{\mathsf{s}}'} \vee  \left( h_n +\big\|\sfw_{\calA^c}\big\|_{\infty}\right)\right)  \frac{q_0'}{\eta_n'}   \overset{p}{\rightarrow} 0, 
$$
and 
\begin{equation}\label{eq:SNR-strong-dep}
    \lambda_{\min}\gg 
    C_{\init}\left(\norm{ R^{-1}  }_\infty + \frac{  \norm{T_{\calH} }}{ \norm{T_{\calH} }_{2,\max} } \left( |\operatorname{supp}(T_{\calH} )|\wedge \sqrt{d_2} \right) \right) \max_{T\in\calH}\left\{ \frac{ \|T\|_{\ell_1} }{ \|T\|_{\mathrm{F}} } \right\} \sigma_\xi \sqrt{\frac{ q d_1^3 \log d_1 }{n}}.
\end{equation}
Then there exists universal constant $C_4>0$ such that if regularization parameter $\lambda=C_4\sqrt{\log d_1 }$ in Algorithm~\ref{alg:matrix-sda}, then
$$
\mathrm{FDP}=\frac{\sum_{T\in \cH_0  } \bbI(\sfw_{T}^{\mathsf{rank}} >L ) }{\left( \sum_{T\in \cH  } \bbI(\sfw_T^{\mathsf{rank}} >L ) \right) \vee 1 } \le \alpha(1+o_p(1))
$$
and
$$
\mathrm{POWER}= \frac{\sum_{T\in \cH_1  } \bbI(\sfw_T^{\mathsf{rank}} >L ) }{q_1} \ge \frac{\eta_n' }{q_1}(1-o_p(1)).
$$   
\end{Theorem}

Note that the covariance matrix $\Sigma=\big(\langle \calP_{M}(T_i), \calP_M(T_j)\rangle\big)_{i,j\in[q]}$ is not known and our whitening procedure uses an estimate in its place. The additional lower bound of $\lambda_{\min}$ in Theorem \ref{thm:matrix-fdr-strong} is in place to ensure that the estimated covariance matrix indeed can be used to ``whiten'' the test statistics. It is also worth pointing out that we do not require the sure-screening condition of Lasso. Such conditions are common in the literature. See , e.g., \cite{roeder2009genome,barber2019knockoff,du2021false,dai2023scale}. For our purpose, weak signals can be entertained as long as $\|\sfw_{\calA^c}\|_{\infty}$ is sufficiently small. We also note that the sample size $n$ in our matrix completion problem has a fundamentally different meaning compared to the classical regression problem. Due to incomplete observations, each sample point provides only limited information for inferring a matrix. Even for a single linear form, constructing a test statistic requires at least $n\gg d_1\log^2 d_1$ samples from Theorem \ref{thm:asymp-normal-varest}. Therefore, given $q\le d_1 d_2$, one cannot expect $n$ to be logarithmically dependent on $q$, as in the classical regression problem.


\section{Numerical Experiments}\label{sec:experiments}

\subsection{Simulation Studies}
To complement our theoretical development, we also conducted several sets of numerical experiments to further demonstrate the practical merits of the proposed methodology. We begin with a series of simulation studies aimed at illustrating the impact of several key aspects of our approach. All the simulations in this section display the averaged performance of multiple independent runs.  The reproduction code can be accessed through Github repo \url{https://github.com/wantengma/MC-FDR-Public.git}.

\subsubsection{Variance of linear forms}\label{sec:variance-lf}
In Section \ref{sec:clt},  we have presented the asymptotic normal test statistics for linear forms with a more accurate characterization of its variance. To justify the accuracy of our variance $\norm{\cP_M(T)}_\tF$, we show the simulation of empirical distribution functions of our test statistics $W_T$ in Theorem \ref{thm:asymp-normal} against former test statistic in \eqref{eq:clt-tangent} whose variance is characterized by $(\|U^\top T\|_{\rm F}^2+\|TV\|_{\rm F}^2)^{1/2}$ in \cite{xia2021statistical}. We plot the difference between empirical distribution functions $\bar{F}_n(z)$ and standard normal distribution function $\Phi(z)$ by sampling 10,000 independent realizations of test statistics. The result is shown in Figure \ref{fig:variance-comparison}. It is clear that our methods share a more precise asymptotic normal rate given smaller errors of $\bar{F}_n(z)-\Phi(z)$, especially for small sample size $n$.
\begin{figure}[H]
\centering
\begin{subfigure}{0.32\textwidth}
    \includegraphics[width=\textwidth]{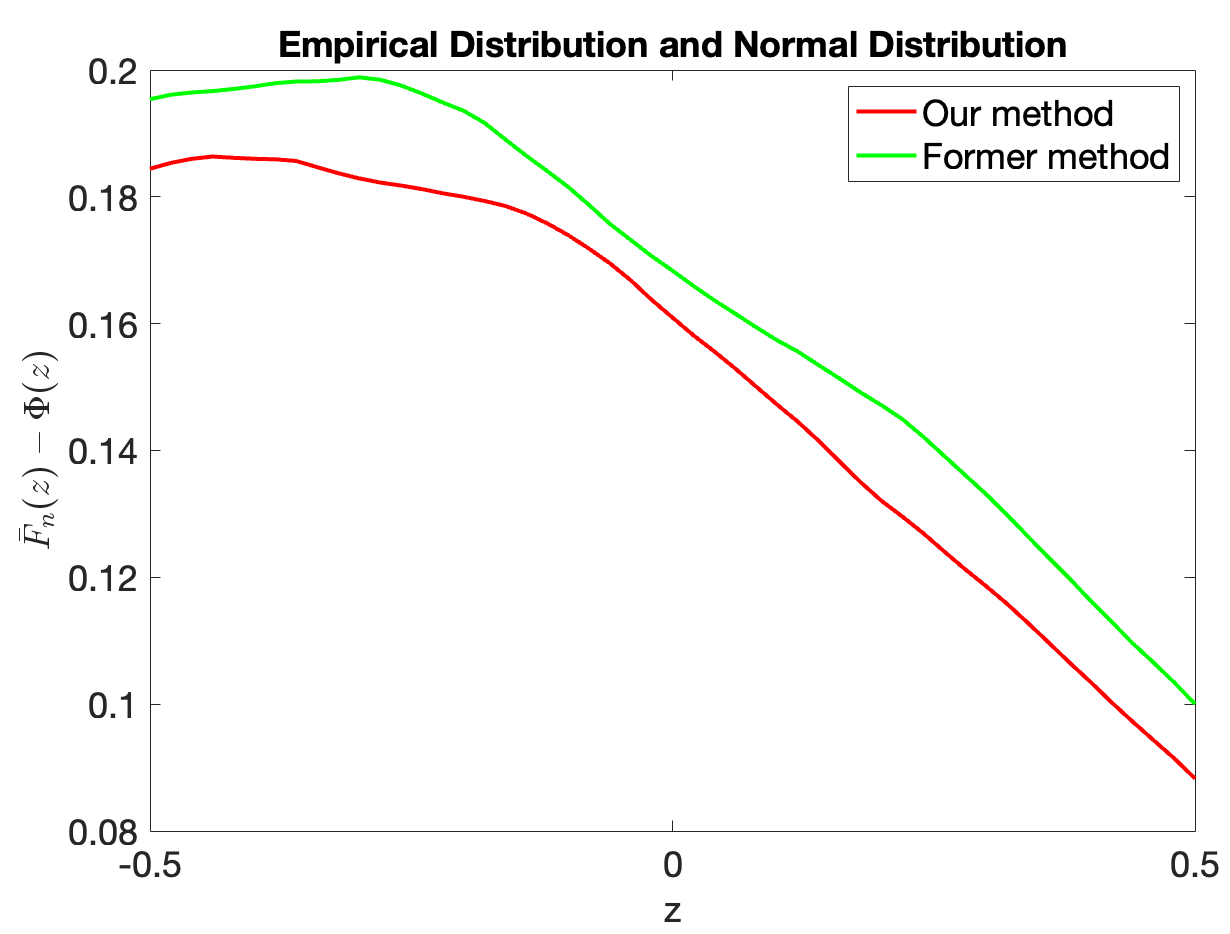}
    \caption{$n=2400$ }
    \label{fig:variance-comparison-1}
\end{subfigure}
 \begin{subfigure}{0.32\textwidth}
    \includegraphics[width=\textwidth]{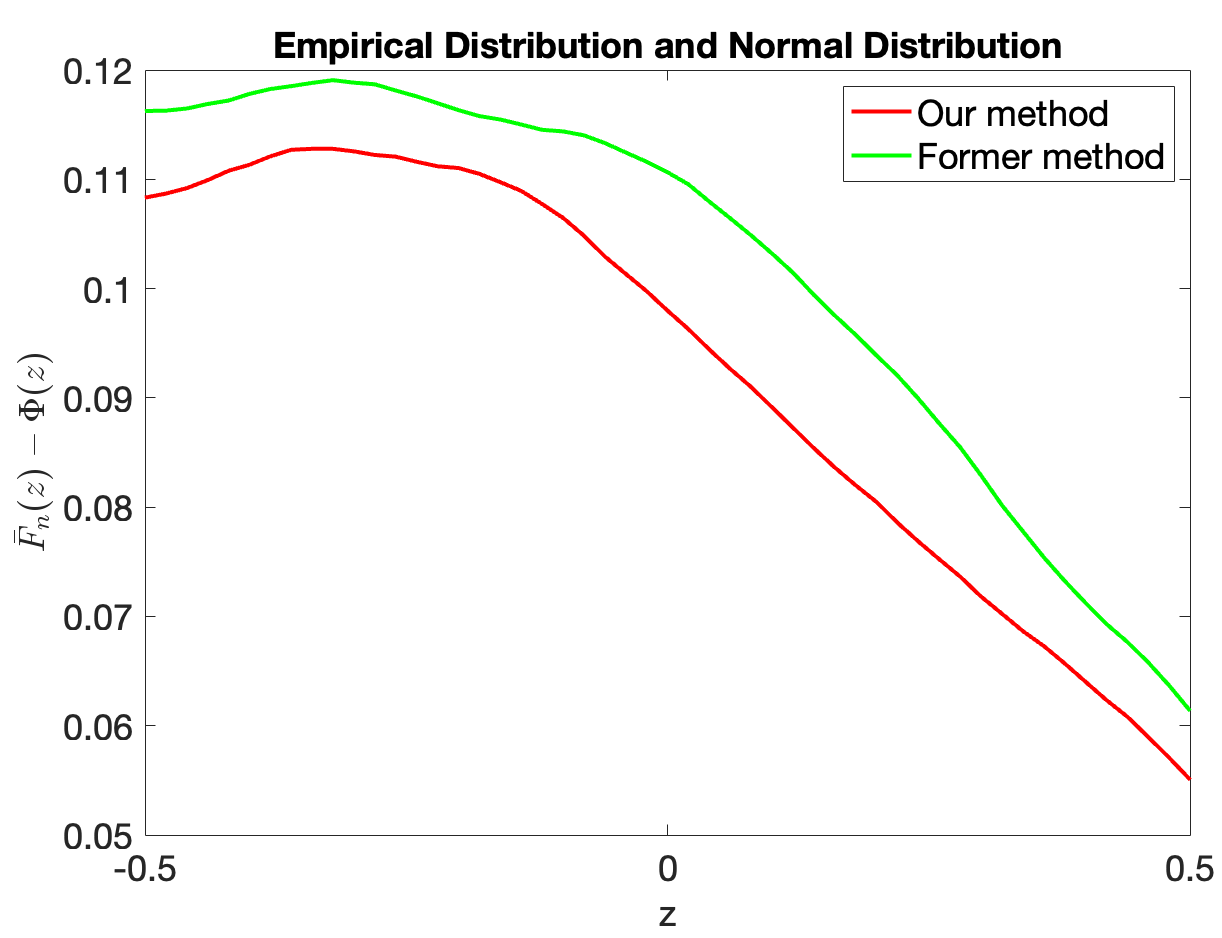}
    \caption{$n=3000$ }
    \label{fig:variance-comparison-2}
\end{subfigure}
     \begin{subfigure}{0.32\textwidth}
    \includegraphics[width=\textwidth]{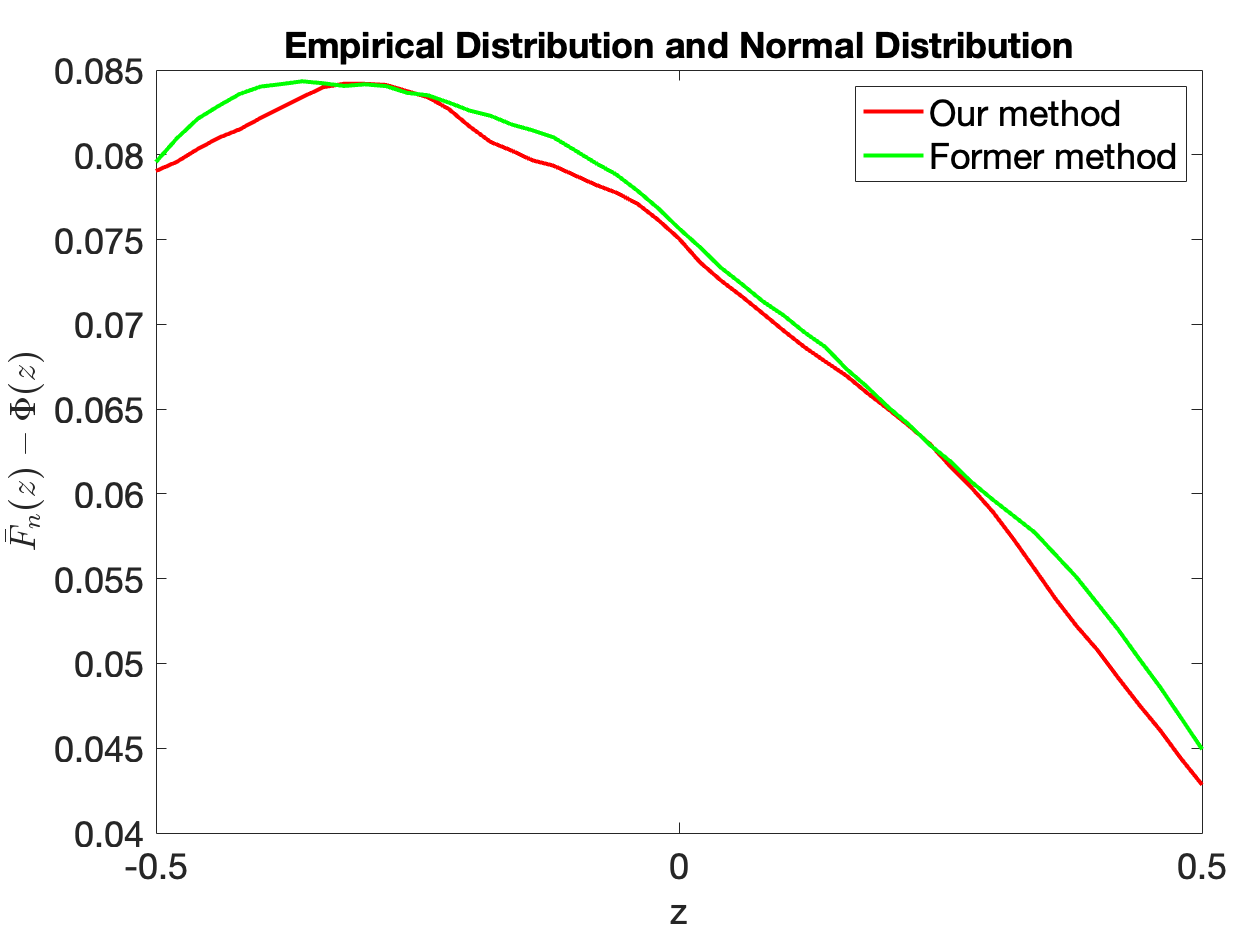}
    \caption{$n=3600$ }
    \label{fig:variance-comparison-3}
\end{subfigure}

 \caption{The difference between empirical distribution functions and $\Phi(z)$.  Here, we compare our $W_T$ with the former method \citep{xia2021statistical}. We set the matrix with $d_1=d_2=\lambda_{\min}=400$, and $r=3$, and vary the number of random samples $n$ in noisy matrix completion. }
 \label{fig:variance-comparison}
\end{figure}

\subsubsection{Data aggregation under weak dependency}\label{sec:weak-simu}
We first evaluate our Algorithm \ref{alg:matrix-fdr} by simulations to corroborate two important properties of the proposed method: (1) the validity of FDR control for multiple testing of linear forms; (2) the power boost by data splitting and data aggregation; we randomly sample a low-rank matrix of dimension $d_1=d_2=1000$, rank $r=3$, with signal strength $\lambda_{\min}=1000$. The number of observations used for Algorithm \ref{alg:matrix-fdr} is $n=50rd_1$, and the noises $\xi\sim N(0,1^2)$. We use the gradient descent \citep{wei2016guarantees,chen2020noisy,cai2022generalized} as initialization. We first verify the FDR control in weak dependency by performing blockwise matrix tests: we test each entry in $M(1:200,1:200)$ by $H_{0,ij}: M_{ij}-m_{ij}=0$ versus $H_{1,ij}: M_{ij}-m_{ij}\neq 0$. We randomly assign non-null hypotheses to these $200\times 200=40,000$ entries with probability $p=0.2$, which leads to the following settings of $m_{ij}$:
\begin{equation}\label{eq:construct_H0}
    M_{ij}-m_{ij}=\begin{cases}
        \mu_{ij}, & \text{ with probability } p=0.2 ;\\
         0, & \text{ otherwise }
    \end{cases}  
\end{equation}
Here $\mu_{ij}$ are randomly-generated signals with a fixed absolute mean: $\E \abs{\mu_{ij}}=\mu$. We run Algorithm \ref{alg:matrix-fdr} and compare different methods of data aggregation (see Section~\ref{sec:compare} for more details): I. multiplication; II. minimum absolute value with sign multiplication; III. adding absolute values with sign multiplication; IV. BH with no data splitting. Here, BH with no data splitting means that we use data $\cD_1$ and $\cD_2$ together to construct asymptotic normal test statistics and then compute their $p$-values by the normal distribution. We defer the detailed BH selection approach to Section \ref{sec:supp-table-alg} in the supplement.

The result presented in Figure \ref{fig:fdr-simu} clearly shows the excellent performance of multiplication in data aggregation with respect to both FDR control and power. By Section \ref{sec:dependence}, the blockwise matrix entry tests here can be treated as the weakly correlated case. Although the BH method \cite{benjamini1995controlling} is guaranteed to be effective in the classical regression model, it fails to provide the highest power in the matrix completion problem. 




\begin{figure}
\centering
     \includegraphics[width=\textwidth]{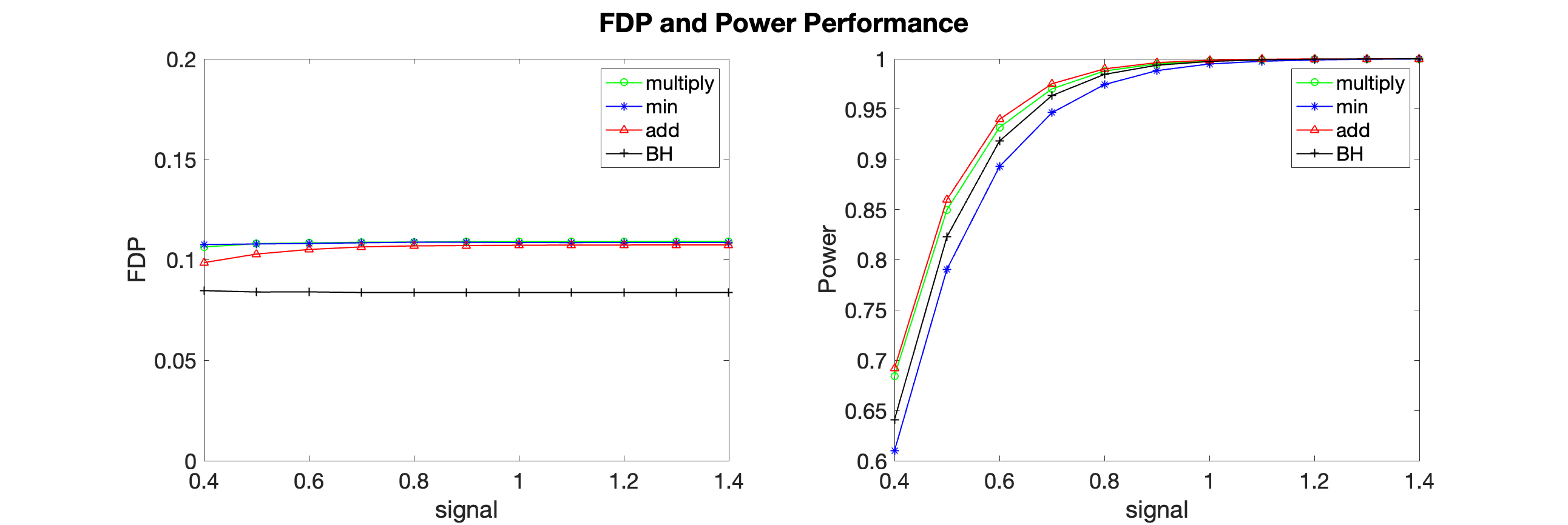}
 \caption{FDR control \& Power of different data aggregation schemes in blockwise matrix tests with $\alpha=0.1$.  Here the signal is defined by $\mu$ in eq. (\ref{eq:construct_H0}). }
 \label{fig:fdr-simu}
\end{figure}

\subsubsection{Whitening and screening}\label{sec:whitening}
We now evaluate Algorithm \ref{alg:matrix-fdr} and Algorithm \ref{alg:matrix-sda} and show the advantages of de-correlation. 
 To this end, we still adopt the data generation mechanism in the previous section,  but apply our methods to the entry comparisons between rows: we test $q=400$ differences between first row $M(1,1:400)$ and second row $M(2,1:400)$, with $H_{0,T_i}$: $M_{1,i}-M_{2,i}=0$. By choosing $T_i$ like these, the linear forms are in the same rows with shared position information, meaning that they are more likely to be correlated (but not highly correlated). Because of the complicated correlation structure of features, here we measure the overall correlation of our case by the proportion of related pairs: 
\begin{equation}\label{eq:rho-emp}
    \varrho^*(z)=\frac{\sum_{i,j\in[q] } \bbI\left( \abs{\rho_{T_i, T_j} }> z\right)   }{q^2},
\end{equation}
where $\rho_{T_i,T_j}$ indicates the correlation of two linear form $M_{T_i}$ and $M_{T_j}$ and is given by \eqref{eq:corr}. 
Here $\varrho^*(z)$ can be treated as a proxy of the strength of correlation $\beta_{\mathsf{s}}$. In this entry comparison problem, we have $\varrho^*(0.2)=0.2745$, which means that an indispensable proportion of feature pairs are correlated. For the SDA method, we use a known correlation matrix. 
 The performance of Algorithms 1 and 2 with different data aggregation methods are summarized in Figure 3. In addition, we compare our methods with $e$-BH \citep{wang2022false} and factor-adjusted multiple testing by principal factor approximation (PFA) \citep{fan2017estimation}. Here, $e$-value for each $T_i$ is taken as $e_{T_i}=\exp\left(\mu W_{T_i}^{\mathsf{all}} - \frac{\mu^2}{2}\right)$, which is claimed to be powerful \citep{ramdas2025hypothesis}.
\begin{figure}
\centering
     \includegraphics[width=\textwidth]{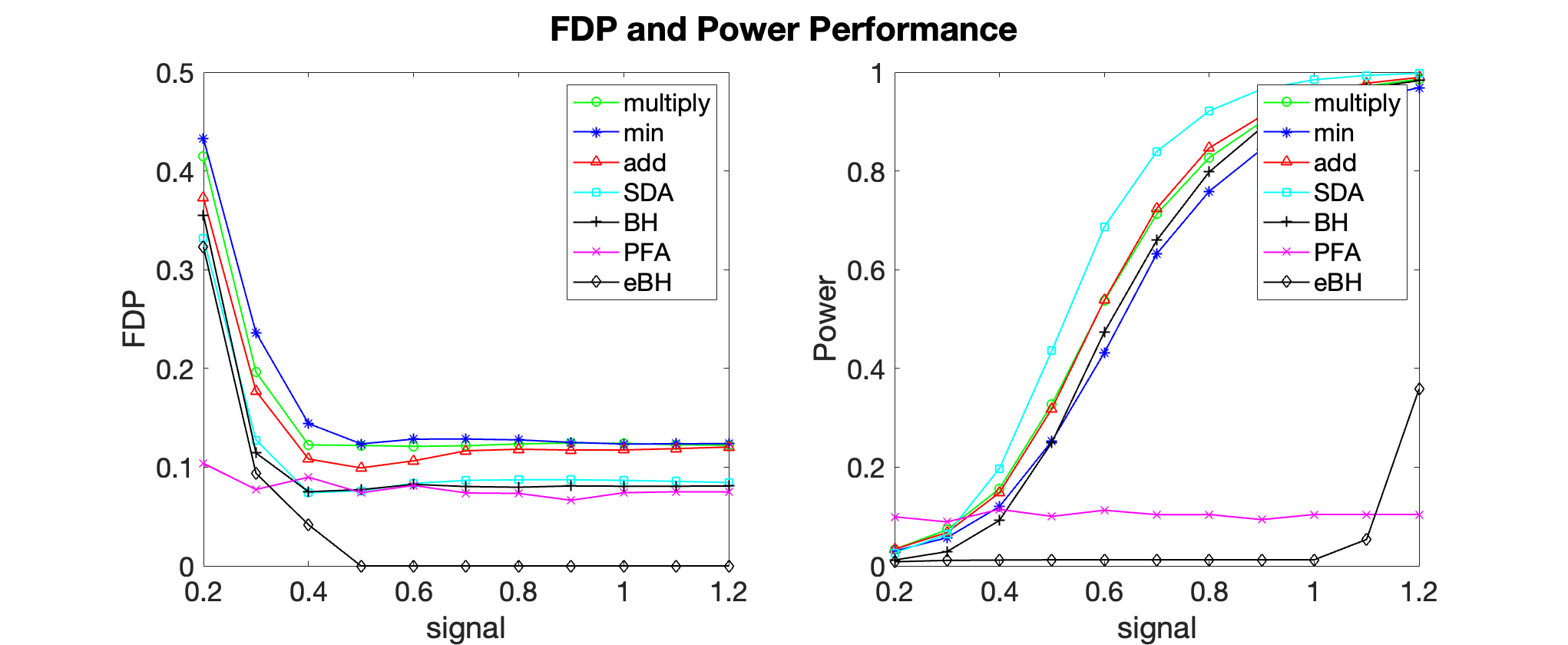}
 \caption{FDR control \& Power of different data aggregation schemes in row tests with $\alpha=0.1$. Here the signal is defined by $\mu$ in eq. (\ref{eq:construct_H0}). }
 \label{fig:fdr-sda-simu}
\end{figure}

In Figure \ref{fig:fdr-sda-simu}, the SDA method can effectively control the FDR level at $\alpha=0.1$, with a notable power enhancement compared with the standard BH method. Moreover, without de-correlation and screening, simple data aggregation methods fail to control the FDR due to dependency. PFA and $e$-BH can control the FDR but at the cost of very conservative rejections.
We can thereby draw the conclusion that our algorithm based on SDA outperforms others in the highly correlated case with the help of screening and de-correlation. 

\subsubsection{Heavy-tailed noise}
While our theories are established for sub-Gaussian noise, we observe that the proposed methods are very robust to heavy-tailed noise. This section showcases the performance of our algorithms in the existence of heavy-tailed noises, e.g., $t$-distribution and exponential distribution, and compares the performances of different methods. We consider moderate and strong correlations, respectively. Here $M$ is randomly generated with dimensions $d_1=d_2=400$, rank $r=3$, $\lambda_{\min}=400$, and the noise is fixed with a standard deviation $\sigma_\xi=0.4$. The sample size is set by $n=3000$. We focus on the following tasks: (i) entry comparisons between rows; (ii) entry comparisons within a block. More specifically, in the entry comparison task between rows, we compare $H_{0,T}$: $M_{i,j}-M_{i+1,1}=0$ for every $1\le i\le 4$ and $j\ge 2$. That is, we compare each entry with the first entry of the next row; in the entry comparison task within a block, we compare $H_{0,T}$: $M_{i,j}-M_{1,1}=0$ for every  $1\le i\le 4$ and $j\ge 2$. For these two tasks, we all have $q=1596$, but the correlation structures and levels are different. That is, (i) entry comparisons between rows, $\varrho^*(0.2)=0.4541$; (ii) entry comparisons within a block, $\varrho^*(0.2)=0.9514$. Here, (i) and (ii) can be viewed as examples of moderate and strong correlations.

\begin{figure}
\centering
\begin{subfigure}{0.8\textwidth}
     \includegraphics[width=1\textwidth]{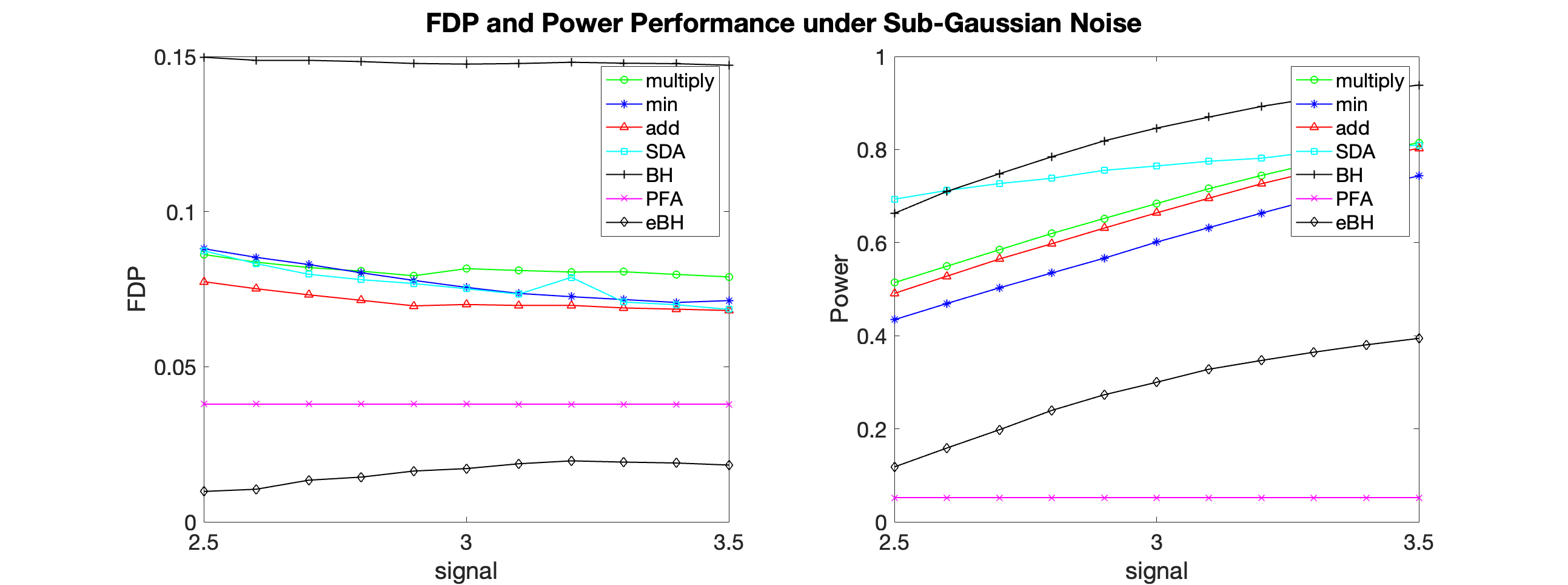}
     \label{fig:fdr-moderate-1}
 \end{subfigure}
 \begin{subfigure}{0.8\textwidth}
     \includegraphics[width=1\textwidth]{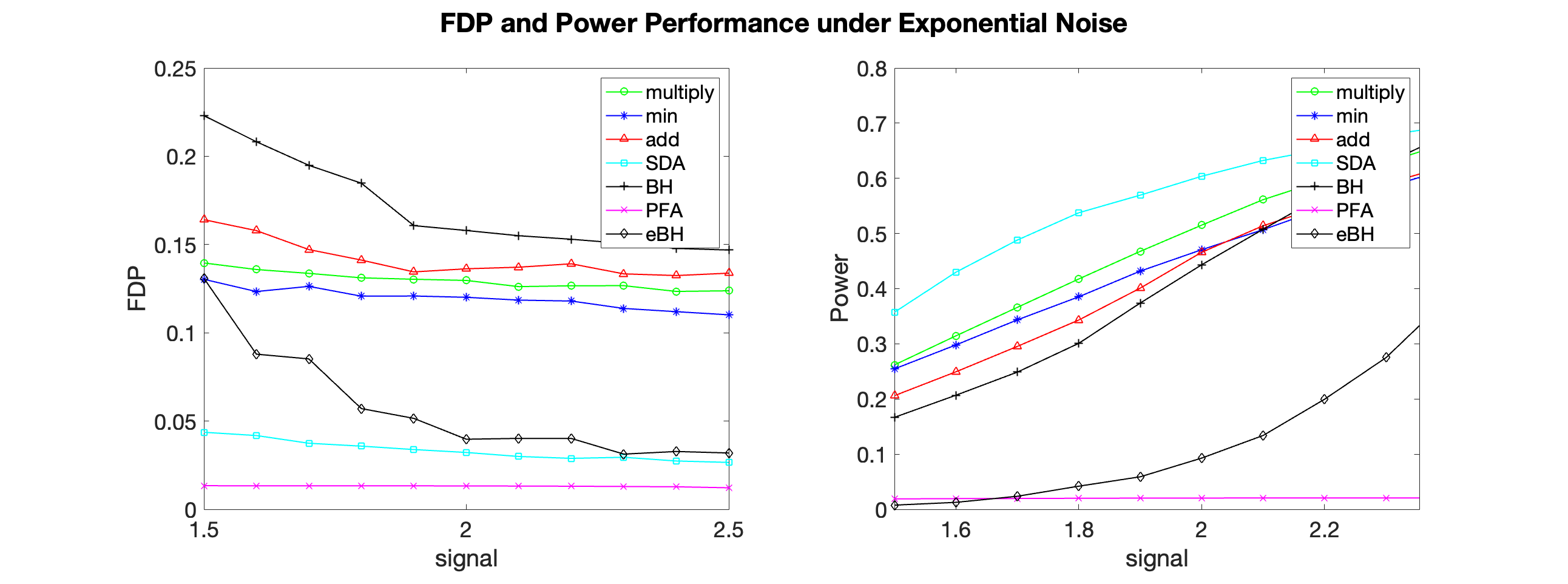}
     \label{fig:fdr-moderate-2}
 \end{subfigure}
  \begin{subfigure}{0.8\textwidth}
     \includegraphics[width=\textwidth]{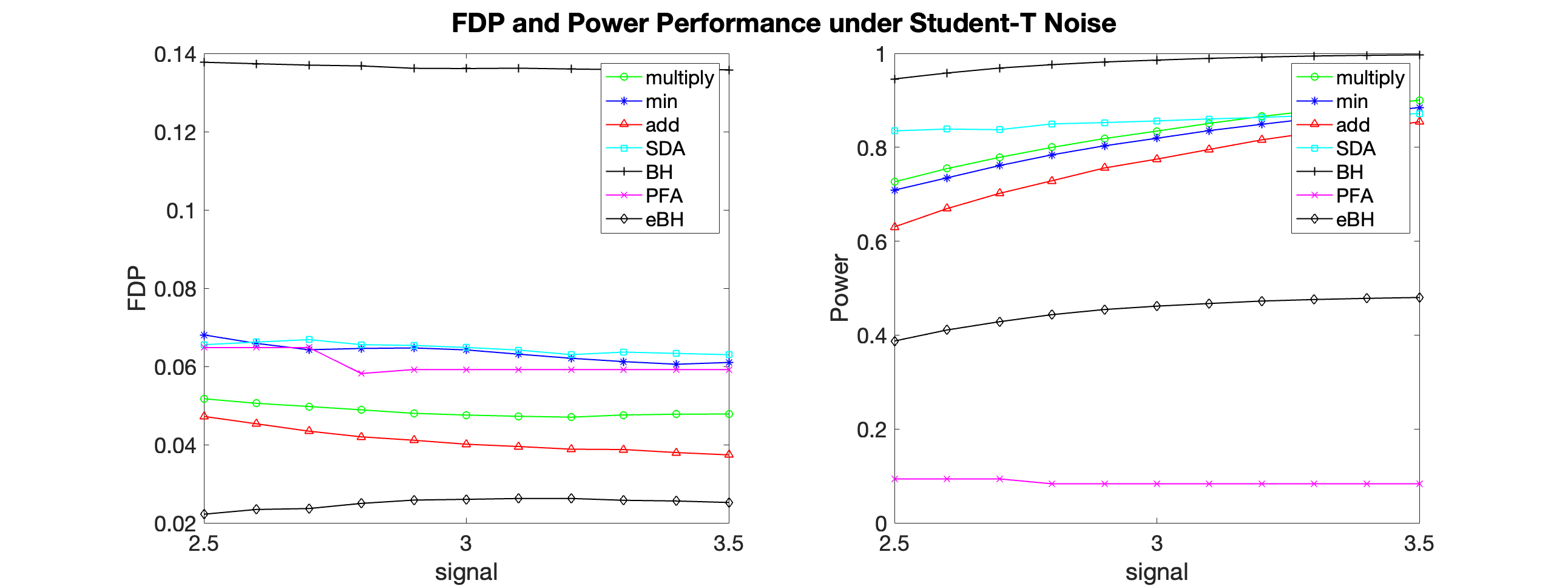}
     \label{fig:fdr-moderate-3}
 \end{subfigure}
 \caption{FDR control \& Power of different data aggregation schemes for entry comparisons between rows with $\alpha=0.1$ when the noises are heavy-tailed distributed}
 \label{fig:fdr-moderate}
\end{figure}
\begin{figure}
\centering
\begin{subfigure}{0.8\textwidth}
     \includegraphics[width=\textwidth]{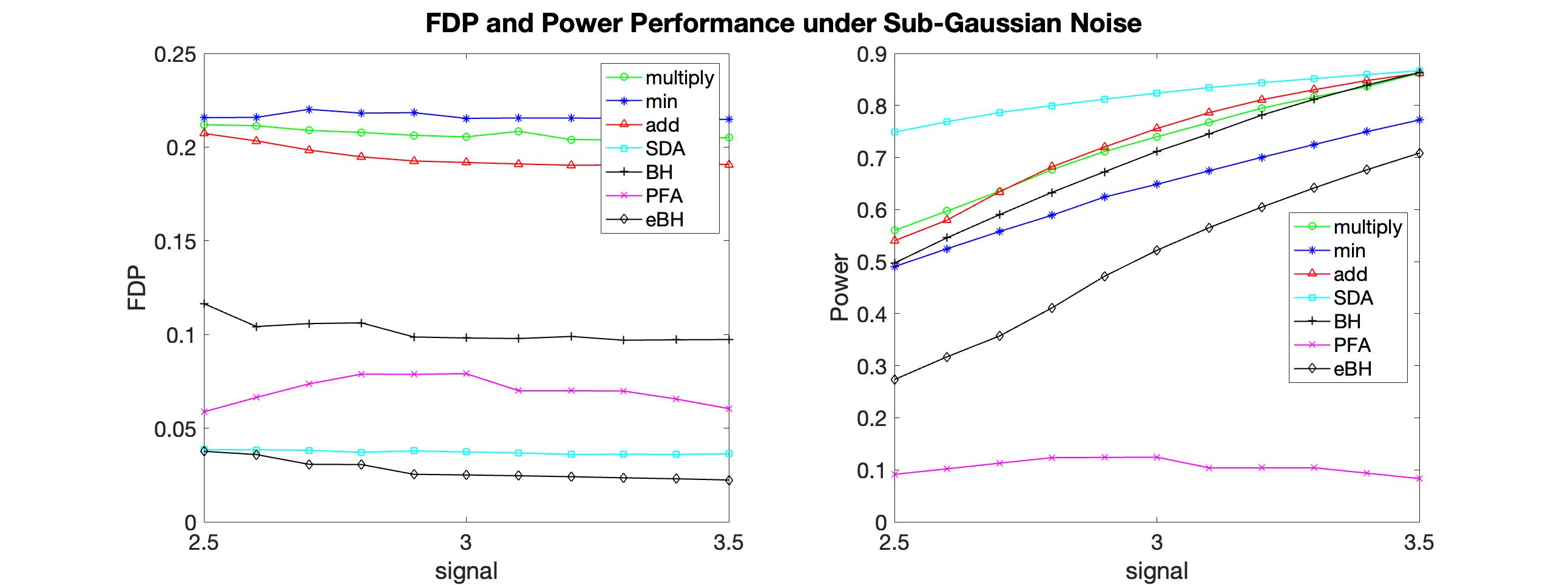}
     \label{fig:fdr-strong-1}
 \end{subfigure}
 \begin{subfigure}{0.8\textwidth}
     \includegraphics[width=\textwidth]{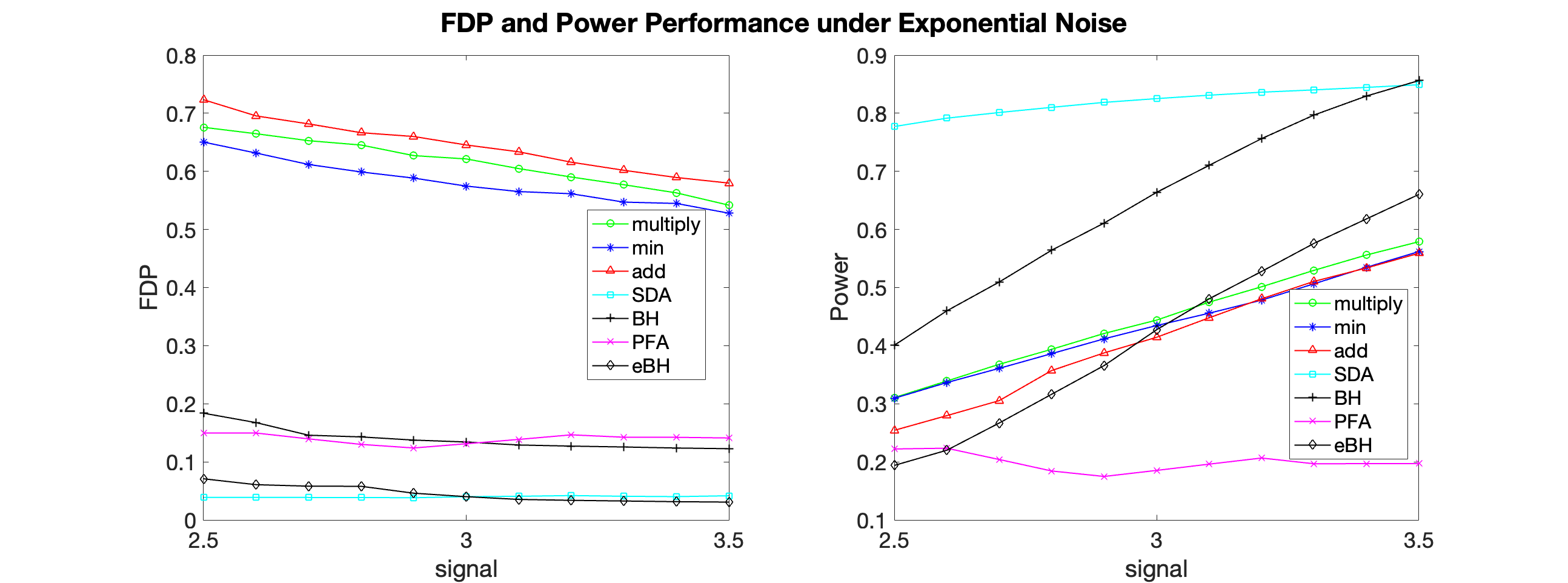}
     \label{fig:fdr-strong-2}
 \end{subfigure}
  \begin{subfigure}{0.8\textwidth}
     \includegraphics[width=\textwidth]{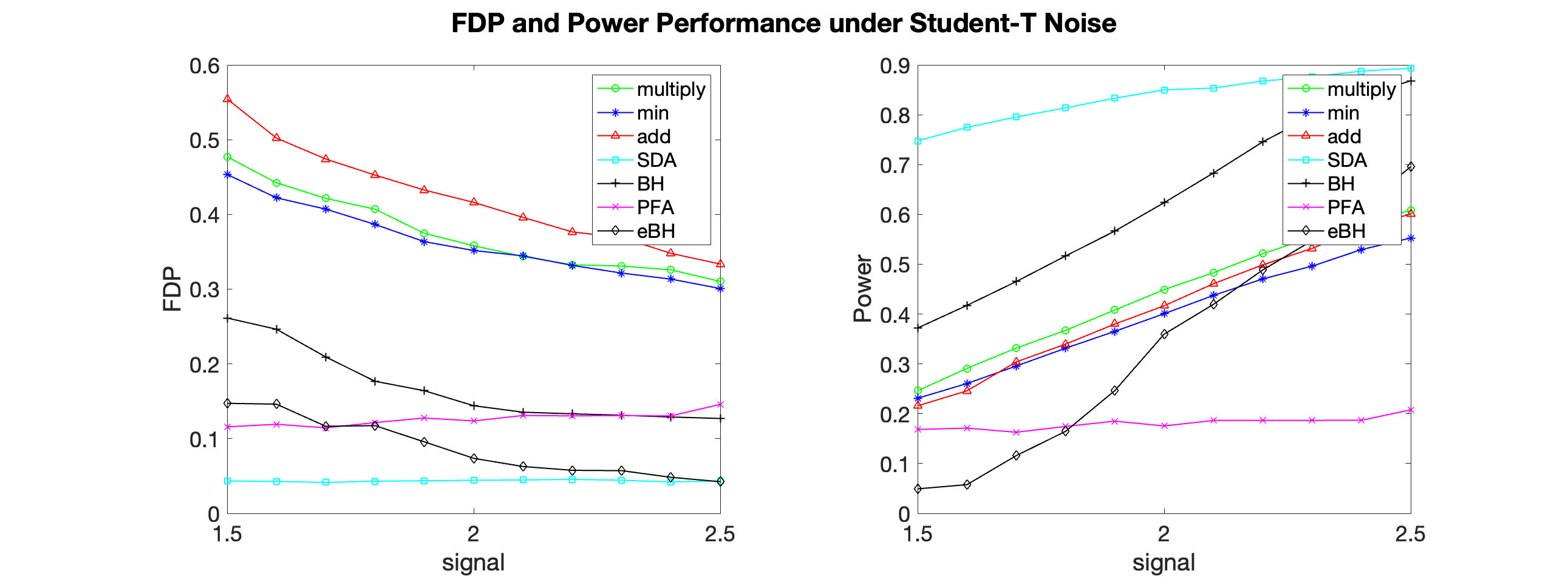}
     \label{fig:fdr-strong-3}
 \end{subfigure}
 \caption{FDR control \& Power of different data aggregation schemes for entry comparisons within a block with $\alpha=0.1$ when the noises are heavy-tailed distributed}
 \label{fig:fdr-strong}
\end{figure}
We report all the results in Figure \ref{fig:fdr-moderate} and Figure \ref{fig:fdr-strong}. In both moderate and strong correlation cases, BH method and other data aggregation scheme show unstable FDR control, while our proposed SDA method always performs well even under strong correlation. The SDA method is also robust with respect to heavy-tailed noises.

\subsection{Real Data Examples}
\subsubsection{MovieLens}
This section applies our methods to the MovieLens dataset for multiple testing and FDR control. MovieLens \citep{harper2015movielens}, as a commonly used dataset in matrix completion problems, records millions of people’s expressed preferences for movies (rated from 1-5). The dataset can be viewed as a huge, sparse matrix with heavily incomplete observations. MovieLens dataset is broadly used in matrix completion \citep{hastie2015matrix,monti2017geometric,xia2021statistical} and other machine learning tasks. The dataset is available on \url{https://grouplens.org/datasets/movielens/}. To improve data quality and modeling reliability, we removed users with fewer than 20 rated movies, resulting in 100{,}000 ratings on a 0–5 scale from 943 users over 1{,}682 movies; entries with value 0 correspond to unrated (missing) items.
We assume the latent low-rank structure of this user-rating matrix with $r=10$. Approximately, this leads to $\kappa_0=6.4225$ for GD initial estimate. We select $q=1000$ 
adjacent and observed entry pairs, aiming to compare 
\begin{equation*}
    H_{0,ij}: M(i,j)-M(i,j+1)=0 \text{ versus } H_{1,ij}: M(i,j)-M(i,j+1)> 0,
\end{equation*}
for a group of suitable entries $(i,j)$.
 Notice that since in the noisy matrix completion problem, we have the observation $Y(i, j)=M(i, j)+\xi(i, j)$, which means that the ground truth $M(i, j)$ is always unknown, we adopt the process in \cite{xia2021statistical} that treats $\mathbb{I}\left(Y\left(i, j+1\right)>Y\left(i, j\right)\right)$ as a proxy to differentiate $H_1$ from $H_0$. 
 
 We first randomly split data into two parts $\cD_1$, $\cD_2$, and then use gradient descent \citep{wei2016guarantees,chen2020noisy,cai2022generalized} on the two parts for initialization. Then, we run Algorithm \ref{alg:matrix-fdr}, \ref{alg:matrix-sda} with data splitting. Still, we consider 3 types of data aggregation on $\cD_1$, $\cD_2$.
We first verify the symmetric property of our test statistics on MovieLens Data. To that end, we first set our hypotheses $m_{ij}= Y(i,j)- Y(i,j+1)$ and construct asymptotic statistics on $\cD_1$, $\cD_2$ to mimic null test statistics. Here we still use $Y(i,j)- Y(i,j+1)$ as a proxy of $M(i,j)- M(i,j+1)$. The distribution of the corresponding $W_T^{(1)}$, $W_T^{(2)}$ can be found in Figure \ref{fig:symm-null-movielens}, showing clearly the symmetric properties of null hypotheses.

\begin{figure}
\centering
\includegraphics[width=0.7\textwidth]{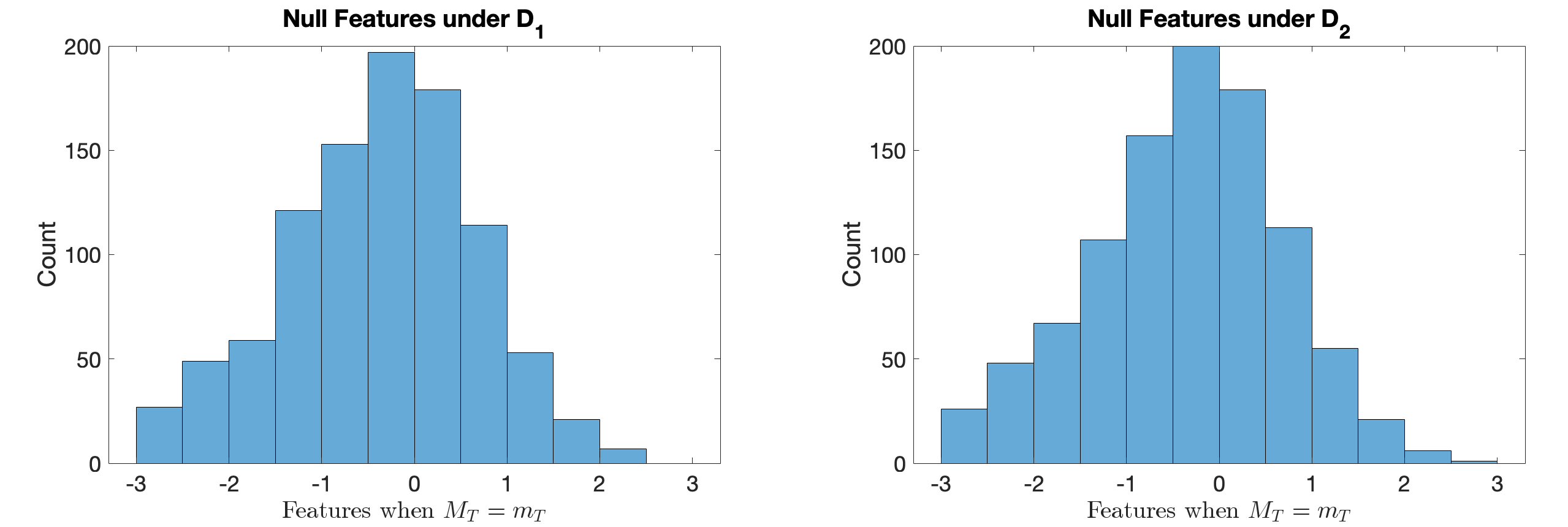}
 \caption{Symmetric distribution of all the test statistics under the null set given $\langle M, T\rangle=m_T$ for all $T$. The test statistics are observed to preserve a good symmetric property both on $\cD_1$ and $\cD_2$.}
 \label{fig:symm-null-movielens}
\end{figure}
We then apply our methods to the entrywise comparison task. Given a total of $q=1000$, the number of instances for $\mathbb{I}\left(Y\left(i, j\right)>Y\left(i, j+1\right)\right)$ is $q_1=262$. We perform the tests for this one-sided hypothesis testing by dropping out hypotheses with negative test statistics on both $\cD_1$ and $\cD_2$. The $p$-values for BH are also adjusted correspondingly. The outcomes are concisely presented in Table \ref{tab:tab1}. The result table clearly shows that the SDA method outperforms other data aggregation methods and the BH method in terms of false discovery rate control. The ineffectiveness of the first three simple data aggregation methods can be attributed to the high correlation of entry pairs, as adjacent entry pairs within a row are selected. When $\alpha$ is significantly small, SDA tends to be more conservative, which leads to good FDR control, while other methods remain to keep large FDRs. The result also shows business implications: instead of excessively recommending movies to users, the SDA can better select target users that are truly interested in the movies to increase the accuracy of the recommendation. By adopting our method for recommendation, the movie company can increase its profit while avoiding losing potential customers.

\begin{table}
	\centering
	\begin{tabular}{lc|cccc}
		\hline\hline
		Level $\alpha$ & Method & False discoveries & True discoveries & FDP  \\ \hline
		$\alpha = 0.01$ & Multiplication & 13  & 59 & 0.1806  \\
		& Minimum & 13 & 58 & 0.1831 \\
		& Addition & 13 & 60 & 0.1781  \\ 
            & SDA & \textbf{0} & {18} & \textbf{0}\\ 
            & BH & 1 & 26 & 0.0370  \\ 
  \hline
  		$\alpha = 0.05$ & Multiplication & 20  & 84 & 0.1923  \\
		& Minimum & 20 & 83 & 0.1942 \\
		& Addition & 20 & 84 & 0.1923  \\ 
            & SDA & \textbf{2} & {25} & \textbf{0.0741}\\ 
            & BH & 10 & 53 & 0.1587  \\ 
  \hline
  		$\alpha = 0.1$ & Multiplication & 24 & 95 & \textbf{0.2017}  \\
		& Minimum & 24 & 94 & 0.2034   \\
		& Addition & 25 & 95 & 0.2083   \\ 
            & SDA & \textbf{8} & {49} & \textbf{0.1404} \\ 
            & BH & 22 & 76 & 0.2245  \\ 
    \hline
    $\alpha = 0.2$ & Multiplication & 33 & 108 & 0.2340  \\
		& Minimum & 33 & 108 & 0.2340   \\
		& Addition & 33 & 108 & 0.2340   \\ 
            & SDA & \textbf{23} & 89 & \textbf{0.2054} \\ 
            & BH & 36 & 115 & 0.2384 \\ 
  \hline \hline
	\end{tabular}
	\caption{Numbers of the discovered entry pairs with FDP by different data aggregation methods under various levels on MovieLens data.}
	\label{tab:tab1}
\end{table}

\subsubsection{Rossmann sales dataset}
We use the Rossmann sales dataset that has recently been studied for uncertainty quantification in matrix completion \citep{farias2022uncertainty,gui2023conformalized}. The Rossmann sales dataset records over 3,000 drug stores run by Rossmann in 7 European countries. The training set contains daily sales of 1115 drug stores on workdays from Jan 1, 2013, to July 31, 2015. The data matrix is thus of dimension $1115\times780$, where two dimensions represent drug stores and workdays, respectively. The unit of sales data is 1K. The dataset is very dense with about $80\%$ valid (non-zero sells) observations of the full matrix; thus, we apply random masking to get sparse observations and use other data only to initialize the algorithm. In this example, we use $20\%$ of the total records as each one split and apply Algorithm \ref{alg:matrix-fdr} on the two splits of the data that are properly processed. Noticing that most observed entries are given, we use the observations as true $M_{ij}$ and perform multiple entrywise tests \eqref{eq:ross}. We select the first $q=20,000$ masked entries sorted by workdays with records in the whole dataset as our target $\cH$.
Since we consider the inference of a submatrix with a relatively large $q$, according to Section \ref{sec:dependence}, the problem can be approximately treated as weakly correlated, which means simple data aggregation is enough to control FDR. We randomly assign null and non-null features by \eqref{eq:construct_H0} but only consider positive signals. In this case, the ratio of non-null is $p=0.3$, and we assume the latent low-rank $r=30$. Specifically, we simultaneously test
\begin{equation}\label{eq:ross}
      H_{0,ij}: M_{ij}=m_{ij}\quad \text{ vs. } \quad H_{1,ij}: M_{ij}> m_{ij}, \quad \forall (i,j)\in \cH,
\end{equation}
for a collection of masked entries $(i,j)\in\cH$. Here $m_{ij}$ serves as a baseline level for entry $(i,j)$. To create a controlled benchmark where ground truth is available, we set $m_{ij}$ based on the observed (pre-masking) values: we take $m_{ij}=M_{ij}$ for $(i,j)\in\cH_0$ and $m_{ij}=M_{ij}-\mu_{\operatorname{signal}}$ for $(i,j)\in\cH_1$, with signal strength $\mu_{\operatorname{signal}}>0$.
We present the results in Figure \ref{fig:fdr-Rossmann} and the ROC curves in Figure \ref{fig:fdr-Rossmann-roc}. The Rossmann sales dataset is available at \url{https://www.kaggle.com/c/rossmann-store-sales}.

\begin{figure}
\centering
\begin{subfigure}{1\textwidth}
     \includegraphics[width=0.9\textwidth]{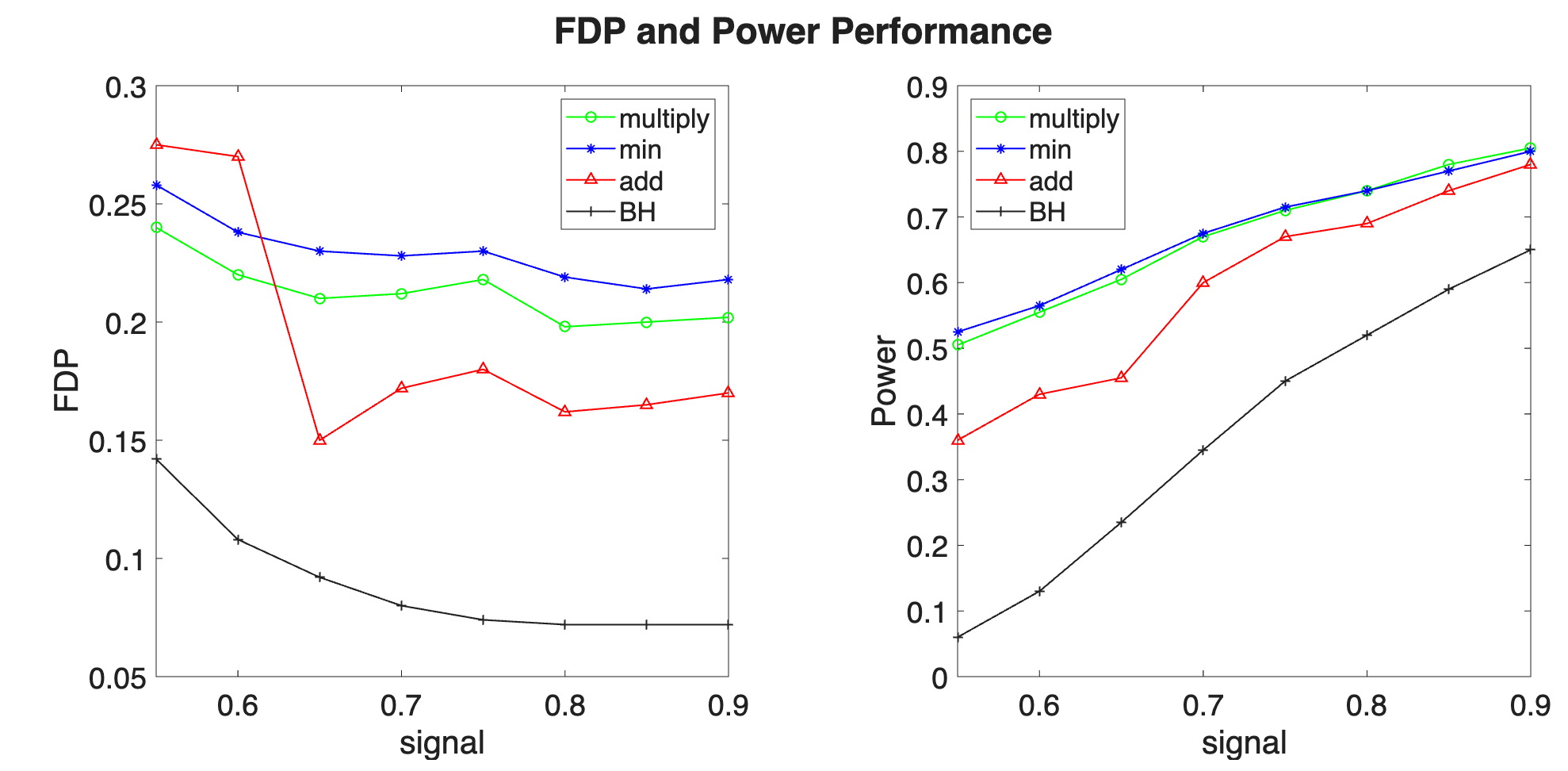}
     \caption{ Empirical FDP and power at  FDR control level $\alpha=0.2$}
     \label{fig:fdr-Rossmann-1}
 \end{subfigure}
 \begin{subfigure}{1\textwidth}
     \includegraphics[width=0.9\textwidth]{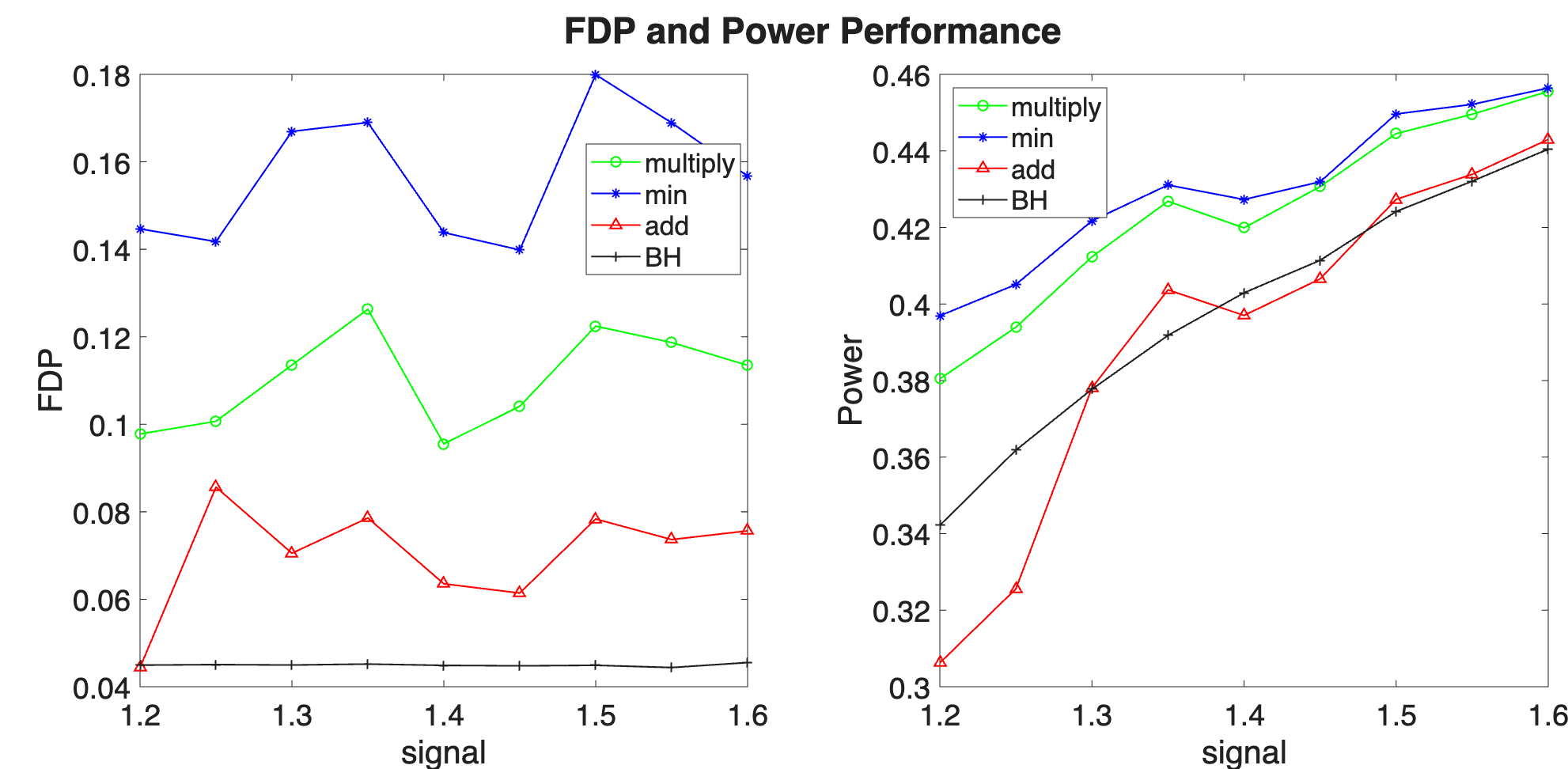}
     \caption{Empirical FDP and power at  FDR control level $\alpha=0.1$}
     \label{fig:fdr-fdr-Rossmann-2}
     
 \end{subfigure}

 \caption{FDR control \& Power of different data aggregation schemes for Rossmann sales testing. Here the signals indicate the sizes of $\mu_{\operatorname{signal}}=\abs{M_{ij}-m_{ij} }$ which are scaled  by $10^3$}
 \label{fig:fdr-Rossmann}

\end{figure}

  \begin{figure}
  \centering
\begin{subfigure}{0.45\textwidth}
         \includegraphics[width=\textwidth]{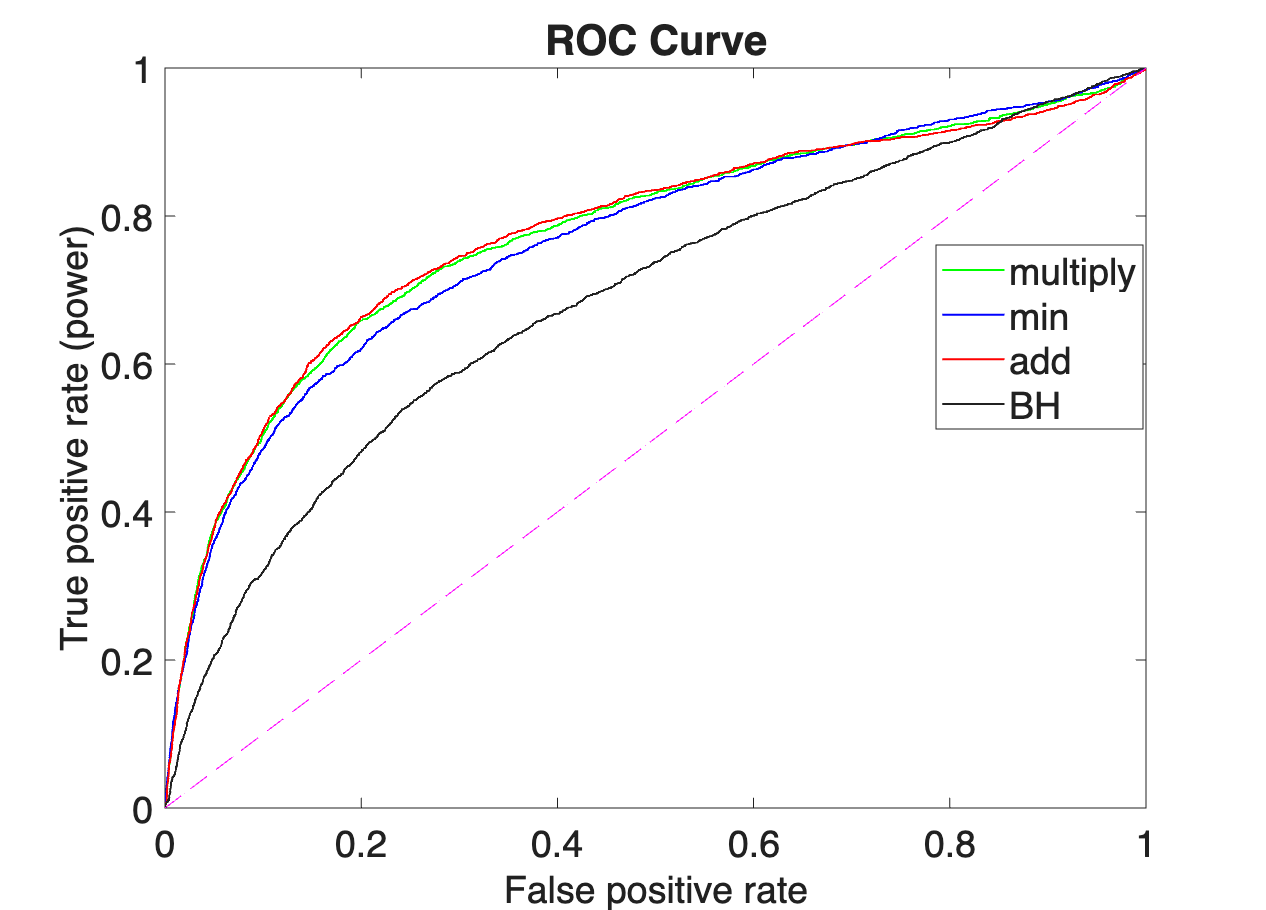}
     \caption{ROC curve with $\text{signal} = 1$}
     \label{fig:fdr-fdr-Rossmann-roc-1}
\end{subfigure}
\begin{subfigure}{0.45\textwidth}
         \includegraphics[width=\textwidth]{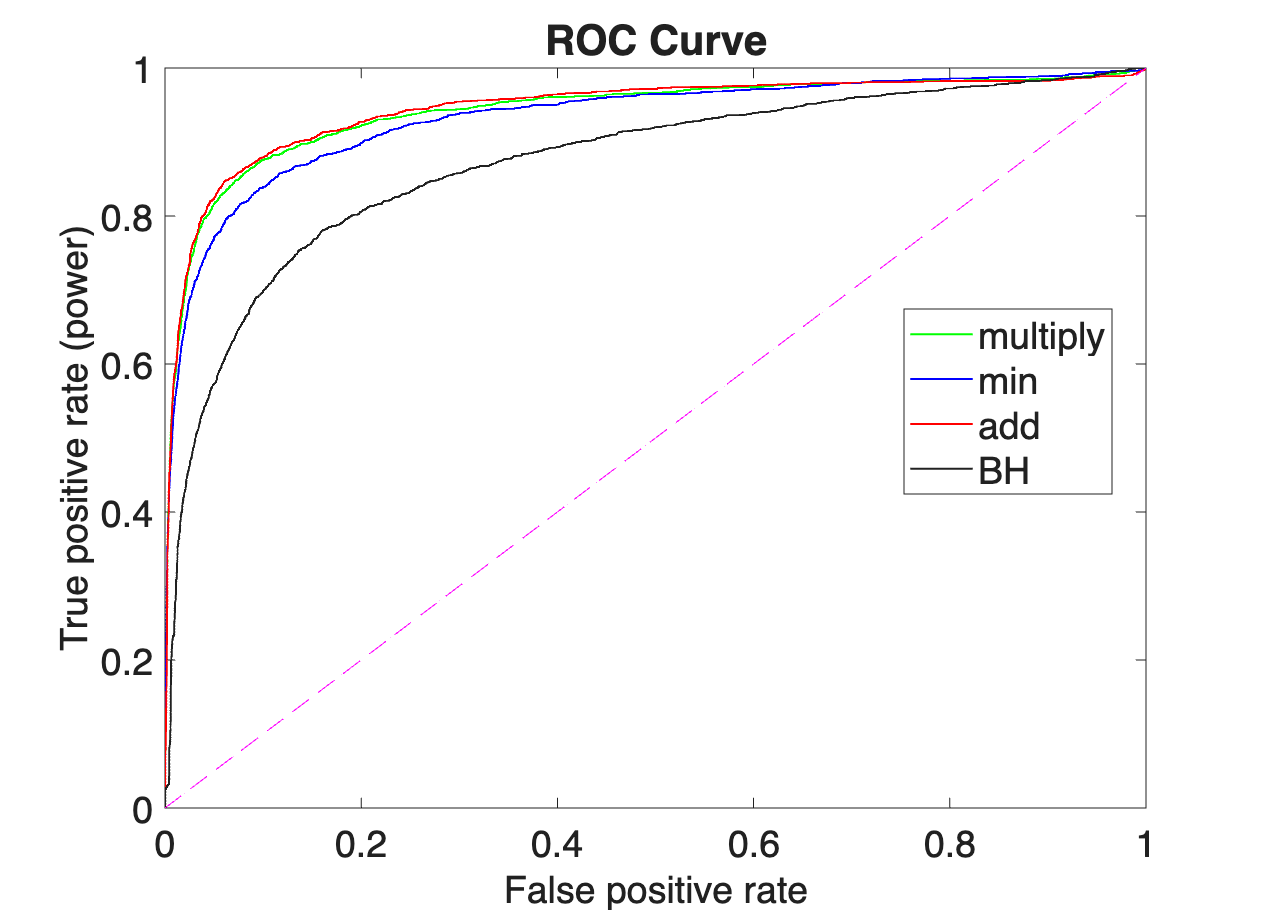}
     \caption{ROC curve with $\text{signal} = 2$}
     \label{fig:fdr-fdr-Rossmann-roc-2}
\end{subfigure}
\caption{ROC curve for different test statistics in Rossmann sales dataset}
     \label{fig:fdr-Rossmann-roc}
     \end{figure}
Here, three different data aggregation methods, together with BH method, are compared. For this one-sided problem, we also drop out features that have negative statistics on $\cD_1$ and $\cD_2$. From Figure \ref{fig:fdr-Rossmann}, it is clear that the data aggregation method with multiplication performs better regarding both FDR control and power. Data aggregation by taking minimum absolute values performs close to our aggregation method with multiplication in power, but it has larger FDPs. Data aggregation by adding absolute values behaves conservatively in the problem. The BH method appears to be more conservative compared to the data aggregation methods, particularly at the FDR control level of $\alpha = 0.2$. Moreover, from the ROC curves in Figure \ref{fig:fdr-Rossmann-roc}, we can observe the obvious advantage of our data aggregation methods against the BH method.

\section{Concluding Remarks}\label{sec:remark}
In this paper, motivated by large-scale recommender systems, we study the problem of multiple testing for linear forms in noisy matrix completion and develop a general framework to control the FDR. Our approach is based upon a new test statistic for testing linear forms that enjoy sharper asymptotics than existing ones in the literature and an effective data splitting and symmetric aggregation scheme that can be shown to be especially suitable in the context of matrix completion.

Our approach can potentially be extended to many other problems with structural high-dimensional features. For example, one possible direction is the FDR control for tensor completion. Indeed, multiple testing in multilinear arrays presents a number of additional technical challenges as it requires much-involved analysis of singular subspace perturbations. As such, inferences in general for low-rank multilinear arrays are largely unexplored.  Additionally, throughout the paper we assume uniform sampling for clarity of exposition and technical simplicity. Our methods and inference framework can be extended to non-uniform sampling when the sampling probabilities are known. For instance, we can construct individual test statistics through inverse propensity weighting (IPW; \cite{duan2025online}). See Section~\ref{sec:hetero}. However, these extensions would substantially complicate the presentation of both the methodology and the theory, especially since establishing the nonasymptotic FDP bounds is already technically involved even under uniform sampling.  Moreover, if the underlying matrix is only approximately rank $r$, we can treat the components beyond rank $r$ as a bias term, which is negligible as long as those components are small in Frobenius norm. We shall leave these intriguing problems for future investigation.

\end{sloppypar}

\bibliography{reference}
\bibliographystyle{apalike}

\newpage

\appendix
\begin{center}
{\bf\LARGE Supplement to ``Multiple Testing of Linear Forms for Noisy Matrix Completion"}
\end{center}
\smallskip
\section{Supplementary Tables and Algorithms}\label{sec:supp-table-alg}

To ease understanding, here we list several important notations frequently encountered while reading our main text and proofs in Table \ref{tab:tab-notations}.
\begin{table}[htbp]
	\centering
	\begin{tabular}{|l|l|}
		\hline\hline
  Notation \ & Meaning \\
  \hline
		$q$, $q_1$, $q_0$ \ & number of all, non-null, and null tests respectively  \\ 
  \hline
  $W_T$,  $W_T^{\sfh}$ \ & asymptotic normal test statistic \\
  \hline
    $W_T^{(i)}$ \ & test statistic of linear form $M_T:=\langle M, T\rangle$ constructed from the $i$th data split \\ 
    \hline
    $W_T^{\mathsf{rank}}$ \ & $W_T^{\mathsf{rank}}=W_{T}^{(1)}W_{T}^{(2)}$ is the  rank statistic before thresholding \\
    \hline
     $\mu$  & parameter for incoherence condition \\ 
     \hline
    $\beta_T$  &  alignment parameter for $T$ defined in \eqref{eq:alignment}\\
     \hline
     $\beta_0$ &  minimum alignment parameter among all $T\in\cH$  \\ 
     \hline
    $\alpha_d$ & dimension ratio of matrix $M$: $\alpha_d:= d_1/d_2$ \\
    \hline
    $\rho_T$& $\rho_T=\norm{T}_{\ell_1}/\norm{T}_{\tF}$ \\
    \hline
    $\kappa_0$ & condition number of matrix $M$ \\
    \hline
    $\gamma_n$  & accuracy of initial estimation $\norm{\widehat{M}^{\mathsf{init}} - M}_{\max} \le  \sigma_\xi \gamma_n$, which may take  $\gamma_n= C_{\init}\sqrt{\frac{ d_1 \log  d_1}{n}}$ \\
    \hline 
    $\cP_M(\cdot)$ & projection operators  $\cP_M^{}(T):= T- \cP_M^{\perp}(T)= T- U_{\perp} U^\top_{\perp} T V_{\perp}V^\top_{\perp}$ \\
    \hline
    $s_T$ & standard deviation of testing $M_T$ induced by random sampling: $s_T=\norm{\cP_M(T) }_\tF$ \\
     \hline
    $h_n$  & asymptotic normal rate defined in \eqref{eq:fixed-asymp} \\
    \hline
    $\beta_{\mathsf{s} }$ &  proportion of strongly correlated linear form pairs defined in \eqref{eq:strong-T} \\
    \hline
    $\eta_n$ &  number of strong signals \\
    \hline
    $\kappa_1$ & condition number of covariance matrix $\Sigma=\big(\big<\calP_M(T_j), \calP_M(T_k)\big>\big)_{1\leq j,k\leq q}$ \\
    \hline
    $\kappa_T$ & shrinkage of variances caused by low-rank projection $\kappa_T = \norm{T_{\cH} }/\norm{\Sigma}^{1/2}$ \\
    \hline
    $\kappa_\infty$ & maximum row-wise $\ell_1$-norm of inverse correlation matrix: $\kappa_\infty=\norm{R^{-1}}_\infty:=\max_{i}\|e_i^{\top}R\|_{\ell_1}$ \\
    \hline
    $q_n$, $q_{0n}$ & cardinality of support after screening $q_n=\abs{\cA}$, and  $q_{0n}=\abs{\cA\cap \cH_0 }$ \\
    \hline
    $\beta_{\mathsf{s} }'$,  $\eta_n'$&  proportion $\beta_{\mathsf{s} }$ and number of strong signals after screening\\

  \hline \hline
	\end{tabular}
	\caption{Important notations used in the main text
 }
	\label{tab:tab-notations}
\end{table}

We also summarize the assumptions required for our major theorems in the following Table \ref{tab:table-assumption}:
\begin{table}[]
    \centering
    \begin{tabular}{|c|c|}
           \hline
       Theorem  & Assumption(s)  \\
       \hline
       CLT in Theorem \ref{thm:asymp-normal}, \ref{thm:asymp-normal-varest}   & sub-gaussian $\xi$, incoherence, SNR and sample size  \eqref{eq:condition-SNR-sample} \\
       \hline
       CLT under heterogeneity in Proposition \ref{prop:hetero} & sub-gaussian $\xi$, incoherence, SNR and sample size   \eqref{eq:condition-SNR-sample-hetero}
 \\
\hline
     FDR under weak dependence in Theorem \ref{thm:weak-cor-fdr}
   &  $h_n\to 0$ and
$\left(\sqrt{\beta_{\mathsf{s}}} \vee  h_n \right)  \frac{q_0}{\eta_n}   \to 0$. \\
       \hline
       FDR under strong dependence in Theorem \ref{thm:matrix-fdr-strong} & 
       $
\left(\sqrt{\beta_{\mathsf{s}}'} \vee  \left( h_n +\big\|\sfw_{\calA^c}\big\|_{\infty}\right)\right)  \frac{q_0'}{\eta_n'}   \overset{p}{\rightarrow} 0$
, where $\abs{\calA_0}\le q_{0}'$, \eqref{eq:SNR-strong-dep}
\\
\hline
 
    \end{tabular}
    \caption{A table summarizing assumptions required for major results}
    \label{tab:table-assumption}
\end{table}

In our Algorithm \ref{alg:matrix-fdr}, gradient descent is applied to obtain an initial estimate for our debiasing purpose. Here, we detail the gradient descent procedure for matrix completion in the following Algorithm \ref{alg:GD}.

\begin{algorithm}[H]
\caption{Gradient Descent for Matrix Completion (\cite{chen2019inference,chen2020noisy})}
\label{alg:GD}
\begin{algorithmic}[1]
\REQUIRE $\{(X_i,Y_i)\}_{i=1}^n$, rank $r$, step size $\eta\asymp1/(\kappa^3d_1\lambda_{\max})$, regularization parameter $\lambda\asymp \sigma_{\xi}\sqrt{n/d_2}$, $t_0$
\STATE \textbf{Initialization:} Obtain initial parameters
$U^0 \in \mathbb{R}^{d_1\times r}$,
$V^0 \in \mathbb{R}^{d_2\times r}$ by spectral SVD initialization
\FOR{$t = 0,1,\dots,t_0-1$}
    \STATE Compute residuals $r_i^t := \langle X_i, U^t (V^t)^\top\rangle - Y_i$ for $i=1,\dots,n$
    \STATE Compute gradient w.r.t.\ $U$:
    \[
    \nabla_U \mathcal{L}(U^t,V^t)
    =
    \frac{1}{n}
    \sum_{i=1}^n
    r_i^t \, X_i V^t
    + \lambda U^t
    \]
    \STATE Update $U$:
    $
    U^{t+1}
    =
    U^t
    -
    \eta \, \nabla_U \mathcal{L}(U^t,V^t)
    $
    \STATE Compute gradient w.r.t.\ $V$:
    \[
    \nabla_V \mathcal{L}(U^t,V^t)
    =
    \frac{1}{n}
    \sum_{i=1}^n
    r_i^t \, X_i^\top U^t
    + \lambda V^t
    \]
    \STATE Update $V$:
    $
    V^{t+1}
    =
    V^t
    -
    \eta \, \nabla_V \mathcal{L}(U^t,V^t)
    $
\ENDFOR
\STATE \textbf{Output:} $\widehat{M} =  U^{t_0} ( V^{t_0})^{\top} $
\end{algorithmic}
\end{algorithm}

For completeness, we describe the BH selection for linear forms used in the simulation as follows:
\begin{enumerate}
    \item Use $\cD=\{\cD_1,\cD_2\}$ to construct an initial estimate $\widehat{M}_{\mathsf{init}}$ 
    \item Following the construction of $W_T$, but use both $\cD=\{\cD_1,\cD_2\}$ to de-bias $\widehat{M}_{\mathsf{init}}$
    \item Project the debiased matrix on the low-rank structure and get test statistics $W_T^{\mathsf{all} }$ for each linear form $T$
    \item Computing two-sided $p$-value $P_i=2(1-\Phi(\abs{W_{T_i}^{\mathsf{all}} }))$
    \item Feature selection by BH method: finding the largest $k$ such that $P_{(k)} \leq \frac{k}{q} \alpha$, and rejecting null hypothesis $H_{0,T_i }$ with  $P_i\le P_{(k)}$
\end{enumerate}
This BH selection relies on the asymptotic normality of high-dimensional features and serves as a counterpart to our methods. 

\section{Additional Results}\label{sec:add-res}



\subsection{Covariance Matrix, Effect of Screening, and Whitening}
Given $T_{\cH}\in \R^{q\times d_1 d_2} $, we have the unnormalized covariance matrix for $W_T$ as
\begin{equation*}
    \Sigma:= \big(\big<\calP_M(T_j),  \calP_M(T_k)\big>\big)_{1\leq j,k\leq q}=T_{\cH}(I_{d_1 d_2} - U_\perp U_\perp^\top \otimes V_\perp V_\perp^\top  ) T_{\cH}^\top.
\end{equation*}
Here, we have $\operatorname{rank}(I_{d_1 d_2} - U_\perp U_\perp^\top \otimes V_\perp V_\perp^\top  )=r(d_1+d_2)-r^2$, and $T_{\cH}$ is of rank $q$. Therefore, to make sure that $ \Sigma$ is of full rank, we must have $q\le r(d_1+d_2)-r^2$.

Based on this covariance matrix representation, we shall discuss an example of testing about a submatrix of $M$ to further illustrate the effect of screening and whitening. In particular, we shall show how whitening can weaken the dependence in $Q^{\ast}$, compared with the un-whitened $R_{\calA,\calA}$, where
$$
Q^{\ast}:=\big(R_{\calA}^{-1/2\top} R_{\calA}^{-1/2}\big)^{-1}=R_{\calA,\calA} - R_{\calA,\calA^c}R_{\calA^c,\calA^c}^{-1}R_{\calA,\calA^c}^\top.
$$
To this end, we define the total test matrix $T_{\calH}=[P_{d\times d},\boldsymbol{0}_{d\times(d^2-d) }]$, where we set $d_1=d_2=q=d$. Thus, the covariance matrix of our un-standardized test statistics is
\begin{equation*}
\begin{aligned}
        \Sigma = & T_{\calH}\left(I_{d^2} -  U_\perp U_\perp^\top \otimes V_\perp V_\perp^\top\right)T_{\calH}^\top =  P\left(I_d - \left(U_\perp U_\perp^\top\right)_{11}V_\perp V_\perp^\top \right) P^\top \\ 
        & = P \left(VV^\top +u_{11}V_\perp V_\perp^\top \right)P^{\top} = P \left[V, V_\perp\right] \left[\begin{matrix} I_r & 0 \\0 & u_{11} I_{d-r}
        \end{matrix}\right] \left[V^\top; V_\perp^\top \right]P^\top.
\end{aligned}
\end{equation*}
Here $u_{11}= \left(UU^\top\right)_{11}$. Without loss of generality, let $P=I_d$ be a diagonal matrix, i.e., testing multiple entries in the first row. The $q\times q$ covariance matrix  $\Sigma=\big[u_{11}I_q+(1-u_{11})VV^{\top}\big]$, showing that the test statistics under noisy matrix completion are always correlated, due to the low-rank projection. This underscores the difficulties of multiple testing in matrix completion problems. Nevertheless, it is clear from textbook results of Multivariate Statistics that the total variance ${\rm tr}\big(\Sigma_{\calA,\calA}-\Sigma_{\calA,\calA^{\rm c}}\Sigma_{\calA^{\rm c},\calA^{\rm c}}^{-1}\Sigma_{\calA^{\rm c}, \calA}\big)\leq {\rm tr}(\Sigma_{\calA,\calA})$, which is smaller than the total variance of the unscreened statistics ${\rm tr}(\Sigma)$.

A special case of  multiple testing is defined by making 
\begin{equation*}
    P=\left[\begin{matrix} I_{\calA} & B \\0 & I_{\calA^{c} }
        \end{matrix}\right]\left[V, V_\perp\right]^{\top},
\end{equation*}
where for simplicity, we assume $\calA$ is just the index set from the first $\abs{\calA}$ dimensions. This gives us the covariance matrix
\begin{equation*}
\begin{aligned}
        \Sigma = \left[\begin{matrix} \Lambda +u_{11} BB^\top & u_{11}B \\u_{11}B^\top & u_{11} I_{\calA^{c} }
        \end{matrix}\right].
\end{aligned}
\end{equation*}
Here $\Lambda$ is a diagonal matrix of the size $\abs{\calA}\times\abs{\calA}$  with the first $r$ diagonals equal to $1$, and others equal to $u_{11}$.
Obviously, this covariance matrix shows that the test statistics can be highly correlated since $R_{\calA,\calA} = D_{\calA}^{-\frac{1}{2}}\left( \Lambda +u_{11} BB^\top \right)D_{\calA}^{-\frac{1}{2}}  $ contains off-diagonal elements determined by $B$. Here $D_{\calA}$ represents the the first $\abs{\calA}$ diagonal elements of $\Sigma$. However, the screening shows that
\begin{equation*}
\begin{aligned}
        Q^{\ast}&=\big(R_{\calA}^{-1/2\top} R_{\calA}^{-1/2}\big)^{-1}=R_{\calA,\calA} - R_{\calA,\calA^c}R_{\calA^c,\calA^c}^{-1}R_{\calA,\calA^c}^\top \\
        &=  D_{\calA}^{-\frac{1}{2}}\left( \Lambda +u_{11} BB^\top \right)D_{\calA}^{-\frac{1}{2}} -   D_{\calA}^{-\frac{1}{2}}  u_{11}^{\frac{1}{2}} B B^\top  u_{11}^{\frac{1}{2}} D_{\calA}^{-\frac{1}{2}} \\
        & =  D_{\calA}^{-\frac{1}{2}} \Lambda D_{\calA}^{-\frac{1}{2}},
\end{aligned}
\end{equation*}
with no off-diagonal elements. This indicates that our screening and whitening procedure in the noisy matrix completion model can reduce the correlation of test statistics.

\subsection{Non-asymptotic Bounds for FDR and Power}\label{sec:non-saymp-results}
Here, we present a specific non-asymptotic version of our theoretical guarantees.

\subsubsection{Weak dependence}

\begin{Theorem}
	Under the conditions of Theorem \ref{thm:weak-cor-fdr},
    \begin{itemize}
\item[(a)] with probability at least   
$$
1-C_2\varepsilon^{-2}\log(\frac{q_0 }{\alpha \eta_n})\left( \left(\frac{\beta_{\mathsf{s}} q_0^2 }{ \alpha^2\eta_n^2}\right)^{\frac{1}{2}} + \left(\frac{h_n q_0}{\alpha \eta_n} +(\alpha\eta_n q_0)^{-\nu/2} \right)^{\frac{1}{2}}\right)- C_2 h_n,
$$
Algorithm \ref{alg:matrix-fdr} achieves false discovery proportion
\begin{equation}\label{eq:thm-weak-cor-fdr}
    \mathrm{FDP}:=\frac{\sum_{T\in \cH_0  } \bbI(W_T^{\mathsf{rank}} >L ) }{\left( \sum_{T\in \cH  } \bbI(W_T^{\mathsf{rank}} >L ) \right) \vee 1 } \le \alpha(1+\varepsilon),
\end{equation}
for any $\varepsilon\in(0,1)$. 
\item[(b)] with probability at least
$$
1-C_2 \log(\frac{q_0 }{\alpha \eta_n})\left( \left(\frac{\beta_{\mathsf{s}} q_0^2 }{ \alpha^2\eta_n^2}\right)^{\frac{1}{2}} + \left(\frac{h_n q_0}{\alpha \eta_n}+(\alpha\eta_n q_0)^{-\nu/2} \right)^{\frac{1}{2}}\right)- C_2 \varepsilon^{-1} h_n,
$$
Algorithm~\ref{alg:matrix-fdr} can select the strong signals with power 
\begin{equation}\label{eq:thm-weak-cor-power}
   \mathrm{POWER}:= \frac{\sum_{T\in \cH_1  } \bbI(W_T^{\mathsf{rank}} >L ) }{q_1} \ge (1-\varepsilon)\frac{\eta_n}{q_1}.
\end{equation}
\end{itemize}
\end{Theorem}

Note that Part (a) also implies that
\begin{equation}\label{eq:fdr-exp-exact}
     \text{FDR}=\E (\text{FDP})\le \alpha+ C_2h_n+C_2\alpha^{\frac{2}{3}}\log(\frac{q_0 }{\alpha \eta_n})\left( \left(\frac{\beta_{\mathsf{s}} q_0^2 }{ \alpha^2\eta_n^2}\right)^{\frac{1}{6}} + \left(\frac{h_n q_0}{\alpha \eta_n}\right)^{\frac{1}{6}} +\left(\alpha\eta_n q_0\right)^{-\frac{\nu}{12}}  \right).
\end{equation}

\subsubsection{Whitening and screening}

\begin{Theorem}
Under the settings of Theorem \ref{thm:matrix-fdr-strong}, suppose that
\begin{equation*}
    \left(\norm{R^{-1}  }_\infty +\kappa_1 \frac{\norm{T_{\calH} } }{\norm{ \Sigma }^{1/2}}\left( \frac{\operatorname{supp}(T_{\calH} )}{\sqrt{d_2}}\wedge 1 \right)  \right)\frac{\beta_T \mu \sigma_\xi }{\beta_0 \lambda_{\min} }\sqrt{\frac{\kappa_1\alpha_d q d_1^2 d_2 \log d_1 }{n}}=o(1).
\end{equation*}
With the regularization level $\lambda=C\sqrt{\log d_1}$, the Algorithm \ref{alg:matrix-sda} attains an FDP 
\begin{equation*}
     \mathrm{FDP}=\frac{\sum_{T\in \cH_0  } \bbI(\sfw_{T}^{\mathsf{rank}} >L ) }{\left( \sum_{T\in \cH  } \bbI(\sfw_T^{\mathsf{rank}} >L ) \right) \vee 1 } \le \alpha(1+\varepsilon),
\end{equation*}
for any $\varepsilon\in(0,1)$ with probability at least 
\begin{equation*}
    1-C_1\varepsilon^{-2}\log(\frac{ q_{0}' }{\alpha \eta_n'})\left( \left(\frac{\beta_{\mathsf{s}}'  q_{0}^{'2} }{ \alpha^2\eta_n^{'2}}\right)^{\frac{1}{2}} + \left(\frac{C_\infty \left( h_n+\norm{ \sfw_{\calA^c} }_{\infty} \right)  q_{0}' }{\alpha \eta_n'} +(\alpha\eta_n'  q_{0}')^{-\nu/2} \right)^{\frac{1}{2}}\right)- C_\infty  \left( h_n+\norm{ \sfw_{\calA^c} }_{\infty} \right),
\end{equation*}
where $C_\infty$ is a constant involving $\wt R$ and $\calA$ only, defined later in Proposition~\ref{prop:OLS-normal}. Moreover, the power is guaranteed to be lower bounded by:
\begin{equation*}
   \mathrm{POWER}= \frac{\sum_{T\in \cH_1  } \bbI(\sfw_T^{\mathsf{rank}} >L ) }{q_1} \ge (1-\varepsilon)\frac{\eta_n' }{q_1},
\end{equation*}
with a probability at least
\begin{equation*}
  1-C_1 \log(\frac{ q_{0}' }{\alpha \eta_n'})\left( \left(\frac{\beta_{\mathsf{s}}'  q_{0}^{'2} }{ \alpha^2\eta_n^{'2}}\right)^{\frac{1}{2}} + \left(\frac{C_\infty  \left( h_n+\norm{ \sfw_{\calA^c} }_{\infty} \right)  q_{0}'}{\alpha \eta_n'}+(\alpha\eta_n'  q_{0}')^{-\nu/2} \right)^{\frac{1}{2}}\right)- C_1 C_\infty  \varepsilon^{-1} \left( h_n+\norm{ \sfw_{\calA^c} }_{\infty} \right).
\end{equation*}
\end{Theorem}

Since we further have 
\begin{equation*}
    \begin{aligned}
        \norm{ \Sigma }^{\frac{1}{2}} \ge \norm{e_i^\top T_{\calH}\left(I_{d_1 d_2} -  U_\perp U_\perp^\top \otimes V_\perp V_\perp^\top\right)  }_2 =\norm{\cP_M(T_i)}_\tF \ge \beta_0\sqrt{\frac{r}{d_1}} \norm{T_i}_\tF,
    \end{aligned}
\end{equation*}
i.e., $ \norm{ \Sigma }^{\frac{1}{2}}  \ge \beta_0\sqrt{\frac{r}{d_1}}\max_{i}\norm{T_i}_\tF= \beta_0\sqrt{\frac{r}{d_1}} \norm{T_{\calH}}_{2,\max} $, and 
\begin{equation*}
    \frac{  \norm{T_{\calH} }  }{ \norm{ \Sigma }^{\frac{1}{2}}  } \left( \frac{\operatorname{supp}(T_{\calH} )}{\sqrt{d_2}}\wedge 1 \right) \le \frac{\sqrt{\alpha_d }}{\sqrt{r}\beta_0}\cdot \frac{  \norm{T_{\calH} } }{ \norm{T_{\calH} }_{2,\max} } \left( \operatorname{supp}(T_{\calH} ) \wedge \sqrt{d_2} \right),
\end{equation*}
we can convert this signal requirement to a stronger but clearer one presented in Theorem \ref{thm:matrix-fdr-strong}. In the subsequent proofs, we shall prove the non-asymptotic versions of Theorem \ref{thm:weak-cor-fdr} and \ref{thm:matrix-fdr-strong}. 

\subsection{Finite-sample Guarantees for Whitening and Screening}
Notice that, in our method of FDR control with whitening and screening, the condition of the correlation structure is defined on the asymptotic correlation matrix $Q^{\ast}:=\big(R_{\calA}^{-1/2\top} R_{\calA}^{-1/2}\big)^{-1}$. However, conditional on  $\cD_1$, the covariance of our test statistics is determined by $\wt \sfw_i^{(2)}$ and is sample-related, which is $Q:=(\wt R_{\calA}^{-1/2\top} \wt R_{\calA}^{-1/2})^{-1}\wt R_{\calA}^{-1/2\top} \wt R^{-1/2} R  \wt R^{-1/2} \wt R^{-1/2}_{\calA} (\wt R_{\calA}^{-1/2\top} \wt R_{\calA}^{-1/2})^{-1}$. The following Proposition \ref{prop:weak-cor-aftscr} will show that, as long as the signal strength of $M$ is strong enough, the estimation of $R$ will be accurate enough such that the data-driven $Q$ is also weakly correlated.
\begin{Proposition}[Finite-sample guarantee of weak correlation after screening]\label{prop:weak-cor-aftscr}
If the matrix signal strength satisfies
\begin{equation*}
    \frac{\kappa_1^{1.5} \norm{T_{\calH} } \sigma_{\xi} }{ \lambda_{\min} \norm{ \Sigma }^{\frac{1}{2}}  } \left( \frac{\operatorname{supp}(T_{\calH} )}{\sqrt{d_2}}\wedge 1 \right)  \cdot \sqrt{\frac{ d_1^2 d_2 \log d_1}{n}} \lesssim \frac{1}{q^{\nu}},
\end{equation*}
then the weak correlation condition also holds for finite-sample covariance matrix $Q$, i.e., $\beta_{\textsf{s}}'$ is defined as the proportion of strongly correlated pairs using $Q$ instead of $Q^{\ast}$.  
\end{Proposition}


\begin{Proposition}[LASSO screening]\label{prop:lasso-scr}
By choosing regularization level $\lambda=C\sqrt{\log d_1 } $, LASSO can recover the signal with precision 
    \begin{equation*}
    \abs{\wt\sfw^{(1)}_i -\sfw_i} \le C \kappa_1^{1.5} \sqrt{q_1 \log d_1 }+ h_n \abs{\sfw_i },
\end{equation*}
uniformly for all $i\in[q]$ with probability at least $1-C d_1^{-2}\log d_1 $ for some universal constant $C>0$, as long as the sample requirement of SDA holds.
Moreover, under this condition, if $T_i\in \cS'$, then LASSO can surely select feature $i$. 
\end{Proposition}

 In our method, LASSO is used for pre-selection.
 In fact, we always deliberately choose a weak regularization level so that most true signals and many false positives are included in $\cA$, at the cost of power loss. Here, we do not require the sure-screening condition of LASSO that is commonly used in \cite{roeder2009genome,barber2019knockoff,du2021false,dai2023scale}. We emphasize that our theory can hold with non-identified signals as long as $\norm{ \sfw_{\calA^c} }_{\infty}$ is small enough.

We exploit the symmetricity of $\wt\sfw^{(2)}$ obtained by linear regression after LASSO. This symmetricity, described in the following Proposition \ref{prop:OLS-normal}, serves as a counterpart of Theorem \ref{thm:asymp-normal} in the weakly correlated case.
\begin{Proposition}[Linear regression after screening]\label{prop:OLS-normal} Suppose $T_i\in \cA\cap\cH_0$. Denote an upper bound of variance shrinkage effect of screening on $\cA$ as 
$$C_\infty: = \sup_{i\in \cA }  \frac{ 1 \vee \norm{e_i^\top \left(\wt R_{\mathcal{A}}^{-1/2\top} \wt R_{\mathcal{A}}^{-1/2}\right)^{-1} \wt R_{\mathcal{A}}^{-1/2\top}\wt R^{-1/2}_{\mathcal{A}^c} }_{\ell_1} }{ \sqrt{Q_{ii}}}.$$
Here, we slightly abuse the notation by treating $\cA$ as an index set of numbers. Conditional on $\cD_1$, we have
\begin{equation*}
   \abs{ \PP\left( \frac{\wt\sfw^{(2)}_i}{\sqrt{Q_{ii}} }\le t \middle|  \cD_1  \right)-\Phi (t)}\le C \cdot C_\infty \left(h_n+\norm{ \sfw_{\calA^c} }_{\infty}\right).
\end{equation*}
\end{Proposition}

Here, $C_\infty $ can be viewed as a special kind of coherence condition that has been broadly used in LASSO selection \citep{donoho2001uncertainty,zhao2006model,wainwright2009sharp}. In this propostion, $\norm{ \sfw_{\calA^c} }_{\infty}$ measures the error caused by inconsistent screening.

\subsection{An Equivalent Version of Algorithm \ref{alg:matrix-sda}}
Note that Algorithm \ref{alg:matrix-sda} involves the computation of the correlation coefficient matrix. To ease the analysis, we rewrite Algorithm \ref{alg:matrix-sda} as the following version that avoids computing the inverse of diagonal elements. Notice that, this is just a change of notation for mathematical analysis, rather than a new algorithm.
\begin{algorithm}[H]
	\caption{Matrix FDP Control with Whitening and Screening}
	\label{alg:matrix-sda-practical}
	\begin{algorithmic}[1]
		\REQUIRE Hypotheses $\left\{H_{0T}: \langle M, T\rangle=\theta_T , T\in \cH\right\}$, data splits $\cD_1$, $\cD_2$, rank $r$, FDR level $\alpha$.
		\STATE{ Use $\cD_0$ to construct an initial estimate $\widehat{M}_{\mathsf{init}}$}
		\STATE{ Apply proposed asymptotic test statistics to the second part of data $\cD_1$ and the third part of data $\cD_2$ respectively to get un-normalized test statistics $\mathbf{W}^{(1)}$ and $\mathbf{W}^{(2)}$, where
			\begin{equation*}
				\mathbf{W}_i^{(k)}= \widehat{s}^{(k)}_{T_i} W_{T_i}^{(k)}= \frac{\widehat{M}^{(k)}_{T_i} - \theta_{T_i}}{ \widehat{\sigma}^{(k)}_{\xi} \sqrt{d_1 d_2/n} }, \ k=1,2, \text{ and }  \widehat{D}=\operatorname{diag}\left(\widehat{s}_{T_1}^{(1)},\dots,\widehat{s}_{T_p}^{(1)} \right).
			\end{equation*}
   Here $\wt s_{T_i}^{(k)}=\big\|\calP_{\wt M^{(k)}_{\init}}(T_i)\big\|_{\rm F}$ is an estimate of $s_{T_i}=\big\|\calP_M(T_i)\big\|_{\rm F}$. 
		}
		\STATE { Obtain a covariance matrix estimate $\widehat{\Sigma }$ by $\widehat{U}^{\init}$, $\widehat{V}^{\init}$ from $\cD_0$ and $\cD_1$:
			\begin{equation}
				\widehat{\Sigma }= T_{\calH}(I_{d_1 d_2} - \widehat U_\perp \widehat U_\perp^\top \otimes \widehat V_\perp \widehat V_\perp^\top  ) T_{\calH}^\top, 
			\end{equation}
			and write $\mathbf{X}= \widehat{\Sigma }^{-\frac{1}{2}}$.  Construct response $\mathbf{y}_1 = \mathbf{X}\mathbf{W}^{(1)}  $, and solve LASSO
			
			\begin{equation*}
				\wt\sfw^{(1)} = \argmin_{\sfw\in \R^q } \left\{ \frac{1}{2}\norm{\mathbf{y}_1 - \mathbf{X}\widehat{D} \sfw }^2 + \lambda\norm{\sfw}_{\ell_1} \right\}.
			\end{equation*}
			
		}
		\STATE {Denote $\cA$ as the support set of LASSO solution $\wt\sfw^{(1)}$. We then have the separation $\mathbf{X} = \left[ \mathbf{X}_{\cA}, \mathbf{X}_{\cA^c}  \right]$. We run linear regression on $\cA$ with new loading matrix $\mathbf{X}_{\cA} $ and response $\mathbf{y}_2 = \mathbf{X}\mathbf{W}^{(2)}$ from $\cD_2$ to get  asymptotic symmetric statistics $\wt\sfw^{(2)}$, where
			
			\begin{equation*}
				\wt\sfw^{(2)}_i=\left\{\begin{array}{cc}
					e_i^{\top}\left(\mathbf{X}_{\mathcal{A}}^{\top} \mathbf{X}_{\mathcal{A}}\right)^{-1} \mathbf{X}_{\mathcal{A}}^{\top} \mathbf{y}_2, & i \in \mathcal{A} \\
					0, & i \in \mathcal{A}^c
				\end{array}\right.
			\end{equation*}
			with variance estimate $\widehat{\sigma}_{w i}^2= e_i^{\top}\left(\mathbf{X}_{\mathcal{A}}^{\top} \mathbf{X}_{\mathcal{A}}\right)^{-1} e_i $.
		}
		\STATE { Compute the final ranking statistics of each $T_i$ by $\wt\sfw_{T_i}^{\mathsf{rank}}=\wt\sfw_{i}^{(1)}\wt\sfw_{i}^{(2)}/\widehat{\sigma}_{w i}$, and then choose a data-driven threshold $L$ by 
			\begin{equation*}
				L=\inf \left\{t>0: \frac{\sum_{T\in\cH}\bbI\left(\wt\sfw_T^{\mathsf{rank}}<-t\right)}{\sum_{T\in\cH} \bbI \left(\wt\sfw_T^{\mathsf{rank}}>t\right) \vee 1} \leq \alpha\right\}.
		\end{equation*}}
		\STATE{Reject $H_{0T_i}$ if $\wt\sfw_{T_i}^{\mathsf{rank}}>L$}
	\end{algorithmic}
\end{algorithm}
It can be easily checked that $\wt\sfw^{(1)}_i$ and $\wt\sfw^{(2)}_i/\widehat{\sigma}_{w i}$ share the same representation as in Algorithm \ref{alg:matrix-sda}, and Algorithm \ref{alg:matrix-sda}, \ref{alg:matrix-sda-practical} are essentially identical. For brevity of notations, our following proofs (presented in Sections~\ref{sec:proof-aftscr}-\ref{sec:proof-fdr-strong}) of theories in Section~\ref{sec:strong-corr} will be based on the quantities and notations in Algorithm \ref{alg:matrix-sda-practical} rather than that in Algorithm \ref{alg:matrix-sda}.

\subsection{Comparison of Data Aggregation Methods } \label{sec:compare}
The empirical success of data splitting in multiple testing leads to the problem of how to choose data aggregation methods for split data and what the theoretical explanations are behind them. In this section, we probe into the power behavior of different data aggregation methods to answer this question. Indeed, existing literature have scarcely compared the power of different FDR control procedures. Here we list some notable attempts: \cite{genovese2006false} found that the $p$-value weighting can improve the power compared with the original BH method; \cite{scott2015false} showed simulation evidence that FDR regression improves the power upon traditional FDR control methods; for knockoff procedure,  \cite{liu2019power,weinstein2020power} focused on explaining the power behavior of knockoff under special designs. 
However, all these attempts have been unsuccessful in transferring to the case of data aggregation methods and in comparing the power enhancement achieved through data splitting.
We compare our methods with other combination schemes in a simple mean-testing problem. Actually, several data aggregation methods have been proposed in \cite{dai2022false} by the so-called ``mirror statistic" design. Namely, for any dimension $i\in [q]$, we derive two independent test statistics $X^1_i$, $X^2_i$ from two groups of data. Then we combine each pair of $X^1_i$, $X^2_i$ by 
\begin{equation}\label{eq:mirror}
	W(X^1_i,X^2_i)=\operatorname{sign}(X^1_i X^2_i )f(\abs{X^1_i},\abs{X^2_i} ).
\end{equation}
Possible candidates of $f(u,v)$ are 
\begin{equation}\label{eq:mirror-choice}
	f_1(u, v)=u v, \quad f_2(u, v)= \min (u, v),\quad f_3(u, v)=u+v.
\end{equation}

Among these choices, $f_2$ and $f_3$ have been discussed in \cite{xing2021controlling,dai2023scale,dai2022false} and $f_3$ is said to be nearly optimal with respect to power under certain conditions \citep{dai2022false}. Our method can be viewed as a special kind of mirror statistic design by choosing $f_1(u, v)=u v$. This amounts to computing the Hadamard product
of two test statistic vectors $X^1$, $X^2\in \R^{p}$. 
Interestingly, in practice, it is sometimes observed that $f_1$ can outperform other methods; see \cite{dai2023scale,du2021false} for examples. Here, we explain this empirical finding from a Bayesian perspective by mixture model. Consider the multiple testing problem that we observe $q$-dimensional vector $X$ sampled from the model 
\begin{equation}\label{eq:toy-model}
	\begin{aligned}
		X=\boldsymbol{\delta} + \boldsymbol{\xi},
	\end{aligned}
\end{equation} 
where noise $\boldsymbol{\xi}$ is independent multivariate gaussian with $\Sigma=I_q$ (or weakly dependent). The signals $\boldsymbol{\delta}$ are sparse and independent from an unknown non-zero prior $\mathbf{\Theta}$ in the sense that in each dimension $i\in [q]$, $\delta_i=0 $ or $\delta_i\sim \mathbf{\Theta} $, and $\pi_1:= \#\{\mu_i\sim \mathbf{\Theta}  \}/q \to 0$.  Our tests are 
$$\cH_{0i}: \delta_i=0 \text{ versus } \cH_{1i}: \delta_i\neq 0, \text{ for every } i\in[q].  $$

To examine the impact of data aggregation, suppose two observations $X^1$, $X^2$ are given, and we aim to control the FDR by data aggregation in \eqref{eq:mirror} with a certain threshold $L_\alpha$. When $q \to \infty$, the performance of this data-splitting-based method can actually be explained by a mixture model. Consider a prior $H_0: \delta=0$, and $H_1: \delta\sim \mathbf{\Theta} $, with $\PP(H_0)=1-\pi_1, \ \PP(H_1)=\pi_1$ and a variable $Y$ with mixture distribution $Y|H_0\sim W(\xi_1,\xi_2)$, and $Y|H_1\sim W(\delta+\xi_1,\delta+\xi_2)$. Here $\xi_1, \xi_2$ are independent standard normal variables. When   all the dimensions of $X$ are independent or weakly correlated, we have 
\begin{equation*}
	\frac{\#\left\{i: W_i>t\right\}}{q} \to \PP(Y>t) 
\end{equation*}
uniformly for any $t$. 
The limiting behavior of data aggregation method $W$ given any threshold $L$ is summarized as follows:
\begin{equation}\label{eq:limiting_fdr}
	\begin{aligned}
		\text{FDR}_W (L) &= \frac{\PP(Y>L,H_0)}{ \PP(Y>L) } = \frac{(1-\pi_1)\PP(Y>L|H_0)}{(1-\pi_1)\PP(Y>L|H_0) + \pi_1 \PP(Y>L|H_1) },  \\
		\text{Power}_W(L) & = \PP(Y>L|H_1),
	\end{aligned}
\end{equation}
where the limiting power is the expectation with respect to $ \mathbf{\Theta}  $: $\PP(Y>L|H_1) = \E_{\mathbf{\Theta}} \PP(Y>L|\delta,H_1)$. Suppose we can specify a  threshold $L_\alpha$ that controls the limiting FDR at exact level $\alpha$, i.e.,
\begin{equation}\label{eq:limiting_thres}
	L_\alpha =\min \left\{L>0: \text{FDR}_W ( L) = \alpha\right\}, 
\end{equation}
where $L_\alpha$ is determined by both FDR level $\alpha$ and aggregation function $W$. Then, at the same FDR level $\alpha$, the power of different data aggregation methods is only decided by the mixture distribution $Y$ induced by aggregation function $W$. To compare the limiting power of different aggregation method $W_j(u,v) = \operatorname{sign}(u v )f_j(\abs{u},\abs{v} )$, $j=1,2,3$, we denote $L_{\alpha j} $ as the corresponding threshold by \eqref{eq:limiting_thres}. It suffices to compare $\operatorname{Power}_{W_j}(L_{\alpha j} )$. This is equivalent to comparing the quantities $\operatorname{Power}_{W_j}(L_{p j} )$ where $L_{p j}$ is the $p$-th quantile of  null distribution $Y_j|H_0\sim W_j(\xi_1,\xi_2)$. The rationale is as follows. For the same quantile $p$, if the $\operatorname{Power}_{W_j}(L_{p j} )$ is larger, then in order to achieve the same FDR level, one must have a smaller threshold, thus the corresponding $\PP(Y_i>L |H_0) $ tends to be larger. It is clear that given the threshold $L_\alpha$ that controls the limiting FDR at exact level $\alpha$, we have $\PP(Y>L_\alpha |H_0) $ proportional to  $\PP(Y>L_\alpha |H_1) $, which implies that larger $\PP(Y>L |H_0) $  leads to a larger power.

If  $\operatorname{FDR}_{W_j}(L_{\alpha j} ) = \alpha $, then we have

\begin{equation*}
	\begin{aligned}
		\PP(Y>L_\alpha|H_0) =  \frac{\pi_1}{1-\pi_1} \PP(Y>L_\alpha|H_1)\frac{\alpha}{1-\alpha}\le c \pi_1,
	\end{aligned}
\end{equation*}
which indicates that to reach any fixed FDR level $\alpha$, the quantity $\PP(Y>L_\alpha|H_0)$ will decrease at the rate $O(\pi_1)$. We thus choose $p = O(\pi_1)$ and $L_{p j}$ satisfying $ \PP(Y_j>L_{p j}|H_0) = p$ for $j=1,2,3$. Let $z= \sqrt{p} \delta$. We will use Talyor expansion and compare the derivatives of $ \PP(Y_j>L_{p j}|z,H_0) = p$ with respect to $z$.

\begin{Theorem}\label{thm:power-comparison}
	Consider the limiting behaviors \eqref{eq:limiting_fdr} of different data aggregation methods in \eqref{eq:mirror-choice} characterized by the mixture model stated above. When achieving the same FDR level $\alpha$ and $\pi_1\to 0$, we have the following asymptotic power relation:	
	\begin{equation*}
		\begin{aligned}
			\operatorname{Power}_{W_1}(L_{\alpha 1} )\ge \operatorname{Power}_{W_2}(L_{\alpha 2} ) \ge \operatorname{Power}_{W_3}(L_{\alpha 3} ),
		\end{aligned}
	\end{equation*}
	for any bounded prior $\mathbf{\Theta} $: $\PP( \abs{\delta}\le \delta_0|\mathbf{\Theta}  )\to 1$ where $\delta_0 = o(\sqrt{\frac{1}{\pi_1}})$.
\end{Theorem}

Here, we allow the bound $\delta_0$ to go to infinity as long as its order is of $o(\sqrt{\frac{1}{\pi_1}})$. This theorem offers a theoretical justification for the superiority of our data aggregation method over other common alternatives, a conclusion that aligns with our empirical findings in \cite{dai2023scale,du2021false}.

Intuitively, when the two-sided tails of mixture distribution are more unbalanced (left-skewed) and $\PP(Y>t)$ decreases slower, the threshold $L_\alpha$ tends to be smaller and thus the null and non-null distributions can be well-separated. In Figure~\ref{fig:c0}, we present a simulation of the density of mixture distributions and $\PP(Y>L_\alpha|H_1)$ given different data aggregation methods.

\begin{figure}
	\centering
	\begin{subfigure}{0.45\textwidth}
		\includegraphics[width=\textwidth]{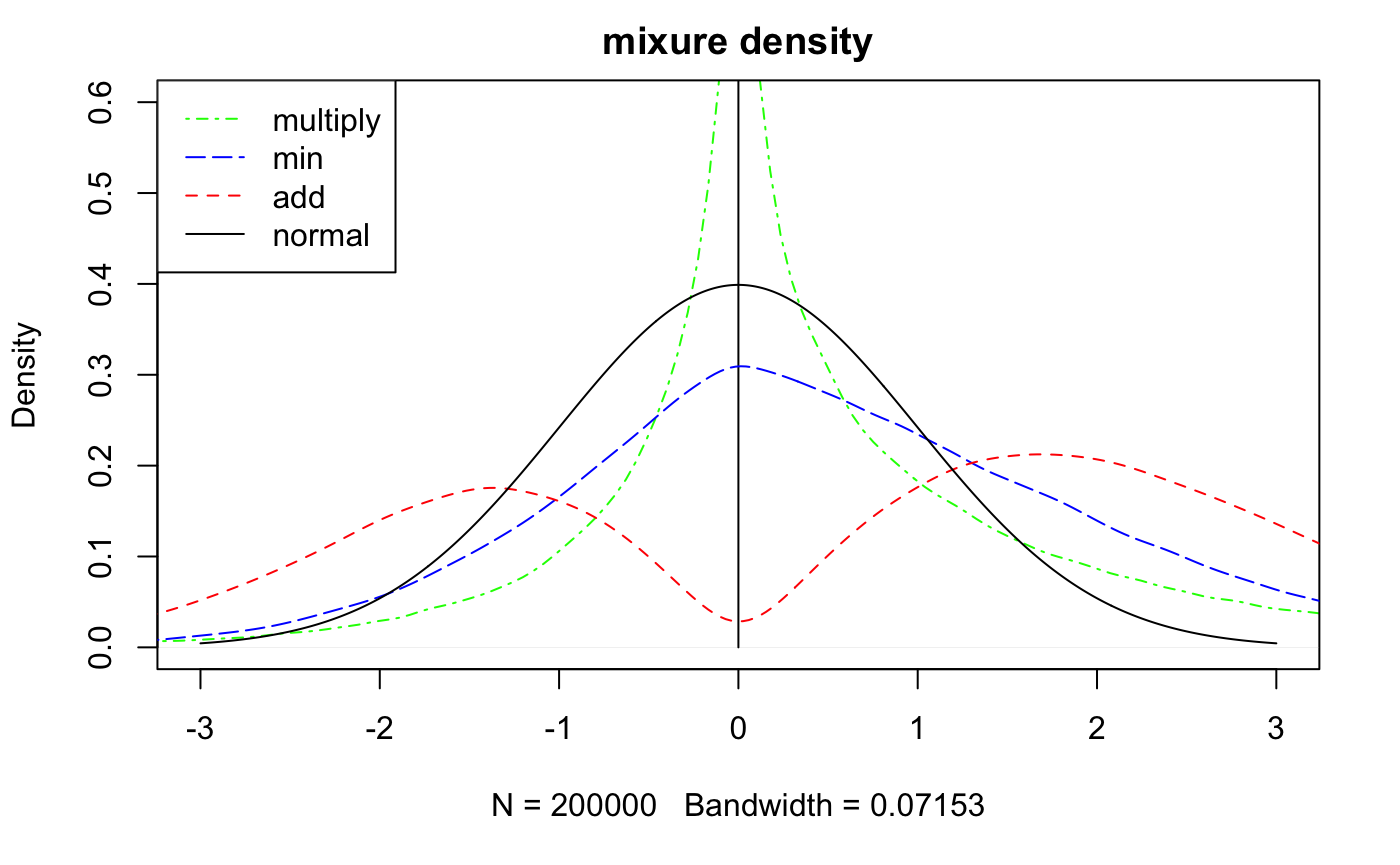}
		\caption{Density of mixture distribution}
		\label{fig:c1}
	\end{subfigure}
	\begin{subfigure}{0.45\textwidth}
		\includegraphics[width=\textwidth]{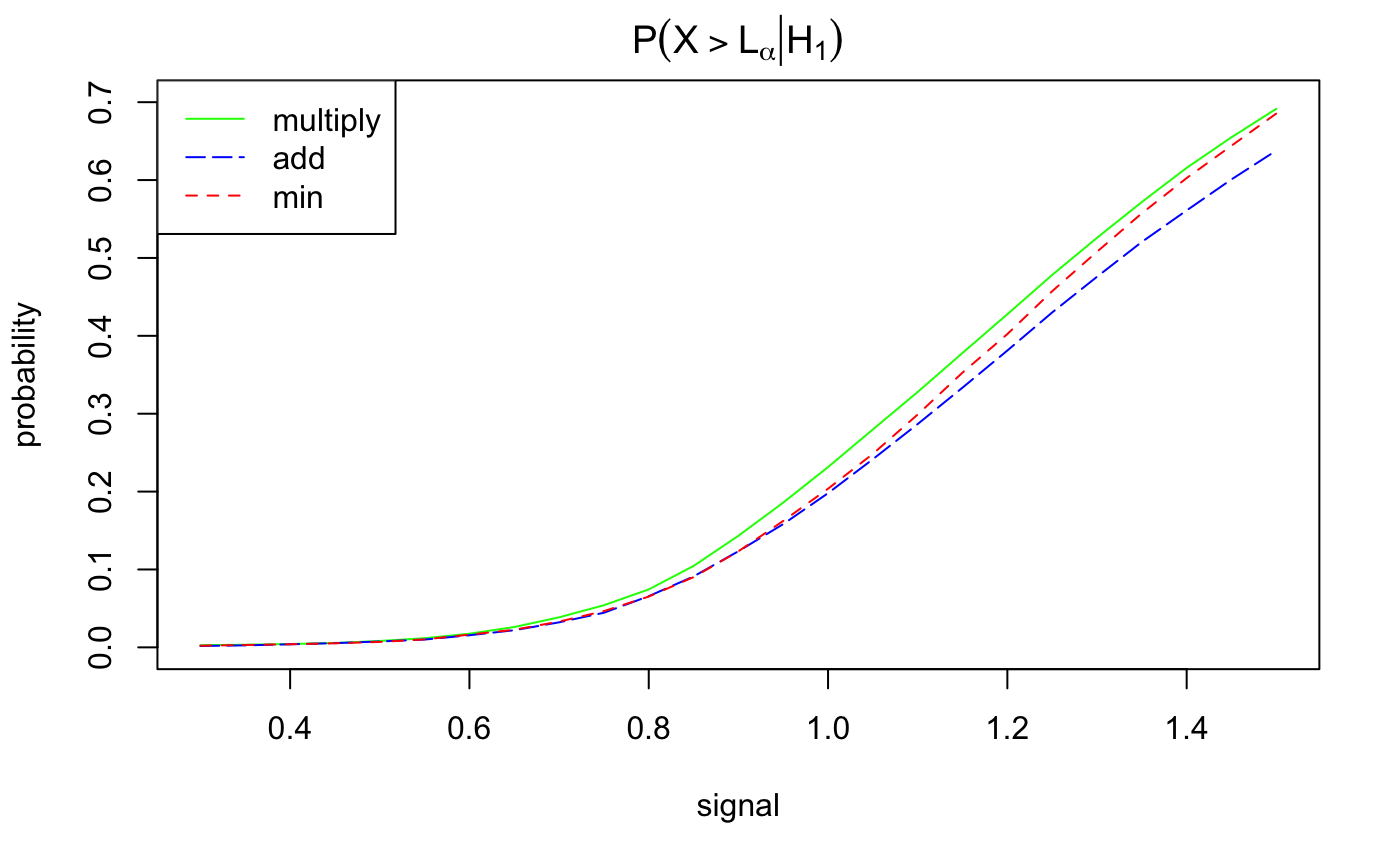}
		\caption{$\PP(Y>L_\alpha|H_1)$ }
		\label{fig:c2}
	\end{subfigure}
	\caption{Simulation of mixture distribution $Y$ under different constructions}
	\label{fig:c0}
\end{figure}
It is observed that $f_1$ generates a narrower mixture distribution with unbalanced tails starting to decrease slowly when $t$ is moderate, and the limiting power of $f_1$ is the highest among the three combinations.

\subsection{Additional Simulation Results}

\subsubsection{More Asymptotic Normality Comparison}
This section provides additional results comparing our methods with existing approaches \citep{xia2021statistical,chen2019inference} in terms of the asymptotic normality $\bar{F}_n(z)-\Phi(z)$, as discussed in Section \ref{sec:variance-lf}. We increase both $d$ and $r$ and consider $d=400$ and $d=800$, with $\lambda_{\min}=d$. The corresponding results are presented in Figure \ref{fig:variance-comparison-add}. As expected, our method exhibits superior asymptotic normality due to the use of the correct variance term. Nevertheless, under the incoherent subspace assumption, for sparse $T$ we have
\begin{equation*}
\abs{ (\|U^\top T\|_{\rm F}^2+\|TV\|_{\rm F}^2)^{1/2}-\norm{\cP_{M}(T)}_{\tF}} \lesssim \frac{1}{d}.
\end{equation*}
Thus, when $d$ is sufficiently large, the numerical difference between these variance terms becomes negligible, resulting in nearly identical empirical performance.
\begin{figure}[H]
\centering
\begin{subfigure}{0.32\textwidth}
    \includegraphics[width=\textwidth]{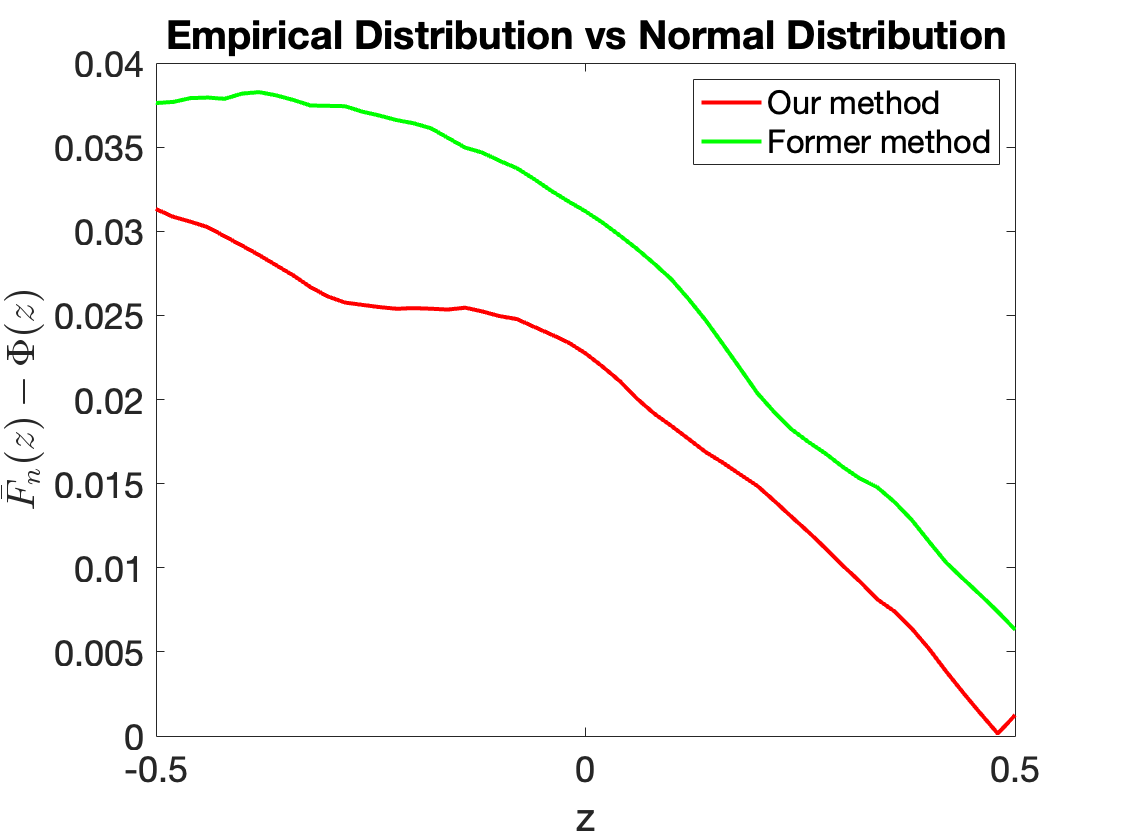}
    \caption{$d=400,r=5,n=3000$ }
    \label{fig:variance-comparison-add-1}
\end{subfigure}
 \begin{subfigure}{0.32\textwidth}
    \includegraphics[width=\textwidth]{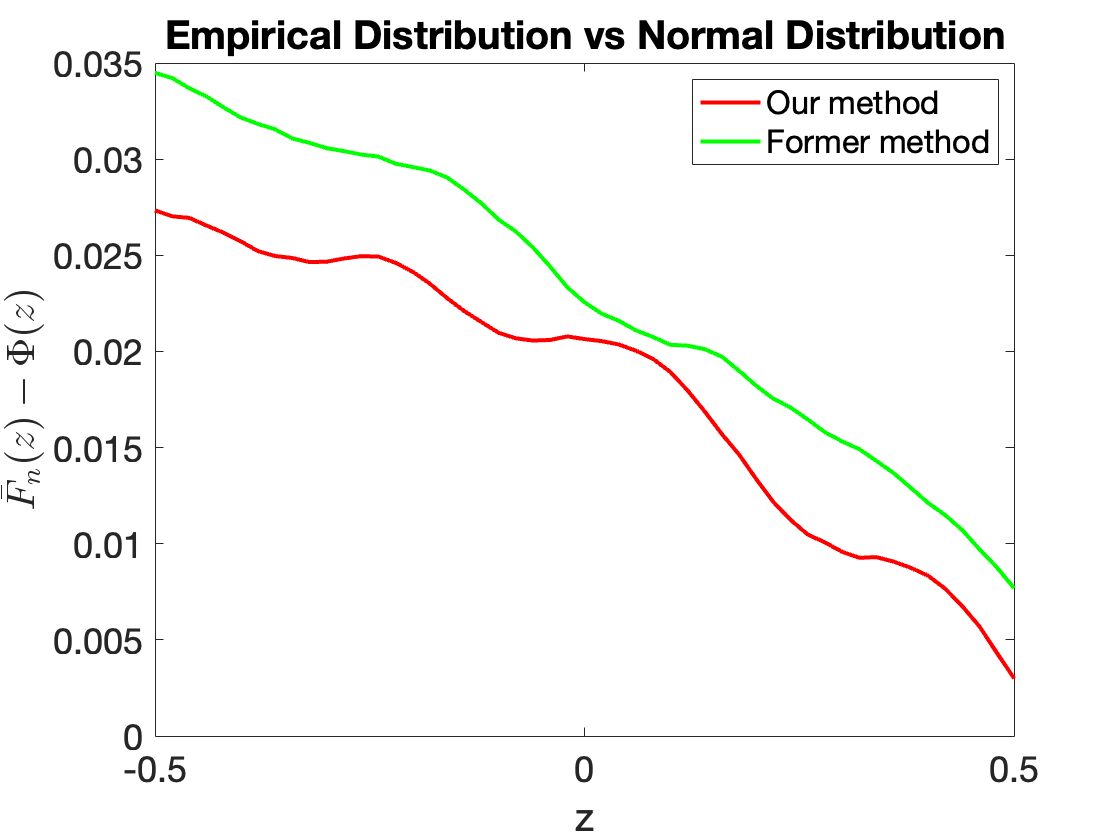}
    \caption{$d=400,r=5,n=4000$ }
    \label{fig:variance-comparison-add-2}
\end{subfigure}
     \begin{subfigure}{0.32\textwidth}
    \includegraphics[width=\textwidth]{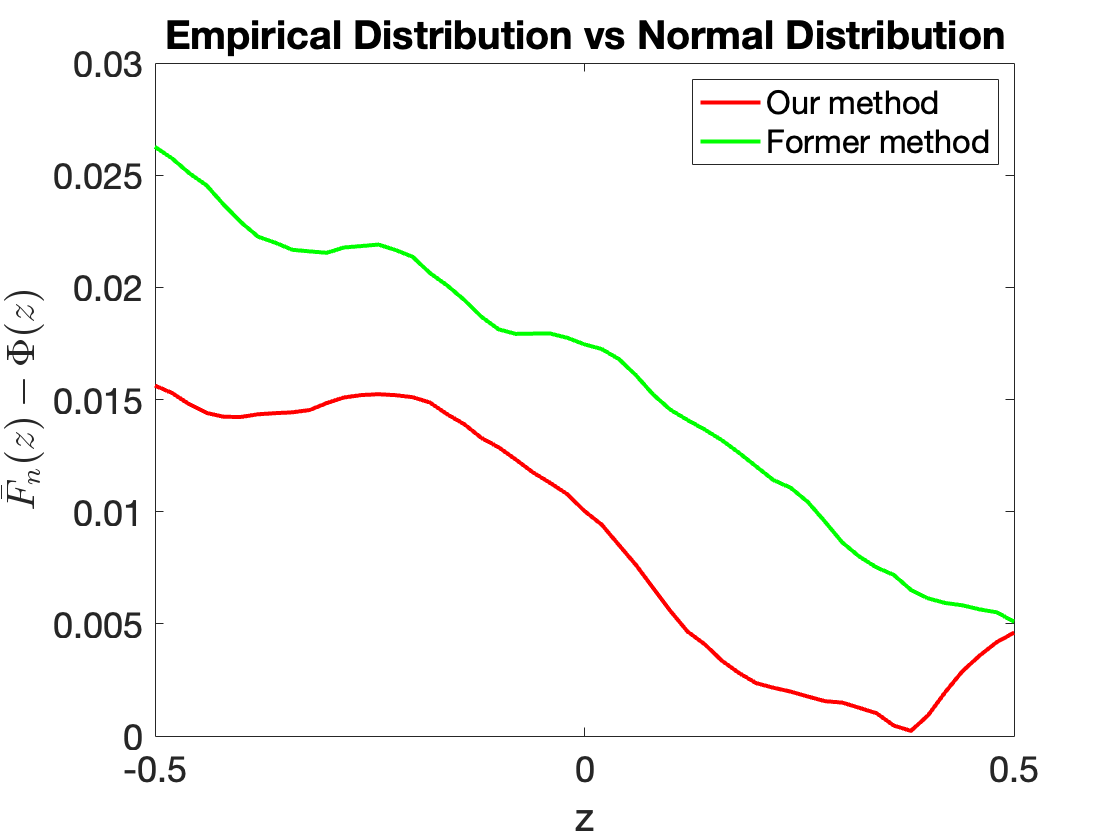}
    \caption{$d=400,r=5,n=6000$ }
    \label{fig:variance-comparison-add-3}
    \end{subfigure}
    \\
    \begin{subfigure}{0.32\textwidth}
    \includegraphics[width=\textwidth]{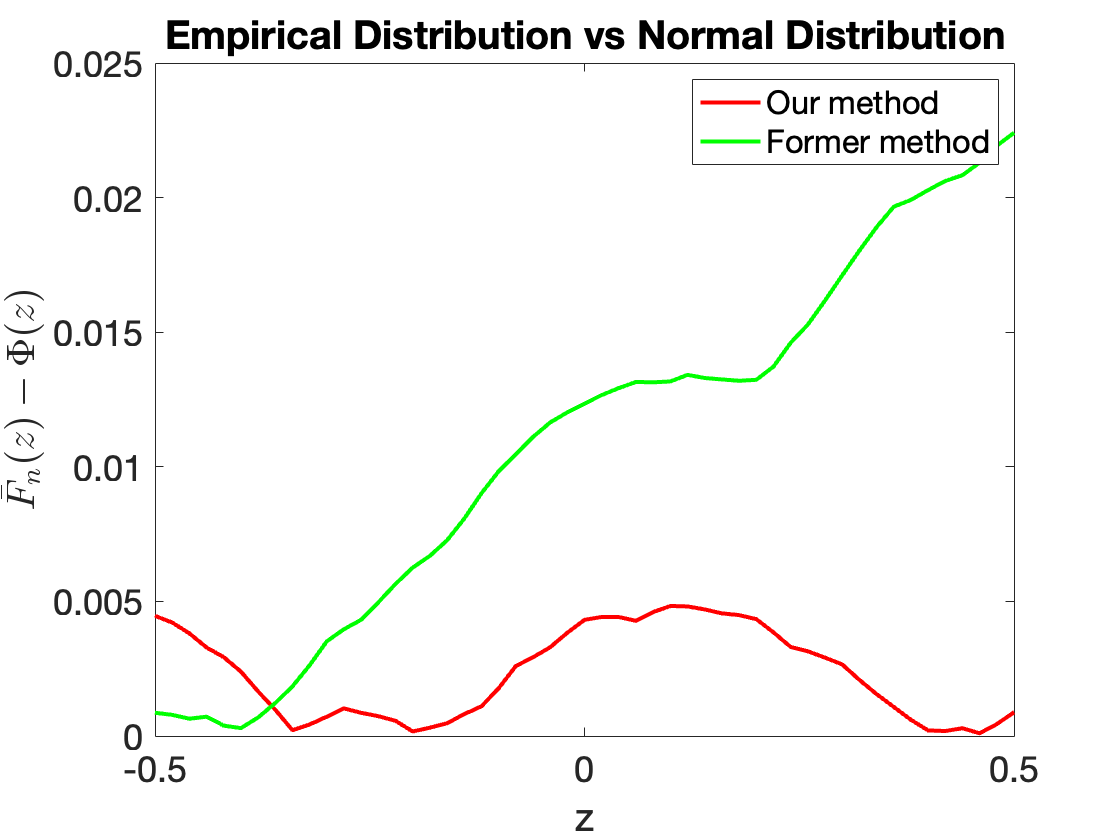}
    \caption{$d=800,r=5,n=6000$ }
    \label{fig:variance-comparison-add-4}
\end{subfigure}
 \begin{subfigure}{0.32\textwidth}
    \includegraphics[width=\textwidth]{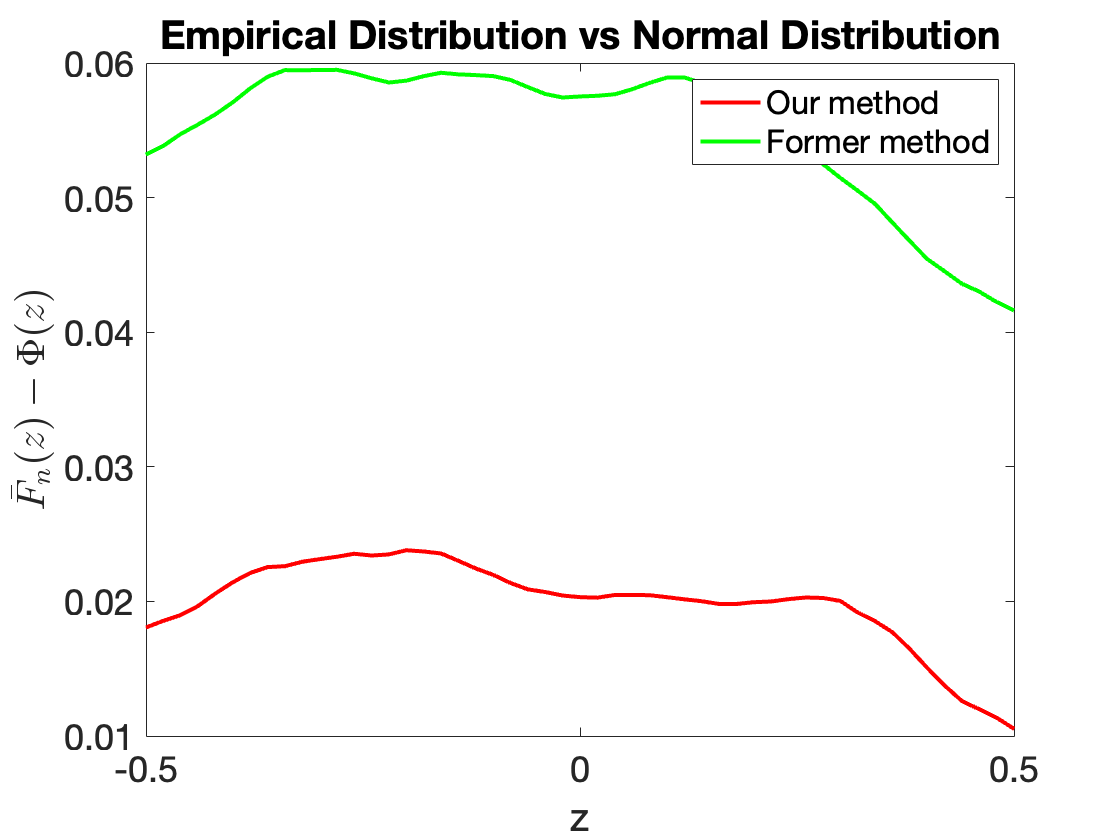}
    \caption{$d=800,r=5,n=8000$ }
    \label{fig:variance-comparison-add-5}
\end{subfigure}
     \begin{subfigure}{0.32\textwidth}
    \includegraphics[width=\textwidth]{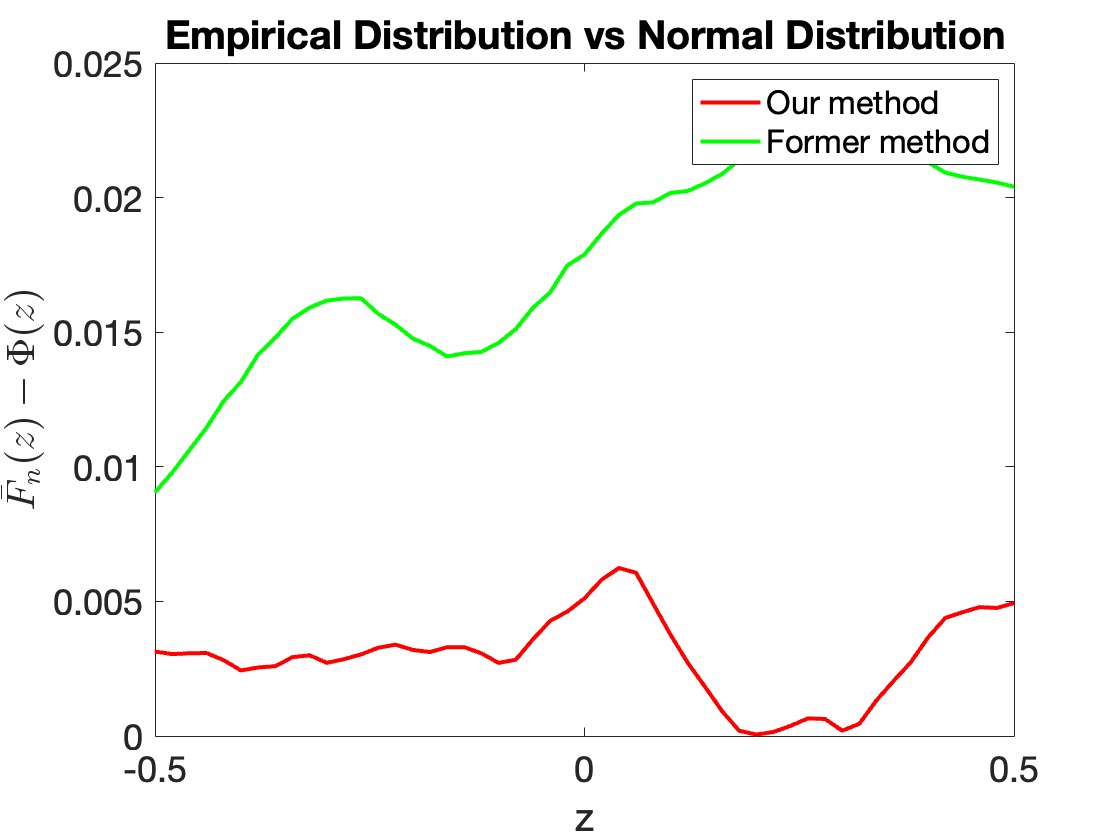}
    \caption{$d=800,r=5,n=12000$ }
    \label{fig:variance-comparison-add-6}
    \end{subfigure}

 \caption{The difference between empirical distribution functions and $\Phi(z)$ under $r=5$, and $d=400$, $d=800$.   }
 \label{fig:variance-comparison-add}
\end{figure}

\subsubsection{Robustness under Model Misspecification}
We evaluate the performance of our methods under model misspecification. To this end, we set the underlying true matrix $M$ to be only approximately low rank rather than having exact rank $r$. Specifically, we let
\begin{equation}
M = M^* + \Delta_{M},
\end{equation}
where $M^*\in\R^ {d_1\times d_2}$ is low rank with $r^*=\operatorname{rank}(M^*)$, and $\Delta_{M}$ is a full-rank random perturbation satisfying $\norm{\Delta_M}_2=d_1^{\gamma}$. Here $\gamma>0$ quantifies the perturbation level, indicating the degree to which the underlying model deviates from our low-rank assumption. In this setting, the true $M$ is full rank but remains well approximated by the low-rank structure $M^*$. We report the performance of our methods when only $r^*$ is provided, varying $\gamma$ to assess the robustness of our feature construction and FDP control procedures.

The data-generating mechanism follows Section \ref{sec:whitening}, with $d_1=d_2=1000$, $r^*=3$, $\lambda_{\min}=1000$, and $q=400$ hypotheses testing differences between the first row $M(1,1:400)$ and the second row $M(2,1:400)$ under $H_{0,T_i}: M_{1,i}-M_{2,i}=0$.

\begin{figure}[H]
\centering
    \includegraphics[width=\textwidth]{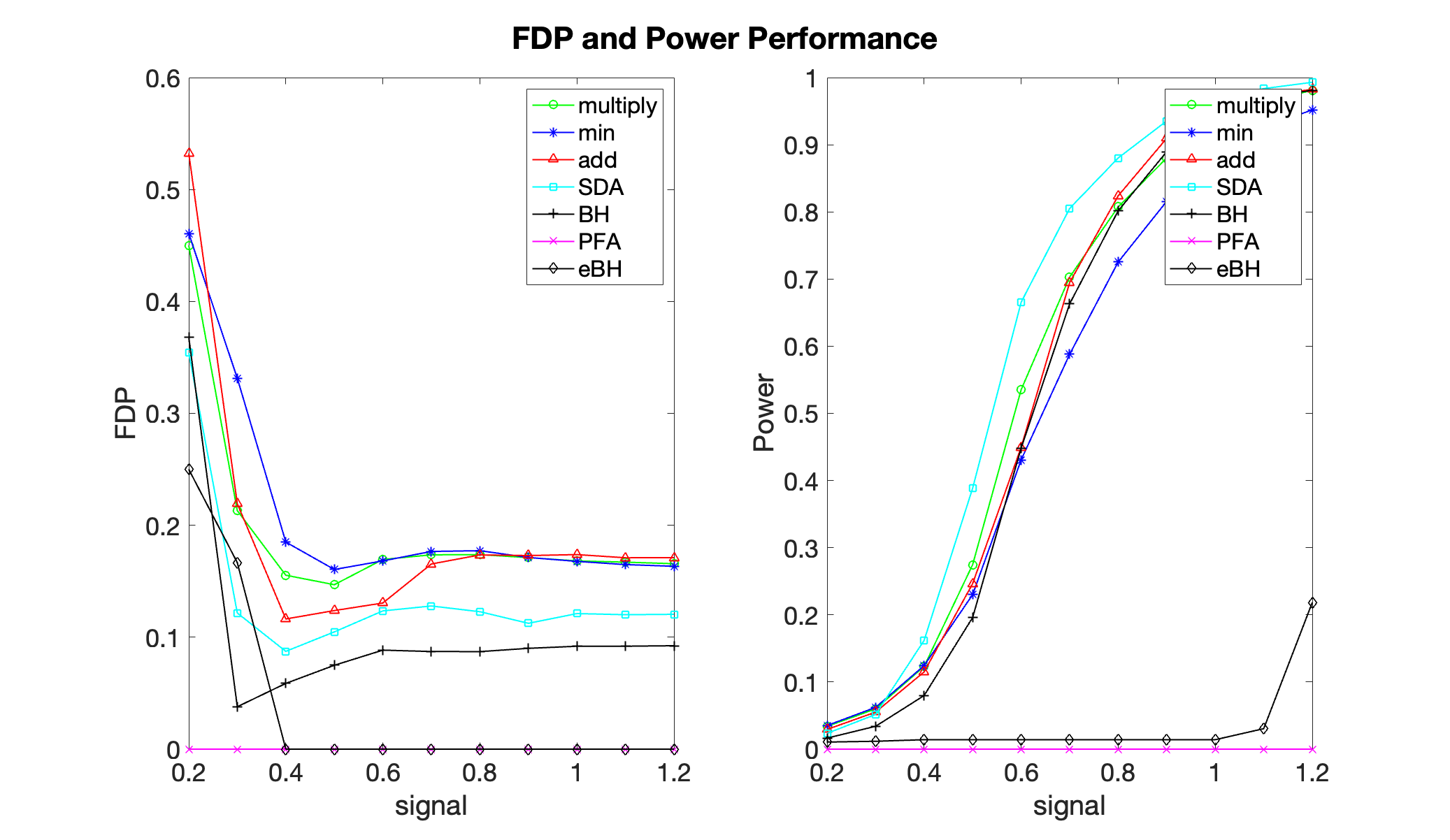}
    \caption{Performance under model misspecification with $\gamma = 0.05$ under $\alpha=0.1$}
    \label{fig:missp-gam=0.05}
\end{figure}

\begin{figure}[H]
\centering
    \includegraphics[width=\textwidth]{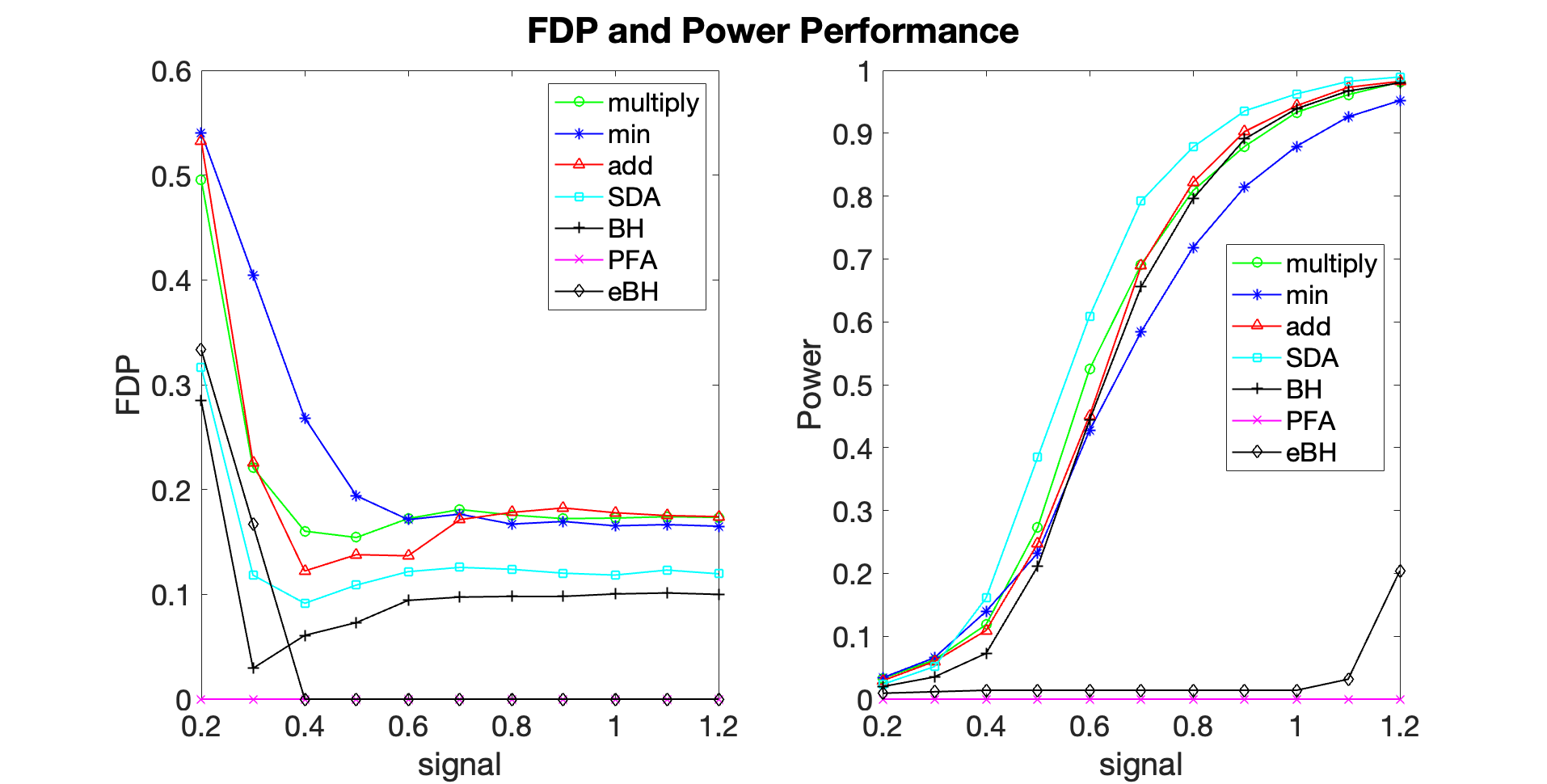}
    \caption{Performance under model misspecification with $\gamma = 0.1$ under $\alpha=0.1$}
    \label{fig:missp-gam=0.1}
\end{figure}

\begin{figure}[H]
\centering
    \includegraphics[width=\textwidth]{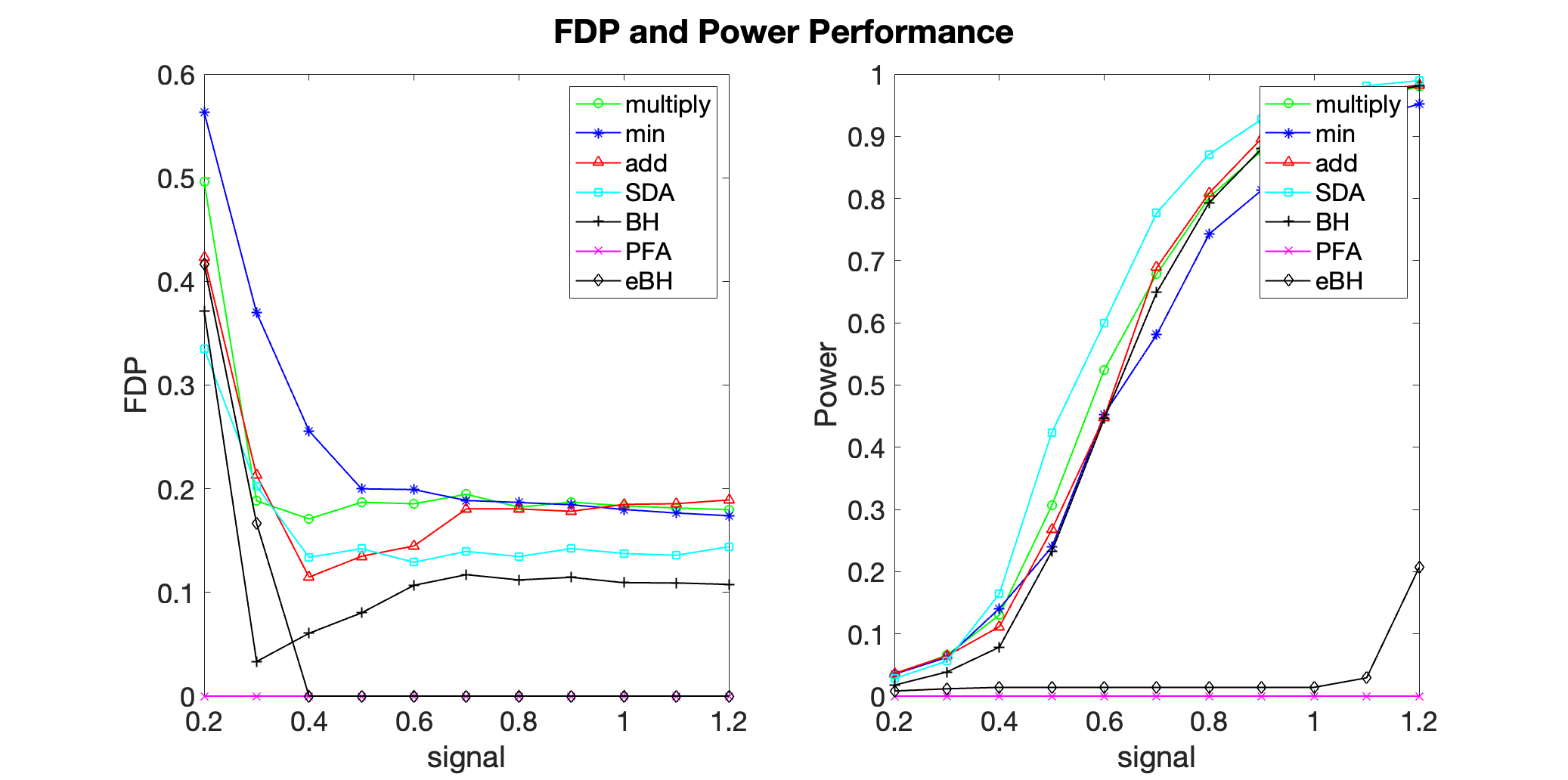}
    \caption{Performance under model misspecification with $\gamma = 0.15$ under $\alpha=0.1$ }
    \label{fig:missp-gam=0.15}
\end{figure}

\begin{figure}[H]
\centering
    \includegraphics[width=\textwidth]{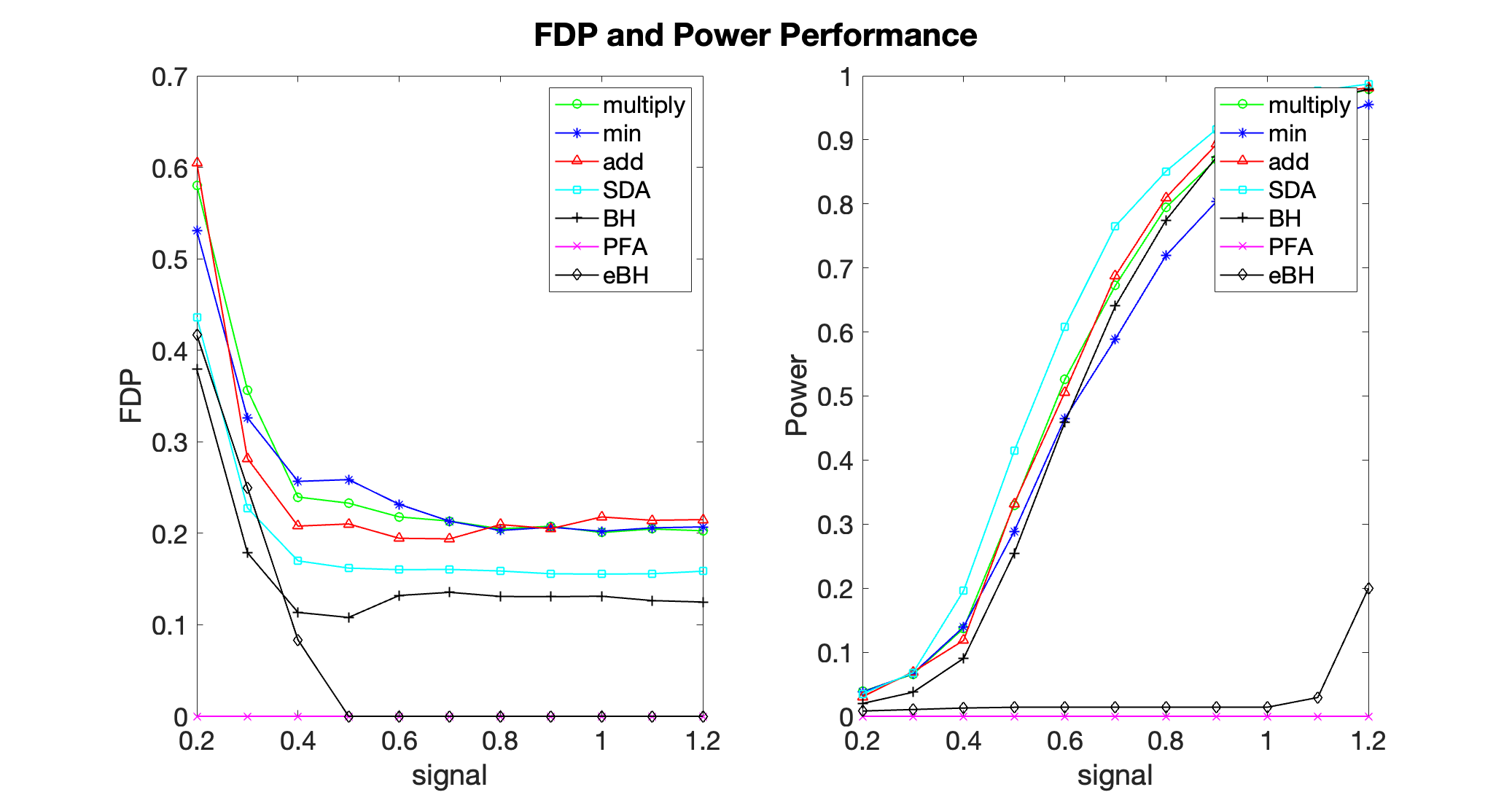}
    \caption{Performance under model misspecification with $\gamma = 0.2$ under $\alpha=0.1$ }
    \label{fig:missp-gam=0.2}
\end{figure}

It can be observed that our SDA method remains robust to model misspecification under small to moderate perturbation levels, specifically when $\gamma \le 0.15$; see Figure \ref{fig:missp-gam=0.05}, \ref{fig:missp-gam=0.1}, and \ref{fig:missp-gam=0.15}. However, when $\gamma$ becomes relatively large, e.g., $\gamma=0.2$, both our methods and the BH procedure fail to control the FDP at $0.1$ due to the substantial perturbation, as shown in Figure \ref{fig:missp-gam=0.2}. In this regime, the factor-adjusted procedure and eBH still control the FDP but are overly conservative, yielding almost no discoveries.

\subsubsection{Performance under Very Weak Dependence}
The previous simulations focus on settings where certain dependence structures are present among the test statistics, allowing our SDA procedure to leverage such dependencies to improve performance. In this section, we present results for scenarios in which the dependence is very weak. In this case, data aggregation does not necessarily outperform BH or other methods in terms of power—a reasonable phenomenon also observed in prior work \cite{du2021false}.

To generate testing cases with weak dependence, we consider $q=400$ diagonal elements $M(i,i)$ for $i=1,\ldots,400$. As discussed in Section \ref{sec:dependence}, these entries lie in different rows and columns, leading to extremely weak correlations among the corresponding test statistics. In our simulations, we generate data following the same setting as in Sections \ref{sec:weak-simu} and \ref{sec:whitening}. For this task, we obtain $\varrho^*(0.2) = 0.0025$, where $\varrho^*(z)$ is defined in \eqref{eq:rho-emp}. This confirms that the vast majority of test statistics exhibit negligible correlation, consistent with the theoretical insights presented in Section \ref{sec:dependence}.

\begin{figure}[H]
    \centering
    \includegraphics[width=0.9\linewidth]{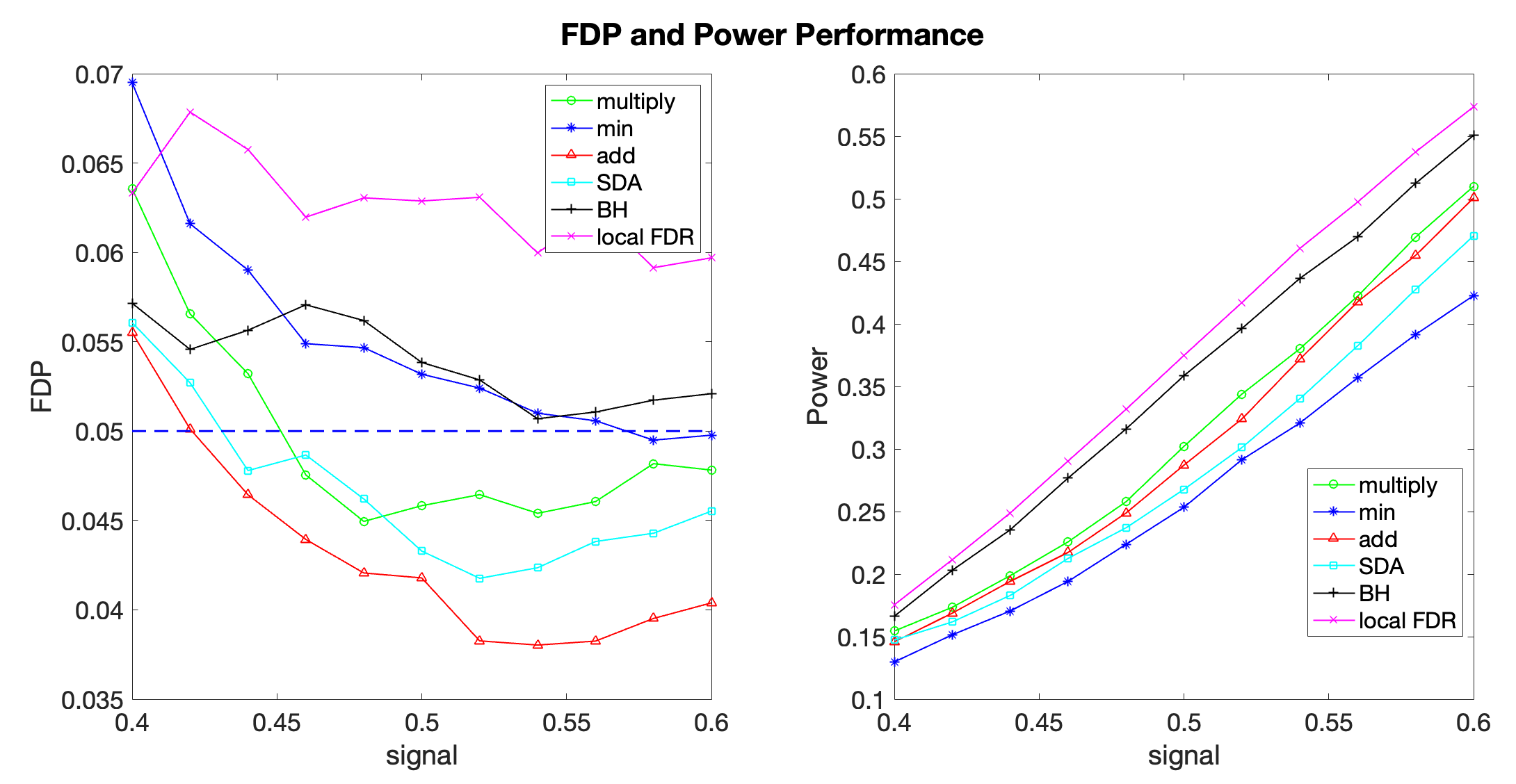}
    \caption{Comparison between our aggretation methods with  local FDR thresholding under very weak dependence, with FDR control level $\alpha=0.05$}
    \label{fig:weak-SDA}
\end{figure}

\subsubsection{Performance under Heterogeneous Noise}
To examine our test statistics under heterogeneous noise introduced in Proposition \ref{prop:hetero}, we conduct some additional numerical simulations. We generate observations where random noise can depend on the position of an observation. Specifically, we let
\begin{equation}\label{eq:hetero-int}
    Y_i \sim \lfloor \langle M,X_i\rangle \rfloor + \operatorname{Ber}\left(\text{decimal}( \langle M,X_i\rangle) \right),
\end{equation}
where $\text{decimal}(x):= x- \lfloor x \rfloor$ represents the decimal part of a real number. In the case \eqref{eq:hetero-int}, we still have $\E[Y_i\mid X_i] =\langle M,X_i\rangle $  , but the noise depends on $X_i$ and $M$ such that $Y_i$ can only take values between $\lfloor \langle M,X_i\rangle \rfloor $ and $\lfloor \langle M,X_i\rangle \rfloor  +1 $. For matrix generating and hypotheses selection, we still adopt cases in Section \ref{sec:whitening}. We report results in Figure \ref{fig:moderate-SDA-hetero}.

\begin{figure}[H]
    \centering
    \includegraphics[width=0.9\linewidth]{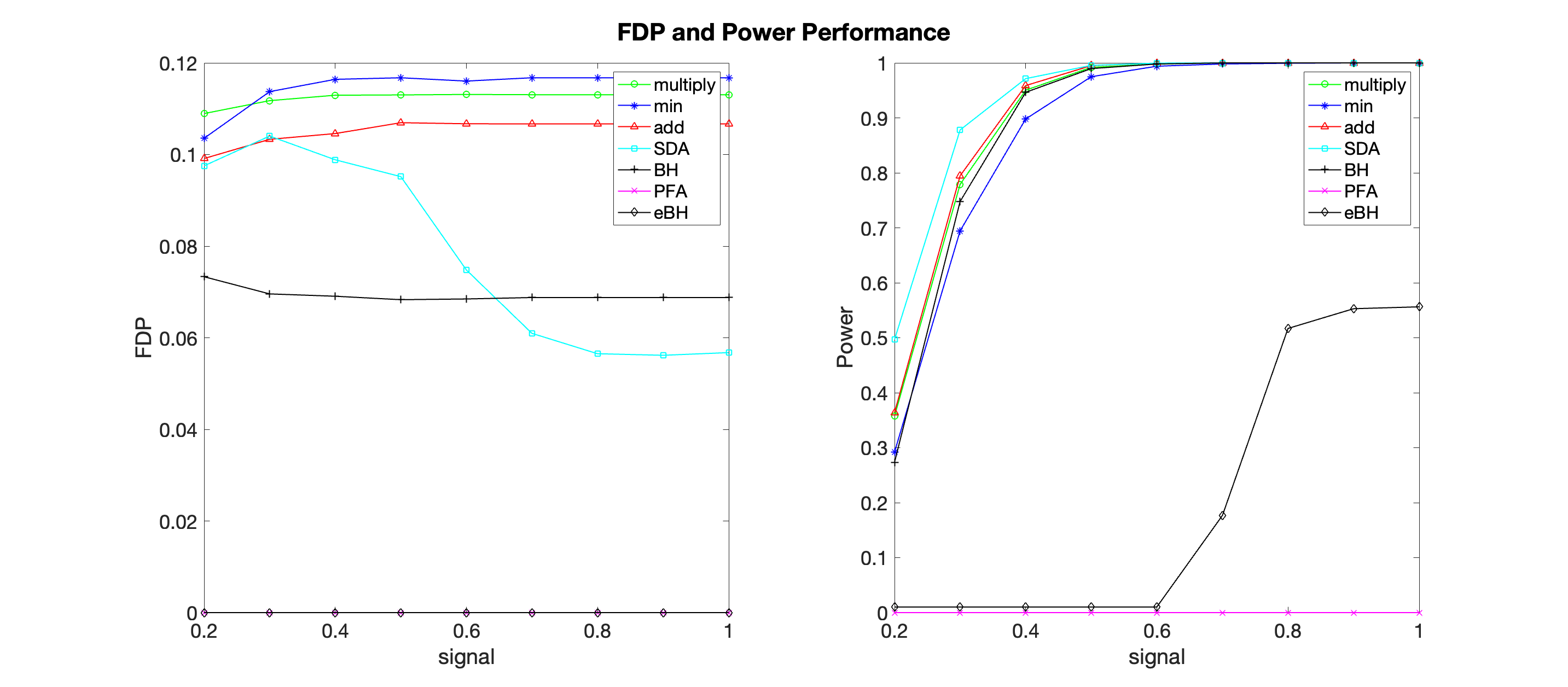}
    \caption{FDP and power performance of methods under moderate dependence with a level $\alpha=0.1$.}
    \label{fig:moderate-SDA-hetero}
\end{figure}

\section{Proofs of Main Results}

\subsection{Proof of Theorems \ref{thm:asymp-normal} and \ref{thm:asymp-normal-varest}}\label{sec:proof-clt}
We first verify the initialization condition \eqref{eq:init-est}. Let $\widehat{M}^{\init}$ be the output of gradient descent till convergence from \cite{chen2020noisy}. According to the leave-one-out analysis in \cite{chen2020noisy} (eq 93, and Lemma 10-15), we have, 
\begin{equation}\label{eq:lol-init-1}
  \begin{aligned}
        \norm{\widehat{M}^{\init}- M}_{\max} & \le \norm{V\Lambda^{\frac{1}{2}}}_{2,\max}\norm{\widehat{U}^{\init}(\widehat{\Lambda}^{\init})^{\frac{1}{2}}\widetilde{O}  - U\Lambda^{\frac{1}{2}} }_{2,\max} + \norm{U\Lambda^{\frac{1}{2}}}_{2,\max}\norm{\widehat{V}^{\init}(\widehat{\Lambda}^{\init})^{\frac{1}{2}}\widetilde{O}  - V\Lambda^{\frac{1}{2}} }_{2,\max} 
        \\
        &\quad + \norm{\widehat{U}^{\init}(\widehat{\Lambda}^{\init})^{\frac{1}{2}}\widetilde{O}  - U\Lambda^{\frac{1}{2}} }_{2,\max} \norm{\widehat{V}^{\init}(\widehat{\Lambda}^{\init})^{\frac{1}{2}}\widetilde{O}  - V\Lambda^{\frac{1}{2}} }_{2,\max} 
        \\
        & \le C \frac{\mu r}{\sqrt{d_1 d_2}}\lambda_{\max}\cdot \frac{\kappa_0\sigma_{\xi}}{\lambda_{\min }} \sqrt{\frac{d_1^2 d_2 \log d_1}{n }} 
        \le C \mu r\kappa_0^2 \sqrt{\frac{d_1\log d_1}{n }} 
  \end{aligned}
\end{equation}
where $\widetilde{O}$ is a rotation matrix with $\widetilde{O}\in\R^{r\times r}$ and $\widetilde{O}^\top \widetilde{O}=I_{r}$. Moreover, for the leave-one-out perturbation series $\widehat{M}^{\init}_{(l)}$ which is constructed by knocking off observations from the $l$-th row, we also have
\begin{equation}\label{eq:lol-init-2}
     \norm{\widehat{M}^{\init}_{(l)}- M}_{\max} \le  \norm{\widehat{M}^{\init}- \widehat{M}^{\init}_{(l)}}_{\max} + \norm{\widehat{M}^{\init}- M}_{\max} \le C \mu r\kappa_0^2 \sqrt{\frac{d_1\log d_1}{n }}.
\end{equation}
Specifically, \cite{chen2020noisy} (eq 93c) suggests the leave-one-out error on the singular subspace by Wedin's $\sin\Theta$ Theorem:
\begin{equation}\label{eq:lol-init-3}
    \norm{\widehat{U}^{\init}_{(l)} \widehat{U}^{\init\top}_{(l)}-\widehat{U}^{\init} \widehat{U}^{\init\top}}_{\tF}\le \frac{\kappa_0^{\frac{1}{2}}\sigma_{\xi}}{\lambda_{\min }} \sqrt{\frac{d_1^2 d_2 \log d_1}{n }} \sqrt{\frac{\mu r}{d_1}}.
\end{equation}
A corresponding bound for $\widehat{V}^{\init}_{(l)}$ can also be established by studying the GD iterations on the first $d_1$ rows and last $d_2$ rows of $F=\left[\begin{array}{c}
         U\Lambda^{\frac{1}{2}}  \\
          V \Lambda^{\frac{1}{2}}
    \end{array}\right]$ respectively. 
    
Combining \eqref{eq:lol-init-1}, \eqref{eq:lol-init-2}, \eqref{eq:lol-init-3}, we know that the assumptions of Theorem 4 in \cite{ma2024statistical} are satisfied, with the initialization error constant at most $C_{\init}\lesssim \mu r\kappa_0^2$.  We now apply the decomposition of SVD operators \citep{xia2021statistical,ma2024statistical} and show CLT with Berry-Esseen theorem: by (47) in \cite{ma2024statistical}, we have
\begin{equation}\label{eq:svd-decomp}
    \langle \widehat{M},T\rangle -  \langle M,T\rangle = \underbrace{\sum_{i=1}^n \frac{d_1 d_2}{n}\xi_i\langle \cP_M( X_i),T\rangle}_{\text{core variance part}} + \underbrace{E_T^{\higher}}_{\text{higher-order part}},
\end{equation}
where the first part is the sum of i.i.d. terms that will lead to asymptotic normality, and the second part is a higher-order term that will vanish. Invoking Berry-Esseen theorem \citep{raivc2019multivariate}, we yield the asymptotic normality for the core variance part:
\begin{equation}\label{eq:BE-core}
    \sup_{t\in\R}\abs{\PP\left( \frac{\sum_{i=1}^n \frac{d_1 d_2}{n}\xi_i\langle \cP_M( X_i),T\rangle}{ \sigma \norm{\cP_M(T)}_{\tF} \sqrt{d_1 d_2/n}}\le t  \right) - \Phi(t)}\le C\sqrt{\frac{\mu r d}{n}}.
\end{equation}
For the higher-order term $E_T^{\higher}$, together with the variance estimation error, they can all be controlled at the level $h_n\times \sigma \norm{\cP_M(T)}_{\tF} \sqrt{d_1 d_2/n}$,  as is shown in Theorem 4, \cite{ma2024statistical}.
Therefore, we have Theorems \ref{thm:asymp-normal} and \ref{thm:asymp-normal-varest}.

\subsection{Proof of Theorem \ref{thm:asymp-two-var}}
\begin{proof}
We write the debiasing estimator as 
\begin{equation*}
    \widehat{M}^{\mathsf {unbs }} = M+ \underbrace{\frac{d_1 d_2}{n} \sum_{i=1 }^n \xi_i X_i}_{\widehat{Z}_1: \text{ i.i.d. noises}} +\underbrace{\left(\frac{d_1 d_2}{n} \sum_{i=1 }^n\left\langle\widehat{M}^{\mathsf{init}} - M, X_i\right\rangle X_i-(\widehat{M}^{\mathsf{init}} -M)\right)}_{\widehat{Z}_2:\text{ initialization error} }.
\end{equation*}
By , we have the inequalities uniformly hold for all $T$, we can write the couple $(W_{T_1},W_{T_2})$ as 
$$
(W_{T_1},W_{T_2}) = \left(\frac{\left\langle  \widehat{Z}_1 , \cP_M(T_1) \right\rangle }{\sigma_{\xi} \norm{\cP_M(T_1) }_\tF  \sqrt{d_1 d_2/n}}+ \Delta_{T_1}, \frac{\left\langle  \widehat{Z}_1 , \cP_M(T_2) \right\rangle }{\sigma_{\xi} \norm{\cP_M(T_2) }_\tF  \sqrt{d_1 d_2/n}}+ \Delta_{T_2} \right),
$$
where with probability at least $1-C\log d_1 d_1^{-\tau}  $,
\begin{equation}\label{eq:2-d-remainder}
\begin{aligned}
     \max\left\{\abs{\Delta_{T_1}}, \abs{\Delta_{T_2}}\right\} \le & C_4 h_n
\end{aligned}
\end{equation}
by the same argument in the proof of Theorem 4 in \cite{ma2024statistical}. Notice that $\langle\widehat{Z}_1,\cP_M(T_i) \rangle$, $i=1,2$ are the sum of i.i.d. random variables, with covariance matrix:

\begin{equation*}
\begin{aligned}
         \E \langle\widehat{Z}_1,\cP_M(T_1) \rangle \langle\widehat{Z}_1,\cP_M(T_2) \rangle &= \frac{ d_1^2 d_2^2}{n^2} \sum_{i\in I_2} \xi_i^2 \left\langle  X_i, \cP_M(T_1) \right\rangle\left\langle X_i, \cP_M(T_2) \right\rangle \\
         & = \frac{ d_1 d_2}{n} \sigma_{\xi}^2 \sum_{i\in [d_1]}\sum_{j\in [d_2]} e_i^\top \cP_M(T_1)e_j e_i^\top \cP_M(T_2)e_j \\
         & = \frac{ d_1 d_2}{n} \sigma_{\xi}^2\left\langle  \cP_M(T_1), \cP_M(T_2) \right\rangle.
\end{aligned}
\end{equation*}
Then we have 
\begin{equation*}
    \operatorname{cov}\left(\frac{\left\langle  \widehat{Z}_1 , \cP_M(T_1) \right\rangle }{\sigma_{\xi} \norm{\cP_M(T_1) }_\tF  \sqrt{d_1 d_2/n}}, \frac{\left\langle  \widehat{Z}_1 , \cP_M(T_2) \right\rangle }{\sigma_{\xi} \norm{\cP_M(T_2) }_\tF  \sqrt{d_1 d_2/n}} \right) =:\rho_{T_1, T_2}.
\end{equation*}
We jointly control the c.d.f. of them by multivariate Berry-Essen theorem \citep{stein1972bound,raivc2019multivariate}

\begin{equation}\label{eq:2-d-be}
    \begin{aligned}
       & \sup_{t_1,t_2\in \R} \abs{ \PP \left( \frac{\left\langle  \widehat{Z}_1 , \cP_M(T_1) \right\rangle }{\sigma_{\xi} \norm{\cP_M(T_1) }_\tF  \sqrt{d_1 d_2/n}}\le t_1, \frac{\left\langle  \widehat{Z}_1 , \cP_M(T_2) \right\rangle }{\sigma_{\xi} \norm{\cP_M(T_2) }_\tF  \sqrt{d_1 d_2/n}} 
        t_2  \right) - \Phi_{\rho_{T_1,T_2}}(t_1,t_2)}
        \\
       &\le C \mu \left(\frac{1}{1-\rho}\right)^{\frac{3}{2}} \sqrt{\frac{r d_1}{n}}.
    \end{aligned}
\end{equation}
The gradient bound  $\norm{\nabla \Phi_{\rho} (t_1,t_2)}\le C$ indicates the Lipschitz property of $\Phi_{\rho} (t_1,t_2)$, which suggests the desired probability bound following \eqref{eq:2-d-remainder}, \eqref{eq:2-d-be}.
\end{proof}

\subsection{Proof of Theorem \ref{thm:weak-cor-fdr}}
We remark that our proof of Theorem \ref{thm:weak-cor-fdr} will give a non-asymptotic bound on the FDR control (Section \ref{sec:non-saymp-results}) with novel techniques, which is totally different from the previous asymptotic analysis in \cite{du2021false,dai2022false}. We proceed to prove the FDR control in the sequel by three steps: we first show that $\bbI(W_{T}^{(i)}>t) $ follows weak dependency and asymptotic symmetricity for $T\in\cH_0$ when $t$ is in a certain region $[0,L_n]$; then we show that with high probability, the data-driven threshold is in the region $[0,L_n]$; finally we control the power when strong signals dominate signals in the non-null set. Since $n\le O(d_1^{2\tau} )$ in general, we treat $h_n=\Omega({\sqrt{\log d_1}/{d_1^{\tau/2} } }\vee (rd_1/n)^{1/4} ) $ in the proof for simplicity. Here $h_n$ can be chosen smaller as long as $n$ is large. Moreover, to make the FDR control nontrivial, we assume that $q_0\ge c q$ for some $0<c<1$, and we consider the case when $q$ is large as $n\to \infty$.  We first start with the asymptotic property of ${W}^{(i)}_{T}$.
\subsubsection{Weak dependence and symmetricity}
From  Theorem \ref{thm:asymp-normal}, \ref{thm:asymp-normal-varest}, and definition of $h_n$, we have the following claim of the asymptotic normality of ${W}^{(1)}_{T}$:
\begin{Proposition}\label{prop:asy-w1}
There exits a constant $C_2$ such that $W^{(1)}_T$ follows the asymptotic normality rate:
\begin{equation}\label{eq:asy-w1}
    \begin{aligned}
    \sup_{t\in\R}\abs{\PP(W^{(1)}_T>t )-\Phi(-t)}\le C_2 h_n,
    \end{aligned}
\end{equation}
for any $T\in \cH_0$.
\end{Proposition}
This proposition implies the asymptotic symmetricity of ${W}^{(1)}_{T}$, ${W}^{(2)}_{T}$, which is crucial for the following analysis. Since ${W}^{(1)}_{T}$, ${W}^{(2)}_{T}$ are asymptotically normal, the c.d.f. of their product $W_T^\mathsf{rank}$ will converge to an asymptotically conditional symmetric random variable. We will show that, conditional on the splits $\cD_1$, $W_T^\mathsf{rank}$ is asymptotically symmetric for $T\in\cH_0$. Define 
\begin{equation*}
    G(t)= \frac{\sum_{T\in\cH_0 } \PP( {W}^{(1)}_{T} Z > t |\cD_1 )   }{q_0},
\end{equation*}
where $Z$ is a standard Gaussian variable. Here since ${W}^{(1)}_{T}\in\sigma(\cD_1)$, ${W}^{(1)}_{T}$ is fixed conditional on $\cD_1$.



Denote $L_n= G^{-1}( \frac{\epsilon_n \eta_n}{q_0})= \inf\left\{t:G(t)\le \frac{\epsilon_n \eta_n}{q_0}\right\}$, where $\epsilon_n$ is a rate to be specified later. We can exploit the following asymptotic symmetric property of ${W}^{(1)}_{T}$ to investigate the population version of the following ratio:
\begin{equation*}
    \sfR_0=\frac{\sum_{T\in \cH_0  } \bbI(W_T^{\mathsf{rank}} >L ) }{ \sum_{T\in \cH_0  } \bbI(W_T^{\mathsf{rank}} <-L )}.
\end{equation*}
Here, we introduce a weaker characterization of strong signals, that is 
 \begin{equation}\label{eq:snr-strong}
        \cS = \left\{ T\in \cH: \frac{  \sqrt{n} \abs{M_T-\theta_T } }{ {\sigma_\xi \sqrt{d_1 d_2}\norm{\cP_M(T) }_{\tF} }  \sqrt{\log d } } \ge C_{\mathsf{gap} } \right\},
    \end{equation}   
    with $\eta_n=\abs{\cS}$ for some large constant $C_{\mathsf{gap} }$. In the following proof, we will actually focus on this definition of strong signals. This condition is actually weaker than in our main text, because $\norm{\cP_M(T)}_{\tF} \le \norm{T}_{\ell_1} \max_{i,j}\norm{ \cP_M(e_i e_j^\top) }_{\tF} \le 3 {\mu} \norm{T}_{\ell_1} \sqrt{\frac{r}{d_2} }$. Thus, all the signals that satisfy condition \eqref{eq:strong-T} can also satisfy condition \eqref{eq:snr-strong}, meaning that the $\eta_n$ defined here is always larger than that defined in \eqref{eq:strong-T}.

\begin{Lemma}\label{lemma:pop-wrank}  Conditional on $\cD_1$, we have
\begin{equation*}
    \sup\limits_{0\le t \le L_n} \abs{\frac{ \sum_{T\in\cH_0}\PP(W_T^\mathsf{rank} >t ) }{q_0 G(t) } -1 }\le C_3 \frac{h_n q_0}{\epsilon_n \eta_n }.
\end{equation*}
\end{Lemma}

\begin{proof}
We only focus on small $h_n$. For each $T\in \cH_0$, conditional on  $\cD_1$, Proposition \ref{prop:asy-w1} implies that 
\begin{equation*}
    \begin{aligned}
    &\abs{\PP(W_T^\mathsf{rank}>t )-\PP( {W}^{(1)}_{T} Z > t |\cD_1 )} \le C_2 h_n.
    \end{aligned}
\end{equation*}
The definition of $L_n$ also implies $G(t)\ge \frac{\epsilon_n \eta_n}{q_0}$.
Then, we can derive the following uniform bound of convergence:

\begin{equation*}
    \sup\limits_{0\le t \le L_n} \abs{\frac{ \sum_{T\in\cH_0}\PP(W_T^\mathsf{rank}>t ) }{q_0 G(t) } -1 }\le \sup\limits_{0\le t \le L_n} \abs{\frac{ \sum_{T\in\cH_0}\PP(W_T^\mathsf{rank}>t )-G(t) }{q_0 G(L_n) } } \le C_3 \frac{h_n q_0}{\epsilon_n \eta_n }.
\end{equation*}
\end{proof}

Then, we explore the weak dependency of linear forms under signals and correlations assumptions. We will show that, with high probability, the $\sfR_0$ can be very close to its population version described in Lemma \ref{lemma:pop-wrank}. Although we already have the intuition of dependency between different $W_T^{(1)}$ by Theorem \ref{thm:asymp-two-var}, the rate provided is not enough for FDR control based on $W_T^{\mathsf{rank} }$. Here, we study the correlation of $W_T^{\mathsf{rank} }$ between different $T$ with a more delicate analysis. Let $T_1$, and $T_2$ be two different indexing matrices in $\cH_0$. To this end, we introduce the following Lemma:

\begin{Lemma}[Weak dependency of null statistics]\label{lemma:weak-cov}
        Conditional on $\cD_1$, 
\begin{equation}
\begin{aligned}
          \sup_{0\le t\le L_n } \frac{\sum_{(T_i,T_j)\in \cH_{0,\text{weak} }^2 } \abs{\operatorname{cov}(\bbI(W_{T_i}^{\mathsf{rank} }>t),\bbI(W_{T_j}^{\mathsf{rank} }>t)) }}{ q_0^2 G^2(t)}\le  C_1 \frac{h_n q_0}{\epsilon_n \eta_n }+C_2 \frac{1}{ \left(\epsilon_n\eta_n q_0\right)^{\nu/2} },
\end{aligned}
\end{equation}
where $\nu$ is the weak correlation parameter defined in \eqref{eq:weak-corr-1}.
\end{Lemma}
\begin{proof}
Suppose we have a pair $(T_1,T_2)\in \cH_{0,\text{weak} }^2$. Here we adopt the notation in the proof of Theorem \ref{thm:asymp-normal-varest}: denote 
$$
\begin{aligned}
     (W_{T_1}^{(i)},W_{T_2}^{(i)}) &= \left(\frac{\left\langle  \widehat{Z}_1^{(i)} , \cP_M(T_1) \right\rangle }{\sigma_{\xi} \norm{\cP_M(T_1) }_\tF  \sqrt{d_1 d_2/n}}+ \Delta_{T_1}^{(i)}, \frac{\left\langle  \widehat{Z}_1^{(i)} , \cP_M(T_2) \right\rangle }{\sigma_{\xi} \norm{\cP_M(T_2) }_\tF  \sqrt{d_1 d_2/n}}+ \Delta_{T_2}^{(i)} \right) \\
     &:= \left(\widetilde{W}^{(i)}_{T_1}+\Delta_{T_1}^{(i)}  , \widetilde{W}^{(i)}_{T_2}+\Delta_{T_2}^{(i)} \right),\ i=1,2,
\end{aligned}
$$
where $\widehat{Z}_1^{(i)} $, $\Delta^{(i)}_T$ are defined analogously as in the proof of Theorem \ref{thm:asymp-normal-varest}. We have $\E (\widetilde{W}^{(1)}_T)^2=\E (\widetilde{W}^{(2)}_T)^2=1 $. Here $\widetilde{W}^{(1)}_T$ and $\widetilde{W}^{(2)}_T$ are standardized averages of $n$ i.i.d. samples and can be regarded as the cores which lead to asymptotic normality of $W^{(1)}_T$, $W^{(2)}_T$. By the proof of Theorem \ref{thm:asymp-normal-varest}, the remainder term $ \Delta^{(i)}_T$ is controlled by:
\begin{equation}\label{eq:remainder}
    \PP\left(\abs{\Delta^{(i)}_T}>c_1 h_n \right)\le C_2 \frac{\log d_1}{d_1^\tau}.
\end{equation}
For $i=1$, as is shown in the proof of Theorem \ref{thm:asymp-normal-varest},  by multivariate Berry--Esseen theorem,   $(\widetilde{W}^{(2)}_{T_1},\widetilde{W}^{(2)}_{T_2})$  converges to normal variable $\omega_1 \sim \cN(0,R)$ conditional on $E_0$ where $R_{11}=R_{22}=1$ and $R_{12}=R_{21}=\operatorname{cov}(\widetilde{W}^1_{T_1},\widetilde{W}^1_{T_2})=\rho_{T_1,T_2}$ with the error bound:

\begin{equation}\label{eq:multi-be}
    \begin{aligned}
    \abs{\PP\left((\widetilde{W}^{(2)}_{T_1},\widetilde{W}^{(2)}_{T_2})\in {A} \middle| E_0 \right)-\PP(\omega_1\in{A})}\le C(1/(1-\rho))^{\frac{3}{2}} \mu \sqrt{\frac{r d_1}{n}},
    \end{aligned}
\end{equation}
for any convex set ${A} \subseteq \R^2$. Here the $\rho:=\rho_{T_1,T_2}$ is the correlation between $W^{(1)}_T$, $W^{(2)}_T$ defined in \eqref{eq:corr}, and it follows $\abs{\rho}\le c q_0^{-\nu}\le \frac{1}{2}$  due to the weak dependency in $\cH_{0,\text{weak} }^2$. By the following calculation of the covariance between $\bbI(W_{T_1}^{\mathsf{rank} }>t)$ and $\bbI(W_{T_2}^{\mathsf{rank} }>t)$ conditional on  $\cD_1$, we have:
\begin{equation*}
\begin{aligned}
     &\abs{\operatorname{cov}(\bbI(W_{T_1}^{\mathsf{rank} }>t),\bbI(W_{T_2}^{\mathsf{rank} }>t)) } \\
     & = \abs{\PP(W^{(1)}_{T_1}W^{(2)}_{T_1} >t, W^{(1)}_{T_2}W^{(2)}_{T_2}>t )- \PP(W^{(1)}_{T_1}W^{(2)}_{T_1} >t )\PP(W^{(1)}_{T_2}W^{(2)}_{T_2} >t ) } \\
     & \le \abs{\PP(W^{(1)}_{T_1}W^{(2)}_{T_1} >t, W^{(1)}_{T_2}W^{(2)}_{T_2}>t )- \PP( {W}^{(1)}_{T_1} w_{11} > t, {W}^{(1)}_{T_2}w_{12} > t  )} \text{ (1)}\\ &+\abs{\PP(W^{(1)}_{T_1}W^{(2)}_{T_1} >t )\PP(W^{(1)}_{T_2}W^{(2)}_{T_2} >t )- \PP( {W}^{(1)}_{T_1} w_{11} > t  )\PP( {W}^{(1)}_{T_2} w_{12} > t  ) } \text{ (2)}\\
     &+\abs{\PP( {W}^{(1)}_{T_1} w_{11} > t, {W}^{(1)}_{T_2}w_{12} > t  )- \PP( {W}^{(1)}_{T_1} w_{11} > t  )\PP( {W}^{(1)}_{T_2} w_{12} > t  ) } \text{ (3)}.\\
\end{aligned}
\end{equation*}
Term (1), (2), (3) can be controlled separately. For (1), conditional on  $\cD_1$, we invoke multivariate Berry--Esseen theorem \eqref{eq:multi-be} to bound the joint c.d.f. of $(W^{(2)}_{T_1}, W^{(2)}_{T_2})$ by  
\begin{equation*}
    \begin{aligned}
    &\PP(W^{(2)}_{T_1} >t_1, W^{(2)}_{T_2}>t_2 ) \\
    &\le  \PP(W^{(2)}_{T_1} >t_1, W^{(2)}_{T_2}>t_2,\abs{\Delta^{(2)}_{T_1}}\le c_1 h_n,\abs{\Delta^{(2)}_{T_2}}\le c_1 h_n  ) +\frac{2c_2 \log d_1}{d_1^\tau}\\ 
    &\le \PP(\widetilde{W}^{(2)}_{T_1} >t_1-c_1 h_n, \widetilde{W}^{(2)}_{T_2}>t_2-c_1 h_n) +\frac{2c_2 \log d_1}{d_1^\tau} \\
    &\le \PP(\omega_{11}>t_1,\omega_{12}>t_2 ) +c_1\left[\phi(t_1)\PP(\omega_{12}>t_2|\omega_{11}=t_1) +\phi(t_2)\PP(\omega_{11}>t_1|\omega_{12}=t_2)\right]h_n +C_2 h_n^2,
    \end{aligned}
\end{equation*}
where we apply Taylor expansion to the c.d.f. of normal distribution $\omega_1 \sim \cN(0,R)$ and apply the upper bound $\log d_1 \cdot d_1^{-\tau} \le h_n^2$. Analogously, it also holds that
\begin{equation*}
    \begin{aligned}
    &\PP(W^{(2)}_{T_1} >t_1, W^{(2)}_{T_2}>t_2)\ge \PP(\widetilde{W}^{(2)}_{T_1} >t_1+c_1 h_n, \widetilde{W}^{(2)}_{T_2}>t_2+c_1 h_n) -\frac{2c_2 \log d_1}{d_1^\tau} \\
    &\ge \PP(\omega_{11}>t_1,\omega_{12}>t_2 ) -c_1\left[\phi(t_1)\PP(\omega_{12}>t_2|\omega_{11}=t_1) +\phi(t_2)\PP(\omega_{11}>t_1|\omega_{12}=t_2)\right]h_n -C_2 h_n^2.
    \end{aligned}
\end{equation*}
We conclude that, conditional on $\cD_1$
\begin{equation*}
\begin{aligned}
&\abs{\PP(W^{(2)}_{T_1} >t_1, W^{(2)}_{T_2}>t_2 )-\PP(\omega_{11}>t_1,\omega_{12}>t_2 ) }\\
& \le c_1\left[\phi(t_1)\PP(\omega_{12}>t_2|\omega_{11}=t_1) +\phi(t_2)\PP(\omega_{11}>t_1|\omega_{12}=t_2)\right]h_n + C_2 h_n^2.
\end{aligned}
\end{equation*}

Using the Lipschitz property of $\Phi(t)$, we have 
\begin{equation}\label{eq:bound-w-omega}
\begin{aligned}
         &\abs{\PP(W^{(1)}_{T_1}W^{(2)}_{T_1} >t, W^{(1)}_{T_2}W^{(2)}_{T_2}>t )- \PP(W^{(1)}_{T_1}\omega_{11}>t, W^{(1)}_{T_2}\omega_{12}>t )} \\
         &\le 2c_1 h_n\left(\PP(W^{(1)}_{T_1}\omega_{11}>t) +\PP(W^{(1)}_{T_2}\omega_{12}>t)  \right) +Ch_n^2.
\end{aligned}
\end{equation}
For (2), the proof of Lemma \ref{lemma:pop-wrank} also implies the following bound 
\begin{equation}\label{eq:bound-w-omega-2}
\begin{aligned}
        &\abs{\PP(W^{(1)}_{T_1}W^{(2)}_{T_1} >t )\PP(W^{(1)}_{T_2}W^{(2)}_{T_2} >t )- \PP( {W}^{(1)}_{T_1} w_{11} > t  )\PP( {W}^{(1)}_{T_2} w_{12} > t  ) } \\
        &\le 2c_1 h_n\left(\PP(W^{(1)}_{T_1}\omega_{11}>t) +\PP(W^{(1)}_{T_2}\omega_{12}>t)  \right) +Ch_n^2, 
\end{aligned}
\end{equation}
by the same argument conditional on $E_0$ and $\cD_1$. Our next step is to compare the c.d.f. of $(\omega_1,\omega_2)$ with standard Gaussian $(Z_1,Z_2)$ to control the term (3), the difference between $\PP(\omega_{11} >t_1, \omega_{12}>t_2)$ and $\Phi(-t_1)\Phi(-t_2) $. Since $(T_1,T_2)\in \cH_{0,\text{weak} }^2$, the covariance between $w_{11},w_{12}$ is thus bounded by: $\abs{\rho}\le c q_0^{-\nu}$.

We invoke the property of bivariate Gaussian copula \citep{meyer2013bivariate}:
\begin{equation*}
    \begin{aligned}
         \abs{\PP(\omega_{11} >t_1, \omega_{12}>t_2)-\Phi(-t_1)\Phi(-t_2) }=\abs{\int_{0}^{\rho} \phi_2(-t_1,-t_2,z ) dz},
    \end{aligned}
\end{equation*}
where $ \phi_2(x,y,z )$ is the p.d.f of bivariate normal distribution with correlation coefficient $z$. Without loss of generality, assume $t_1,t_2>0$ are away from $0$. Thus, it is clear that
\begin{equation*}
    \begin{aligned}
         \abs{\PP(\omega_{11} >t_1, \omega_{12}>t_2)-\Phi(-t_1)\Phi(-t_2) } &\le \int_{0}^{\rho } \phi_2(-t_1,-t_2,z ) dz, \\
         & \le \frac{\rho}{2 \pi \sqrt{1-\rho^2}} \exp \left(-\frac{t_1^2+t_2^2}{2}+\frac{\rho t_1 t_2}{\left(1-\rho^2\right)}\right) \\
         & \le \frac{2\rho}{2 \pi }  \exp \left(-\frac{t_1^2+t_2^2}{2}(1-c\rho)\right) \\
         & = 2\rho \left[\phi(-t_1)\phi(-t_2)\right]^{1-c\rho}.
    \end{aligned}
\end{equation*}
For any $\nu>0$, there exist $C_\nu>0$ such that $\Phi(-t)^\nu \le C_\nu/t $ for all $t>0$. Because by Mill's ratio, we have:
\begin{equation*}
    \Phi(-t)^\nu \le \frac{\phi (-t)^\nu}{t^\nu}\le C_\nu \frac{1}{t^{1-\nu}}\frac{1}{t^\nu}=C_\nu \frac{1}{t},
\end{equation*}
where we use the fact that $\phi (-t)^\nu\le C_\nu t^{-(1-\nu)} $. Now combine this with the upper bound of $\phi (-t)$: $\phi (-t)\le C (t+1)\Phi (-t) $ , we have:
\begin{equation}\label{eq:bound-w-omega-3}
    \begin{aligned}
         \abs{\PP(\omega_{11} >t_1, \omega_{12}>t_2)-\Phi(-t_1)\Phi(-t_2) } &\le 2\rho \left[\phi(-t_1)\phi(-t_2)\right]^{1-c\rho} \\
         & \le 2\rho\left[ C (\Phi(-t_1)^{-\nu}+1 )\Phi(-t_1) (\Phi(-t_2)^{-\nu}+1 )\Phi(-t_2) \right]^{1-c\rho} \\
         & \le C\rho \left[\Phi(-t_1)\Phi(-t_2)\right]^{(1-\nu)(1-c\rho) },
    \end{aligned}
\end{equation}
for the term (3). Together with \eqref{eq:bound-w-omega}, \eqref{eq:bound-w-omega-2}, we can show that
\begin{equation*}
\begin{aligned}
       &\sup_{0\le t\le L_n } \sum_{(T_i,T_j)\in \cH_{0,\text{weak} }^2 }\frac{ \abs{\operatorname{cov}(\bbI(W_{T_i}^{\mathsf{rank} }>t),\bbI(W_{T_j}^{\mathsf{rank} }>t)) }}{q_0^2 G^2(t)}\le  \frac{8c_1 h_n q_0 G(t)  }{q_0^2 G(t)^2} \\
       & +\sup_{0\le t\le L_n }\frac{ \sum_{(T_i,T_j)\in\cH_{0,\text{weak}} ^2 } C\rho \left[\PP({W}^{(1)}_{T_i} Z >t )\PP({W}^{(1)}_{T_j} Z>t)\right]^{(1-\nu)(1-c\rho) }  }{q_0^2 G(t)^2}\\
       & \le  C \frac{h_n q_0}{\epsilon_n \eta_n } + \sup_{0\le t\le L_n }\frac{ C\rho \left(\sum_{T\in \cH_0 }\PP({W}^{(1)}_{T} Z >t )^{(1-\nu)(1-c\rho) }\right)^2  }{q_0^2 G(t)^2} \\ 
       & \le  C \frac{h_n q_0}{\epsilon_n \eta_n } + \sup_{0\le t\le L_n }C\rho \frac{ \left( G(t) \right)^{2(1-\nu)(1-c\rho) }  }{ G(t)^2} \\
       & \le  C \frac{h_n q_0}{\epsilon_n \eta_n } +C {\rho}( \frac{q_0}{\epsilon_n\eta_n})^{3\nu}.
\end{aligned}
\end{equation*}
The argument above is valid for any $\nu$, thus, we choose $3\nu$ to be the $\frac{\nu}{2}$, where $\nu$ is defined in \eqref{eq:weak-corr-1}. It thus finishes the proof.
\end{proof}
We now apply the weak dependency yielded in Lemma \ref{lemma:weak-cov} to derive a uniform bound between $R$ and its population version: 

\begin{Lemma}\label{lemma:conv-prob}
For any $\varepsilon>0$, conditional on $\cD_1$, it holds that
\begin{equation*}
   \PP\left(\sup\limits_{0\le t\le L_n}\abs{\frac{ \sum_{T\in\cH_0}\bbI(W_T^{\mathsf{rank} }>t ) }{q_0 G(t) } -1} \ge \varepsilon\right) \le \frac{C}{\varepsilon^2} \log(\frac{q_0 }{\epsilon_n \eta_n}) \left( \left(\frac{\beta_{\mathsf{s}} q_0^2 }{ \epsilon_n^2\eta_n^2}\right)^{\frac{1}{2}} + \left(\frac{h_n q_0}{\epsilon_n \eta_n}+ \frac{1}{ \left(\epsilon_n\eta_n q_0\right)^{v/2} } \right)^{\frac{1}{2}}\right) .
\end{equation*}
\end{Lemma}

\begin{proof}
To prove the uniform convergence in probability, we define a grid on $[0,L_n]$:
$$\left\{ t_k= G^{-1}\left( \frac{1}{2} (2G(L_n))^{\frac{k}{K}} \right) \right\}_{k=0}^{K},  $$
which equates each $G(t_k)$ with $\frac{1}{2} (2G(L_n))^{\frac{k}{K}}$. Then for each $t\in [t_{k-1},t_{k})$, the ratio can be bounded by:
\begin{equation*}
    \frac{ \sum_{T\in\cH_0}\bbI(W_T^{\mathsf{rank} }>t_{k} ) }{q_0 G(t_{k-1}) } \le \frac{ \sum_{T\in\cH_0}\bbI(W_T^{\mathsf{rank} }>t ) }{q_0 G(t) } \le  \frac{ \sum_{T\in\cH_0}\bbI(W_T^{\mathsf{rank} }>t_{k-1} ) }{q_0 G(t_{k}) }.
\end{equation*}
Define $(2G(L_n))^{\frac{1}{K}}=r_K$, we have $G(t_k)/G(t_{k-1})=r_K$, and 
\begin{equation*}
    \abs{\frac{ \sum_{T\in\cH_0}\bbI(W_T^{\mathsf{rank} }>t ) }{q_0 G(t) } -1}\le \sup_{i=k-1,k}\frac{1}{r_K} \abs{\frac{ \sum_{T\in\cH_0}\bbI(W_T^{\mathsf{rank} }>t_i ) }{q_0 G(t_{i}) }-1} +\abs{r_K-1}\vee\abs{\frac{1}{r_K}-1 },
\end{equation*}
for each $t\in [t_{k-1},t_{k})$. Then for any $t\in [0,L_n]$, it suffices to control the quantities 
\begin{equation*}
    \sup_{k=0,\dots, K}\frac{1}{r_K} \abs{\frac{ \sum_{T\in\cH_0}\bbI(W_T^{\mathsf{rank} }>t_k ) }{q_0 G(t_{k}) }-1}\le  \sup_{k=0,\dots, K}\frac{1}{r_K} \abs{\frac{ \sum_{T\in\cH_0}\bbI(W_T^{\mathsf{rank} }>t_k )-\PP(W_T^{\mathsf{rank} }>t_k ) }{q_0 G(t_{k}) }} + C_3 \frac{h_n q_0}{\epsilon_n \eta_n }
\end{equation*}
and $\abs{r_K-1}\vee\abs{\frac{1}{r_K}-1 }$ by Proposition \ref{prop:asy-w1}.  Denote 
$$D_n=\sup\limits_{k=0,\dots, K}\frac{1}{r_K} \abs{\frac{ \sum_{T\in\cH_0}\bbI(W_T^{\mathsf{rank} }>t_k )-\PP(W_T^{\mathsf{rank} }>t_k ) }{q_0 G(t_{k}) }}.$$ 
It follows that
\begin{equation}
    \begin{aligned}
              &\E D_n^2\le \frac{K}{r_K^2}\E \abs{\frac{ \sum_{T\in\cH_0}\bbI(W_T^{\mathsf{rank} }>t_{k} )-\PP(W_T^{\mathsf{rank} }>t_{k} ) }{q_0 G(t_{k}) }}^2\\
              &\le  \frac{K}{r_K^2} \frac{\sum\limits_{(T_1,T_2)\in\cH^2_{0,\text{weak}}} \abs{\operatorname{cov}(\bbI(W_{T_1}^{\mathsf{rank} }>t),\bbI(W_{T_2}^{\mathsf{rank} }>t)) } +\sum\limits_{T_1,T_2\in\cH^2_{0,\text{strong}} } \abs{\operatorname{cov}(\bbI(W_{T_1}^{\mathsf{rank} }>t),\bbI(W_{T_2}^{\mathsf{rank} }>t)) } }{q_0^2 G^2(t) },
    \end{aligned}
\end{equation}
for any $t\in \{t_k\}$. Since the number of strong dependency pairs $\abs{\cH^2_{0,\text{strong}}}\le \beta_{\mathsf{s}}q_0^2 $, we have
\begin{equation*}
    \frac{\sum\limits_{T_1,T_2\in\cH^2_{0,\text{strong}} } \abs{\operatorname{cov}(\bbI(W_{T_1}^{\mathsf{rank} }>t),\bbI(W_{T_2}^{\mathsf{rank} }>t)) }}{q_0^2 G^2(t) }\le \frac{\beta_{\mathsf{s}}q_0^2 }{ \epsilon_n^2\eta_n^2},
\end{equation*}
for any $t\in [0,L_n]$. For the weak dependency pair, 

\begin{equation*}
    \frac{\sum\limits_{T_1,T_2\in\cH^2_{0,\text{weak}} } \abs{\operatorname{cov}(\bbI(W_{T_1}^{\mathsf{rank} }>t),\bbI(W_{T_2}^{\mathsf{rank} }>t)) }}{q_0^2 G^2(t)}\le C_1 \frac{h_n q_0}{\epsilon_n \eta_n }+C_2 \frac{1}{ \left(\eta_n q_0\right)^{v/2} },
\end{equation*}
 where we apply our previous results in Lemma \ref{lemma:weak-cov}. What remains for us is to specify the density of grid $K$. Choose a constant $\varsigma$ and we set
 $$K= \log(\frac{q_0 }{\epsilon_n \eta_n}) \min \left\{ \left(\frac{ q_0^2 \beta_{\mathsf{s}} }{ \eta_n^2 \epsilon_n  }\right)^{-\varsigma} , \left( \frac{q_0 h_n}{\eta_n\epsilon_n}+\frac{1}{ \left(\epsilon_n\eta_n q_0\right)^{v/2} } \right)^{-\varsigma} \right\}  ,$$ 
 then it is clear that $1\le\frac{1}{r_k}\le \left[ \frac{q_0}{\epsilon_n \eta_n} \right]^{1/K}\to 1$, and $ K (\frac{\beta_{\mathsf{s}}q_0^2 }{ \epsilon_n^2\eta_n^2} + \frac{h_n q_0}{\epsilon_n \eta_n })\to 0$. Therefore 
 \begin{equation*}
\begin{aligned}
      &\abs{r_K-1}\vee\abs{\frac{1}{r_K}-1 } \le C \frac{1}{K}\log(\frac{q_0 }{\epsilon_n \eta_n}) \le  \left( \left(\frac{\beta_{\mathsf{s}} q_0^2 }{ \epsilon_n^2\eta_n^2}\right)^{\varsigma} + \left(\frac{h_n q_0}{\epsilon_n \eta_n} +\frac{1}{ \left(\epsilon_n\eta_n q_0\right)^{v/2} }\right)^{\varsigma}\right) \\
    &\E D^2_n  \le C K \left(\frac{\beta_{\mathsf{s}}q_0^2 }{ \eta_n^2} + \frac{h_n q_0}{\epsilon_n \eta_n} \right)\le C \log(\frac{q_0 }{\epsilon_n \eta_n}) \left( \left(\frac{\beta_{\mathsf{s}} q_0^2 }{ \epsilon_n^2\eta_n^2}+\frac{1}{ \left(\epsilon_n\eta_n q_0\right)^{v/2} }\right)^{1-\varsigma} + \left(\frac{h_n q_0}{\epsilon_n \eta_n} \right)^{1-\varsigma}\right).
\end{aligned}
 \end{equation*}
We can finish the proof of uniform convergence by using Markov's inequality with $\varsigma=\frac{1}{2}$.
\end{proof}

Recall the main theorem. For the ratio $\frac{\sum_{T\in \cH_0  } \bbI(W_T^{\mathsf{rank} } >L ) }{ \sum_{T\in \cH_0  } \bbI(W_T^{\mathsf{rank} } <-L )}$, we have

\begin{equation*}
    \sfR_0=\frac{\sum_{T\in \cH_0  } \bbI(W_T^{\mathsf{rank} } >L ) }{ \sum_{T\in \cH_0  } \bbI(W_T^{\mathsf{rank} } <-L )} =  \frac{\sum_{T\in \cH_0  } \bbI(W_T^{\mathsf{rank} } >L ) }{ q_0 G(t) } \cdot \frac{q_0 G(t) }{  \sum_{T\in \cH_0  } \bbI(W_T^{\mathsf{rank} } <-L ) }.
\end{equation*}
Then, it's clear that, under the event that $L\le L_n$, if Lemma \ref{lemma:conv-prob} holds for a $\varepsilon$, then we have

\begin{equation*}
     \sfR_0\le \frac{1+\varepsilon}{1-\varepsilon}\le 1+\frac{2\varepsilon}{1-\varepsilon}\le 1+3\varepsilon,
\end{equation*}
with probability at least $1-\frac{C}{\varepsilon^2} \log(\frac{q_0 }{\epsilon_n \eta_n}) \left( \left(\frac{\beta_{\mathsf{s}} q_0^2 }{ \epsilon_n^2\eta_n^2}\right)^{\frac{1}{2}} + \left(\frac{h_n q_0}{\epsilon_n \eta_n}+(\epsilon_n\eta_n q_0)^{-\nu/2} \right)^{\frac{1}{2}}\right) $. 
By Lemma \ref{lemma:conv-prob}, we now successfully reduce our problem to proving our data-driven threshold $L\le L_n$ with high probability.

\subsubsection{ Threshold control}\label{sec:prof-thres}
 The gist of asymptotic threshold control is that when we choose $L_n$ as the threshold and $d_1,d_2,n$ go large, entries with strong signals in $\cS$ can always pass the test, and other entries with weak signals or no signal can pass the test will little possibility. We first focus on the entries with strong signals. Denote the standardized signal $\delta_T=(M_T-\theta_T)/({\sigma}_{\xi}\|\calP_M(T)\|_{\rm F}\sqrt{d_1d_2/n} )$, and $\widehat{W}^1_T= W^{(1)}_T- \delta_T$, $\widehat{W}^2_T= W^{(2)}_T- \delta_T$. Given any $T\in\cH_1$, following the argument that is similar to the proof in Lemma \ref{lemma:pop-wrank}, we have 
 \begin{equation*}
    \sup_{t\in \R}\abs{ \PP(W_T^\mathsf{rank}>t )- \PP\left( (Z_1+\delta_T)(Z_2+\delta_T)>t \right) }\le C h_n.
 \end{equation*}
Here $(Z_1,Z_2)$ are independent standard Gaussian. Without loss of generality, assume $M_T-\theta_T>0$. Then, 
 \begin{equation*}
\begin{aligned}
\PP(W_T^\mathsf{rank}< L_n ) \le &\PP\left( (Z_1+\delta_T)(Z_2+\delta_T)< L_n \right) +C h_n \\
\le & 1- \PP\left( (Z_1+\delta_T)(Z_2+\delta_T) \ge L_n \right) +C h_n \\
\le & 1- \PP\left( Z_1\ge -\delta_T+\sqrt{L_n} \right)^2 +C h_n.
\end{aligned}
 \end{equation*}
Here, we use the fact that
\begin{equation*}
    \left\{ Z_1\ge -\delta_T+\sqrt{L_n} \right\} \cap \left\{ Z_2\ge -\delta_T+\sqrt{L_n} \right\} \subseteq \left\{(Z_1+\delta_T)(Z_2+\delta_T) \ge L_n \right\}.
\end{equation*}
 
An upper bound of $G(t)$ is given by 
\begin{equation*}
    G(t)= \frac{\sum_{T\in\cH_0 } \PP( {W}^{(1)}_{T} Z > t |\cD_0,\cD_1 )   }{q_0}\le \frac{\sqrt{2}}{\sqrt{\pi}} \exp{\left( -\frac{t^2}{ 2\max_{T\in\cH_0} \abs{{W}^{(1)}_{T} }^2 } \right)}.
\end{equation*}
From Theorem \ref{thm:asymp-normal}, an uniform upper bound of $\abs{{W}^{(1)}_{T} }$ is given by:
\begin{equation*}
    \PP\left(\max_{T\in\cH_0}\abs{{W}^{(1)}_{T} } \ge C(h_n +\sqrt{\log d_1 })  \right)\le \frac{1}{d_1^{\tau-1} }.
\end{equation*}
If $T\in \cS$, then $\delta_T\ge C_{\mathsf{gap}} \sqrt{\log d_1} $ by the definition of $\cS$. The definition of $L_n$ implies that $L_n  \le C \sqrt{\log(\frac{q_0}{\epsilon_n \eta_n})}\cdot  \sqrt{\log d_1} \ll \log(\frac{1}{h_n})\vee (\log d_1) $. Generally, we have $d^{-10}\le h_n$, thus the term $\log (\frac{1}{h_n})$ can be omitted. Assume $C_{\mathsf{gap}}$ is large. It is clear that 
\begin{equation*}
    \PP\left( Z_1\ge -\delta_T+\sqrt{L_n} \right)^2\ge  \PP\left( Z_1\ge -C\sqrt{ (\log d_1) }  \right)^2\ge (1- c h_n)^2,
\end{equation*}
i.e., $\PP(W_T^\mathsf{rank}< L_n ) \le C h_n$.  For any $\varepsilon>0$, we compute the probability that $\sum\limits_{T\in\cS} \bbI(W_T^\mathsf{rank}> L_n)=\eta_n$ by finding its complement:
 \begin{equation*}
\begin{aligned}
               \PP(\sum\limits_{T\in\cS} \bbI(W_T^\mathsf{rank}> L_n)\le (1-\varepsilon)\eta_n )& =\PP(\sum\limits_{T\in\cS} \bbI(W_T^\mathsf{rank}< L_n)> \varepsilon\eta_n )  \le \frac{\sum\limits_{T\in\cS} \PP(W_T^\mathsf{rank}< L_n)}{\varepsilon\eta_n}\le C h_n/\varepsilon,
\end{aligned}
 \end{equation*}
 i.e., $ \PP(\sum\limits_{T\in\cS} \bbI(W_T^\mathsf{rank}> L_n)\le (1-\varepsilon)\eta_n )\to 0$, $ \PP(\sum\limits_{T\in\cS} \bbI(W_T^\mathsf{rank}> L_n)\ge \eta_n )\to 1$. This indicates that, all the signals in $\cS$ can pass our test. For our data-driven threshold \eqref{eq:dd-threshold}, we have 
 
  \begin{equation}\label{eq:dd-thre-lb}
\begin{aligned}
             \sum\limits_{T\in\cH} \bbI(W_T^\mathsf{rank}> L_n)\ge   \sum\limits_{T\in\cS} \bbI(W_T^\mathsf{rank}> L_n)\ge \frac{3}{4}\eta_n,
\end{aligned}
 \end{equation}
 with probability at least $1-C h_n$

Consider the probability $ \PP(\sum\limits_{T\in\cH_0} \bbI(W_T^\mathsf{rank}<- L_n)\ge \frac{\alpha}{4}\eta_n )$ for the no-signal linear forms $T\in\cH_0$. As we have shown in the proof of Lemma \ref{lemma:conv-prob}, we have 
  \begin{equation*}
\begin{aligned}
            \PP(\sum\limits_{T\in\cH_0} \bbI(W_T^\mathsf{rank}<- L_n)\ge 2\epsilon_n\eta_n  )\le \log(\frac{q_0 }{\epsilon_n \eta_n}) \left( \left(\frac{\beta_{\mathsf{s}} q_0^2 }{ \epsilon_n^2\eta_n^2}\right)^{\frac{1}{2}} + \left(\frac{h_n q_0}{\epsilon_n \eta_n} +(\epsilon_n\eta_n q_0)^{-\nu/2}\right)^{\frac{1}{2}}\right) \to 0,
\end{aligned}
 \end{equation*}
and consequently, by taking $\epsilon_n=\alpha/8$,
  \begin{equation}\label{eq:dd-thre-up}
\begin{aligned}
            \PP(\sum\limits_{T\in\cH} \bbI(W_T^\mathsf{rank}<- L_n)\ge\frac{3}{4}\alpha\eta_n )& \le \PP(2\sum\limits_{T\in\cH_0} \bbI(W_T^\mathsf{rank}<- L_n)\ge \frac{\alpha}{2}\eta_n )+\PP(\sum\limits_{T\in\cS} \bbI(W_T^\mathsf{rank}<- L_n)\ge \frac{\alpha}{4}\eta_n )\\
            &\le \log(\frac{q_0 }{\alpha \eta_n})\left( \left(\frac{\beta_{\mathsf{s}} q_0^2 }{ \alpha^2\eta_n^2}\right)^{\frac{1}{2}} + \left(\frac{h_n q_0}{\alpha \eta_n} +(\epsilon_n\eta_n q_0)^{-\nu/2}\right)^{\frac{1}{2}}\right)+ C h_n.
\end{aligned}
 \end{equation}
Combining \eqref{eq:dd-thre-lb} and \eqref{eq:dd-thre-up}, it is sufficient to conclude that 
$$
\PP\left(\frac{\sum\limits_{T\in\cH}\bbI\left(T: W_T^\mathsf{rank}<-L_n\right)}{\left(\sum\limits_{T\in\cH}\bbI\left(T: W_T^\mathsf{rank}>L_n\right)\right)\vee 1 } \ge \alpha \right)\le \log(\frac{q_0 }{\alpha \eta_n})\left( \left(\frac{\beta_{\mathsf{s}} q_0^2 }{ \alpha^2\eta_n^2}\right)^{\frac{1}{2}} + \left(\frac{h_n q_0}{\alpha \eta_n}+(\alpha\eta_n q_0)^{-\nu/2} \right)^{\frac{1}{2}}\right)+ C h_n,
$$
i.e., $\PP(L\le L_n)\to 1$.

\subsubsection{Power analysis}
From the discussion on the threshold control, it is clear that for any $\varepsilon$,

\begin{equation*}
    \PP(\sum\limits_{T\in\cS} \bbI(W_T^\mathsf{rank}> L_n)\le (1-\varepsilon)\eta_n ) \le \frac{\sum\limits_{T\in\cS} \PP(W_T^\mathsf{rank}< L_n)}{\varepsilon\eta_n}\le Ch_n/\varepsilon.
\end{equation*}
Under the event that $L\le L_n$, this also implies that with probability at least $1-C h_n/\varepsilon$, 

\begin{equation*}
  (1-\varepsilon)\eta_n \le   \sum\limits_{T\in\cS} \bbI(W_T^\mathsf{rank}> L_n) \le \sum\limits_{T\in\cS} \bbI(W_T^\mathsf{rank}> L).
\end{equation*}
The probability of $\{L\le L_n\}$ is lower bounded in Section \ref{sec:prof-thres}. We can, therefore, get the power:
\begin{equation*}
    \text{POWER}= \frac{\sum\limits_{T\in\cH_1} \bbI(W_T^\mathsf{rank}> L)}{q_1}\ge \frac{\sum\limits_{T\in\cS} \bbI(W_T^\mathsf{rank}> L)}{\eta_n} \cdot \frac{\eta_n}{q_1}\ge (1-\varepsilon)\frac{\eta_n}{q_1},
\end{equation*}
with probability at least:
\begin{equation*}
      1-C \log(\frac{q_0 }{\alpha \eta_n})\left( \left(\frac{\beta_{\mathsf{s}} q_0^2 }{ \alpha^2\eta_n^2}\right)^{\frac{1}{2}} + \left(\frac{h_n q_0}{\alpha \eta_n}+(\alpha\eta_n q_0)^{-\nu/2} \right)^{\frac{1}{2}}\right)- C \varepsilon^{-1} h_n.
\end{equation*}

\subsection{Proof of Proposition \ref{prop:hetero}}
This follows the proof of Theorems \ref{thm:asymp-normal}, \ref{thm:asymp-normal-varest} in Section \ref{sec:proof-clt}, and Theorem 7 in \cite{ma2024statistical}. By \eqref{eq:svd-decomp}, we have the following variance computation for the core variance part:
\begin{equation*}
\begin{aligned}
        & \operatorname{Var}\left(\sum_{i=1}^n \frac{d_1 d_2}{n}\xi_i\langle \cP_M( X_i),T\rangle \right)
        \\
        & = \E\left(\sum_{i=1}^n \frac{d_1 d_2}{n}\xi_i\langle \cP_M( X_i),T\rangle\right)^2 
       \\
       & = \E_{X_i} \E\left[\left(\sum_{i=1}^n \frac{d_1 d_2}{n}\xi_i\langle \cP_M( X_i),T\rangle\right)^2\middle| X_i\right]
       \\
       & = \frac{d_1^2 d_2^2}{n}\E \left( \langle S,X_i \rangle \right)^2 \left(\langle\cP_M( X_i),T\rangle\right)^2 =\frac{d_1^2 d_2^2}{n}\E  \left( \langle S,X_i \rangle \right)^2\left(\langle X_i,\cP_M(T) \rangle\right)^2 
       \\
       &=  \frac{d_1 d_2}{ n}\sum_{x\in \mathfrak{E}}  \langle S,x \rangle^2   \langle \cP_M(T),x \rangle^2 = \frac{d_1 d_2}{ n}\norm{S\odot \cP_M(T)}_{\tF}^2.
\end{aligned}
\end{equation*}
Therefore, the Berry-Esseen  bound \eqref{eq:BE-core}  also holds for heterogeneous $\xi_i$ with CLT standard deviation $\sqrt{\frac{d_1 d_2}{n}}\norm{S\odot \cP_M(T)}_{\tF}$.  For the higher-order term $E_T^{\higher}$, together with the variance estimation error, they can all be controlled at the level $h_n\times \sigma \norm{\cP_M(T)}_{\tF} \sqrt{d_1 d_2/n}$,  as is shown in Theorem 7, \cite{ma2024statistical}. Here $h_n$ is  the new rate described in \eqref{eq:hetero-CLT}. Therefore, we prove the claim.

\subsection{Proof of Proposition \ref{prop:weak-cor-aftscr}}\label{sec:proof-aftscr}
\begin{proof}
By definition, we can equally use the covariance matrix $\mathbf{Q}^* = (\mathbf{X}_{\cA}^{*\top}\mathbf{X}_{\cA}^{*} )^{-1}= \left( \Sigma_{\calA}^{ -\frac{1}{2}  \top} \Sigma_{\calA}^{-\frac{1}{2} }\right)^{-1}$ to derive the correlation coefficient matrix. Here in the proof, we use bold symbols like $\mathbf{Q}$ to distinguish our analysis from the $Q$ in the Algorithm \ref{alg:matrix-sda} of the main text, although they lead to the same correlation structure. We will show that, if two linear forms indexed by $T_i$, $T_j$ are weakly correlated in $ \mathbf{Q}^*$, i.e., 
$$ \abs{\frac{\mathbf{Q}^*_{ij} }{\sqrt{\mathbf{Q}^*_{ii}\mathbf{Q}^*_{jj}}} }= \frac{\abs{e_i^\top \left(\mathbf{X}_{\mathcal{A}}^{*\top} \mathbf{X}_{\mathcal{A}}^*\right)^{-1}e_j} }{\sqrt{\mathbf{Q}^*_{ii}\mathbf{Q}^*_{jj}}} \le C_1 q_n^{-\nu}, $$ 
then, in the data-driven covariance matrix $\mathbf{Q}$, they are also weakly correlated:
\begin{equation*}
    \abs{\frac{\mathbf{Q}_{ij} }{\sqrt{\mathbf{Q}_{ii}\mathbf{Q}_{jj}}} } \le C_2 q_n^{-\nu},
\end{equation*}
with probability at least $1- Cd_1^{-2}\log d_1 $.
By definition, the covariance matrix of $\wt\sfw^{(2)}$ without normalization is 
$$
\begin{aligned}
     \mathbf{Q}=&\left(\mathbf{X}_{\mathcal{A}}^{\top} \mathbf{X}_{\mathcal{A}}\right)^{-1} \mathbf{X}_{\mathcal{A}}^{\top} \mathbf{X} \Sigma  \mathbf{X}^{\top} \mathbf{X}_{\mathcal{A}}\left(\mathbf{X}_{\mathcal{A}}^{\top} \mathbf{X}_{\mathcal{A}}\right)^{-1} \\
     =& \left(\mathbf{X}_{\mathcal{A}}^{\top} \mathbf{X}_{\mathcal{A}}\right)^{-1}+\left(\mathbf{X}_{\mathcal{A}}^{\top} \mathbf{X}_{\mathcal{A}}\right)^{-1}\mathbf{X}_{\mathcal{A}}^{\top} \Delta \Sigma \mathbf{X}_{\mathcal{A}}\left(\mathbf{X}_{\mathcal{A}}^{\top} \mathbf{X}_{\mathcal{A}}\right)^{-1},
\end{aligned}
$$
where we define $\Delta \Sigma= \mathbf{X} \Sigma  \mathbf{X}^{\top}- I = \widehat{\Sigma}^{-\frac{1}{2}}(\Sigma-\widehat{\Sigma})\widehat{\Sigma}^{-\frac{1}{2}} $.
The following Lemma characterizes the precision of our covariance estimation:
\begin{Lemma}\label{lemma:cov-est-prec}
    Suppose that we use $\widehat{U}=\widehat{U}^{\init}$, $\widehat{V}=\widehat{V}^{\init}$ obtained from  $\cD_1$ to estimate $\Sigma$: 
    $$\widehat{\Sigma}= T_{\calH}(I_{d_1 d_2} - \widehat U_\perp \widehat U_\perp^\top \otimes \widehat V_\perp \widehat V_\perp^\top  ) T_{\calH}^\top. $$ 
    Then with probability at least $1-Cd_1^{-\tau}\log d_1 $, we have 
    \begin{equation}\label{eq:cov-est-prec-1}
        \norm{\Sigma^{-\frac{1}{2} }(\Sigma-\widehat{\Sigma})\Sigma^{-\frac{1}{2} } }\le CC_{\init}\frac{\kappa_T \sigma_\xi }{\lambda_{\min} }\left( \frac{\operatorname{supp}(T_{\calH} )}{\sqrt{d_2}}\wedge 1 \right)  \sqrt{\frac{ \kappa_1  d_1^2 d_2 \log d_1 }{n}}.
    \end{equation}
\end{Lemma}
For simplicity, we denote $\kappa_T' = \kappa_T\left( \frac{\operatorname{supp}(T_{\calH} )}{\sqrt{d_2}}\wedge 1 \right)  $. Lemma \ref{lemma:cov-est-prec} implies the bound of eigenvalue : $\abs{\lambda_i(\Sigma^{-\frac{1}{2} }\widehat{\Sigma} \Sigma^{-\frac{1}{2} } )-1}=o_p(1)$ for all eigenvalues. Thus, the  eigenvalues of its inverse can also be bounded by the rate in \eqref{eq:cov-est-prec-1}, i.e.,
\begin{equation*}
    \norm{\Delta \Sigma}\le CC_{\init}\frac{\kappa_T' \sigma_\xi }{\lambda_{\min} }\sqrt{\frac{ \kappa_1  d_1^2 d_2 \log d_1 }{n}}.
\end{equation*}
We then have
\begin{equation}\label{eq:Qij-decomp}
    \abs{\mathbf{Q}_{ij}- \mathbf{Q}_{ij}^*} \le \abs{e_i^\top \left( \left(\mathbf{X}_{\mathcal{A}}^{\top} \mathbf{X}_{\mathcal{A}}\right)^{-1}-\left(\mathbf{X}_{\mathcal{A}}^{*\top} \mathbf{X}_{\mathcal{A}}^*\right)^{-1}\right)e_j} + \abs{e_i^\top \left(\mathbf{X}_{\mathcal{A}}^{\top} \mathbf{X}_{\mathcal{A}}\right)^{-1}\mathbf{X}_{\mathcal{A}}^{\top} \Delta \Sigma \mathbf{X}_{\mathcal{A}}\left(\mathbf{X}_{\mathcal{A}}^{\top} \mathbf{X}_{\mathcal{A}}\right)^{-1}e_j}.
\end{equation}
Denote $\mathbf{Q}'= \left(\mathbf{X}_{\mathcal{A}}^{\top} \mathbf{X}_{\mathcal{A}}\right)^{-1} $. The first term in \eqref{eq:Qij-decomp} can be controlled by:
\begin{equation}\label{eq:}
   \begin{aligned}
      &\abs{e_i^\top \left( \left(\mathbf{X}_{\mathcal{A}}^{\top} \mathbf{X}_{\mathcal{A}}\right)^{-1}-\left(\mathbf{X}_{\mathcal{A}}^{*\top} \mathbf{X}_{\mathcal{A}}^*\right)^{-1}\right)e_j} \\
    &= \abs{e_i^\top  \left(\mathbf{X}_{\mathcal{A}}^{\top} \mathbf{X}_{\mathcal{A}}\right)^{-1}\left( \mathbf{X}_{\mathcal{A}}^{*\top} \mathbf{X}_{\mathcal{A}}^* - \mathbf{X}_{\mathcal{A}}^{\top} \mathbf{X}_{\mathcal{A}}\right)\left(\mathbf{X}_{\mathcal{A}}^{*\top} \mathbf{X}_{\mathcal{A}}^*\right)^{-1}e_j} \\
    & \le CC_{\init}\frac{\kappa_1^{1.5} \kappa_T' \sigma_{\xi} }{ \lambda_{\min}  }\cdot \sqrt{\frac{ d_1^2 d_2 \log d_1}{n}}\sqrt{\mathbf{Q}_{ii}'\mathbf{Q}_{jj}^* },
   \end{aligned}
\end{equation}
where we use the fact that 
\begin{equation*}
\begin{aligned}
      \norm{\mathbf{X}_{\mathcal{A}}^{*\top} \mathbf{X}_{\mathcal{A}}^* - \mathbf{X}_{\mathcal{A}}^{\top} \mathbf{X}_{\mathcal{A}}}&\le \norm{\widehat{\Sigma}^{-1}-\Sigma^{-1} }\le \norm{\Sigma^{-1}(\widehat{\Sigma}-\Sigma)\Sigma^{-1}  } + O\left(\norm{\widehat{\Sigma}-\Sigma }^2 \right) \\
      & \le \frac{1}{\lambda_{\min}(\Sigma) } \norm{\Sigma^{-\frac{1}{2} }(\Sigma-\widehat{\Sigma})\Sigma^{-\frac{1}{2} }} + o\left(\norm{\widehat{\Sigma}-\Sigma } \right) \\
      & \le CC_{\init}\frac{\kappa_T' \sigma_\xi }{\lambda_{\min}(\Sigma)\lambda_{\min} }\sqrt{\frac{ \kappa_1  d_1^2 d_2 \log d_1 }{n}},
\end{aligned}
\end{equation*}
by Fréchet derivative \citep{higham2008functions,al2009computing} and Lemma \ref{lemma:cov-est-prec}, and also we have
\begin{equation*}
\begin{aligned}
      \norm{\left(\mathbf{X}_{\mathcal{A}}^{*\top} \mathbf{X}_{\mathcal{A}}^*\right)^{-1}e_j}^2 & \le \frac{1}{\lambda_{\min}\left(\mathbf{X}_{\mathcal{A}}^{*\top} \mathbf{X}_{\mathcal{A}}^*\right)  } \norm{ e^\top_j \left(\mathbf{X}_{\mathcal{A}}^{*\top} \mathbf{X}_{\mathcal{A}}^*\right)^{-1}\mathbf{X}_{\mathcal{A}}^{*\top} \mathbf{X}_{\mathcal{A}}^*\left(\mathbf{X}_{\mathcal{A}}^{*\top} \mathbf{X}_{\mathcal{A}}^*\right)^{-1}e_j } \\
      & \le \lambda_{\max}(\Sigma) \mathbf{Q}_{jj}^*,
\end{aligned}
\end{equation*}

\begin{equation*}
\begin{aligned}
      \norm{\left(\mathbf{X}_{\mathcal{A}}^{\top} \mathbf{X}_{\mathcal{A}}\right)^{-1}e_i}^2 & \le \frac{1}{\lambda_{\min}\left(\mathbf{X}_{\mathcal{A}}^{*\top} \mathbf{X}_{\mathcal{A}}^*\right)  } \big| e^\top_i \left(\mathbf{X}_{\mathcal{A}}^{\top} \mathbf{X}_{\mathcal{A}}\right)^{-1}\mathbf{X}_{\mathcal{A}}^{\top} \mathbf{X}_{\mathcal{A}}\left(\mathbf{X}_{\mathcal{A}}^{\top} \mathbf{X}_{\mathcal{A}}\right)^{-1}e_i \big|  \\
      & + \frac{1}{\lambda_{\min}\left(\mathbf{X}_{\mathcal{A}}^{*\top} \mathbf{X}_{\mathcal{A}}^*\right)  } \big| e^\top_i \left(\mathbf{X}_{\mathcal{A}}^{\top} \mathbf{X}_{\mathcal{A}}\right)^{-1}\left( \mathbf{X}_{\mathcal{A}}^{*\top} \mathbf{X}_{\mathcal{A}}^* - \mathbf{X}_{\mathcal{A}}^{\top} \mathbf{X}_{\mathcal{A}}\right) \left(\mathbf{X}_{\mathcal{A}}^{\top} \mathbf{X}_{\mathcal{A}}\right)^{-1}e_i \big| \\
      & \le \lambda_{\max}(\Sigma) \mathbf{Q}_{ii}'+ CC_{\init}\frac{\kappa_1^{1.5} \kappa_T'  \sigma_{\xi} }{ \lambda_{\min}  }\cdot \sqrt{\frac{ d_1^2 d_2 \log d_1}{n}}\norm{\left(\mathbf{X}_{\mathcal{A}}^{\top} \mathbf{X}_{\mathcal{A}}\right)^{-1}e_i}^2,
\end{aligned}
\end{equation*}
which is equivalent to 
\begin{equation*}
    \begin{aligned}
       \norm{\left(\mathbf{X}_{\mathcal{A}}^{*\top} \mathbf{X}_{\mathcal{A}}^*\right)^{-1}e_j} &\le \sqrt{\lambda_{\max}(\Sigma) \mathbf{Q}_{jj}^*}\\
      \norm{\left(\mathbf{X}_{\mathcal{A}}^{\top} \mathbf{X}_{\mathcal{A}}\right)^{-1}e_i} &\le (1+c)\sqrt{\lambda_{\max}(\Sigma) \mathbf{Q}_{ii}'}.
    \end{aligned}
\end{equation*}
The second term in \eqref{eq:Qij-decomp} can be bounded given that
\begin{equation*}
    \begin{aligned}
      \abs{e_i^\top \left(\mathbf{X}_{\mathcal{A}}^{\top} \mathbf{X}_{\mathcal{A}}\right)^{-1}\mathbf{X}_{\mathcal{A}}^{\top} \Delta \Sigma \mathbf{X}_{\mathcal{A}}\left(\mathbf{X}_{\mathcal{A}}^{\top} \mathbf{X}_{\mathcal{A}}\right)^{-1}e_j} &\le \norm{\mathbf{X}_{\mathcal{A}}\left(\mathbf{X}_{\mathcal{A}}^{\top} \mathbf{X}_{\mathcal{A}}\right)^{-1}e_j} \norm{\mathbf{X}_{\mathcal{A}}\left(\mathbf{X}_{\mathcal{A}}^{\top} \mathbf{X}_{\mathcal{A}}\right)^{-1}e_i} \norm{\Delta \Sigma} \\
      & = \sqrt{\mathbf{Q}_{ii}'\mathbf{Q}_{jj}'}\norm{\Delta \Sigma}  \\
      & \le CC_{\init}\frac{\kappa_1^{1.5} \kappa_T'  \sigma_{\xi} }{ \lambda_{\min}  }\cdot \sqrt{\frac{ d_1^2 d_2 \log d_1}{n}} \sqrt{\mathbf{Q}_{ii}'\mathbf{Q}_{jj}'}.
    \end{aligned}
\end{equation*}
However, notice that, 
\begin{equation*}
    \abs{\frac{\mathbf{Q}_{ii}- \mathbf{Q}_{ii}'}{\mathbf{Q}_{ii}' } }\le CC_{\init}\frac{\kappa_1^{1.5} \kappa_T'  \sigma_{\xi} }{ \lambda_{\min}  }\cdot \sqrt{\frac{ d_1^2 d_2 \log d_1}{n}}.
\end{equation*}
We can conclude that 
\begin{equation*}
    \abs{\mathbf{Q}_{ij}-\mathbf{Q}_{ij}^*} \le CC_{\init}\frac{\kappa_1^{1.5} \kappa_T'  \sigma_{\xi} }{ \lambda_{\min}  }\cdot \sqrt{\frac{ d_1^2 d_2 \log d_1}{n}} \left(\sqrt{\mathbf{Q}_{ii}\mathbf{Q}_{jj}^* }+ \sqrt{\mathbf{Q}_{ii}\mathbf{Q}_{jj} } \right).
\end{equation*}
Setting $i=j$, we also have
\begin{equation*}
    \frac{\abs{\mathbf{Q}_{jj}-\mathbf{Q}_{ij}^*}}{ \mathbf{Q}_{jj}} \le CC_{\init}\frac{\kappa_1^{1.5} \kappa_T'  \sigma_{\xi} }{ \lambda_{\min}  }\cdot \sqrt{\frac{ d_1^2 d_2 \log d_1}{n}}\left(\sqrt{1+{\frac{\abs{\mathbf{Q}_{jj}-\mathbf{Q}_{jj}^*}}{ \mathbf{Q}_{jj}}} }+ 1 \right),
\end{equation*}
i.e.,
\begin{equation*}
      \frac{\abs{\mathbf{Q}_{jj}-\mathbf{Q}_{jj}^*}}{ \mathbf{Q}_{jj}} \le CC_{\init}\frac{\kappa_1^{1.5} \kappa_T'  \sigma_{\xi} }{ \lambda_{\min}  }\cdot \sqrt{\frac{ d_1^2 d_2 \log d_1}{n}}.
\end{equation*}
We now compare the difference of correlation coefficients:
\begin{equation*}
\begin{aligned}
       \abs{\frac{\mathbf{Q}_{ij} }{\sqrt{\mathbf{Q}_{ii}\mathbf{Q}_{jj}}} - {\frac{\mathbf{Q}_{ij}^* }{\sqrt{\mathbf{Q}_{ii}^*\mathbf{Q}_{jj}^*}} } } &\le \frac{  \abs{\mathbf{Q}_{ij}-\mathbf{Q}_{ij}^*}  }{ \sqrt{\mathbf{Q}_{ii}\mathbf{Q}_{jj}} }+ \abs{\mathbf{Q}_{ij}^*}\frac{\abs{\sqrt{\mathbf{Q}_{ii}\mathbf{Q}_{jj}}-\sqrt{\mathbf{Q}_{ii}^*\mathbf{Q}_{jj}^*}}}{\sqrt{\mathbf{Q}_{ii}\mathbf{Q}_{jj}} \sqrt{\mathbf{Q}_{ii}^*\mathbf{Q}_{jj}^*}  }\\
    & + \abs{\mathbf{Q}_{ij}-\mathbf{Q}_{ij}^*} \frac{\abs{\sqrt{\mathbf{Q}_{ii}\mathbf{Q}_{jj}}-\sqrt{\mathbf{Q}_{ii}^*\mathbf{Q}_{jj}^*}}}{\sqrt{\mathbf{Q}_{ii}\mathbf{Q}_{jj}} \sqrt{\mathbf{Q}_{ii}^*\mathbf{Q}_{jj}^*}  } \\
    &\le  CC_{\init}\frac{\kappa_1^{1.5} \kappa_T'  \sigma_{\xi} }{ \lambda_{\min}  }\cdot \sqrt{\frac{ d_1^2 d_2 \log d_1}{n}}.
\end{aligned}
\end{equation*}
If the assumption on the signal strength, i.e.,
\begin{equation*}
     \frac{\kappa_1^{1.5} \kappa_T'  \sigma_{\xi} }{ \lambda_{\min}  }\cdot \sqrt{\frac{ d_1^2 d_2 \log d_1}{n}} \lesssim \frac{1}{q^{\nu}},
\end{equation*}
is satisfied, we also have $ { \abs{\mathbf{Q}_{ij}} }/{\sqrt{\mathbf{Q}_{ii}\mathbf{Q}_{jj}}}\lesssim q^{-\nu} $, which indicates that these two linear forms are also weakly correlated in data-driven covariance matrix $\mathbf{Q}$.
\end{proof}

\subsection{Proof of Proposition \ref{prop:lasso-scr}}
\begin{proof}
We start with the decomposition of LASSO response $\mathbf{y}_1 = \mathbf{X}\mathbf{W}^{(1)}  $:
\begin{equation*}
   \mathbf{y}_1= \widehat{\Sigma}^{-\frac{1}{2}}\widehat{D}\widehat{\sfw} + \widehat{\Sigma}^{-\frac{1}{2}}\widehat{D}\widetilde{\mathbf{W}},
\end{equation*}
where $\widehat{\sfw}_i=  \frac{M_{T_i}-\theta_{T_i}}{\widehat\sigma_\xi^{(1)} \sqrt{d_1 d_2}\widehat{s}_{T_i}^{(1)} } \sqrt{n}$ is the standardized signals with variance estimation with respect to $T_i$, $\widetilde{\mathbf{W}}_i=\mathbf{W}_i^{(1)}/\widehat s_{T_i}^{(1)} -\sfw_i $ is the asymptotic normal noise. Here recall that $M_{T_i}:=\langle M, T_i\rangle$ and $\wt s_{T_i}^{(1)}=\big\|\calP_{\wt M^{(1)}_{\init}}(T_i)\big\|_{\rm F}$.

Our loading matrix is $\widehat{\Sigma}^{-\frac{1}{2}}\widehat{D}$, with
\begin{equation*}
\lambda_{\min}(\widehat{\Sigma}^{-\frac{1}{2}}\widehat{D})   = \frac{1}{\norm{\widehat{\Sigma}^{\frac{1}{2}}\widehat{D}^{-1} } }=\frac{1}{\sqrt{\norm{\widehat{D}^{-1}\widehat{\Sigma}^{}\widehat{D}^{-1} } }}\ge \frac{1}{\sqrt{\norm{\widehat{D}^{-1}{\Sigma}^{}\widehat{D}^{-1} } + \norm{\widehat{D}^{-1}\left({\Sigma}^{} -\widehat{\Sigma}^{} \right) \widehat{D}^{-1} } }}
\end{equation*}
We now denote $\rho_T=\norm{T}_{\ell_1}/\norm{T}_{\tF}$. By \cite{xia2021statistical,ma2024statistical}, we have $\abs{1- \widehat{s}_T^{(1)}/s_T}\le  C_2  \frac{\mu\rho_T }{\beta_0} \cdot \frac{\sigma_{\xi}}{\lambda_{\min}} \sqrt{\frac{\alpha_d d_1^2 d_2 \log d_1}{n}} $ with probability at least $1-Cd_1^{-\tau}\log d_1 $. Here $D:={\rm diag}(s_{T_1},\cdots, s_{T_q})$. 
Thus, the absolute value of the diagonal matrix can be controlled by:
\begin{equation}
   \abs{ D^{-1}-\widehat{D}^{-1}}\preceq C_2  \frac{\mu\rho_T }{\beta_0} \cdot \frac{\sigma_{\xi}}{\lambda_{\min}} \sqrt{\frac{\alpha_d d_1^2 d_2 \log d_1}{n}} D^{-1}.
\end{equation}
This indicates that 
\begin{equation*}
   \norm{\widehat{D}^{-1}{\Sigma}^{}\widehat{D}^{-1} }\le (1+c)\norm{D^{-1}\Sigma D^{-1} }\le \frac{3}{2}\kappa_1,
\end{equation*}
for a small $c$ as long as $  \frac{\mu\rho_T }{\beta_0} \cdot \frac{\sigma_{\xi}}{\lambda_{\min}} \sqrt{\frac{\alpha_d d_1^2 d_2 \log d_1}{n}} \to 0$; and also
\begin{equation*}
    \norm{\widehat{D}^{-1}\left({\Sigma}^{} -\widehat{\Sigma}^{} \right) \widehat{D}^{-1} } \le (1+c)  \norm{{D}^{-1}\left({\Sigma}^{} -\widehat{\Sigma}^{} \right) {D}^{-1} }\le C C_{\init}\frac{\rho_T \mu \sigma_\xi }{\beta_0 \lambda_{\min} }\sqrt{\frac{\alpha_d \kappa_1 q d_1^2 d_2 \log d_1 }{n}},
\end{equation*}
which can be derived following the same steps as Lemma \ref{lemma:cov-est-prec}. It thus gives the well-conditioning of our loading matrix in LASSO:
\begin{equation*}
    \lambda_{\min} \left(\widehat{\Sigma}^{-\frac{1}{2}}\widehat{D}\right)\ge \frac{1}{\sqrt{2\kappa_1}}.
\end{equation*}
Following a classic argument on the LASSO precision analysis \citep{van2009conditions,buhlmann2011statistics}, we have

\begin{equation*}
\begin{aligned}
  \norm{\widehat{\Sigma}^{-\frac{1}{2}}\widehat{D}\left(\wt\sfw^{(1)}  - \widehat{\sfw} \right) }^2\le & 2\left\langle \widehat{D} \widehat{\Sigma}^{-1}\widehat{D} \widetilde{\mathbf{W}},\wt\sfw^{(1)}-\widehat{ \sfw}  \right\rangle   +2\lambda \left(\norm{\widehat{ \sfw} }_{\ell_1} -\norm{\wt\sfw^{(1)}}_{\ell_1} \right) \\
  \le & \mathcal{\lambda} \norm{\wt\sfw^{(1)}-\widehat{ \sfw} }_{\ell_1} +2\lambda \left(\norm{\widehat{ \sfw} }_{\ell_1} -\norm{\wt\sfw^{(1)}}_{\ell_1} \right),
\end{aligned}
\end{equation*}
where we define $\lambda$ as the value that $\PP\left(2\norm{\widehat{D} \widehat{\Sigma}^{-1}\widehat{D} \widetilde{\mathbf{W}}}_{\infty} \ge \lambda \right)\le d_1^{-2}$. It is thus clear that 
\begin{equation*}
\begin{aligned}
  \norm{\widehat{\Sigma}^{-\frac{1}{2}}\widehat{D}\left(\wt\sfw^{(1)}  - \widehat{\sfw} \right) }^2 \le & 3\lambda \norm{\wt\sfw^{(1)}_{s}-\widehat{ \sfw}_{s} }_{\ell_1}  \le 3\lambda\sqrt{q_1} \norm{\wt\sfw^{(1)}-\widehat{ \sfw} }.
\end{aligned}
\end{equation*}
Here, we use the subscript $s$ to denote the support set of $\sfw$. Combined with the well-conditioning property of $\widehat{\Sigma}^{-\frac{1}{2}}\widehat{D}$, we have
\begin{equation*}
    \frac{1}{2\kappa_1 }\norm{\wt\sfw^{(1)}  - \widehat{\sfw} }^2\le   \norm{\widehat{\Sigma}^{-\frac{1}{2}}\widehat{D}\left(\wt\sfw^{(1)}  - \widehat{\sfw} \right) }^2 \le 3\lambda\sqrt{q_1} \norm{\wt\sfw^{(1)}-\widehat{ \sfw} },
\end{equation*}
i.e., $\norm{\wt\sfw^{(1)}  - \widehat{\sfw} }\le 6\lambda \kappa_1 \sqrt{q_1}$. Then, it amounts to determining the regularization level $\lambda$. Notice that $\widehat{D} \widetilde{\mathbf{W}}={D} \widehat{\mathbf{W}}$, where $\widehat{\mathbf{W}}_i=\mathbf{W}_i^{(1)}/s_{T_i}-  \frac{M_{T_i}-\theta_{T_i}}{\widehat \sigma_\xi s_{T_i} \sqrt{d_1 d_2} } \sqrt{n} $. Here $\widehat{\mathbf{W}}_i$ and $\widetilde{\mathbf{W}}_i$ only differ  in the sampling variance $s_{T_i}$.
We adopt the notation in the proof of Theorem \ref{thm:asymp-normal}: we define an average of i.i.d. matrix as $\widehat{Z}_1= \frac{d_1 d_2}{n} \sum_{i\in I_2 } \xi_i X_i $, and split the noise $\widehat{\mathbf{W}}=\widehat{\mathbf{W}}_1+\widehat{\mathbf{W}}_2$, where 
\begin{equation}\label{eq:noise-W-decomp}
    \widehat{\mathbf{W}}_{1i}= \frac{\left\langle \widehat{Z}_1, \cP_M(T_i) \right\rangle }{\sigma_\xi s_{T_i} \sqrt{d_1 d_2/n} }, \text{~for~each~} i\in[q].
\end{equation}
By Theorem \ref{thm:asymp-normal}, we have $\norm{\widehat{\mathbf{W}}_2}_\infty \le C h_n$, with probability at least $1-C d_1^2$. Therefore,
\begin{equation*}
    \begin{aligned}
      2\norm{\widehat{D} \widehat{\Sigma}^{-1}\widehat{D} \widetilde{\mathbf{W}}}_{\infty} & =2\norm{\widehat{D} \widehat{\Sigma}^{-1}{D} \widehat{\mathbf{W}}}_{\infty} \le 2(1+ch_n)\norm{{D} \widehat{\Sigma}^{-1}{D} \widehat{\mathbf{W}}}_{\infty} \\
      & \le 3 \left(\norm{{D} {\Sigma}^{-1}{D} \widehat{\mathbf{W}}}_{\infty} +\norm{{D} (\widehat{\Sigma}^{-1}-{\Sigma}^{-1}) {D} \widehat{\mathbf{W}}}_{\infty} \right)\\
      & \le 3\left( \norm{{D} {\Sigma}^{-1}{D} \widehat{\mathbf{W}}_1}_{\infty}+\norm{{D} {\Sigma}^{-1}{D} \widehat{\mathbf{W}}_2}_{\infty}+\norm{{D} (\widehat{\Sigma}^{-1}-{\Sigma}^{-1}) {D} \widehat{\mathbf{W}}}_{\infty} \right).
    \end{aligned}
\end{equation*}
For any $i$, it is clear that 
\begin{equation*}
    \begin{aligned}
      e_i^\top {D} {\Sigma}^{-1}{D} \widehat{\mathbf{W}}_1 = \frac{\left\langle \operatorname{Vec}(\widehat{Z}_1)^\top, e_i^\top {D} {\Sigma}^{-1}T_{\calH} (I_{d_1 d_2} - U_\perp U_\perp^\top \otimes V_\perp V_\perp^\top ) \right\rangle }{\sigma_\xi \sqrt{d_1 d_2/n} },
    \end{aligned}
\end{equation*}
with 
\begin{equation*}
    \begin{aligned}
      \E \left(e_i^\top {D} {\Sigma}^{-1}{D} \widehat{\mathbf{W}}_1 \right)^2 &= e_i^\top D \Sigma^{-1} D e_i\le \kappa_1. \\
    \end{aligned}
\end{equation*}
According to Bernstein inequality, we have

\begin{equation*}
    \frac{1}{\sqrt{n}} \abs{\frac{  e_i^\top {D} {\Sigma}^{-1}{D} \widehat{\mathbf{W}}_1}{\left(e_i^\top D \Sigma^{-1} D e_i\right)^{\frac{1}{2}}} } \le C_1 \sqrt{\frac{(\log d_1) }{n}}+ C_2\frac{\sqrt{r d_1}(\log d_1) }{n},
\end{equation*}
with probability at least $1-q^{-1}d_1^{-\tau}$.
This indicates that 
\begin{equation*}
    \PP\left(\norm{{D} {\Sigma}^{-1}{D} \widehat{\mathbf{W}}_1}_{\infty}\ge C\sqrt{\kappa_1(\log d_1)}  \right) \le d_1^{-\tau}.
\end{equation*}
If we use $\widehat U$, $\widehat V$ to estimate $\Sigma$, then a corresponding accuracy in $\|\cdot\|_{\infty}$-norm is given by:
\begin{Lemma}\label{lemma:pres-DSD} If we use $\widehat{\Sigma}$ to approximate $\Sigma$, then
\begin{equation*}
            \norm{{D} (\widehat{\Sigma}^{-1}-{\Sigma}^{-1}) {D}}_{\infty} \le C C_{\init}\left(\kappa_\infty\sqrt{\kappa_1} +\kappa_1^{1.5}\kappa_T \left( \frac{\operatorname{supp}(T_{\calH} )}{\sqrt{d_2}}\wedge 1 \right)  \right)\frac{\rho_T \mu \sigma_\xi }{\beta_0 \lambda_{\min} }\sqrt{\frac{\alpha_d q d_1^2 d_2 \log d_1 }{n}}.
\end{equation*}
Here $\|M\|_{\infty}:=\max_i\|e_i^{\top}M\|_{\ell_1}$ and $\kappa_{\infty}:=\|R^{-1}\|_{\infty}$ where $R=D^{-1}\Sigma D^{-1}$. 
\end{Lemma}
Notice that, since $\widehat{\mathbf{W}}_{1i}$ is standardized, Bernstein inequality also gives the bound:
\begin{equation*}
    \norm{\widehat{\mathbf{W}}_1  }_\infty \le C\sqrt{\log d_1 },
\end{equation*}
with probability at least $1-d_1^{-\tau}$. This indicates that, with probability at least $1-C d^{-\tau}_1$, we have
\begin{equation*}
      2\norm{\widehat{D} \widehat{\Sigma}^{-1}\widehat{D} \widetilde{\mathbf{W}}}_{\infty} \le C \left( \sqrt{\kappa_1(\log d_1)} +\kappa_\infty h_n \right) \le C \sqrt{\kappa_1(\log d_1)},
\end{equation*}
as long as $\left(\kappa_\infty h_n \right)\vee \left(\left(\kappa_\infty\sqrt{\kappa_1} +\kappa_1^{1.5}\kappa_T \left( \frac{\operatorname{supp}(T_{\calH} )}{\sqrt{d_2}}\wedge 1 \right)  \right)\frac{\rho_T \mu \sigma_\xi }{\beta_0 \lambda_{\min} }\sqrt{\frac{\alpha_d q d_1^2 d_2 \log d_1 }{n}}\right)\le c \sqrt{\kappa_1}$ for some small constant $c$. Here we use the fact $\norm{{D} {\Sigma}^{-1}{D} \widehat{\mathbf{W}}_2}_{\infty} \le C\kappa_\infty h_n$.This leads to the error bound of $\wt\sfw^{(1)}$:

\begin{equation*}
   \norm{\wt\sfw^{(1)}  - \widehat{\sfw} }_\infty \le \norm{\wt\sfw^{(1)}  - \widehat{\sfw} }\le 6\lambda \kappa_1 \sqrt{q_1}\le C \kappa_1^{1.5}\sqrt{q_1(\log d_1 )}.
\end{equation*}
Since for each $i$, 
$$\abs{\widehat{\sfw}_i-\sfw_i}\le \left(\frac{C_1 \tau \log d_1}{\sqrt{n}}+C_2 \gamma_n^2 +C_3 \mu \frac{\|T\|_{\ell_1}}{\|T\|_\tF \beta_0} \cdot \frac{\sigma_{\xi}}{\lambda_{\min}} \sqrt{\frac{\tau \alpha_d d_1^2 d_2 \log d_1}{n}}\right)\abs{\sfw_i} \le C h_n \abs{\sfw_i}, $$
we finish the proof.
\end{proof}

\subsection{Proof of Proposition \ref{prop:OLS-normal}}
\begin{proof}

We proceed to discuss the asymptotic normality of each $e_i^{\top}\left(\mathbf{X}_{\mathcal{A}}^{\top} \mathbf{X}_{\mathcal{A}}\right)^{-1} \mathbf{X}_{\mathcal{A}}^{\top} \mathbf{y}_2$: since $\mathbf{y}_2 = \mathbf{X}\mathbf{W}^{(2)}  $, with
\begin{equation*}
   \mathbf{y}_2= \widehat{\Sigma}^{-\frac{1}{2}}{D}\widehat{\sfw} + \widehat{\Sigma}^{-\frac{1}{2}}{D}\widehat{\mathbf{W}},
\end{equation*}
where, with a slight abuse of notation, we define $\widehat{\sfw}_i=  \frac{M_T-\theta_T}{\widehat\sigma_\xi \sqrt{d_1 d_2}{s}_T } \sqrt{n}$ is the standardized signals with variance estimation, $\widehat{\mathbf{W}}_i=\mathbf{W}_i/s_{T_i} -\widehat\sfw_i $ is the asymptotic normal noise. From the proof of Theorem \ref{thm:asymp-normal}, it is clear that $\widehat{\sfw}_i$ is close enough to ${\sfw}_i$.
Notice that, here, we do not assume $\cH_1\subseteq \cA$. For any $i\in\cA$ , we have

$$
\begin{aligned}
 & e_i^{\top}\left(\mathbf{X}_{\mathcal{A}}^{\top} \mathbf{X}_{\mathcal{A}}\right)^{-1} \mathbf{X}_{\mathcal{A}}^{\top} \mathbf{y}_2  = e_i^{\top}\left(\mathbf{X}_{\mathcal{A}}^{\top} \mathbf{X}_{\mathcal{A}}\right)^{-1} \mathbf{X}_{\mathcal{A}}^{\top} \left[ \mathbf{X}_{\mathcal{A}}, \mathbf{X}_{\mathcal{A}^c} \right] D (\widehat{\sfw}+ \widehat{\mathbf{W}}) \\
 &= s_{T_i} \widehat\sfw_i + e_i^{\top}\left(\mathbf{X}_{\mathcal{A}}^{\top} \mathbf{X}_{\mathcal{A}}\right)^{-1} \mathbf{X} D\widehat{\mathbf{W}} + e_i^\top \left(\mathbf{X}_{\mathcal{A}}^{\top} \mathbf{X}_{\mathcal{A}}\right)^{-1} \mathbf{X}_{\mathcal{A}}^{\top}\mathbf{X}_{\mathcal{A}^c}D_{\cA^c}\widehat\sfw_{\cA^c} \ \\
 &= s_{T_i}\widehat\sfw_i+ e_i^{\top}\left(\mathbf{X}_{\mathcal{A}}^{\top} \mathbf{X}_{\mathcal{A}}\right)^{-1} \mathbf{X} D\left(\widehat{\mathbf{W}}_1+\widehat{\mathbf{W}}_2\right)+e_i^\top \left(\mathbf{X}_{\mathcal{A}}^{\top} \mathbf{X}_{\mathcal{A}}\right)^{-1} \mathbf{X}_{\mathcal{A}}^{\top}\mathbf{X}_{\mathcal{A}^c}D_{\cA^c}\widehat \sfw_{\cA^c},
\end{aligned}
$$
where the noise decomposition $\widehat{\mathbf{W}}=\widehat{\mathbf{W}}_1+\widehat{\mathbf{W}}_2$ is the same as \eqref{eq:noise-W-decomp}. 
If $T_i\in \cA \cap \cH_0$, we have $\sfw_i =0$, thus $e_i^{\top}\left(\mathbf{X}_{\mathcal{A}}^{\top} \mathbf{X}_{\mathcal{A}}\right)^{-1} \mathbf{X}_{\mathcal{A}}^{\top} \mathbf{y}_2 = e_i^{\top}\left(\mathbf{X}_{\mathcal{A}}^{\top} \mathbf{X}_{\mathcal{A}}\right)^{-1} \mathbf{X} D\left(\widehat{\mathbf{W}}_1+\widehat{\mathbf{W}}_2\right)$. We investigate the following terms: (i) the asymptotic normality of $e_i^{\top}\left(\mathbf{X}_{\mathcal{A}}^{\top} \mathbf{X}_{\mathcal{A}}\right)^{-1} \mathbf{X} D\widehat{\mathbf{W}}_1$, (ii) the vanishing of $e_i^{\top}\left(\mathbf{X}_{\mathcal{A}}^{\top} \mathbf{X}_{\mathcal{A}}\right)^{-1} \mathbf{X} D\widehat{\mathbf{W}}_2$, and (iii) the bias introduced by inconsistent screening $e_i^\top \left(\mathbf{X}_{\mathcal{A}}^{\top} \mathbf{X}_{\mathcal{A}}\right)^{-1} \mathbf{X}_{\mathcal{A}}^{\top}\mathbf{X}_{\mathcal{A}^c}D_{\cA^c}\widehat\sfw_{\cA^c}$.

\noindent \textbf{(i)} the asymptotic normality of $\widehat{\beta}_i:=e_i^{\top}\left(\mathbf{X}_{\mathcal{A}}^{\top} \mathbf{X}_{\mathcal{A}}\right)^{-1} \mathbf{X} D\widehat{\mathbf{W}}_1$.
Conditional on $\cD_0$ and $\cD_1$, $\widehat{\beta}_i$ can be viewed as sum of i.i.d. independent random variables:

\begin{equation}
   \widehat{\beta}_i= \frac{\left\langle \Vect(\widehat{Z}_1),  e_i^\top \left(\mathbf{X}_{\mathcal{A}}^{\top} \mathbf{X}_{\mathcal{A}}\right)^{-1} \mathbf{X} T_{\calH} \left(I_{d_1 d_2} -  U_\perp U_\perp^\top \otimes V_\perp V_\perp^\top\right) \right\rangle }{\sigma_\xi  \sqrt{d_1 d_2/n} }.
\end{equation}
The variance of $\widehat{\beta}_i$ is given by 
\begin{equation*}
\begin{aligned}
     \E  \widehat{\beta}_i^2& = \norm{e_i^\top \left(\mathbf{X}_{\mathcal{A}}^{\top} \mathbf{X}_{\mathcal{A}}\right)^{-1} \mathbf{X} T_{\calH} \left(I_{d_1 d_2} -  U_\perp U_\perp^\top \otimes V_\perp V_\perp^\top\right)}^2 \\
    &=  e_i^\top\left(\mathbf{X}_{\mathcal{A}}^{\top} \mathbf{X}_{\mathcal{A}}\right)^{-1} \mathbf{X}_{\mathcal{A}}^{\top} \mathbf{X} \Sigma  \mathbf{X}^{\top} \mathbf{X}_{\mathcal{A}}\left(\mathbf{X}_{\mathcal{A}}^{\top} \mathbf{X}_{\mathcal{A}}\right)^{-1} e_i=\mathbf{Q}_{ii}.
\end{aligned}
\end{equation*}
The third-order moment of each component is also derived by 

\begin{equation*}
\begin{aligned}
     \E & \abs{\sqrt{d_1 d_2/n}\frac{\left\langle \Vect(\xi_i X_i),  e_i^\top \left(\mathbf{X}_{\mathcal{A}}^{\top} \mathbf{X}_{\mathcal{A}}\right)^{-1} \mathbf{X} T_{\calH} \left(I_{d_1 d_2} -  U_\perp U_\perp^\top \otimes V_\perp V_\perp^\top\right) \right\rangle }{\sigma_\xi \mathbf{Q}_{ii}^{\frac{1}{2}}  }}^3 \\
     & \le C\frac{\sqrt{d_1 d_2}}{n^{1.5} }\frac{ \abs{\left\langle \Vect( X_i),  e_i^\top \left(\mathbf{X}_{\mathcal{A}}^{\top} \mathbf{X}_{\mathcal{A}}\right)^{-1} \mathbf{X} T_{\calH} \left(I_{d_1 d_2} -  U_\perp U_\perp^\top \otimes V_\perp V_\perp^\top\right) \right\rangle}  }{\mathbf{Q}_{ii}^{\frac{1}{2}} } \\
     & = C\frac{\sqrt{d_1 d_2}}{n^{1.5} }\frac{ \abs{\left\langle \Vect( X_i)\left(I_{d_1 d_2} -  U_\perp U_\perp^\top \otimes V_\perp V_\perp^\top\right) ,  e_i^\top \left(\mathbf{X}_{\mathcal{A}}^{\top} \mathbf{X}_{\mathcal{A}}\right)^{-1} \mathbf{X} T_{\calH} \left(I_{d_1 d_2} -  U_\perp U_\perp^\top \otimes V_\perp V_\perp^\top\right) \right\rangle}  }{\mathbf{Q}_{ii}^{\frac{1}{2}} } \\
     &\le C\frac{\sqrt{d_1 d_2}}{n^{1.5}} \norm{\Vect( X_i)\left(I_{d_1 d_2} -  U_\perp U_\perp^\top \otimes V_\perp V_\perp^\top\right)  }_\tF \\
    &\le  C \frac{\mu \sqrt{r d_1}}{n^{1.5}},
\end{aligned}
\end{equation*}
where we use the incoherence condition in the last inequality. It is thus suggested that:
\begin{equation}
    \abs{ \PP\left(\frac{\widehat{\beta}_i}{\sqrt{\mathbf{Q}_{ii}}  } \le t \middle| \cD_0,\cD_1\right) - \Phi(t) } \le C\mu \sqrt{\frac{r d_1}{n}}.
\end{equation}

\noindent \textbf{(ii)} the vanishing of $\Delta{\beta}_i:= e_i^{\top}\left(\mathbf{X}_{\mathcal{A}}^{\top} \mathbf{X}_{\mathcal{A}}\right)^{-1} \mathbf{X} D\widehat{\mathbf{W}}_2$. By the proof of Theorem \ref{thm:asymp-normal}, we have $\norm{\widehat{\mathbf{W}}_2}_\infty \le C h_n$, with probability at least $1-Cd_1^{-\tau}\log d_1 $. Thus, by writing $\mathbf{X}=[\mathbf{X}_{\mathcal{A}},\mathbf{X}_{\mathcal{\cA}^c}]$, we have
\begin{equation*}
    \frac{\abs{\Delta{\beta}_i }}{\sqrt{\mathbf{Q}_{ii}} } = \frac{\abs{e_i^{\top}\left(\mathbf{X}_{\mathcal{A}}^{\top} \mathbf{X}_{\mathcal{A}}\right)^{-1} \mathbf{X} D\widehat{\mathbf{W}}_2}}{\sqrt{\mathbf{Q}_{ii}} } \le \frac{\abs{s_{T_i }\widehat{\mathbf{W}}_{2i}} + \abs{e_i^\top \left(\mathbf{X}_{\mathcal{A}}^{\top} \mathbf{X}_{\mathcal{A}}\right)^{-1} \mathbf{X}_{\mathcal{A}}^{\top}\mathbf{X}_{\mathcal{A}^c}D_{\cA^c}\widehat{\mathbf{W}}_{2,\cA^c} }  }{\sqrt{\mathbf{Q}_{ii}}}. \\
\end{equation*}
Using the definition of $C_\infty$, it follows that
\begin{equation*}
    \frac{\abs{\Delta{\beta}_i }}{ \sqrt{\mathbf{Q}_{ii}} } \le C C_\infty h_n,
\end{equation*}
uniformly for all $i$ with probability at least $1-C\log d_1 d_1^{-\tau}$. 

\noindent \textbf{(iii)} the bias $e_i^\top \left(\mathbf{X}_{\mathcal{A}}^{\top} \mathbf{X}_{\mathcal{A}}\right)^{-1} \mathbf{X}_{\mathcal{A}}^{\top}\mathbf{X}_{\mathcal{A}^c}D_{\cA^c}\widehat\sfw_{\cA^c}$ can be surely controlled by 
\begin{equation*}
    \frac{\abs{e_i^\top \left(\mathbf{X}_{\mathcal{A}}^{\top} \mathbf{X}_{\mathcal{A}}\right)^{-1} \mathbf{X}_{\mathcal{A}}^{\top}\mathbf{X}_{\mathcal{A}^c}D_{\cA^c}\widehat\sfw_{\cA^c} } }{ \sqrt{\mathbf{Q}_{ii}} } \le C\cdot C_\infty (h_n+\norm{\sfw_{\cA^c} }_{\infty}).
\end{equation*}
Then, combing (i), (ii), and (iii) by the Lipschiz property of $\Phi(t)$, we have
\begin{equation*}
       \abs{ \PP\left( \frac{\wt\sfw^{(2)}_i}{\sqrt{\mathbf{Q}_{ii}} }\le t \middle| \cD_0, \cD_1  \right)-\Phi (t)}\le C \cdot C_\infty ( h_n+\norm{\sfw_{\cA^c} }_{\infty}) + C\mu \sqrt{\frac{r d_1}{n}} \le C \cdot C_\infty( h_n+\norm{\sfw_{\cA^c} }_{\infty}).
\end{equation*}

\end{proof}

\subsection{Proof of Theorem \ref{thm:matrix-fdr-strong}} \label{sec:proof-fdr-strong}
\begin{proof}
In the following proof, we write $h_n+\norm{\sfw_{\cA^c} }_{\infty}$ as $h_n$ for notational simplicity.  The proof essentially follows the proof of Theorem \ref{thm:weak-cor-fdr}. Define the expected false rejection:
\begin{equation*}
    \widetilde{G}(t)= \frac{\sum_{{T_i}\in\cH_0 \cap \cA } \PP( {\wt\sfw_i^{(1)}\frac{\sqrt{\mathbf{Q}_{ii}}}{ \widehat{\sigma}_{w i} } } Z > t |\cD_1 )   }{q_{0n}},
\end{equation*}
where $\wt\sigma_{wi}^2=e_i^{\top}\big(\mathbf{X}_{\calA}^{\top}\mathbf{X}_{\calA}\big)^{-1}e_i$ is defined in Algorithm~\ref{alg:matrix-sda-practical}. 
Denote $$L_n'= \widetilde{G}^{-1}\left( \frac{\epsilon_n \eta_n'}{q_{0n} }\right)= \inf\left\{t:\widetilde{G}(t)\le \frac{\epsilon_n \eta_n'}{q_{0n} }\right\},$$ where $\epsilon_n$ is a rate to be specified later, and $q_{0n}=\abs{\cA \cap \cH_0 }$. We can rewrite Lemma \ref{lemma:pop-wrank}, \ref{lemma:weak-cov}, and \ref{lemma:conv-prob} as:

\begin{Lemma}\label{lemma:pop-wrank-sda} Conditional on $E_0$ and $\cD_1$, we have
\begin{equation*}
    \sup\limits_{0\le t \le L_n} \abs{\frac{ \sum_{{T_i}\in\cH_0 \cap \cA }\PP(\wt\sfw_i^\mathsf{rank} >t ) }{q_{0n} \widetilde{G} (t) } -1 }\le C_3 \frac{C_\infty h_n q_{0n} }{\epsilon_n \eta_n' }.
\end{equation*}
\end{Lemma}
Here we use $\wt\sfw_i^\mathsf{rank}$ to indicate the combined statistics $\wt\sfw_{T_i}^\mathsf{rank}$
\begin{Lemma}[Weak dependency of null features]\label{lemma:weak-cov-sda}
        Conditional on $\cD_1$, 
\begin{equation}
\begin{aligned}
          \sup_{0\le t\le L_n' } \frac{\sum_{(T_i,T_j)\in \cH_{0\cA,\text{weak} }^2 } \abs{\operatorname{cov}(\bbI(\wt\sfw_i^\mathsf{rank} >t),\bbI(\wt\sfw_i^\mathsf{rank}  >t)) }}{ q_{0n}^2 \widetilde G^2(t)}\le  C_1 \frac{C_\infty h_n q_{0n} }{\epsilon_n \eta_n' }+C_2 \frac{1}{ \left(\epsilon_n\eta_n' q_{0n}\right)^{v/2} }.
\end{aligned}
\end{equation}
\end{Lemma}

\begin{Lemma}\label{lemma:conv-prob-sda }
For any $\varepsilon>0$, conditional on  $\cD_1$, it holds that
\begin{equation*}
\begin{aligned}
    &\PP\left(\sup\limits_{0\le t\le L_n'}\abs{\frac{ \sum_{T_i\in\cH_0 \cap \cA }\bbI(\wt\sfw_i^\mathsf{rank}  >t))}{q_{0n} \widetilde G(t) } -1} \ge \varepsilon\right) \\
    &\le \frac{C}{\varepsilon^2} \log(\frac{q_{0n} }{\epsilon_n \eta_n'}) \left( \left(\frac{\beta_{\mathsf{s}}' q_{0n}^2 }{ \epsilon_n^2\eta_n^{'2}}\right)^{\frac{1}{2}} + \left(\frac{ C_\infty h_n q_{0n} }{\epsilon_n \eta_n'}+ \frac{1}{ \left(\epsilon_n\eta_n' q_{0n}\right)^{v/2} } \right)^{\frac{1}{2}}\right) .
\end{aligned}
\end{equation*}
\end{Lemma}
The proof of Lemma \ref{lemma:pop-wrank-sda}, \ref{lemma:weak-cov-sda}, and \ref{lemma:conv-prob-sda } is same as that in Lemma \ref{lemma:pop-wrank}, \ref{lemma:weak-cov}, and \ref{lemma:conv-prob}, and thus omitted.
These lemmas imply that, if $L\le L_n'$, then we have $\sfR_0\le 1+3\varepsilon$ with probability at least 
\begin{equation*}
    1-\frac{C}{\varepsilon^2} \log(\frac{q_{0n} }{\epsilon_n \eta_n'}) \left( \left(\frac{\beta_{\mathsf{s}}' q_{0n}^2 }{ \epsilon_n^2\eta_n^2}\right)^{\frac{1}{2}} + \left(\frac{ C_\infty h_n q_{0n} }{\epsilon_n \eta_n'}+ \frac{1}{ \left(\epsilon_n\eta_n' q_{0n}\right)^{v/2} } \right)^{\frac{1}{2}}\right).
\end{equation*}

We then prove the probability of $\PP(L\le L_n')$ can be very large. A matching upper bound of $\widetilde{G}(t)$ is given by, similarly as in the proof Theorem \ref{thm:weak-cor-fdr},
 \begin{equation*}
    \widetilde{G}(t)= \frac{\sum_{{T_i}\in\cH_0 \cap \cA } \PP( {\wt\sfw_i^{(1)}\frac{\sqrt{\mathbf{Q}_{ii}}}{ \widehat{\sigma}_{w i} } } Z > t |\cD_0,\cD_1 )   }{q_{0n}} \le \frac{\sqrt{2}}{\sqrt{\pi}} \exp{\left( -\frac{t^2}{2 \max_{T\in\cH_0 \cap \cA} \abs{\wt\sfw_i^{(1)}\frac{\sqrt{\mathbf{Q}_{ii}}}{ \widehat{\sigma}_{ w i}  }}^2 } \right)}.
\end{equation*}

The LASSO results presented in Proposition \ref{prop:lasso-scr} show that, the $\abs{\wt\sfw_i^{(1)}}$ can be uniformly bounded by:
\begin{equation*}
\begin{aligned}
     \max_{T_i\in\cH_0 \cap \cA} \abs{\wt\sfw_i^{(1)}\frac{\sqrt{\mathbf{Q}_{ii}}}{ \widehat{\sigma}_{w i}   }} & \le  \max_{T\in\cH_0 \cap \cA}\abs{\wt\sfw_i^{(1)}} \frac{ \sqrt{e_i^\top\left(\mathbf{X}_{\mathcal{A}}^{\top} \mathbf{X}_{\mathcal{A}}\right)^{-1} \mathbf{X}_{\mathcal{A}}^{\top} \mathbf{X} \Sigma  \mathbf{X}^{\top} \mathbf{X}_{\mathcal{A}}\left(\mathbf{X}_{\mathcal{A}}^{\top} \mathbf{X}_{\mathcal{A}}\right)^{-1} e_i}  }{\sqrt{e_i^{\top}\left(\mathbf{X}_{\mathcal{A}}^{\top} \mathbf{X}_{\mathcal{A}}\right)^{-1} e_i} } \\
     & \le \max_{T\in\cH_0 \cap \cA}\abs{\wt\sfw_i^{(1)}} \left(1+ \frac{\norm{ e_i^\top\left(\mathbf{X}_{\mathcal{A}}^{\top} \mathbf{X}_{\mathcal{A}}\right)^{-1} \mathbf{X}_{\mathcal{A}}^{\top}  } \sqrt{\norm{\mathbf{X}\Sigma \mathbf{X}^{\top} -I } } }{\sqrt{e_i^{\top}\left(\mathbf{X}_{\mathcal{A}}^{\top} \mathbf{X}_{\mathcal{A}}\right)^{-1} e_i}}  \right) \\
     & \le (1+c)\max_{T\in\cH_0 \cap \cA}\abs{\wt\sfw_i^{(1)}} \\
     & \le C \kappa_1^{1.5}\sqrt{q_1(\log d_1 )}, \\
\end{aligned}
\end{equation*}
with probability at least $1-C d_1^{-\tau}$. Here we use the fact that $\norm{\mathbf{X}\Sigma \mathbf{X}^{\top} -I }\le \frac{1}{1-c}$ if we have its inverse $\norm{\Sigma^{-\frac{1}{2}} \widehat{\Sigma} \Sigma^{-\frac{1}{2}}-I }\le c$. The definition of $L_n'$ implies that 
$$L_n'  \le C \sqrt{\log(\frac{q_{0n} }{\epsilon_n \eta_n'})}\cdot  C \kappa_1^{1.5}\sqrt{q_1(\log d_1 )} \ll \sqrt{\log(\frac{1}{h_n})}\cdot \kappa_1^{1.5}\sqrt{q_1(\log d_1 )}. $$  
If $T_i\in \cS$, then $\abs{\delta_{T_i} }\ge C_{\mathsf{gap}}\sqrt{\log \frac{1}{h_n} }\vee \kappa_1^{1.5}\sqrt{q_1(\log d_1 )} $ by the definition of $\cS$, and also the LASSO estimation: 
$$  \abs{\wt\sfw_i^{(1)}}\ge C \kappa_1^{1.5}\sqrt{q_1(\log d_1 )},$$
by our assumption. Assume $C_{\mathsf{gap}}$ is large enough, and $\delta_{T_i}>0$. Then we have 
 \begin{equation*}
\begin{aligned}
\PP(\wt\sfw_i^\mathsf{rank}< L_n' ) \le &\PP\left(\wt\sfw_i^{(1)} (Z_2+\delta_{T_i} \frac{s_{T_i} }{\sqrt{\mathbf{Q}_{ii} } } )< L_n' \right) +C_\infty h_n \\
\le & 1- \PP\left( (Z_2+\delta_{T_i}\frac{s_{T_i} }{\sqrt{\mathbf{Q}_{ii} } }) \ge L_n'/\wt\sfw_i^{(1)} \right) +C_\infty h_n \\
\le & 1- \PP\left( Z_2\ge -\delta_{T_i}+\sqrt{\log \frac{1}{h_n} } \right) +C_\infty h_n \\
\le & \PP\left(Z_2\le -2\sqrt{\log \frac{1}{h_n} } \right) +C_\infty h_n  \\
\le & 2C_\infty h_n.
\end{aligned}
 \end{equation*}
 We compute the probability:
  \begin{equation*}
\begin{aligned}
               \PP(\sum\limits_{T\in\cS} \bbI(\wt\sfw_i^\mathsf{rank}> L_n')\le (1-\varepsilon)\eta_n' )& =\PP(\sum\limits_{T\in\cS} \bbI(\wt\sfw_i^\mathsf{rank}< L_n)> \varepsilon\eta_n' )  & \\
               \le \frac{\sum\limits_{T\in\cS} \PP(\wt\sfw_i^\mathsf{rank}< L_n')}{\varepsilon\eta_n'}\le C C_\infty h_n/\varepsilon,
\end{aligned}
 \end{equation*}
 i.e., $ \PP(\sum\limits_{T\in\cS} \bbI(\wt\sfw_T^\mathsf{rank}> L_n')\le (1-\varepsilon)\eta_n' )\to 0$, $ \PP(\sum\limits_{T\in\cS} \bbI(\wt\sfw_T^\mathsf{rank}> L_n')\ge \eta_n' )\to 1$. By taking $\epsilon_n=\alpha/8$, other steps essentially follow the proof of Theorem \ref{thm:weak-cor-fdr}.

\end{proof}


\section{Additional Technical Lemmas}
We now give an error bound on the generalized SVD for any mode-$m$ low-rank tensor, which is instantly applicable to the low-rank matrix with $m=2$.
\begin{Lemma}[Perturbation of general HOSVD] \label{eq:lemma:tensor-diff} Given a mode-$m$ Tucker low-rank tensor with $\cM=\cS\times_{j=1}^m U_j\in \R^{d_1\times\cdots\times d_m}$, where  $\cS\in\R^{ r_1\times r_2\cdots\times r_m}$, and  $U_j\in\R^{d_j\times r_j}$ are incoherent singular subspaces, i.e., $U_j^\top U_j=I_{d_j}$, $\norm{U_j}_{2,\max}\le\sqrt{\mu r_j/d_j}$. We denote its  maximum and minimum singular values as $\lambda_{\max}=\max_{j\in [m]}\{ \lambda_{1}( \Mat_j(\cM) ) \}$, $\lambda_{\min}=\min_{j\in [m]}\{ \lambda_{r_j}( \Mat_j(\cM) ) \}$, with condition number $\kappa_0=\lambda_{\max}/\lambda_{\min}$. Here, $\Mat_j(\cdot)$ means the mode-$j$ unfolding. Denote $d^*=\prod_{j=1}^m d_j$, $r^*=\prod_{j=1}^m r_j$, $\dmax=\max\{d_j\}$. Write the tangent space of the $r_1\times r_2\cdots\times r_m$ low-rank manifold at point $\cM$  as $\TT$, with the projection onto it as $\cP_{\TT}(\cdot)$. For any perturbation $\cE$, we denote its higher-order SVD as 
$$\widehat{\cM}=\operatorname{HOSVD}(\cM+\cE,r_1\times r_2\cdots\times r_m)= (\cM+\cE) \times_{j=1}^m \cP_{\widehat{U}_j}, \text{ where } \widehat{U}_j=\SVD_{r_j}(\Mat_j(\cM+\cE)).
$$
 Then,  when 
 $$\norm{\cE}_{\ell_{\infty}}= \varepsilon_{\infty} \le \frac{\lambda_{\min}}{48\kappa_0 m\sqrt{d^*}}, 
 $$
 for any tensor $\cI\in \R^{d_1\times\cdots\times d_j}$, we have
\begin{equation*}
    \abs{\left\langle \widehat{\cM},\cI \right\rangle-\left\langle\cM +\cP_{\TT}(\cE),\cI \right\rangle} \le \frac{37e^2 m^2 d^*\varepsilon_{\infty}^2}{\lambda_{\min} }\sqrt{\frac{\mu^m r^*}{d^*} }\norm{\cI}_{\ell_1}
\end{equation*}
\end{Lemma}
For more introduction on the Tucker low-rank tensor and the related definitions/notations, please refer to \cite{ma2024statistical}.
\begin{proof}
    The proof relies on the spectral decomposition that has been studied in \cite{ma2024statistical}. Define the mode-$j$ unfolding of $\cM$, $\cE$ as $M_j$, $E_j$, correspondingly. Without loss of generality, we write each unfolding of $\cS$ as $S_j=\Lambda_j V_j^\top$. Define $\norm{\cE}_{\tF }= \varepsilon_{\tF}$. According to the proof of Lemma 1 in \cite{ma2024statistical}, for any mode $j$, we have
    \begin{equation*}
       \cP_{\widehat{U}_j}-\cP_{U_j}= \sum_{k\ge 1} \cS_{j }^{(k)} =  \sum_{k\ge 1} \sum_{\mathbf{s}: s_1+\cdots+s_{k+1}=k}(-1)^{1+\tau(\mathbf{s})} \cdot \mathfrak{P}_j^{-s_1} {\Delta}_{j}\mathfrak{P}_j^{-s_2} {\Delta}_{j} \cdots {\Delta}_{j}\mathfrak{P}_1^{-s_{k+1}},
    \end{equation*}
    where  $s_1\ge 0, \cdots, s_{k+1}\ge 0$ are non-negative integers with $
\tau(\mathbf{s})=\sum_{j=1}^{k+1} \mathbb{I}\left(s_j>0\right)$, and
    \begin{equation*}
        {\Delta}_{j} = (M_j+E_j)(M_j+E_j)^{\top} -\cM_j\cM_j^\top = \cM_j E_j^\top + E_j\cM_j^\top +E_jE_j^\top,
    \end{equation*}
    and $\mathfrak{P}_j$ are the power series whose definition can be found in \cite{xia2021normal,ma2024statistical}. For the series $ \cS_{j }^{(k)} $, its first order is
    \begin{equation*}
\begin{aligned}
             \cS_{j }^{(1)} & =  U_j\Lambda_j^{-2}U_j^\top {\Delta}_{j}
         \cP_{U_j}^\perp  +  \cP_{U_j}^\perp {\Delta}_{j} U_j\Lambda_j^{-2}U_j^\top
         \\
         & = \underbrace{U_j\Lambda_j^{-1} V_j^\top\left(\otimes_{k\neq j}\cP_{U_k}\right) E_j^\top \cP_{U_j}^\perp  +\cP_{U_j}^\perp  E_j \left(\otimes_{k\neq j}\cP_{U_k}\right) V_j \Lambda_j^{-1} U_j^\top }_{:=\mathfrak{A}_j} 
         \\
         & \quad + \underbrace{U_j\Lambda_j^{-2}U_j^\top (E_jE_j^\top) 
         \cP_{U_j}^\perp  + \cP_{U_j}^\perp (E_jE_j^\top)U_j\Lambda_j^{-2}U_j^\top}_{:=\mathfrak{B}_j}
\end{aligned}
    \end{equation*}
From Lemma 1 of \cite{ma2024statistical}, we can extract that:
    \begin{equation}\label{eq:tensor-inco}
   \norm{ \cP_{\widehat{U}_{j}} - \cP_{U_j} }_{2,\infty} \le \frac{8  \sqrt{d^*}\varepsilon_{\infty}}{\lambda_{\min}}\sqrt{\frac{\mu r_j}{d_j} } \le \frac{1}{2} \sqrt{\frac{\mu r_j}{d_j} },
\end{equation}
and for any $e_j$ as the canonical basis in $\R^{d_j}$,
\begin{equation}\label{eq:tensor-hod}
\begin{gathered}
    \norm{e_{j}^\top\mathfrak{B}_j}_2 \le  4 \frac{  \sqrt{d^*}\varepsilon_{\infty}}{\lambda_{\min}} \sqrt{\frac{\mu r_j}{d_j} } \left(\frac{\varepsilon_{\tF}}{\lambda_{\min}}\right), \   \norm{e_{j}^\top\mathfrak{B}_j M_j}_2 \le  4 \sqrt{d^*}\varepsilon_{\infty} \sqrt{\frac{\mu r_j}{d_j} } \left(\frac{\varepsilon_{\tF}}{\lambda_{\min}}\right),
\\
   \norm{ e_{j}^\top\sum_{k\ge 2} \cS_{j }^{(k)} }_2\le \sqrt{\frac{\mu r_j}{d_j} } \cdot\frac{8  \sqrt{d^*}\varepsilon_{\infty}}{\lambda_{\min}} \left(\frac{8 \varepsilon_{\tF}  }{\lambda_{\min}}\right)^{k-1} .
\end{gathered}
\end{equation}
We then have
\begin{equation}\label{eq:tensor-decomp-1}
    \begin{aligned}
       \widehat{\cM}-\cM &=   (\cM+\cE) \times_{j=1}^m \cP_{\widehat{U}_j} - \cM\times_{j=1}^m \cP_{{U}_j} 
       \\
       &= \cE\times_{j=1}^m \cP_{{U}_j} + \sum_{k=1}^{m} \cM \times_k \left( \cP_{\widehat{U}_{k}} - \cP_{U_k}\right)\times_{j\neq k} \cP_{{U}_j} 
       \\
       & + \underbrace{\sum_{\QQ\subseteq [m],\abs{\QQ}\ge 1 } \cE \times_{k\in\QQ} \left( \cP_{\widehat{U}_{k}} - \cP_{U_k}\right)\times_{j\notin \QQ} \cP_{{U}_j}   + \sum_{\QQ\subseteq [m],\abs{\QQ}\ge 2 } \cM \times_{k\in\QQ} \left( \cP_{\widehat{U}_{k}} - \cP_{U_k}\right)\times_{j\notin \QQ} \cP_{{U}_j}  }_{:=\mathfrak{C}_1}
       \\
       & = \cE\times_{j=1}^m \cP_{{U}_j} + \sum_{k=1}^{m} \cM \times_k \mathfrak{A}_k\times_{j\neq k} \cP_{{U}_j} + \sum_{k=1}^{m} \cM \times_k \mathfrak{B}_k\times_{j\neq k} \cP_{{U}_j}  +\mathfrak{C}_1
       \\
       & = \cP_{\TT}(\cE) + \sum_{k=1}^{m} \cM \times_k \mathfrak{B}_k\times_{j\neq k} \cP_{{U}_j}  +\mathfrak{C}_1.
    \end{aligned}
\end{equation}
Therefore, for any $\cI$, we have  
$$ \abs{\left\langle \widehat{\cM},\cI \right\rangle-\left\langle\cM +\cP_{\TT}(\cE),\cI \right\rangle} =  \abs{\left\langle  \sum_{k=1}^{m} \cM \times_k \mathfrak{B}_k\times_{j\neq k} \cP_{{U}_j}  +\mathfrak{C}_1,\cI \right\rangle}. $$
According to \eqref{eq:tensor-inco}, \eqref{eq:tensor-hod}, given any single entry $\calW$, we can control the term above by:
\begin{equation}\label{eq:tensor-decomp-2}
\begin{aligned}
         \abs{\left\langle  \sum_{k=1}^{m} \cM \times_k \mathfrak{B}_k\times_{j\neq k} \cP_{{U}_j}  ,\calW \right\rangle} \le  4m \sqrt{d^*}\varepsilon_{\infty} \sqrt{\frac{\mu^m r^*}{d^*} } \left(\frac{\varepsilon_{\tF}}{\lambda_{\min}}\right),
\end{aligned}
\end{equation}
and
\begin{equation}\label{eq:tensor-decomp-3}
\begin{aligned}
         \abs{\left\langle \mathfrak{C}_1,\calW \right\rangle} & \le  \sum_{k\ge1} (em)^k \left(\frac{8  \sqrt{d^*}\varepsilon_{\infty}}{\lambda_{\min}}\right)^k\sqrt{\frac{\mu^m r^*}{d^*} } \varepsilon_{\tF}+\sum_{k\ge2}  \left(\frac{em}{2}\right)^k \left(\frac{8  \sqrt{d^*}\varepsilon_{\infty}}{\lambda_{\min}}\right)^{k-1}\sqrt{\frac{\mu^m r^*}{d^*} }8  \sqrt{d^*}\varepsilon_{\infty}
         \\
         &\le \left(\frac{16 e m d^*\varepsilon_{\infty}^2}{\lambda_{\min} }+\frac{32e^2 m^2 d^*\varepsilon_{\infty}^2}{\lambda_{\min} }\right) \sqrt{\frac{\mu^m r^*}{d^*} } \le \frac{36e^2 m^2 d^*\varepsilon_{\infty}^2}{\lambda_{\min} }\sqrt{\frac{\mu^m r^*}{d^*} }.
\end{aligned}
\end{equation}
Combining \eqref{eq:tensor-decomp-1}, \eqref{eq:tensor-decomp-2}, \eqref{eq:tensor-decomp-3}, we know that for $\widehat{\cM}$,
\begin{equation*}
    \abs{\left\langle \widehat{\cM},\cI \right\rangle-\left\langle\cM +\cP_{\TT}(\cE),\cI \right\rangle} \le \frac{37e^2 m^2 d^*\varepsilon_{\infty}^2}{\lambda_{\min} }\sqrt{\frac{\mu^m r^*}{d^*} }\norm{\cI}_{\ell_1}.
\end{equation*}
\end{proof}

\section{Proof of Minimax CI Length}\label{Minimax}
\begin{Theorem}[Minimax optimal CI length for tensor completion] Consider the tensor completion model: 
\begin{equation*}
    Y_i=\left\langle\cX_i,\cM\right\rangle +\xi_i, \quad i\in[n],
\end{equation*}
where $\xi_i\sim \cN(0,\sigma^2)$ are i.i.d., and $\cX_i$ are independent and uniformly distributed over all the canonical bases in $\R^{d_1\times\cdots d_j}$, which are independent of $\{\xi_i\}_{i=1}^n$. We use the notation in Lemma \ref{eq:lemma:tensor-diff}, with $\underline{r}=\min\{r_j\}$ Denote $s_0=\norm{\cP_{\TT}(\cI)}_{\tF}$.
Define the parameter space as 
$$
\begin{aligned}
    \boldsymbol{\Theta}=&\big\{ \cM\in\R^{d_1\times\cdots d_j}: \operatorname{rank}(\Mat_j(\cM))\le r_j, \norm{U_j}_{2,\max}\le\sqrt{\mu r_j/d_j}, 
    \\
    &\quad\quad\lambda_{\min}(\cM)\ge \lambda_{\min}, \kappa(\cM)\le \kappa_0,\norm{\cP_{\TT(\cM)}(\cI)}_{\tF}\ge s_0 \big\}.
\end{aligned}
$$
Here $\cP_{\TT(\cM)}(\cdot)$ means the projection onto the tangent space at any given $\cM$.
Consider the set of any valid $1-\alpha$ confidence interval  with $\alpha<\frac{1}{4}$ as:
$$
\mathcal{I}_\alpha(\boldsymbol{\Theta}, \cI):=\left\{\operatorname{CI}_{\mathcal{I}}^\alpha\left(\cM, \{(\cX_i,Y_i)\}_{i=1}^n \right)=[l, u]: \inf _{\cM \in \boldsymbol{\Theta}} \mathbb{P}(l \leq\langle\cM, \cI\rangle \leq u ) \geq 1-\alpha\right\},
$$
where $l, u$ are any functions of observations $\{(\cX_i,Y_i)\}_{i=1}^n$. Then, when the SNR satisfies
\begin{equation*}
         \frac{\lambda_{\min}}{\sigma} \ge C_{\gap} \kappa_0\left(\frac{\norm{\cI}_{\ell_1}}{ \norm{\cP_{\TT}(\cI)}_{\tF}\sqrt{d^*/\dmax} }\bigvee 1\right) \sqrt{\frac{ m^5 (2\mu)^{3 m} (r^*)^{3 }\dmax d^* }{\underline{r}^2 n } }
\end{equation*}
for a numeric constant $C_{\gap}$, the length of the confidence interval has the minimax lower bound:
\begin{equation*}
    \inf_{\operatorname{CI}_{\mathcal{I}}^\alpha\left(\cM, \{(\cX_i,Y_i)\}_{i=1}^n \right) \in \mathcal{I}_\alpha(\boldsymbol{\Theta}, \cI) } \sup _{\cM\in \boldsymbol{\Theta} } \mathbb{E} L\left(\operatorname{CI}_{\mathcal{I}}^\alpha\left(\cM, \{(\cX_i,Y_i)\}_{i=1}^n \right)\right) \ge c\sigma \sqrt{\frac{d^*}{n}}s_0=c\sigma \sqrt{\frac{d^*}{n}}\norm{\cP_{\TT}(\cI)}_{\tF}
\end{equation*}
\end{Theorem}
\begin{proof}
   Invoking the Lemma 1 of \cite{cai2017confidence}, we have 
    \begin{equation*}
       \begin{aligned}
            &\inf_{\operatorname{CI}_{\mathcal{I}}^\alpha\left(\cM, \{(\cX_i,Y_i)\}_{i=1}^n \right) \in \mathcal{I}_\alpha(\boldsymbol{\Theta}, \cI) } \sup _{\cM\in \boldsymbol{\Theta} } \mathbb{E} L\left(\operatorname{CI}_{\mathcal{I}}^\alpha\left(\cM, \{(\cX_i,Y_i)\}_{i=1}^n \right)\right) 
        \\
        & \ge  \inf_{\operatorname{CI}_{\mathcal{I}}^\alpha\left(\cM, \{(\cX_i,Y_i)\}_{i=1}^n \right) \in \mathcal{I}_\alpha(\boldsymbol{\Theta}, \cI) } \sup _{\cM\in \{\cM_1,\cM_2 \} } \mathbb{E} L\left(\operatorname{CI}_{\mathcal{I}}^\alpha\left(\cM, \{(\cX_i,Y_i)\}_{i=1}^n \right)\right)
        \\
        & \ge \abs{\left\langle \cM_1-\cM_2,\cI\right\rangle} \left(1-2\alpha-\operatorname{TV}(\pi(\cM_1),\pi(\cM_2)) \right)
        \\
        & \ge \abs{\left\langle \cM_1-\cM_2,\cI\right\rangle} \left(1-2\alpha-\sqrt{2\operatorname{KL}(\pi(\cM_1),\pi(\cM_2))} \right)
       \end{aligned}
    \end{equation*}
where we use Pinsker’s inequality for the last step. Now, with a slightly abuse of notation, we choose a $\cM\in\boldsymbol{\Theta}$ such that $\lambda_{\min}(\cM)=2\lambda_{\min}$, $\kappa(\cM)=\frac{1}{2}\kappa_0$, $\norm{U_j}_{2,\max}\le\sqrt{\mu r_j/d_j}/2$, and $\norm{\cP_{\TT(\cM)}(\cI)}_{\tF}\ge 2s_0$. Let a new $\widehat{\cM}$ be 
\begin{equation*}
    \widehat{\cM} = \operatorname{HOSVD}\left(\cM+\varepsilon\frac{\cP_{\TT(\cM)}(\cI)}{\norm{\cP_{\TT(\cM)}(\cI)}_{\tF}},r_1,\dots,r_m\right).
\end{equation*}
For each entry $\calW$, we have
\begin{equation*}
    \varepsilon\abs{\left\langle\frac{\cP_{\TT(\cM)}(\cI)}{\norm{\cP_{\TT(\cM)}(\cI)}_{\tF}},\calW\right\rangle}=  \varepsilon\abs{\left\langle\frac{\cP_{\TT(\cM)}(\cI)}{\norm{\cP_{\TT(\cM)}(\cI)}_{\tF}},\cP_{\TT(\cM)}(\calW)\right\rangle} \le \varepsilon\sqrt{\frac{mr^*\mu^{m-1}\dmax}{\underline{r}d^*}},
\end{equation*}
i.e.,
\begin{equation*}
    \norm{\varepsilon\frac{\cP_{\TT(\cM)}(\cI)}{\norm{\cP_{\TT(\cM)}(\cI)}_{\tF}}}_{\ell_{\infty}} \le \varepsilon\sqrt{\frac{mr^*\mu^{m-1}\dmax}{\underline{r}d^*}}.
\end{equation*}
Now, we set 
\begin{equation}\label{eq:minimax-eps-1}
    \varepsilon\sqrt{\frac{mr^*\mu^{m-1}\dmax}{\underline{r}d^*}} \le \frac{\lambda_{\min}}{48\kappa_0 m\sqrt{d^*}}, \text{ i.e., } \varepsilon \le \frac{\lambda_{\min}\sqrt{\underline{r}}}{48\kappa_0 m^{1.5 }\sqrt{r^* \dmax}},
\end{equation}
and apply \eqref{eq:tensor-inco}, Lemma \ref{eq:lemma:tensor-diff}. This tells us that $ \widehat{\cM} $ is also with rank $r_1,\cdots,r_m$, and $\mu$-incoherence. Moreover, from Lemma 1 of \cite{li2023online}, we know that
\begin{equation*}
    \norm{\widehat{\cM}-\cM}_{\tF}\le \varepsilon + \frac{59 m\varepsilon^2}{\lambda_{\min}} \le 2\varepsilon\le \frac{1}{2}\lambda_{\min}.
\end{equation*}
Also, we need to check the variance condition. By Lemma 14 of \cite{ma2024statistical}, we have
\begin{equation*}
\begin{aligned}
        &\norm{\cP_{\TT(\widehat{\cM})}(\cI)}_{\tF}\ge   \norm{\cP_{\TT({\cM})}(\cI)}_{\tF} -\norm{\cP_{\TT(\widehat{\cM})}(\cI)-\cP_{\TT({\cM})}(\cI)}_{\tF}  
         \\
    & \ge 2 s_0 - \frac{C  m^2\kappa_0   }{\lambda_{\min}}\sqrt{\frac{(2\mu)^m r^*\dmax }{n}}\cdot\varepsilon\sqrt{\frac{mr^*\mu^{m-1}\dmax}{\underline{r}d^*}}\norm{\cI}_{\ell_1}
    \\
    & \ge s_0,
\end{aligned}
\end{equation*}
where we set
\begin{equation}\label{eq:minimax-eps-2}
    \frac{C  m^2\kappa_0   }{\lambda_{\min}}\sqrt{{\mu^m r^*\dmax } }\cdot\varepsilon\sqrt{\frac{mr^*\mu^{m-1}\dmax}{\underline{r}d^*}}\norm{\cI}_{\ell_1} \le s_0, \text{ i.e., } \varepsilon \le  c\frac{s_0}{\norm{\cI}_{\ell_1}} \lambda_{\min} \frac{\sqrt{\underline{r}  d^*}}{m^{2.5}\kappa_0\mu^m r^* \dmax}.
\end{equation}
Therefore, ${\widehat{\cM}}\in\boldsymbol{\Theta}$. Set $\cM_1=\cM$, $\cM_2=\widehat{\cM}$. Now, we compute $\abs{\left\langle \cM_1-\cM_2,\cI\right\rangle}$ and $\operatorname{KL}(\pi(\cM_1),\pi(\cM_2))$. For the KL-divergence, clearly we have the following chain rule:
\begin{equation*}
\begin{aligned}
        \operatorname{KL}(\pi(\cM_1),\pi(\cM_2)) & = \E_{\cX}[\operatorname{KL}(\pi(\cM_1)|\{\cX_i\}_{i=1}^n,\pi(\cM_2)|\{\cX_i\}_{i=1}^n)  ] + \operatorname{KL}(\{\cX_i\}_{i=1}^n,\{\cX_i\}_{i=1}^n) 
        \\
        & =\E_{\cX} \frac{\sum_{i=1}^n\left(\left\langle\widehat{\cM}-\cM,\cX_i\right\rangle\right)^2 }{2\sigma^2} = \frac{n\norm{\widehat{\cM}-\cM}_{\tF}^2}{2\sigma^2 d^*}.
\end{aligned}
\end{equation*}
Thus, $\sqrt{  2\operatorname{KL}(\pi(\cM_1),\pi(\cM_2)) }\le 2\sqrt{n/d^*}/\sigma\cdot \varepsilon$. We select 
$$\varepsilon=\frac{\sigma\sqrt{d^*/n}}{8},$$
which gives  $\left(1-2\alpha-\sqrt{2\operatorname{KL}(\pi(\cM_1),\pi(\cM_2))} \right)\ge \frac{1}{4}$.

For the term $\abs{\left\langle \cM_1-\cM_2,\cI\right\rangle} $, we use Lemma \ref{eq:lemma:tensor-diff}, which gives 
\begin{equation*}
\begin{aligned}
         \abs{\left\langle \widehat{\cM},\cI \right\rangle-\left\langle\cM ,\cI \right\rangle} & \ge      \abs{\left\langle\cP_{\TT(\cM)}\left(\varepsilon\frac{\cP_{\TT(\cM)}(\cI)}{\norm{\cP_{\TT(\cM)}(\cI)}_{\tF}} \right) ,\cI \right\rangle}-  \abs{\left\langle \widehat{\cM},\cI \right\rangle-\left\langle\cM +\cP_{\TT(\cM)}\left(\varepsilon\frac{\cP_{\TT(\cM)}(\cI)}{\norm{\cP_{\TT(\cM)}(\cI)}_{\tF}} \right),\cI \right\rangle} 
         \\
         &\ge \varepsilon \norm{\cP_{\TT(\cM)}(\cI)}_{\tF}-  
     \frac{37e^2 m^2 d^*\varepsilon_{\infty}^2}{2\lambda_{\min} }\sqrt{\frac{\mu^m r^*}{ 2^md^*} }\norm{\cI}_{\ell_1}
     \\
     & \ge \varepsilon \norm{\cP_{\TT(\cM)}(\cI)}_{\tF}-  \varepsilon^2{37e^2 \frac{m^3r^*\mu^{m-1}\dmax}{2^{m}\lambda_{\min}\underline{r}}}\sqrt{\frac{\mu^m r^*}{ 2^md^*} }\norm{\cI}_{\ell_1}.
\end{aligned}
\end{equation*}
Set 
\begin{equation*}
     \frac{2\cdot 37e^2 m^3r^*\mu^{1.5m-1}\sqrt{r^*}\dmax}{2^{1.5m}\underline{r}\lambda_{\min}\sqrt{d^*}}\frac{\norm{\cI}_{\ell_1}}{\norm{\cP_{\TT(\cM)}(\cI)}_{\tF}} \le \frac{74 e^2 (r^*)^{1.5}\mu^{1.5m-1}\dmax}{\underline{r}\lambda_{\min}\sqrt{d^*}}\frac{\norm{\cI}_{\ell_1}}{2s_0}\le \frac{1}{\varepsilon },
\end{equation*}
i.e.,
\begin{equation}\label{eq:minimax-eps-3}
    \varepsilon\le c \frac{\sqrt{d^*}s_0}{\dmax\norm{\cI}_{\ell_1}} \frac{\underline{r}\lambda_{\min}}{(r^*)^{1.5}\mu^{1.5m-1}},
\end{equation}
we have 
\begin{equation*}
      \abs{\left\langle \widehat{\cM},\cI \right\rangle-\left\langle\cM ,\cI \right\rangle} \ge \frac{1}{2}\varepsilon \norm{\cP_{\TT(\cM)}(\cI)}_{\tF} \ge \varepsilon s_0.
\end{equation*}
Combining \eqref{eq:minimax-eps-1}, \eqref{eq:minimax-eps-2}, \eqref{eq:minimax-eps-3}, and $\varepsilon=\frac{\sigma\sqrt{d^*/n}}{8}$, we know that when 
\begin{equation*}
     \frac{\lambda_{\min}}{\sigma} \ge C_{\gap} \kappa_0\left(\frac{\norm{\cI}_{\ell_1}}{ \norm{\cP_{\TT}(\cI)}_{\tF}\sqrt{d^*/\dmax} }\bigvee 1\right) \sqrt{\frac{ m^5 (2\mu)^{3 m} (r^*)^{3 }\dmax d^* }{\underline{r}^2  n } },
\end{equation*}
the minimax lower bound is given by
\begin{equation*}
    \abs{\left\langle \cM_1-\cM_2,\cI\right\rangle} \left(1-2\alpha-\sqrt{2\operatorname{KL}(\pi(\cM_1),\pi(\cM_2))} \right) \ge  \frac{1}{4}\frac{\sigma\sqrt{d^*/n}}{8} s_0 = \frac{\sigma \norm{\cP_{\TT}(\cI)}_{\tF}}{32}\sqrt{\frac{d^*}{n}}.
\end{equation*}
This finishes the proof.
\end{proof}

\section{Proofs of Auxiliary Results}
\subsection{Verification of \eqref{eq:W-kth-moment}}\label{sec:proof-w-k-mom}
It has been shown in the proof of Theorem \ref{thm:asymp-normal} that the test statistics $W_T$ can be decomposed as 
\begin{equation*}
	W_{T} = \frac{\left\langle  \widehat{Z}_1 , \cP_M(T) \right\rangle }{\sigma_{\xi} \norm{\cP_M(T) }_\tF  \sqrt{d_1 d_2/n}}+ \Delta_{T},
\end{equation*}
where $\Delta_T$ is a vanishing term with the rate of convergence described in \eqref{eq:2-d-remainder}. Suppose also the distribution of $\xi$ is symmetric. Denote $I_1$ the index set of observations in sample $\calD_1$. 
Therefore, for any integer $k\ge 2$, we have
\begin{equation*}
	\begin{aligned}
		&\E \abs{W_{T} }^{2k} \gtrsim \E \abs{\frac{\left\langle  \widehat{Z}_1 , \cP_M(T) \right\rangle }{\sigma_{\xi} \norm{\cP_M(T) }_\tF  \sqrt{d_1 d_2/n}}}^{2k} = \E \abs{\frac{\sqrt{d_1 d_2/n} \sum_{i\in I_1} \xi_i \left\langle  X_i , \cP_M(T) \right\rangle }{\sigma_{\xi} \norm{\cP_M(T) }_\tF }}^{2k} \\
		& \ge \left(\E \abs{\frac{\sqrt{d_1 d_2/n} \sum_{i\in I_1} \xi_i \left\langle  X_i , \cP_M(T) \right\rangle }{\sigma_{\xi} \norm{\cP_M(T) }_\tF }}^{4}  \right)^{k/2} \\
		& \ge \left(  \frac{d_1^2 d_2^2 \left( \sum_{i\in I_1} \E\xi_i^4 \left\langle  X_i , \cP_M(T) \right\rangle^4 + \sum_{i,j\in I_1, i\neq j}\E\xi_i^2 \left\langle  X_i , \cP_M(T) \right\rangle^2\xi_j^2 \left\langle  X_j , \cP_M(T) \right\rangle^2 \right) }{n^2\sigma_{\xi}^4 \norm{\cP_M(T) }^4_\tF }  \right)^{k/2} \\
		& \gtrsim \left(  \frac{d_1^2 d_2^2 \E\left\langle  X_i , \cP_M(T) \right\rangle^4 }{n \norm{\cP_M(T) }^4_\tF } +1 \right)^{k/2} \\
		& = \left(  \frac{d_1 d_2 \sum_{i\in [d_1],j\in[d_2]} \cP_M(T)_{i,j}^4 }{n \norm{\cP_M(T) }^4_\tF } +1 \right)^{k/2}. 
	\end{aligned}
\end{equation*}
If the energy of $\cP_M(T)$ is concentrated in a few entries, e.g., there exists an index set $J$ such that the entries in $J$ can dominate other entries, i.e., 
\begin{equation*}
	\sum_{(i,j)\in J} \cP_M(T)_{i,j}^2 \ge  \sum_{(i,j)\notin J} \cP_M(T)_{i,j}^2,
\end{equation*}
with  $s_0:=\abs{J}=O(1)$, then we have 
\begin{equation*}
	\begin{aligned}
		\frac{\sum_{i\in [d_1],j\in[d_2]} \cP_M(T)_{i,j}^4}{\norm{\cP_M(T) }^4_\tF } \ge \frac{\sum_{(i,j)\in J} \cP_M(T)_{i,j}^4}{  4 \left( \sum_{(i,j)\in J} \cP_M(T)_{i,j}^2 \right)^2 } \ge \frac{1}{4 s_0} \ge \Omega(1),
	\end{aligned}
\end{equation*}
and thus, we have
\begin{equation*}
	\begin{aligned}
		&\sqrt[2k]{\E \abs{W_{T} }^{2k}} \gtrsim \left( \frac{d_1 d_2}{n}\right)^{1/4}. 
	\end{aligned}
\end{equation*}

\subsection{Proof of Theorem \ref{thm:power-comparison}}
\begin{proof}
	Define the c.d.f. of the product of two standard normal random variables as $\Psi(t)$, also $\Tilde \Psi(t) :=1-\Psi(t)$. The c.d.f. of standard normal distribution is denoted by $\Phi(t)$ by convention, with $\Tilde \Phi(t) :=1-\Phi(t)$. For $j=1$, we have 
	\begin{equation*}
		\begin{aligned}
			\PP(Y_1>t|H_0) &= \Psi(-t) = \Tilde{\Psi}(t)\\
			F_1(z,t):= \PP(Y_1>t|z,H_1) &=  \PP\left( \frac{(\xi_1+\xi_2+2\delta )^2}{4} -\frac{(\xi_1-\xi_2)^2}{4} >t \right)= \PP \left( \frac{(Z_1+\sqrt{2}\delta )^2}{2}-\frac{Z_2^2}{2} >t \right) \\
			& = \int_{\R} \left[ \Tilde{\Phi}(\sqrt{2t + y_2^2}-\sqrt{\frac{2}{p}} z ) + \Tilde{\Phi}(\sqrt{2t + y_2^2}+\sqrt{\frac{2}{p}}z ) \right]d y_2.
		\end{aligned}
	\end{equation*}
	Here $\xi_1$, $\xi_2$, and $Z_1$, $Z_2$ are all standard normal random variables. Thus $ L_{p 1} = \Tilde{\Psi}^{-1}(p)$. Calculate the first order and second order derivative of $F_1(z,t)$ with respect to $z$ when $t= L_{p 1}$ :
	\begin{equation*}
		\begin{aligned}
			\partial_z   F_1(0,L_{p 1}) & = 0 \\
			\partial_z^2   F_1(0,L_{p 1}) & = \frac{8}{q} \int_{0}^{+\infty} -  f'(\sqrt{2L_{p 1} + y_2^2} )  d y_2= \frac{8}{p} f(-\sqrt{2L_{p 1}}).
		\end{aligned}
	\end{equation*}
	Since $\Tilde{\Psi}(t)<\sqrt{2}\Tilde{\Phi} (\sqrt{2t})$, we have $L_{p 1}< \frac{1}{2} \Tilde{\Phi}^{-1}(p/\sqrt{2})^2 $. When $x\to 0$, we have 
	\begin{equation*}
		\sqrt{2(\log(\frac{1-r_1}{x}) -\frac{1}{2}\log\log (\frac{1-r_1}{x})  ) } \le\Tilde{\Phi}^{-1}(x)\le \sqrt{2(\log(\frac{1}{x}) -\frac{1}{2+r_2}\log\log (\frac{1}{x})  ) }.
	\end{equation*}
	for any small $r_1,r_2>0$. Thus we have $L_{p 1}< \frac{1}{2} \Tilde{\Phi}^{-1}(p/\sqrt{2})^2 \le \log(\frac{\sqrt{2}}{p}) - \frac{1}{2+r_2}\log\log (\frac{\sqrt{2}}{p})$, and the second order derivative 
	\begin{equation*}
		\begin{aligned}
			\partial_z^2   F_1(0,L_{p 1}) & = \frac{8}{p} f(-\sqrt{2L_{p 1}}) \ge c (\log(\frac{\sqrt{2}}{p}) )^{1/(2+r_2)},
		\end{aligned}
	\end{equation*}
	is non-vanishing.
	
	For $j=2$,  we have
	\begin{equation*}
		\begin{aligned}
			\PP(Y_2>t|H_0) &= \PP(\xi_1>t,\xi_2>t )+\PP(\xi_1<-t,\xi_2<-t ) = 2 \Tilde{\Phi}^2(t) \\
			F_2(z,t):= \PP(Y_2>t|z,H_1) &=  \PP(\xi_1 +\mu>t,\xi_2+\mu >t ) +\PP(\xi_1 +\mu<-t,\xi_2+\mu <-t )\\ &=\Tilde{\Phi}^2(t+\sqrt{\frac{1}{p}}z ) +\Tilde{\Phi}^2(t-\sqrt{\frac{1}{p}}z ).
		\end{aligned}
	\end{equation*}
	In this case, the threshold $L_{p 2}= \Tilde{\Phi}^{-1}(\sqrt{\frac{p}{2}} ) $. Compute the derivatives of $F_2$:
	\begin{equation*}
		\begin{aligned}
			\partial_z   F_2(0,L_{p 2}) & = 0 \\
			\partial_z^2   F_2(0,L_{p 2}) & = \frac{4 }{p} (f^2(-L_{p 2}) +\Tilde{\Phi}(L_{p 2})f'(-L_{p 2})  )\ge c ((\log(\sqrt{\frac{2}{p}}) )^{1/(2+r_2) }  +1 ),
		\end{aligned}
	\end{equation*}
	which also has a non-vanishing second-order derivative.
	
	For $j=3$, we have
	\begin{equation*}
		\begin{aligned}
			\PP(Y_3>t|H_0) &= \PP(\xi_1+\xi_2>t, \xi_1>0, \xi_2>0 )+\PP(\xi_1+\xi_2<-t, \xi_1<0, \xi_2<0 ) \\
			&= 2 \PP(Y_1 >\frac{t}{\sqrt{2}}, -Y_1<Y_2<Y_1)=  2\Tilde{\Phi}(\frac{t}{\sqrt{2}})(1-\Tilde{\Phi}(\frac{t}{\sqrt{2}}))  \\
			\PP(Y_3>t|z,H_1) &=  \PP(Z_1 >\frac{t-2\mu}{\sqrt{2}}, -Z_1-\sqrt{2}\mu<Z_2<Z_1+\sqrt{2}\mu) \\
			& + \PP(Z_1 <\frac{-t-2\mu}{\sqrt{2}}, Z_1+\sqrt{2}\mu<Z_2<-Z_1-\sqrt{2}\mu)\\ 
			&\le \Tilde{\Phi}(\frac{t}{\sqrt{2}} )\left(\phi(\frac{t}{\sqrt{2}}+\sqrt{2}\mu)+\phi(\frac{t}{\sqrt{2}}-\sqrt{2}\mu) \right)  \\
			F_3(z,t):= & \Tilde{\Phi}(\frac{t}{\sqrt{2}} )\left(\phi(\frac{t}{\sqrt{2}}+\sqrt{\frac{2}{p}}z)+\phi(\frac{t}{\sqrt{2}}-\sqrt{\frac{2}{p}}z) \right).
		\end{aligned}
	\end{equation*}
	Compute the derivatives of $F_3$:
	
	\begin{equation*}
		\begin{aligned}
			\partial_z   F_3(0,L_{p 3}) & = 0 \\
			\partial_z^2   F_3(0,L_{p 3}) & = \frac{4 }{p} \Tilde{\Phi}(\frac{L_{p 3}}{\sqrt{2}} )f'(\frac{L_{p 3}}{\sqrt{2}})\le 0.
		\end{aligned}
	\end{equation*}
	If $\delta_0=o(\sqrt{\frac{1}{\pi} }) $, we have $z=\sqrt{p}\delta\to 0$. By Taylor's theorem, we have
	$$\operatorname{Power}_{W_j}(L_{p j} ) = p + \E_{\mathbf{\Theta}} \partial_z   F_j(0,L_{p j})z +  \E_{\mathbf{\Theta}} \frac{1}{2}\partial_z^2   F_j(0,L_{p j})z^2 +o( \E_{\mathbf{\Theta}} z^2),$$
	(or $\le$ for $j=3$ ). Plugging in the derivatives of $j=1,2,3$, clearly we have $\operatorname{Power}_{W_1}(L_{p 1} )\ge \operatorname{Power}_{W_3}(L_{p 3} )$, and $\operatorname{Power}_{W_2}(L_{p 2} )\ge \operatorname{Power}_{W_3}(L_{p 3} )$; for the second order derivative of $F_1$ and $F_2$, we also have 
	\begin{equation*}
		\begin{aligned}
			\partial_z^2   F_1(0,L_{p 1})- \partial_z^2   F_2(0,L_{p 2}) 
			& = \frac{4 }{p} (2f(-\sqrt{2L_{p 1}}) -f^2(-L_{p 2}) -\Tilde{\Phi}(L_{p 2})f'(-L_{p 2})  ) \\
			&\ge c \frac{1}{p}\exp\left(- \frac{1}{2} \Tilde{\Phi}^{-1}(p/\sqrt{2})^2 \right)\left(1 - \exp\left(\frac{1}{2} \Tilde{\Phi}^{-1}(p/\sqrt{2})^2  - \Tilde{\Phi}^{-1}(\sqrt{\frac{p}{2} })^2 \right) \right).
		\end{aligned}
	\end{equation*}
	
	Since
	\begin{equation*}
		\begin{aligned}
			&\frac{1}{2} \Tilde{\Phi}^{-1}(p/\sqrt{2})^2  - \Tilde{\Phi}^{-1}(\sqrt{\frac{p}{2} })^2= (\frac{1}{\sqrt{2}} \Tilde{\Phi}^{-1}(p/\sqrt{2}) + \Tilde{\Phi}^{-1}(\sqrt{\frac{p}{2} }) )(\frac{1}{\sqrt{2}} \Tilde{\Phi}^{-1}(p/\sqrt{2}) - \Tilde{\Phi}^{-1}(\sqrt{\frac{p}{2} }) ) \\
			&\le (\frac{1}{\sqrt{2}} \Tilde{\Phi}^{-1}(p/\sqrt{2}) + \Tilde{\Phi}^{-1}(\sqrt{\frac{p}{2} }) ) \\
			& \cdot ( \sqrt{\log(\frac{\sqrt{2}}{p})  - \frac{1}{1+r_2}\log \log(\frac{\sqrt{2}}{p}) } - \sqrt{\log (\frac{2(1-r_1)^2}{p} ) -\log\log( (1-r_1)\sqrt{\frac{2}{p}} )  } ) \\
			&\to -\infty,
		\end{aligned}
	\end{equation*}
	we have $\partial_z^2   F_1(0,L_{p 1})- \partial_z^2   F_2(0,L_{p 2}) \ge 0$, thus $\operatorname{Power}_{W_1}(L_{p 1} )\ge \operatorname{Power}_{W_2}(L_{p 2} )$. Translating the $\operatorname{Power}_{W_j}(L_{p j} )$ to $\operatorname{Power}_{W_j}(L_{\alpha j} )$, we finish our proof. 
\end{proof}

\subsection{Proof of Lemma \ref{lemma:cov-est-prec}}
\begin{proof}
For simplicity, we omit $C_{\init}$ in the proof. Notice that both $\left(I_{d_1 d_2} - \widehat U_\perp \widehat U_\perp^\top \otimes \widehat V_\perp \widehat V_\perp^\top\right)$ and $\left(I_{d_1 d_2} -  U_\perp U_\perp^\top \otimes V_\perp V_\perp^\top\right)$ are projection matrices with $P=P^2$. We thus have
\begin{equation}\label{eq:decomp-detS}
   \begin{aligned}
     \widehat{\Sigma}-\Sigma&=T_{\calH}\left((I_{d_1 d_2} - \widehat U_\perp \widehat U_\perp^\top \otimes \widehat V_\perp \widehat V_\perp^\top  )-(I_{d_1 d_2} -  U_\perp U_\perp^\top \otimes V_\perp V_\perp^\top ) \right) T_{\calH}^\top \\
     & = T_{\calH}\left(\widehat U_\perp \widehat U_\perp^\top \otimes \widehat V_\perp \widehat V_\perp^\top - U_\perp U_\perp^\top \otimes V_\perp V_\perp^\top\right) \left(I-U_\perp U_\perp^\top \otimes V_\perp V_\perp^\top\right) T_{\calH}^\top\\
     & + T_{\calH}\left(I-U_\perp U_\perp^\top \otimes V_\perp V_\perp^\top\right) \left(\widehat U_\perp \widehat U_\perp^\top \otimes \widehat V_\perp \widehat V_\perp^\top - U_\perp U_\perp^\top \otimes V_\perp V_\perp^\top \right) T_{\calH}^\top \\
     & + T_{\calH}\left(\widehat U_\perp \widehat U_\perp^\top \otimes \widehat V_\perp \widehat V_\perp^\top - U_\perp U_\perp^\top \otimes V_\perp V_\perp^\top\right) \left(\widehat U_\perp \widehat U_\perp^\top \otimes \widehat V_\perp \widehat V_\perp^\top - U_\perp U_\perp^\top \otimes V_\perp V_\perp^\top\right) T_{\calH}^\top.\\
\end{aligned} 
\end{equation}
We apply \eqref{eq:decomp-detS} to the error $\Sigma^{-\frac{1}{2}}(\widehat{\Sigma}-\Sigma)\Sigma^{-\frac{1}{2}}$:
\begin{equation*}
    \begin{aligned}
      &\norm{\Sigma^{-\frac{1}{2}}(\widehat{\Sigma}-\Sigma)\Sigma^{-\frac{1}{2}} }\le 2\norm{\Sigma^{-\frac{1}{2}}T_{\calH}\left(\widehat U_\perp \widehat U_\perp^\top \otimes \widehat V_\perp \widehat V_\perp^\top - U_\perp U_\perp^\top \otimes V_\perp V_\perp^\top\right) \left(I-U_\perp U_\perp^\top \otimes V_\perp V_\perp^\top\right) T_{\calH}^\top \Sigma^{-\frac{1}{2}} }\\
      &+  \norm{\Sigma^{-\frac{1}{2}}T_{\calH}\left(\widehat U_\perp \widehat U_\perp^\top \otimes \widehat V_\perp \widehat V_\perp^\top - U_\perp U_\perp^\top \otimes V_\perp V_\perp^\top\right) \left(\widehat U_\perp \widehat U_\perp^\top \otimes \widehat V_\perp \widehat V_\perp^\top - U_\perp U_\perp^\top \otimes V_\perp V_\perp^\top\right) T_{\calH}^\top \Sigma^{-\frac{1}{2}}}.
    \end{aligned}
\end{equation*}
Notice that $\norm{\left(I-U_\perp U_\perp^\top \otimes V_\perp V_\perp^\top\right) T_{\calH}^\top \Sigma^{-\frac{1}{2}}}\le 1$. We only need to focus on the term
\begin{equation*}
    \begin{aligned}
      &\norm{\Sigma^{-\frac{1}{2}}T_{\calH}\left(\widehat U_\perp \widehat U_\perp^\top \otimes \widehat V_\perp \widehat V_\perp^\top - U_\perp U_\perp^\top \otimes V_\perp V_\perp^\top\right)}  \\
      & \le  \norm{\Sigma^{-\frac{1}{2}} T_{\calH} }\norm{\left(\widehat U_\perp \widehat U_\perp^\top \otimes \widehat V_\perp \widehat V_\perp^\top - U_\perp U_\perp^\top \otimes V_\perp V_\perp^\top\right)} \\
      & \le \sqrt{\kappa_1}\kappa_T\left( \norm{ \left(\widehat U_\perp \widehat U_\perp^\top -U_\perp U_\perp^\top \right)\otimes V_\perp V_\perp^\top} +\norm{U_\perp U_\perp^\top \otimes \left( \widehat V_\perp \widehat V_\perp^\top-V_\perp V_\perp^\top\right) }\right. \\
      & + \left.\norm{ \left(\widehat U_\perp \widehat U_\perp^\top -U_\perp U_\perp^\top \right)\otimes \left( \widehat V_\perp \widehat V_\perp^\top-V_\perp V_\perp^\top\right)}\right)\\
      & \le C_2 \sqrt{\kappa_1}\kappa_T \frac{\sqrt{\tau}\left(1+\gamma_{n}\right) \sigma_{\xi}}{\lambda_{\min}} \cdot \sqrt{\frac{d_1^2 d_2 \log d_1}{n}},  
    \end{aligned}
\end{equation*}
where we use the definition of $\kappa_T$ and the perturbation of singular subspaces in \cite{xia2021statistical}. Moreover, when $T_{\calH}$ is sparse, we use $e_{T, k}\in \R^{d_1\times d_2} $, $k\in[\operatorname{supp}(T_{\calH} ) ]$ to indicate the collective supports of  all the $\operatorname{vec}(T_i)$. We then have
\begin{equation}\label{eq:sparse-T}
    \begin{aligned}
      &\norm{\Sigma^{-\frac{1}{2}}T_{\calH}\left(\widehat U_\perp \widehat U_\perp^\top \otimes \widehat V_\perp \widehat V_\perp^\top - U_\perp U_\perp^\top \otimes V_\perp V_\perp^\top\right)}  \\
          & = \norm{\Sigma^{-\frac{1}{2}}T_{\calH} \sum_{k=1}^{ \operatorname{supp}(T_{\calH} )  } e_{T, k}e_{T,k}^\top  \left(\widehat U_\perp \widehat U_\perp^\top \otimes \widehat V_\perp \widehat V_\perp^\top - U_\perp U_\perp^\top \otimes V_\perp V_\perp^\top\right)} \\
          & \le \sqrt{\kappa_1}\kappa_T \operatorname{supp}(T_{\calH} ) \max_{k}\norm{e_{T,k}^\top  \left(\widehat U_\perp \widehat U_\perp^\top \otimes \widehat V_\perp \widehat V_\perp^\top - U_\perp U_\perp^\top \otimes V_\perp V_\perp^\top\right)} \\
          & \le \sqrt{\kappa_1}\kappa_T\operatorname{supp}(T_{\calH} )  \max_{k} \left( \norm{e_{T,k}^\top  \left(\widehat U_\perp \widehat U_\perp^\top -U_\perp U_\perp^\top \right)\otimes V_\perp V_\perp^\top}\right. \\
          &  +\norm{e_{T,k}^\top  U_\perp U_\perp^\top \otimes \left( \widehat V_\perp \widehat V_\perp^\top-V_\perp V_\perp^\top\right) }+ \left.\norm{ e_{T,k}^\top \left(\widehat U_\perp \widehat U_\perp^\top -U_\perp U_\perp^\top \right)\otimes \left( \widehat V_\perp \widehat V_\perp^\top-V_\perp V_\perp^\top\right)}\right).\\
    \end{aligned}
\end{equation}
Since each $e_{T, k}$ can also be represented as $e_{T, k} = e_{T, k}^1 \otimes  e_{T, k}^2$, where $e_{T, k}^1\in\R^{d_1}$ and  $e_{T, k}^2\in\R^{d_2}$ are also canonical bases, we then have 
\begin{equation*}
    \begin{aligned}
      &\norm{\Sigma^{-\frac{1}{2}}T_{\calH}\left(\widehat U_\perp \widehat U_\perp^\top \otimes \widehat V_\perp \widehat V_\perp^\top - U_\perp U_\perp^\top \otimes V_\perp V_\perp^\top\right)} \\
          &  \lesssim \sqrt{\kappa_1}\kappa_T\operatorname{supp}(T_{\calH} )  \left(\norm{UU^\top-\widehat{U}\widehat{U}^\top }_{2,\max} + \norm{VV^\top-\widehat{V}\widehat{V}^\top }_{2,\max} \right)  \\
          & \le C \sqrt{\kappa_1}\kappa_T\frac{\operatorname{supp}(T_{\calH} )}{\sqrt{d_2}} \frac{\sqrt{\tau}\left(1+\gamma_{n}\right) \sigma_{\xi}}{\lambda_{\min}} \cdot \sqrt{\frac{d_1^2 d_2 \log d_1}{n}},
    \end{aligned}
\end{equation*}
because the higher-order error can be dominated.  The rate $\gamma_{n}$ converges to 0, which means that the whole error can be controlled by:
\begin{equation*}
   \norm{ \Sigma^{-\frac{1}{2}}(\widehat{\Sigma}-\Sigma)\Sigma^{-\frac{1}{2}}}\le C\frac{\kappa_T \sigma_\xi }{\lambda_{\min} } \cdot \left(\frac{\operatorname{supp}(T_{\calH} )}{\sqrt{d_2}}\wedge 1\right) \sqrt{\frac{ \kappa_1  d_1^2 d_2 \log d_1 }{n}}.
\end{equation*}
\end{proof}

\subsection{Proof of Lemma \ref{lemma:pres-DSD} }
\begin{proof}
For simplicity, we omit $C_{\init}$ in the proof. Denote $E=\Sigma-\widehat{\Sigma}$. By Fréchet derivative, as long as $\norm{E}=\norm{\Sigma-\widehat{\Sigma}}$ is small for any operator norm,  $\widehat{\Sigma}^{-1}-{\Sigma}^{-1}$ can be dominated by its Fréchet derivative ${\Sigma}^{-1}E{\Sigma}^{-1}$. Therefore, We have
\begin{equation*}
            \norm{{D} (\widehat{\Sigma}^{-1}-{\Sigma}^{-1}) {D}}_{\infty}\le \norm{{D}{\Sigma}^{-1}E{\Sigma}^{-1}{D}}_\infty+o(\norm{E}_\infty).
\end{equation*}
We only need to study the convergence rate of $\norm{{D}{\Sigma}^{-1}E{\Sigma}^{-1}{D}}_\infty$ as $E$ is small. This term, however, can be decomposed following \eqref{eq:decomp-detS}, i.e.,
\begin{equation*}
    \begin{aligned}
     & \norm{{D}{\Sigma}^{-1}E{\Sigma}^{-1}{D}}_\infty \\
       & \le  \norm{{D}{\Sigma}^{-1}D}_\infty \norm{D^{-1}
       T_{\calH}\left(\widehat U_\perp \widehat U_\perp^\top \otimes \widehat V_\perp \widehat V_\perp^\top - U_\perp U_\perp^\top \otimes V_\perp V_\perp^\top\right) \left(I-U_\perp U_\perp^\top \otimes V_\perp V_\perp^\top\right) T_{\calH}^\top {\Sigma}^{-1}{D}
       }_{\infty}\\
     & + 
     \norm{{D}{\Sigma}^{-1} T_{\calH}\left(I-U_\perp U_\perp^\top \otimes V_\perp V_\perp^\top\right)\left(\widehat U_\perp \widehat U_\perp^\top \otimes \widehat V_\perp \widehat V_\perp^\top - U_\perp U_\perp^\top \otimes V_\perp V_\perp^\top \right) T_{\calH}^\top {\Sigma}^{-1}{D}
     }_\infty
     \\
     & + \norm{{D}{\Sigma}^{-1}T_{\calH}\left(\widehat U_\perp \widehat U_\perp^\top \otimes \widehat V_\perp \widehat V_\perp^\top - U_\perp U_\perp^\top \otimes V_\perp V_\perp^\top\right) \left(\widehat U_\perp \widehat U_\perp^\top \otimes \widehat V_\perp \widehat V_\perp^\top - U_\perp U_\perp^\top \otimes V_\perp V_\perp^\top\right) T_{\calH}^\top{\Sigma}^{-1}{D}}_\infty \\
     & \le \kappa_\infty \sqrt{q} \norm{ D^{-1}T_{\calH}\left(\widehat U_\perp \widehat U_\perp^\top \otimes \widehat V_\perp \widehat V_\perp^\top - U_\perp U_\perp^\top \otimes V_\perp V_\perp^\top\right) }_{2,\max} \sqrt{\kappa_1} \\
     & + \sqrt{\kappa_1}\sqrt{q}\norm{ \left(\widehat U_\perp \widehat U_\perp^\top \otimes \widehat V_\perp \widehat V_\perp^\top - U_\perp U_\perp^\top \otimes V_\perp V_\perp^\top \right) T_{\calH}^\top {\Sigma}^{-1}{D}
     }\\
     & + q \kappa_1^2 \norm{ D^{-1}T_{\calH}\left(\widehat U_\perp \widehat U_\perp^\top \otimes \widehat V_\perp \widehat V_\perp^\top - U_\perp U_\perp^\top \otimes V_\perp V_\perp^\top\right) }_{2,\max}^2 \\
     & \le C\left(\kappa_\infty\sqrt{\kappa_1} +\norm{T_{\calH}}_2\kappa_1/\sqrt{\lambda_{\min}(\Sigma)} \right)\frac{\rho_T \mu \sigma_\xi }{\beta_0 \lambda_{\min} }\sqrt{\frac{\alpha_d q d_1^2 d_2 \log d_1 }{n}} \\
     & \le C\left(\kappa_\infty\sqrt{\kappa_1} +\kappa_1^{1.5}\kappa_T \right)\frac{\rho_T \mu \sigma_\xi }{\beta_0 \lambda_{\min} }\sqrt{\frac{\alpha_d q d_1^2 d_2 \log d_1 }{n}}, \\
    \end{aligned} 
\end{equation*}
where the 2-max norm here can be bounded by:
\begin{equation}
    \begin{aligned}
      &\norm{ D^{-1}T_{\calH}\left(\widehat U_\perp \widehat U_\perp^\top \otimes \widehat V_\perp \widehat V_\perp^\top - U_\perp U_\perp^\top \otimes V_\perp V_\perp^\top\right) }_{2,\max} \\
      &\le \max_{T_i\in \cH }\frac{ \sum_{j\in[d_1] } \sum_{k\in [d_2] } \abs{T_i(j,k) \left[\left(U U^\top-\widehat U \widehat U^\top\right)\otimes V_{\perp}V_\perp^\top +  U_\perp U_\perp^\top \otimes \left(\widehat V \widehat V^\top -V V^\top \right) \right]\cdot e_j \otimes e_k} }{s_{T_i}  }\\
      & \quad + \max_{T_i\in \cH }\frac{ \sum_{j\in[d_1] } \sum_{k\in [d_2] } \abs{T_i(j,k) \left(U U^\top-\widehat U \widehat U^\top\right)\otimes \left(\widehat V \widehat V^\top -V V^\top \right) \cdot e_j \otimes e_k} }{s_{T_i}  }\\
      & \le C\max_{T_i\in \cH } \frac{\norm{T}_{\ell_1 } }{ \norm{T}_\tF \beta_0 \sqrt{r/d_1} } \frac{ \mu\left(1+\gamma_{n}\right) \sigma_{\xi}}{\lambda_{\min}} \cdot \sqrt{\frac{rd_1^2 \log d_1}{n}}   \\
      & \le C\frac{\rho_T \mu \sigma_\xi }{\beta_0 \lambda_{\min} }\sqrt{\frac{\alpha_d d_1^2 d_2 \log d_1 }{n}}.
    \end{aligned}
\end{equation}
Here, we use the 2-max norm bound in \cite{xia2021statistical}, the alignment assumption, and the definition of $\kappa_T$. Moreover, the norm $\norm{ \left(\widehat U_\perp \widehat U_\perp^\top \otimes \widehat V_\perp \widehat V_\perp^\top - U_\perp U_\perp^\top \otimes V_\perp V_\perp^\top \right) T_{\calH}^\top {\Sigma}^{-1}{D}
     }$ can also bounded by 

\begin{equation*}
    \begin{aligned}
      &\norm{D\Sigma^{-1}T_{\calH}\left(\widehat U_\perp \widehat U_\perp^\top \otimes \widehat V_\perp \widehat V_\perp^\top - U_\perp U_\perp^\top \otimes V_\perp V_\perp^\top\right) \left(I-U_\perp U_\perp^\top \otimes V_\perp V_\perp^\top\right)}  \\
      & \le \norm{D\Sigma^{-\frac{1}{2} }\cdot \Sigma^{-\frac{1}{2} }T_{\calH}\left(\widehat U_\perp \widehat U_\perp^\top \otimes \widehat V_\perp \widehat V_\perp^\top - U_\perp U_\perp^\top \otimes V_\perp V_\perp^\top\right)}\\
          &  \lesssim {\kappa_1}\kappa_T\operatorname{supp}(T_{\calH} )  \left(\norm{UU^\top-\widehat{U}\widehat{U}^\top }_{2,\max} + \norm{VV^\top-\widehat{V}\widehat{V}^\top }_{2,\max} \right)  \\
          & \le C \kappa_1\kappa_T\frac{\operatorname{supp}(T_{\calH} )}{\sqrt{d_2}} \frac{\sqrt{\tau}\left(1+\gamma_{n}\right) \sigma_{\xi}}{\lambda_{\min}} \cdot \sqrt{\frac{d_1^2 d_2 \log d_1}{n}},
    \end{aligned}
\end{equation*}
where we use the sparsity of $T_{\calH}$ following \eqref{eq:sparse-T}. This gives the desired bound 
\begin{equation*}
    \norm{{D}{\Sigma}^{-1}E{\Sigma}^{-1}{D}}_\infty \le C\left(\kappa_\infty\sqrt{\kappa_1} +\kappa_1^{1.5}\kappa_T\left( \frac{\operatorname{supp}(T_{\calH} )}{\sqrt{d_2}}\wedge 1 \right)  \right)\frac{\rho_T \mu \sigma_\xi }{\beta_0 \lambda_{\min} }\sqrt{\frac{\alpha_d q d_1^2 d_2 \log d_1 }{n}}.
\end{equation*}

\end{proof}

\end{document}